\documentclass[12pt,a4paper,twoside,openright, titlepage]{book} 

\usepackage[utf8]{inputenc}
\usepackage[italian,english]{babel}

\usepackage[style=numeric-comp,sorting=none,url=false,doi=false,isbn=false]{biblatex} %numeric-comp
\usepackage{csquotes}
\renewbibmacro{in:}{}
\renewbibmacro*{publisher+location+date}{%
  \printlist{location}%
  \iflistundef{publisher}
    {\setunit*{\addsemicolon\space}}
    {\setunit*{\addcolon\space}}%
  \printlist{publisher}%
  \setunit*{\space}%
  \usebibmacro{date}%
  \newunit}
\DeclareFieldFormat{date}{\mkbibparens{#1}}
\DeclareFieldFormat[article,periodical]{date}{#1}
\usepackage{amsmath}
\usepackage{amsfonts}
\usepackage{amssymb}
\usepackage{amsthm}
\usepackage{dsfont}	 % for \mathds{}: similar to \mathbb{} but prints ok the unit with \mathds{1}
\usepackage{slashed} % use \slashed{} for Feynman-like notation
  \usepackage{faktor}  % for quotients: \faktor{}{}
\usepackage{tikz-cd} % for commutative diagrams
\usepackage{graphicx}
     \newcommand{\mapsfrom}{\mathrel{\reflectbox{\ensuremath{\mapsto}}}} % left version of \mapsto
\usepackage{mathtools}  % for commands like \xhookrightarrow

\usepackage{indentfirst}  % indent first line of sections
\usepackage[left=2cm,right=2cm,top=3cm,bottom=2.5cm,headheight=15pt]{geometry}   % for frames
\usepackage{enumitem}  % enumerate, itemize, description environments
\usepackage{caption}
\usepackage{floatrow,subfig}
\newsavebox{\tempbox}

\usepackage{emptypage} %per togliere headings da pagine bianche (twoside)
\usepackage{fancyhdr}
     \fancypagestyle{plain}{% 
        \fancyhf{}\fancyfoot[C]{\thepage} } %\small\bfseries
     \makeatletter 
     \renewcommand{\chaptermark}[1]{\markboth{% 
        \ifnum\value{secnumdepth}>\m@ne 
        \if@mainmatter\@chapapp\ \thechapter. \ \fi\fi #1}{}} 
     \renewcommand{\sectionmark}[1]{\markright{% 
        \ifnum\value{secnumdepth}>\z@ \thesection.\ \fi#1}} 
     \makeatother

\theoremstyle{definition}
\newtheorem{thm}{Theorem}[section]
\newtheorem{defn}{Definition}[section]
\newtheorem{prop}{Proposition}[section]
\newtheorem{lemma}{Lemma}[section]
\newtheorem{ex}{Example}[section]
\newtheorem*{rmk}{Remark}

\usepackage{pdfpages}

\begin{document}
\frontmatter
%%%%%%%%%%%%%%%%%%%%% FRONTESPIZIO: modificare in Tesi_RossiP-frn.tex %%%%%%%%%%%%%%%%%%%%%%%%%%%%%%%
%\begin{frontespizio}
%%\Preambolo{\renewcommand{\frontfixednamesfont}{\fontsize{11}{13}\rmfamily\itshape}}
%\Margini{2.5cm}{4.5cm}{2.5cm}{4cm}
%\NCandidato{Laureando}
%\Universita{Modena e Reggio Emilia}
%\Dipartimento {scienze fisiche, informatiche e matematiche}
%\Corso[Laurea Magistrale]{Physics}
%\Annoaccademico{2019-2020}
%%\Titoletto{Tesi di laurea}
%\Titolo{Equivariant and supersymmetric localization in QFT}
%\Candidato{Paolo Rossi}
%\Relatore{Prof. Diego Trancanelli}
%\end{frontespizio}
\includepdf{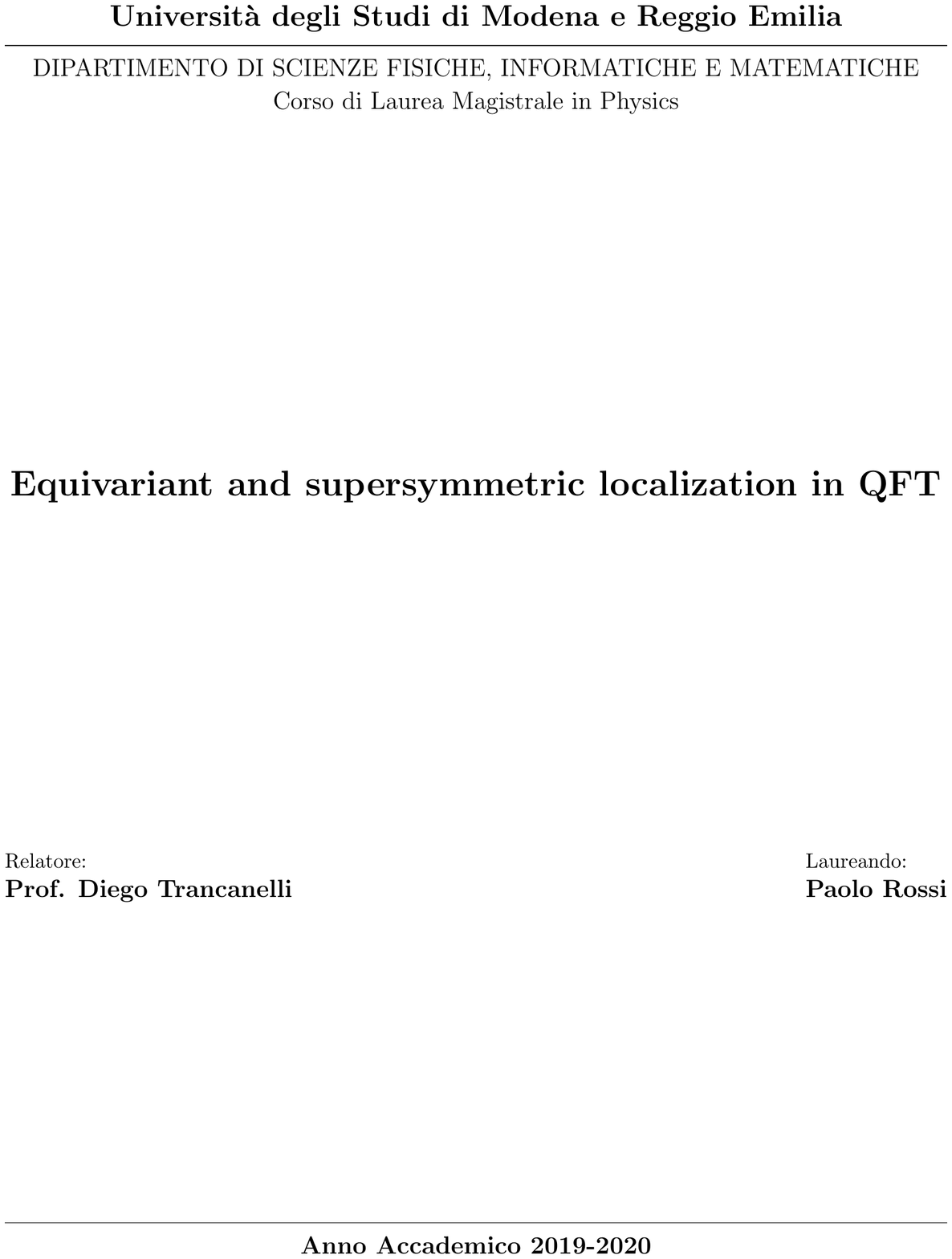}
%%%%%%%%%%%%%%%%%%%%%%%%%%%%%%%%%%%%%%%%%%%%%%%%%%%%%%%%%%%%%%%%%%%%%%%%%%%%%%%%%%%%%%%%%%%%%%%%%

%\let\cleardoublepage\clearpage %serve a non lasciare la pagina bianca in mezzo

%%%%%%%%%%%%%%%%%%%%%%%%%%%%%%%%%%%%%%%%%%%%%%%%% abstract
\cleardoublepage
%\phantomsection
\addcontentsline{toc}{chapter}{Abstract}
\markboth{Abstract}{Abstract}
\chapter*{Abstract}

Equivariant localization theory is a powerful tool that has been extensively used in the past thirty years to elegantly obtain exact integration formulas, in both mathematics and physics. These integration formulas are proved within the mathematical formalism of equivariant cohomology, a variant of standard cohomology theory that incorporates the presence of a symmetry group acting on the space at hand. A suitable infinite-dimensional generalization of this formalism is applicable
to a certain class of Quantum Field Theories (QFT) endowed with  \emph{supersymmetry}.

In this thesis we review the formalism of equivariant localization and some of its applications in Quantum Mechanics (QM) and QFT. We start from the mathematical description of equivariant cohomology and related localization theorems of finite-dimensional integrals in the case of an Abelian group action, and then we discuss their formal application to infinite-dimensional path integrals in QFT. We summarize some examples from the literature of computations of partition functions and expectation values of supersymmetric operators in various dimensions. For 1-dimensional QFT, that is QM, we review the application of the localization principle to the derivation of the Atiyah-Singer index theorem applied to the Dirac operator on a twisted spinor bundle. In 3 and 4 dimensions, we examine the computation of expectation values of certain Wilson loops in supersymmetric gauge theories and their relation to 0-dimensional theories described by \lq\lq matrix models\rq\rq. Finally, we review the formalism of non-Abelian localization applied to 2-dimensional Yang-Mills theory and its application in the mapping between the standard \lq\lq physical\rq\rq\ theory and a related \lq\lq cohomological\rq\rq\ formulation.

\cleardoublepage
%\phantomsection
\addcontentsline{toc}{chapter}{Acknowledgments}
\markboth{Acknowledgments}{Acknowledgments}
\chapter*{Acknowledgments}

First and foremost, I would like to express my sincere gratitude to professor Diego Trancanelli, who supervised me during the draft of this thesis. I thank him for his constant availability, his valuable advice and his sincere interest for my learning process, as well as his great patience in meticulously reviewing my work step by step. It is also a pleasure for me to thank professor Olindo Corradini, who initiated me to the wonderful world of QFT with two brilliant courses, and was always available to discuss and answer my questions.

I would like to thank all the friends and colleagues that grew up with me in this journey. Even though our paths and interests separated, they have always been a precious source of inspiration to me.
A heartfelt thanks goes to Giulia, whose presence alone makes everything easier. Finally I would like to thank my family, for the love and support, no matter what, during all these years.

%%%%%%%%%%%%%%%%%%%%%%%%%%%%%%%%%%%%%%%%%%%%%%%%% indice
\tableofcontents
\mainmatter

%%%%%%%%%%%%%%%%%%%%%%%%%%%%%%%%%%%%%%%%%%%%%%%%% introduzione
%\addcontentsline{toc}{chapter}{Introduction}
%\markboth{Introduction}{Introduction}
\chapter{Introduction}

Quantum Field Theory (QFT) is the framework in which modern theoretical physics describes fundamental interactions between elementary particles and it is also central in the study of condensed matter physics and statistical mechanics. QFT has made the most precise predictions ever in the history of science and it has been tested against a huge amount of experimental data. Nowadays, the most useful formulation of QFT is made in terms of  \emph{path integrals}, which are integrals over the space of all possible field configurations. A \lq\lq field\rq\rq\ mathematically speaking can be thought roughly as a function over  spacetime, so these integrals are computed over functional spaces, that are infinite-dimensional. This makes the exact computation of such objects a complicated task, except for some very special cases, and in fact their precise mathematical formulation is still an open problem. Despite these formal difficulties, many interesting results can be extracted from these objects, that can describe \emph{partition functions} and \emph{expectation values} of physical observables in QFT. The favorite approach to deal with such computations is perturbation theory, applied in the case in which the QFT is weakly coupled. In this regime, one can compute approximately the expectation values in a perturbative expansion, order by order in the coupling constant. This method, applied to the computation of the partition function, is the infinite-dimensional analogous of a \lq\lq saddle-point\rq\rq\ or \lq\lq stationary-phase\rq\rq\ approximation. Intuitively it represents a semi-classical approach to the quantum dynamics.

There are however many cases in which perturbation theory is not applicable, mainly when the QFT is strongly coupled, \textit{i.e.}\ the coupling constant is of order 1. This is not a very rare situation. For example we know that one of the fundamental interactions of the Standard Model of particle physics, the strong nuclear force, is well described by \lq\lq quantum chromodynamics\rq\rq\ (QCD), a QFT that is strongly coupled at low energies (so in the \lq\lq phenomenological\rq\rq\ regime).
%This is not a very rare situation, since we know that one of the fundamental interactions of the Standard Model of particle physics, the strong nuclear force, is described by a QFT, quantum chromodynamics (QCD), that is strongly coupled at low energies (so in the \lq\lq phenomenological\rq\rq\ regime).
Understanding the behavior of QFT in the strong coupling regime is then a major problem from the physical point of view, and one is lead to develop techniques that permit to study the path integral in a non-perturbative approach.

In this thesis we describe features of one of these techniques, that has been exploited in the last few decades for a class of special QFTs, those who exhibit some kind of \emph{supersymmetry}. This technique is called \emph{supersymmetric localization}, or \emph{equivariant localization}, or simply localization. Its name derives from the fact that, when one is able to use this method, the path integral of the QFT at hand simplifies (so \lq\lq localizes\rq\rq) to an integral over a smaller domain, sometimes even a finite-dimensional integral over constant field configurations. This result can be viewed as an exact stationary-phase approximation, so that the full quantum spectrum of the localized theory is completely determined by its semi-classical limit. Without entering in the technical details of this localization phenomenon, we just point out that this method fundamentally relies on the presence of a large amount of symmetry of the theory, that can be described by the presence of a \emph{group action} on the space of fields. When the space of fields is graded, \textit{i.e.}\ there is a distinction between \lq\lq bosonic\rq\rq\ and \lq\lq fermionic\rq\rq\ degrees of freedom, the symmetry group action can exchange these two types of fields and in this case it is called \emph{super}symmetry. This situation arises mainly in BRST-fixed and topological field theories, where the grading is regarded as the \emph{ghost number}, and in Poincaré-supersymmetric QFT, where the grading distinguishes between bosons and fermions in the standard sense of particle physics. In both cases, the action of a supersymmetry transformation \lq\lq squares\rq\rq\ to a canonical (bosonic) one, \textit{i.e.}\ a gauge transformation or a Poincaré transformation. It is reasonable that such a huge amount of symmetry can simplify the dynamics of the theory, but it can be non-trivial a priori how to translate this in a simplification of the path integral.

At this point, one wishes to understand if there is a theoretical framework that allows to systematically understand why and when such a drastic simplification of the path integral can occur, and at which level this is related to (super-)symmetries in QFT. To be mathematically more rigorous, we can think in terms of integration over finite-dimensional spaces, and then try to extrapolate and generalize the important results to the infinite-dimensional case. Integrals of differential forms over manifolds are built up technically from the smooth (so local) structure of the space, but it is a well-known fact that their result can describe and is regulated by topological (so global) properties. It is a consequence of de Rham's theorem and Stokes' theorem that they really depend on the \emph{cohomology class} of the integrand and not on the particular differential form that represents the class. It is thus reasonable that a theory of integration that embeds the presence of a symmetry group action should arise from a topological construction. Indeed, the mathematical framework in which the localization formulas
%analogous to the ones used in supersymmetric QFT 
were firstly derived is a suitable  modification of the standard cohomology theory. This is called \emph{equivariant cohomology}. As de Rham's theorem relates the usual cohomology to differential forms on a smooth manifold, equivariant cohomology can be associated to a modification of them, called \emph{equivariant differential forms}.

Equivariant cohomology theory was initiated in the mathematical literature by Cartan, Borel and others during the 50's \cite{cartan-1,cartan-2,borel-1,borel-2}, but the first instance of a localization formula was presented by Duistermaat and Heckman in 1982 \cite{duistermaat-heckman}. In this paper, they proved the exactness of a stationary-phase approximation  in the context of symplectic geometry and Hamiltonian Abelian group actions. Subsequently, Atiyah and Bott realized that the Duistermaat-Heckman localization formula can be viewed as a special case of a more general theorem, that they proved in the topological language of equivariant cohomology \cite{atiyah-bott-localization}. Almost at the same time, Berline and Vergne derived an analogous localization formula valid for Killing vectors on general compact Riemannian manifolds \cite{berline-vergne-loc}. Roughly, the Atiyah-Bott-Berline-Vergne (ABBV) formula says that the integral over a manifold that is acted upon by an Abelian group localizes as a sum of contributions arising only from the fixed points of the group action.

The first infinite-dimensional generalization of this localization formula was given soon after by Atiyah and Witten in 1985 \cite{Atiyah-circular-symmetry}, applied to supersymmetric Quantum Mechanics (QM). This turns out to be an example of topological theory, since their localization formula relates the partition function to the index of a Dirac operator. Many generalizations of this approach followed, first mainly in the context of topological field theories. Here localization allows to get closed formulas relating partition functions of physical QFTs to topological invariants of the spaces where they live. In all these cases, the BRST cohomology is interpreted as the equivariant structure of the theory, and is responsible for the localization of the path integral.

In recent years, the formal application of the ABBV localization formula for Abelian symmetry actions was employed in the context of Poincaré-supersymmetric theories, whose explicit supersymmetry is an extension of the spacetime symmetry that is common to all QFTs. In this case, the equivariant structure is generated by the cohomology of a \emph{supersymmetry charge}, that acts as a differential on the space of Poincaré symmetric field configurations. This formal structure of supersymmetric field theories was firstly realized by Niemi, Palo and Morozov in the 90's \cite{morozov-supersymplectic-QFT, palo}. Starting from the work of Pestun in 2007 \cite{pestun-article}, localization has been applied to the computation of partition functions and expectation values of supersymmetric operators on curved compact manifolds \cite{LocQFT}. In many of these cases, the partition function localizes to a finite-dimensional integral over matrices, a so-called \emph{matrix model}. To carry out this procedure, a number of technical difficulties have to be overcome, the most urgent being to understand how to define Poincaré-supersymmetric theory on curved spaces. Nowadays there is a well-defined and well-understood procedure that allows to do that, essentially deforming an original theory defined in flat space through the coupling to a non-trivial \lq\lq rigid\rq\rq\ supersymmetric background. This in general reduces the degree of supersymmetry of the original theory, but if some of it is preserved on the new background then one is able in principle to perform localization. Beside the general and abstract motivation of understanding strongly coupled QFT, the importance of these computations can be viewed in a string theory perspective, and in particular as possible tests of the so called \emph{AdS/CFT correspondence} \cite{Maldacena-ads_cft}.

Some generalizations of the ABBV formula for non-Abelian group actions have been proposed both in the mathematical and physical literature. The first generalization of the Dui\-ster\-maat-Heckman theorem to non-Abelian group actions was presented by Guillemin and Prato \cite{guillemin-prato}, restricting the localization principle to the action of the maximal Abelian subgroup. An infinite-dimensional generalization of the localization principle was proposed by Witten in 1992 \cite{witten-2dYM}, applied to the study of 2-dimensional Yang-Mills theory, and a more rigorous proof appeared in the mathematical literature in 1995 by Jeffrey and Kirwan \cite{Jeffrey-nonabLoc}. Of course these localization formulas find many interesting applications not only in physics but also in pure mathematics, but this side of the story is far from the purpose of this work.

%In recent years, localization was exploited in many examples of Poincaré-supersymmetric theories. Niemi and others, interpretation of susy QFT in equivariant cohomological terms. The recent applications involve the computations of quantities in susy QFT on curved spaces. This makes the path integral converge, the difficulty is then to find some preserved susy. In this spirit, one computes stuff that can be useful for AdS/CFT tests for examples and string theory... [CITE SOME GEOMETRY, THE APPLICATION FOR BLACK HOLES, WILSON LOOPS, MATRIX MODELS]

\section*{An exact saddle-point approximation}

To give a feeling of what we mean by \lq\lq exact saddle-point approximation\rq\rq, we present a very simple but instructive example here, in the finite-dimensional setting. The saddle-point (or stationary-phase) method  is applied to oscillatory integrals of the form
\begin{equation}
\label{intro-1}
I(t) = \int_{-\infty}^{+\infty} dx\ e^{itf(x)} g(x) ,
\end{equation}
when one is interested in the asymptotic behavior of \(I(t)\) at positive large values of the real parameter \(t\). In this limit, the integral is dominated by the critical points of \(f(x)\), where its first derivative vanishes and it can be expanded in Taylor series as
\begin{equation}
f(x) = f(x_0) + \frac{1}{2} f''(x_0) (x-x_0)^2 + \cdots .
\end{equation}
If \(F\subset \mathbb{R}\) is the set of critical points, that for simplicity we assume to be discrete, the leading contribution to \eqref{intro-1} is then given by a Gaussian integral,
\begin{equation}
\begin{aligned}
I(t) &\approx \sum_{x_0\in F}  g(x_0) e^{itf(x_0)} \int_{-\infty}^{+\infty} dx\ e^{\frac{it}{2} f''(x_0) (x-x_0)^2} \\
&= \sum_{x_0\in F}  g(x_0) e^{itf(x_0)} e^{\frac{i\pi}{4}\mathrm{sign}(f''(x_0))} \sqrt{\frac{2\pi}{t|f''(x_0)|}} .
\end{aligned}
\end{equation}
If the integral is performed over \(\mathbb{R}^n\), this last formula generalizes easily to
\begin{equation}
\label{intro-2}
I(t) \approx \left(\frac{2\pi}{t}\right)^{n/2} \sum_{x_0\in F}  g(x_0) e^{itf(x_0)} \frac{ e^{\frac{i\pi}{4}\sigma(x_0)} }{|\det(\mathrm{Hess}_f(x_0))|^{1/2}} ,
\end{equation}
where \(\mathrm{Hess}_f(x)\) is the matrix of second derivatives of \(f\) at \(x\in \mathbb{R}^n\), and \(\sigma(x)\) denotes the sum of the signs of its eigenvalues.

Of course, there is no reason for the RHS of \eqref{intro-2} to be the exact answer for \(I(t)\), but the claimed property of localization is that in some cases this turns out to be true! To see this, let us consider the integration over the 2-sphere \(\mathbb{S}^2\), defined by its embedding in \(\mathbb{R}^3\) as the set of points whose distance from the origin is 1. For this example we chose \(f(x,y,z) = z\), the \lq\lq height function\rq\rq, and \(g(x,y,z)=1\). The resulting oscillatory integral is then
\begin{equation}
\label{intro-3}
I(t) = \int_{\mathbb{S}^2} dA\ e^{itz} ,
\end{equation}
where \(dA\) is the volume form on the sphere, normalized such that \(\int_{\mathbb{S}^2} dA = 4\pi\). The critical points of the height function are the North and the South poles, where 
\begin{equation}
z \approx \pm \left( 1 -\frac{1}{2}(x^2+y^2) \right) .
\end{equation}
The volume form at the poles is just \(dA = dxdy\), so if we apply the saddle-point approximation to \eqref{intro-3} we get
\begin{equation}
\label{intro-height-function}
\begin{aligned}
\int_{\mathbb{S}^2} dA\ e^{itz} &\approx e^{it}\int dxdy\ e^{-\frac{it}{2}(x^2+y^2)} + e^{-it}\int dxdy\ e^{\frac{it}{2}(x^2+y^2)} \\
&= \frac{2\pi}{it} e^{it} - \frac{2\pi}{it} e^{-it} \\
&= 4\pi \frac{\sin(t)}{t} .
\end{aligned}
\end{equation}

Now, since this integral is rather easy, in this case we can actually compare the result of the saddle-point approximation with  its exact value. Using spherical coordinates,
\begin{equation}
I(t)= \int_{-1}^{+1} d\cos(\theta) \int_0^{2\pi} d\varphi\ e^{it\cos(\theta)} =  4\pi \frac{\sin(t)}{t} .
\end{equation}
As promised, the result coincides with the stationary-phase result \eqref{intro-height-function}. This is the simplest example of equivariant localization! The scope of the first chapters of the thesis is to describe in  general the structure underlying this result, how it can be related to symmetry properties of the specific function and space under consideration, and the localization theorems that generalize this specific computation to possibly more complicated examples. In the remaining part, we will deal instead with examples of the infinite-dimensional analog of this exact stationary-phase approximation.

\section*{Structure of the thesis}

The aim of this project is to summarize some results in the context of equivariant localization applied to physics. We will draw a line from the first mathematical results concerning the theory of equivariant cohomology and associated powerful localization theorems of finite-dimensional integrals, to the generalization to path integration in quantum mechanical models and finally to applications in quantum field theories. One of the purposes of this work is to provide a suitable reference  for other students  in theoretical and mathematical physics with a background at the level of a master degree, who are interested in approaching the subject. In this spirit, we will try to expose the material in a pedagogical order, being as self-contained as possible, and otherwise giving explicit references to background material.

\medskip
In Chapter 2, we will review the basics of equivariant cohomology theory, starting from its construction in algebraic topology. After having recalled some notions of basic homology and cohomology, we will introduce group actions, and define equivariant cohomology with the so-called \emph{Borel construction}. Then, we will describe the most common algebraic models that generalize de Rham's theorem in the equivariant setup, the \emph{Weil}, \emph{Cartan} and \emph{BRST models}. These give a description of equivariant cohomology in terms of a suitable modification of the complex of differential forms.

\medskip
In Chapter 3, we will describe the common rationale behind the localization property of equivariant integrals in finite-dimensional geometry, the so-called \emph{equivariant localization principle}. Then we will state and explain the Abelian localization formula derived by Atiyah-Bott and Berline-Vergne. In the final part of the chapter, we will connect the discussion with the context of symplectic geometry that, as we will see, can be rephrased in terms of equivariant cohomology. We will review the basic notions of symplectic manifolds, symmetries and Hamiltonian systems, and then state the Duistermaat-Heckman localization formula as a special case of the ABBV theorem.

\medskip
Chapter 4 can be viewed as a long technical aside. Here we will review some notions about supergeometry that are needed to understand a proof of the ABBV theorem, and then specialize the discussion to Poincaré-supersymmetric theories. We will discuss their construction from the perspective of superspace, give some practical examples, and then generalize their description over general curved backgrounds. This is achieved by coupling the given theory to the supersymmetric version of Einstein gravity, \emph{supergravity}, and then requiring the gravitational sector of the resulting theory to decouple from the rest in a \lq\lq rigid limit\rq\rq, analogous to \(G_N\to 0\). This method will bring up to the notion of \emph{Killing spinors}, a special type of spinorial fields whose existence ensures the preservation of some supersymmetry on the curved background. 
Finally, we will comment about a possible \lq\lq super-interpretation\rq\rq\ of the models of equivariant cohomology described in Chapter 2, that connects them to the usual BRST formalism for quantization of constrained Hamiltonian systems.

\medskip
In Chapter 5 we will discuss examples of Abelian supersymmetric localization of path integrals in the infinite-dimensional setting of QFT. The first case we will report is 1-dimensional, \textit{i.e.}\ QM. In an Hamiltonian formulation on phase space, we will describe how it is possible to give a supersymmetric (so equivariant) interpretation to the path integral in a model-independent way, the supersymmetry arising as a \lq\lq hidden\rq\rq\ BRST symmetry which is linked to the Hamiltonian dynamics. This results in a localization of the path integral over the space of classical field configurations or over constant field configurations, that applied to supersymmetric QM gives an alternative proof the Atiyah-Singer index theorem for the Dirac operator over a twisted spinor bundle. Next, we will review two modern applications of the localization principle to the computation of expectation values of \emph{Wilson loop operators} in supersymmetric gauge theories. The first application computes the expectation value in \(\mathcal{N}=4,2,2^*\) Super Yang-Mills theory on the 4-sphere \(\mathbb{S}^4\), the second one in \(\mathcal{N}=2\) Super Chern-Simons theory on the 3-sphere \(\mathbb{S}^3\). In both cases, the partition function and the Wilson loop expectation value can be reduced to finite-dimensional integrals over the Lie algebra of the gauge group of the theory. This makes the supersymmetric theories in exam equivalent to a suitable \lq\lq matrix model\rq\rq, whose path integral can be computed exactly with some special regularization. We will give an example of computations of such a matrix model, since this class of objects arises in many important areas of modern theoretical physics.

\medskip
In Chapter 6 we will introduce Witten's non-Abelian localization formula, and briefly describe its possible application in the study of 2-dimensional Yang-Mills theory. This more general formalism is able to show the mapping between the standard \lq\lq physical\rq\rq\ version of Yang-Mills theory and its \lq\lq cohomological\rq\rq\ (\textit{i.e.}\ topological, in some sense) formulation, and is at the base of the localization of the 2-dimensional Yang-Mills partition function.

\medskip
Some technical asides are relegated to the appendices.
Appendix \ref{app:diff_geo} is devoted to some background in differential geometry, concerning principal bundles and the definition of spinors in curved spacetime. In Appendix \ref{app:eq_coho} we report more details about equivariant cohomology, equivariant vector bundles and characteristic classes. This can be seen as a completion of the discussion of Chapter 2, from a more mathematical point of view.

%In appendix A we report  more details about equivariant cohomology, mainly from the point of view of algebraic topology. This can be seen as a completion of the discussion of Chapter 2, from a more mathematical point of view. Appendix B is devoted to some background in differential geometry, concerning principal bundles, equivariant vector bundles, characteristic classes, and the definition of spinors in curved spacetime.

%%%%%%%%%%%%%%%%%%%%%%%%%%%%%%%%%%%%%%%%%%%%%%%%% cap 1
\chapter{Equivariant cohomology}
\label{cha:equivariant cohomology}

In this chapter we review the theory of \emph{equivariant cohomology}, as a modification of the standard cohomology theory applied to spaces that are equipped with the action of a \emph{symmetry group} \(G\) on them, the so-called \(G\)-manifolds. First, we will review the basic notions  about cohomology and homology, from their algebraic definition to the application in topology and differential geometry. The main result that we need to care about, and extend to the equivariant case, is \emph{de Rham's theorem} \cite{bott-tu-differential_forms}, that gives an \emph{algebraic model} for the cohomology of a smooth manifold in terms of the complex of its differential forms. This and  Stokes' theorem relate the theory of cohomology classes to the integration on smooth manifolds. Next, we will extend this to the equivariant setting, giving a topological definition of equivariant cohomology, and then discussing, in the smooth case, an equivariant version of  de Rham's theorem. This, analogously to the standard case, will give an equivalence between the topological definition of equivariant cohomology and the cohomology of some suitable differential complex built from the smooth structure of the space at hand. There are different, but equivalent, possibilities of such algebraic models for the equivariant cohomology of a \(G\)-manifold: we will see the \emph{Weil model}, the \emph{Cartan model} and the \emph{BRST model}, and discuss how they are related one to each other, since at the end they have to describe the same equivariant cohomology.

The purpose of all this, from the physics point of view, is that with equivariant cohomology we can describe a theory of cohomology and integration over manifolds that are acted upon by a symmetry group, the standard setup of classical mechanics and QFT. In the next chapter we will review one of the climaxes of this theory applied to the problem of integration over \(G\)-manifolds: the famous \emph{localization formulas} of Berline-Vergne \cite{berline-vergne-loc} and Atiyah-Bott \cite{atiyah-bott-localization}, that permit to highly simplify a large class of integrals thanks to the equivariant structure of the underlying manifold. The aim and the core of this thesis will be then the description of some generalizations and application of those theorems to the context of QM and QFT, where the integrals of interest are the infinite-dimensional \emph{path integrals} describing partition functions or expectation values of operators.

For this introductory chapter, we follow mainly \cite{tu-equiv_cohom, szabo, Cordes-2dYMTFT, Bott-equiv-cohom}. Another classical reference is \cite{guillemin-sternberg-book}.
Some background tools from differential geometry that are needed can be found in Appendix \ref{app:diff_geo}.

\section{A brief review of standard cohomology theory}
\label{sec:cohomology-review}
In this section we will review some of the basic facts about standard homology and cohomology theory, and in particular its application to topological spaces with the definition of \emph{singular} homology and cohomology groups. Since this is after all standard material, we refer to any book of topology/geometry/algebra (for example \cite{nakahara, hatcher,tu-intro_to_mfd,rotman-algebra}) for the various proofs, while we will give some intuitive examples to help making concrete the various abstract definitions. The main result that we aim to recall is  \emph{de Rham's theorem}, that relates the cohomology theory to differential forms over smooth manifolds, and that will be extended in the next sections to the modified equivariant setup.

We start with the abstract definition of homology and cohomology as algebraic constructions. From this point of view, (co)homology groups are defined in relation to \emph{(co)chain complexes} (or \emph{differential complexes}). 
\begin{defn}  Given a ring \(R\), a \emph{chain complex} is an ordered sequence \(A=( A_p, d_p)_{p\in \mathbb{N}}\) of \(R\)-modules \(A_p\) and homomorphisms \(d_p :A_p \to A_{p-1}\) such that \(d_{p-1}\circ d_p = 0\). A \emph{cochain complex} has the same structure but with homomorphisms \(d_p : A_p \to A_{p+1}\), and \(d_{p+1}\circ d_p = 0\). 
\[
\begin{aligned}
\text{Chain complex}&: \qquad \cdots A_{p-1} \xleftarrow{d_{p}} A_{p} \xleftarrow{d_{p+1}} A_{p+1} \cdots \\
\text{Cochain complex}&: \qquad \cdots  A_{p-1} \xrightarrow{d_{p-1}} A_{p} \xrightarrow{d_{p}} A_{p+1} \cdots 
\end{aligned}
\] \end{defn}

An element \(\alpha\) of a (co)chain complex is called \emph{(co)cycle} or \emph{closed} if \(\alpha\in \mathrm{Ker}(d_p)\) for some \(p\). It is instead called \emph{(co)boundary} or \emph{exact} if \(\alpha\in \mathrm{Im}(d_p)\) for some \(p\). By definition \[
\mathrm{Im}(d_p) \subseteq \mathrm{Ker}(d_{p\pm 1}) , \]
where the \(-\) is for chain and the \(+\) for cochain complexes,
%with \(\pm\) we collected the properties of chain and cochain complex, 
so the quotient sets \(\faktor{\mathrm{Ker}}{Im}\)  of \(p\)-(co)cycles modulo \(p\)-(co)boundries are well defined.
\begin{defn} Given a (co)chain complex \(A\), the \emph{\(p^{th}\) (co)homology group} of \(A\) is \[
\left( H^p(A):= \faktor{\mathrm{Ker}(d_p)}{\mathrm{Im}(d_{p+1})} \right) \qquad H_p(A):= \faktor{\mathrm{Ker}(d_p)}{\mathrm{Im}(d_{p-1})} . \] \end{defn} 
(Co)homology groups are called like that because they inherit a natural Abelian group structure (or equivalently \(\mathbb{Z}\)-module structure) from the sum in the original chain complex.\footnote{Notice that \([\alpha]+[\beta]:=[\alpha+\beta]\) is well defined.} As usual, once we have a definition of a class of mathematical objects, a prime interest lies in the study of structure preserving maps between them. A morphism between (co)chain complexes \(A\) and \(B\) is then a sequence of homomorphisms \(\left(f_p:A_p \to B_p\right)_{p\in\mathbb{N}}\) such that, schematically, \(f\circ d^{(A)} = d^{(B)}\circ f\). It is easy to see that every such a morphism induces an homomorphism of (co)homology groups, since for example \(f^*:H^p(A)\to H^p(B)\) such that \(f^*([\alpha]):= [f(\alpha)]\) is well defined.

Notice that, in many applications one considers (co)chain complexes defined by \emph{graded} modules or algebras with a suitable \emph{differential}. An example of this type is the complex of differential forms \((\Omega(M),d)\) over a smooth manifold, that we will recover later on. For a general (co)chain complex \((A_p,d_p)_{p\in\mathbb{N}}\), we can always see \(A:=\bigoplus_p A_p\) as a \emph{graded} \(R\)-module, whose elements as \(\alpha \in A_p \subset A\) are said to have pure \emph{degree} \(\textrm{deg}(\alpha):= p\). A generic element will be a sum of elements of pure degree.
\begin{defn} \label{def:dga}
A \emph{differential graded algebra} (\emph{dg-algebra} for short) over \(R\) is then an \(R\)-algebra with the decomposition (grading) \(A=\bigoplus_p A_p\), the product satisfying \(
A_p A_q \subseteq A_{p+q} \),
and a \emph{differential} \(d:A \to A\) such that
\begin{enumerate}[label=(\roman*)]
\item it has degree \(\textrm{deg}(d)=\pm 1\), meaning that for every \(\alpha\) of degree \(p\), \(\textrm{deg}(d\alpha) = p\pm 1\);
\item it is nilpotent, \(d^2 = 0\);
\item it satisfies the graded Leibniz rule, \(d(\alpha \beta) = (d\alpha)\beta + (-1)^{\textrm{deg}(\alpha)} \alpha (d\beta)\).
\end{enumerate} \end{defn}
Every such an algebra clearly defines an underlying complex (cochain if \(\textrm{deg}(d)=1\), chain if \(\textrm{deg}(d)=-1\)) and thus has associated (co)homology groups. Even if the algebra structure (the product and the Leibniz rule) are not needed to define the complex, we included it in the definition because this is the kind of structure that arises in physics or in differential geometry. Morphisms of dg-algebras are naturally defined as structure preserving maps between them, analogously to the above discussion.

\bigskip
Beside the abstract algebraic definitions of above, one of the most important applications of homology and cohomology groups is in the study and classification of topological spaces. In order to define these groups in a topological setup, the complexes one takes into consideration are the \emph{simplicial complexes}, that intuitively represents a formal way of constructing \lq\lq polyhedra\rq\rq\ over \(\mathbb{R}^n\), and that can be used in turn to study properties of topological spaces. Given \(\mathbb{R}^\infty\) with the standard basis \(\lbrace e_i\rbrace_{i=0,1,\cdots}\) (\(e_0=0\)), a \emph{standard q-simplex} is 
\begin{equation}
\Delta_q := \left\lbrace \left. x=\sum_{i=0}^q \lambda_i e_i \right| \sum_{i=0}^q \lambda_i =1, \lambda_i\in[0,1]\ \forall i=0,\cdots,q \right\rbrace .
\end{equation}
Although this definition takes into account any possible dimensionality, we can embed these simplices in in finite-dimensional Euclidean spaces, giving them a more practical interpretation.
Given \(q+1\) points \(v_0,\cdots,v_q \in \mathbb{R}^n\), the associated  \emph{affine singular q-simplex} in \(\mathbb{R}^n\) is the map 
\begin{equation}
\begin{aligned}
\left[ v_0\cdots v_q \right]: \Delta_q &\to \mathbb{R}^n \\
\sum_{i=0}^q \lambda_i e_i &\mapsto \sum_{i=0}^q \lambda_i v_i .
\end{aligned}
\end{equation}
This is the convex hull in \(\mathbb{R}^n\) generated by the \emph{vertices} \((v_i)\). Geometrically, the 0-simplex is just the point \(0\in \mathbb{R}^n\), 1-simplices are line segments, 2-simplices are triangles and so on. Notice that \(\Delta_{q-1}\subset \Delta_q\), and its image through \([v_0\cdots v_q]\) is a \lq\lq face\rq\rq\ of the resulting polygon. More precisely, \([v_0\cdots\hat{v_i}\cdots v_q]: \Delta_{q-1} \to \Delta_q\) (the hat means we take away that point from the list) is regarded as the \emph{\(i^{th}\) face map}, denoted concisely as \(F^{(i)}_q\).
\begin{figure}[ht]
\centering
\sbox{\tempbox}{\includegraphics[width=.45\textwidth]{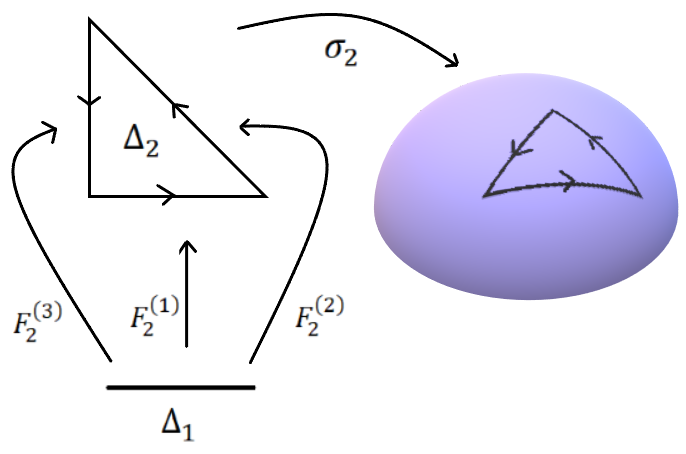}}
\subfloat[]{\vbox to \ht\tempbox{\vfil \hbox{\includegraphics[width=.45\textwidth]{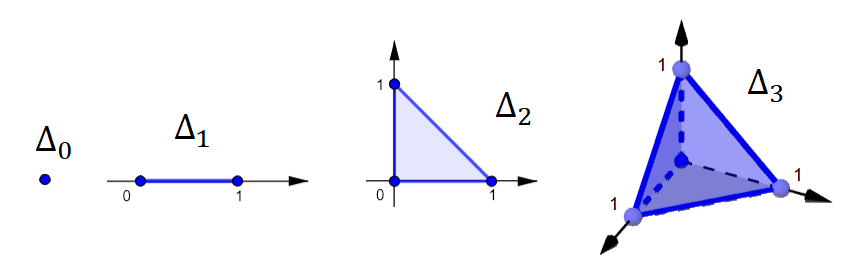}}\vfil}}
\qquad
\subfloat[]{\usebox{\tempbox}}
\caption{(a) First standard simplices. (b) An oriented affine 2-simplex, its face maps and a singular 2-simplex \(\sigma_2\) on a 2-dimensional topological space.}
\label{fig:simplices}
\end{figure}

The same idea can be used to embed the simplices in a generic topological space \(M\), changing the codomain of the simplex map. A \emph{singular q-simplex} in \(M\) is then a continuous map\footnote{The standard topology on \(\mathbb{R}^q\) is induced on \(\Delta_q\).}
\begin{equation}
\sigma_q: \Delta_q \to M 
\end{equation}
where now \(\lbrace \sigma_q(e_0),\cdots,\sigma_q(e_q)\rbrace\) are the \emph{vertices} of \(\sigma_q\). Two simplices are said to have the same/opposite \emph{orientation} if the vertex sets are respectively even/odd permutations of each other. The word \lq\lq singular\rq\rq\ is there because only continuity is required, thus from a \lq\lq smooth\rq\rq\ point of view these simplices can present singularities.

With this setup, we can construct chain complexes on topological spaces in terms of singular simplices. In fact, defining the sum of two singular simplices \(\sigma_q, \rho_p\) as 
\begin{equation}
\begin{aligned}
(\sigma_q + \rho_p) : \Delta_q \sqcup \Delta_p &\to M \\
\lambda &\mapsto \begin{cases}
\sigma_q (\lambda) & \text{if}\ \lambda\in \Delta_q \\
\rho_q (\lambda) & \text{if}\ \lambda\in \Delta_p  ,
\end{cases}
\end{aligned}
\end{equation}
whose image is the (disjoint) union in \(M\) of the images of the two starting simplices. Since the \lq\lq +\rq\rq\ is clearly commutative, \(\mathcal{C}_q(M):=C^0(\Delta_q,M)\) is an Abelian group, called the \emph{(singular) q-chain group} of \(M\). We can define a \emph{boundary operator} as a group homomorphism
\begin{equation}
\begin{aligned}
\partial : \mathcal{C}_q(M) &\to \mathcal{C}_{q-1}(M) \\
\sigma &\mapsto \partial\sigma := \sum_{i=0}^q (-1)^i \left(\sigma\circ F^{(i)}_q\right) ,
\end{aligned}
\end{equation}
that restricts to the (oriented) sum of faces of a given simplex, and happens to satisfy the nilpotency condition \(\partial\circ\partial =0\). This means that \(\mathcal{C}(M):= \bigoplus_{q=0}^\infty \mathcal{C}_q(M)\) with the operator \(\partial\) defines a dg-module over \(\mathbb{Z}\), and an associated chain complex, that we use to define the homolgy groups of \(M\).
\begin{defn} The \emph{singular \(q^{th}\) homology group} of \(M\) is \[
H_q(M;\mathbb{Z}) := \faktor{\mathrm{Ker}(\partial_q)}{\mathrm{Im}(\partial_{q+1})} .\]
It is often  useful to work with homology groups \emph{with coefficients} in some \(\mathbb{Z}\)-module \(A\) (like the real numbers), that is considering \(H_q(M;A)\) as defined from the simplicial complex \(\mathcal{C}(M)\otimes A\). \end{defn}

\begin{ex}[Homology of spheres]\label{ex:homology-spheres}
In practice, the strategy to get \(H_*(M)\) is the so-called \emph{triangulation} of \(M\), \textit{i.e.}\ constructing a suitable simplicial complex \(K\) in \(\mathbb{R}^{\dim(M)}\) as a set of standard simplices, whose union gives a polyhedron that is homeomorphic to \(M\). Then, one can count and classify all the cycles and the boundaries in \(K\), and then get \(H_*(K)\cong H_*(M)\). Some examples of this rigorous approach can be found in \cite{nakahara}. We can still give some examples, less rigorously, by looking directly at simple topological spaces, just to help building some intuition. Remember that a \(q\)-cycle \(\sigma\) on \(M\) is a boundary-less singular \(q\)-simplex \emph{up to continuous deformations}, and it is also a boundary if it can be seen as the border of a \((q+1)\)-simplex. 
\begin{enumerate}
\item[\((\mathbb{S}^1)\)] On the circle there is no place for simplices of dimension higher than 1, so we look at the 1-simplices. There are two inequivalent ways of deforming the standard 1-simplex onto the circle: it can join at the end points covering all \(\mathbb{S}^1\) or not. In the first case, that we call \(\sigma_1\), we have \(\partial \sigma_1 = 0\) since \(\mathbb{S}^1\) has no boundaries, in the second the boundaries are the end points of the singular 1-simplex. The first boundary-less case cannot be seen as a boundary of something else, by dimensionality, so the Abelian group \(H_1(\mathbb{S}^1;\mathbb{Z})=\mathrm{Ker}(\partial_1)/\mathrm{Im}(\partial_2)\) is generated by a single element, \([\sigma_1]\). In other words, \(H_1(\mathbb{S}^1;\mathbb{Z})\cong \mathrm{span}_{\mathbb{Z}}\lbrace [\sigma_1]\rbrace\cong \mathbb{Z}\).

The case \(q=0\) is trivial, since we have only one way of drawing a point on the circle, and every point is boundary-less. We have just said that the boundary of a 1-simplex is either zero or two points, so a single point is never a boundary. Thus the Abelian group \(H_0(\mathbb{S}^1;\mathbb{Z})\cong \mathbb{Z}\) since it is generated by only one element. Generalizing a little, we can already see from this example that the homology group in 0-degree will always follow this trend for \emph{connected} topological spaces. If the space has \(n\) connected components, there will be \(n\) inequivalent ways of drawing a point on it, so \(n\) generators for the homology group, giving \(H_0(M_{(n)};\mathbb{Z})\cong \mathbb{Z}\oplus \cdots \oplus \mathbb{Z}\) (\(n\) factors).

\item[\((\mathbb{S}^2)\)] The 2-sphere does not necessitate of much more work, at least with this level of rigor. Again, by dimensionality the homology groups in degrees higher than \(\dim(\mathbb{S}^2)=2\) are empty. For \(q=2\), the only way we can construct a boundary-less figure on the sphere from \(\Delta_2\) is joining the vertices and the edges together and cover the whole sphere. All other singular 2-simplices have boundaries, and the 2-sphere cannot be seen as a boundary of something else by dimensionality, so analogously to the previous case \(H_2(\mathbb{S}^2;\mathbb{Z})\cong \mathbb{Z}\).

For the 1-simplices, we notice that the only two ways of drawing a segment on the sphere (up to continuous deformations) is to close it or not  at the end points. In the first case, the 1-simplex has no boundary, but can be seen as the boundary of its internal area, so it is in fact exact. In the second case, the 1-simplex has boundaries so it is outside \(\mathrm{Ker}(\partial)\). This means that every 1-cycle is also a boundary, and thus \(H_1(\mathbb{S}^2;\mathbb{Z})\cong 0\). In 0-degree we can argue in the same way as for the circle that \(H_0(\mathbb{S}^2;\mathbb{Z})\cong \mathbb{Z}\).

\item[\((\mathbb{S}^n)\)] It turns out that all spheres follow this trend, giving homology groups \[
H_q(\mathbb{S}^n;\mathbb{Z}) \cong \begin{cases} 
\mathbb{Z} & q=0,n \\
0 & \text{otherwise}. \end{cases}
\]
If interested in the case with real coefficient, the homology of spheres are again very simple, since \(\mathbb{Z}\otimes \mathbb{R}\cong \mathbb{R}\).
\end{enumerate}
\end{ex}

Now we can turn to the construction of singular cohomology groups on topological spaces. This is done considering the dual spaces \(\text{Hom}(\mathcal{C}_q(M),A)\) with values in a \(\mathbb{Z}\)-module \(A\). The simplest choice is of course \(A=\mathbb{Z}\). Notice that \(\text{Hom}(\mathcal{C}_q(M),A)\) itself is a \(\mathbb{Z}\)-module. The \emph{coboundary operator} in this case  is defined as the \(\mathbb{Z}\)-module homomorphism 
\begin{equation}
\begin{aligned}
\delta: \text{Hom}(\mathcal{C}_q(M),A) &\to \text{Hom}(\mathcal{C}_{q+1}(M),A) \\
f &\mapsto \delta f \quad s.t. \quad \delta f(\sigma_{q+1}):= f(\partial \sigma_{q+1})
\end{aligned}
\end{equation}
and from the nilpotency \(\partial^2=0\) we get easily \(\delta^2=0\).
\begin{defn} The \emph{singular \(q^{th}\) cohomology group} of \(M\), \emph{with coefficients in \(A\)}, is \[
H^q(M;A) := \faktor{\mathrm{Ker}(\delta_q)}{\mathrm{Im}(\delta_{q-1})} . \]
\end{defn}
Note that for a commutative ring \(A\) (as for example \(\mathbb{R}\)), the cohomology groups are naturally \(A\)-modules.
Although the definition is less practical than the one for homology groups, there is an important theorem that allows to relate the two, so that homology computations can be used to infer the structure of singular cohomology groups. This is the so-called \emph{universal coefficient theorem} \cite{hatcher}. Since for applications to smooth manifolds we will be primarily interested in cohomology groups with real coefficients (as it will become clearer later) we take \(A=\mathbb{R}\). For this special case, the theorem says
\begin{equation}
H^q(M;\mathbb{R}) \cong H^q(M;\mathbb{Z})\otimes \mathbb{R} \cong H_q(M;\mathbb{R})^*,
\end{equation}
so that the cohomolgy groups are exactly the dual spaces of the homology groups. Notice that, from the example above, \(H^q(\mathbb{S}^n;\mathbb{R}) \cong \mathbb{R}\) in degree \(q=0,n\).

The above construction relates the topological properties of the space \(M\) to the algebraic concept of (co)homology groups. In general we said that morphisms of complexes induce morphisms of associated (co)homologies, and this extends to the present topological case: if we consider a continuous map (morphism of topological spaces) \(F:M\to N\) between the topological spaces \(M, N\), we can lift it to \(F_\# : \mathcal{C}_q(M) \to \mathcal{C}_q(N)\) such that \(F_\#(\sigma):= F\circ\sigma\), that is a morphism of dg-modules. This in turn induces the morphism of homology groups as descrbed above. For the cohomology groups we have, analogously, the lifted map in the opposite direction \(F^\#: \text{Hom}(\mathcal{C}_q(N),A)\to \text{Hom}(\mathcal{C}_q(M),A)\) such that \(F^\#(f):= f\circ F_\#\), giving the morphism of cochain complexes. This induces \(F^*: H^q(N)\to H^q(M)\) such that \(F^*([f]):=[F^\#(f)]\).\footnote{In the context of smooth manifolds and de Rham cohomology, this is analogous to the \emph{pull-back} of differential forms.} Also, we notice that if we have two continuous maps \(F,G\) between the topological spaces, then\footnote{In category theory language, we can summarize these properties saying that singular homology \(H_*(\cdot )\) is a \emph{covariant functor} between the categories \textbf{Top} of topological spaces and \textbf{Ab} of Abelian groups, and singular cohomology \(H^*(\cdot )\) is a \emph{contravariant functor} between \textbf{Top} and \textbf{Ab}. Anyway, we will not need such a terminology for what follows. See for example \cite{marsh}, Appendix A, for a quick introduction to the subject.}
\begin{equation}
\begin{aligned}
&(F\circ G)_\# = F_\# \circ G_\# \quad \Rightarrow \quad (F\circ G)_* = F_* \circ G_* \\
&(F\circ G)^\# =  G^\# \circ F^\# \quad \Rightarrow \quad (F\circ G)_* = G_* \circ F_* .
\end{aligned}
\end{equation}

An important fact that permits to use homology and cohomology groups to classify and characterize topological spaces, is that these objects are \emph{topological invariants}, meaning that isomorphic spaces have the same (co)homology groups. Moreover, a stricter result holds: two homotopy-equivalent topological spaces have the same cohomology and homology groups. We recall that two continuous maps \(F,G:M\to N\) between topological spaces are \emph{homotopic} if it exists a continuous map \(H:[0,1]\times M\to N\) that deforms continuously \(F\) in \(G\), \textit{i.e.}\ \(H(0,x)=F(x)\) and \(H(1,x)=G(x)\) for every \(x\in M\). Homotopy of maps is an equivalent relation, and we denote it by \(F\sim G\). Two topological spaces \(M,N\) are said to be \emph{homotopy-equivalent}, or of the same homotopy type, if there exist two maps \(F:M\to N\) and \(G:N\to M\) such that \( (G\circ F)\sim id_M\). Homotopy-equivalence is also an equivalence relation, that we denote also as \(M\sim N\). The result stated above is then, for cohomologies 
\begin{equation}
\label{cohomology-1}
M\sim N \Rightarrow H^*(M)\cong H^*(N). 
\end{equation}

\begin{ex}[More singular homologies] \label{ex:homology-contractibility}
\hspace{1em}
\begin{enumerate}%[label=(\roman*)]
\item[\((pt)\)] We can consider the very trivial case of \(M\) being just a point. In this case, the informal discussion of Example \ref{ex:homology-spheres} can be carried out just for the 0-dimensional simplices: \(H_*(pt;\mathbb{Z})\cong \mathbb{Z}\) in degree 0. It follows by definition and by the universal coefficient theorem that for any commutative ring \(A\) (as \(\mathbb{R}\)), \(H^0(pt;A)\cong H_0(pt;A) \cong A\) and \(H^q(pt;A)\cong H_q(pt;A) \cong 0\) for \(q>0\). By the homotopy-invariance property discussed above, any contractible space will have the same trivial cohomology and homology as the point!
 
\item[\((\mathcal{C})\)] Let us look at another simple case, the cylinder \(\mathcal{C}=\mathbb{S}^1\times [0,1]\). To compute its homology groups we could follow the intuitive discussion of Example \ref{ex:homology-spheres}, or we can just notice that since the interval \([0,1]\) is contractible, \[
\mathbb{S}^1\times [0,1] \sim \mathbb{S}^1.\]
This means that \(H^q(\mathbb{S}^1\times [0,1];\mathbb{R})\cong H^q(\mathbb{S}^1;\mathbb{R})\cong H_q(\mathbb{S}^1;\mathbb{R})\).

\item[\((T)\)] A less trivial example is the 2-torus \(T= \mathbb{S}^1\times \mathbb{S}^1\). In this case no one of the factors is contractible, so we cannot use the homotopy invariance to get the result from a simpler space. We can anyway get the answer using the same method of Example \ref{ex:homology-spheres}. Starting from the top-degree homology group, we notice that the only boundary-less surface on the torus is the torus itself. Thus analogously to all the other cases, \(H_2(T;\mathbb{Z})\cong \mathbb{Z}\) or \(H_2(T;\mathbb{R})\cong \mathbb{R}\). Since the torus is connected, in degree 0 we get trivially \(H_0(T;\mathbb{R})\cong \mathbb{R}\). In degree 1 we see the difference with the other cases. On the torus there are two inequivalent ways of drawing a closed line that is not a boundary of any 2-dimensional surface, following essentially the two factors of \(\mathbb{S}^1\) (see figure \ref{fig:homology-torus}). This means that the \(1^{st}\) homology group is generated by two elements, and thus \(H_1(T;\mathbb{Z})\cong \mathbb{Z}\oplus \mathbb{Z}\). The same reasoning can be applied to higher genus surfaces \(\Sigma_g\), giving \(H_1(\Sigma_g;\mathbb{Z})\cong (\mathbb{Z})^{\oplus  2g}\).
\end{enumerate}
\end{ex}
\begin{figure}
\centering
\sbox{\tempbox}{\includegraphics[width=.30\textwidth]{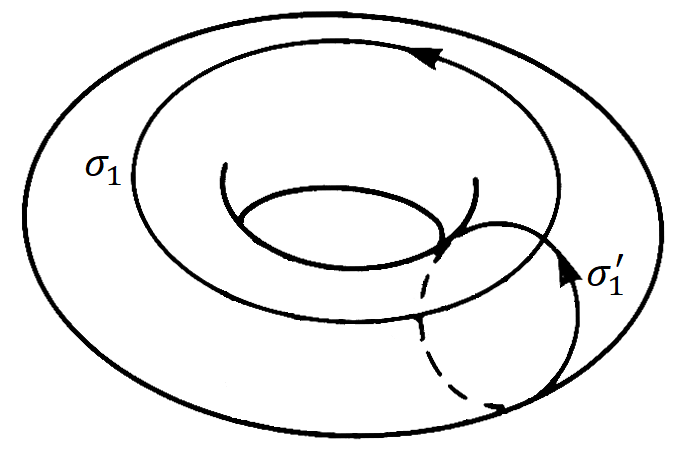}}
\subfloat[]{\usebox{\tempbox}} \qquad
\subfloat[]{\vbox to \ht\tempbox{\vfil \hbox{\includegraphics[width=.60\textwidth]{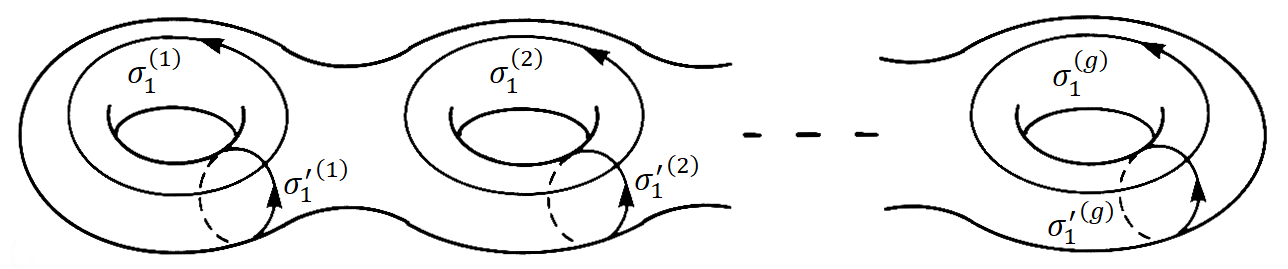}}\vfil}}
\caption{(a) The 2 inequivalent non-exact 1-cycles \(\sigma_1\) and \(\sigma_1'\) on the 2-torus. (b) The genus-\(g\) surface \(\Sigma_g\) has \(2g\) inequivalent non-exact 1-cycles. Figures adapted from \cite{nakahara}.}
\label{fig:homology-torus}
\end{figure}

\bigskip
We leave now the purely topological setup, since in physics we are mostly interested in studying local properties, \textit{i.e.}\ from the differential geometry point of view. We assume to work in the smooth setting and consider \(M\) to be a \(d\)-dimensional \(C^\infty\)-manifold. With \(TM\) we denote its \emph{tangent bundle}, and with \(T^*M\) its \emph{cotangent bundle}.\footnote{Sections of any bundle \(E\to M\) over \(M\) will be denoted in the following with \(\Gamma(M,E)\), or \(\Gamma(E)\) when the base space is clear from the context. For example, vector fields are elements of \(\Gamma(TM)\).} At every point \(p\in M\), \(T_p M\) and \(T^*_p M\) are the dual vector spaces of tangent vectors and 1-forms at \(p\), respectively. We consider the \emph{exterior algebra} \(\bigwedge(T^*_p M)\), with the wedge product \(\wedge\) making it in a graded-commutative algebra, and the exterior derivative \(d: \bigwedge(T^*_p M)\to \bigwedge(T^*_p M)\) acting as a graded derivation of \(\textrm{deg}(d)=+1\). Extending these operations point-wise for every point \(p\in M\), we have the bundle of \emph{differential forms} over \(M\), 
\begin{equation}
\Omega(M) := \bigoplus_{k=0}^d \Omega^k(M) \qquad \text{with}\ \Omega^k(M) := \bigsqcup_{p\in M} \bigwedge^k(T^*_p M).
\end{equation}
\((\Omega(M),\wedge,d)\) is thus a dg-algebra over the commutative ring \(C^\infty(M)\), and naturally defines a cochain complex called the \emph{de Rham complex}. The associated cohomology groups are the \emph{de Rham cohomology groups}, constituting the graded-commutative ring\footnote{Analogously to the singular cohomology, in  category theory language the de Rham cohomology \(H_{dR}(\cdot)\) is a contravariant functor between the categories \textbf{Man} of smooth manifolds and \textbf{Ab} of Abelian groups.}
\begin{equation}
H_{dR}(M) = \bigoplus_{k=0}^d H_{dR}^k(M) \qquad \text{with} \quad H_{dR}^k(M) := H^k(\Omega(M),d) = \faktor{\mathrm{Ker}(d_k)}{\mathrm{Im}(d_{k-1})}.
\end{equation}
The ring structure of \(H_{dR}(M)\) is naturally inherited from the wedge product of differential forms, that lifts at the level of cohomology classes. In fact, for two closed forms \(\omega,\eta\in \Omega(M)\),
\begin{equation}
[\omega] \wedge [\eta] := [\omega\wedge\eta]
\end{equation}
is well-defined.\footnote{This result can be seen also in the topological setup for singular cohomology groups, as it should be by de Rham's theorem. The operation that corresponds to the wedge product between singular cohomology classes is called \emph{cup-product} \cite{hatcher}. It is important to remember that cohomology in general has a ring structure.}
The final important result that we state, and that will be crucial to extend to the equivariant setting in the following section, is the so called \emph{de Rham's theorem}:
\begin{thm}[de Rham] The de Rham cohomology of the smooth manifold \(M\) is isomorphic to its singular cohomology with real coefficients:
\[ H_{dR}(M) \cong H^*(M;\mathbb{R}). \]
\end{thm}
The power of this theorem is that it allows to study topological properties of the manifold (recall that \(H^*(M;\mathbb{R})\) are homotopy-invariants) using differential geometric (so local) objects, the differential forms. We say that the de Rham complex \((\Omega(M),d)\) constitute an \emph{algebraic model} for the singular cohomology of \(M\). Notice that, by dimensionality reasons, we get trivially also in this case that the cohomology groups \(H_{dR}^q(M)\) for \(q>\dim(M)\) are automatically zero. Another important property that is intuitively very clear from the de Rham complex is the \emph{Poincaré duality}. For a closed connected manifold \(M\) this states that, as vector spaces
\begin{equation}
H^k_{dR}(M) \cong H^{\dim(M)-k}_{dR}(M) .
\end{equation}
A crucial tool for the proof of de Rham's theorem is the so-called \emph{Stokes' theorem}, that relates the integral of an exact \(d\)-form over a \(d\)-dimensional manifold to the integral of its primitive over the \((d-1)\)-dimensional boundary,
\begin{equation}
\int_M d\omega = \int_{\partial M} \omega .
\end{equation}
Notice that integration over \(M\) when \(\partial M = \emptyset\) can be regarded as a function on the \(d^{th}\) de Rham cohomology \(
\int : H^d_{dR}(M) \to \mathbb{R} \). 

\begin{ex}[Cohomology rings]
With the help of de Rham's theorem, we can compute some of the previous example directly at the level of cohomology using differential forms and integration.\footnote{Another powerful tool to practically compute cohomology groups and rings, at the topological level, goes by the name of \emph{spectral sequences}. See for example \cite{tu-equiv_cohom}.} Let us consider the case of the tours \(T=(\mathbb{S}^1)^2\). We can parametrize it with coordinates \((x,y)\) taking values in \([0,1)^2\subset \mathbb{R}^2\). If we call \(\alpha:= dx\) and \(\beta:=dy\) in \(\Omega^1(T)\), a natural choice of volume form is \(\omega:=\alpha\wedge\beta\), that gives \(\mathrm{vol}(T)=1\). The volume form is of course closed by dimensionality, but it cannot be exact since otherwise by Stokes' theorem the volume of the torus would be 0, so it defines a non-trivial cohomology class \([\alpha\wedge\beta]\). Any other 2-form is of the type \(\omega'=f\omega\) for some \(f\in C^\infty(T)\), but closed forms must satisfy \(df=0\), so \(f\in \mathbb{R}\) constant. We conclude that any other independent closed 2-form has to be \lq\lq cohomologous\rq\rq\ to \([\alpha\wedge\beta]\), so that in top-degree \(H^2_{dR}(T)\cong \mathrm{span}\lbrace [\alpha\wedge\beta]\rbrace\cong \mathbb{R}\). 

In degree 1, any closed form must be a combination of \(\alpha\) and \(\beta\) with real coefficients (since again \(d(f\alpha)=0\ \Leftrightarrow\ df=0\)), so they are the only independent closed 1-forms (they correspond to the volume forms for the two \(\mathbb{S}^1\) factors). To see whether or not they are exact, we can use Stokes' theorem: if they are, then their integral over \emph{any} closed curve on \(T\) must be zero. But we can take the two curves \(\sigma_1(t)=(t,0)\) and \(\sigma_1'(t)=(0,t)\) of Figure \ref{fig:homology-torus} and see that \[
\int_{\sigma_1}\alpha = 1 = \int_{\sigma_1'}\beta ,\]
so they define two independent cohomology classes \([\alpha]\) and \([\beta]\). This means that \(H^1_{dR}(T)\cong \mathrm{span}\lbrace [\alpha],[\beta]\rbrace\cong \mathbb{R}\oplus \mathbb{R}\).

Since the torus is connected, the only closed 0-form is a constant number, that we can chose to be \(1\). Thus, \(H^0_{dR}\cong \mathbb{R}\).
We can further easily get the ring structure of \(H_{dR}(T)\) by looking at the multiplication rules between the generators. If we call \(a:=[\alpha]\) and \(b:=[\beta]\), the wedge product of differential forms gives the following rules  \[
a^2 \sim 0, \qquad b^2\sim 0, \qquad ab + ba\sim 0. \]
Thus we can rewrite the cohomology ring as a polynomial ring over the indeterminates \((a,b)\), taken in degree 1, that satisfy the above rules:
\[ H^*(T;\mathbb{R})\cong \faktor{\mathbb{R}[a,b]}{(a^2,b^2,ab+ba)} ,\]
where \((a^2,b^2,ab+ba)\) denotes the quotient by the ideal generated by the corresponding expressions.

In the same fashion we can rewrite the cohomology rings of the other examples that we gave above for the \(n\)-sphere. Introducing an indeterminate \(u\) of degree \(n\), and the multiplication rule \(u^2\sim 0\), its cohomology ring can be expressed as \[
H^*(\mathbb{S}^n;\mathbb{R}) \cong \faktor{\mathbb{R}[u]}{u^2}.\]
We quote another example, that will enter in the case of equivariant cohomology with respect to a circle action by \(U(1)\cong\mathbb{S}^1\). For the complex projective plane \(\mathbb{C}P^n\), it turns out that \[
H^*(\mathbb{C}P^n;\mathbb{R}) \cong \faktor{\mathbb{R}[u]}{u^{(n+1)}} \]
where \(\mathrm{deg}(u):=2\). In the limiting case \(n\to\infty\), one has thus \(H^*(\mathbb{C}P^\infty;\mathbb{R})\cong \mathbb{R}[u]\), the polynomials in \(u\).
\end{ex}

\section{Group actions and equivariant cohomology}
\label{sec:equivariant-cohomology}

As already mentioned, equivariant cohomology is an extension of the standard cohomology theory, partly reviewed in the last section, to the cases in which the space \(M\) is acted upon by some group \(G\). This is the common setup in physics, from the finite-dimensional cases of classical Lagrangian or Hamiltonian mechanics to the infinite dimensional case of Quantum Field Theory, where \(M\) can be the configuration space, the phase space, or the space of fields, and \(G\) is a Lie group representing a \emph{symmetry} of the physical system. In gauge theory for example, we want to identify those physical configurations that are equivalent modulo a gauge transformations, so the moduli space of gauge orbits \(M/G\). In Poincaré-supersymmetric theories, the group \(G\) is actually the Poincaré group of spacetime symmetries. In all these cases we are interested in the cohomology of \(M\) modulo these symmetry transformations, since many primary objects of study (partition functions, expectation values...) are usually given in terms of integrals over \(M\).
Before moving to the technical definition of \(G\)-equivariant cohomology of \(M\), we recall some terminology about group actions.

\begin{defn} \begin{enumerate}[label=(\roman*)]
\item Given a group \(G\) and a topological space \(M\),\footnote{We are going to work practically always with smooth manifolds and (compact) Lie groups, but for the moment we do not need this level of structure on \(M\) and \(G\).}  a \emph{\(G\)-action} on \(M\) is given by a group homomorphism (\emph{left} action) or anti-homomorphism (\emph{right} action) \[
\rho: G \to \mbox{Homeo}(M)\quad (\text{or Diff}(M)\text{ for smooth manifolds}). \]
If \(m\in M, g\in G\), the left action of \(g\) on \(m\) can be denoted \(\rho(g)m\equiv g\cdot m\), and the right action \(m\cdot g\), if this causes no confusion. \(M\) is said to be a (left or right) \emph{\(G\)-space}.

\item If \(M, N\) are two \(G\)-spaces, on the product \(M\times N\) it is canonically defined the \emph{diagonal \(G\)-action} \[
\rho^{M\times N}(g)(m,n):= (\rho^M(g)m, \rho^N(g)n)\quad \text{for }m\in M, n\in N, g\in G .\]

\item Given a point \(m\in M\), the \emph{orbit} of \(m\) is the subset of \(M\) of all points that are reached from \(m\) by the action of \(G\). The \emph{orbit space} with respect to the \(G\)-action is \(M/G\).\footnote{It is easy to check that \(m\sim m' \Leftrightarrow m' = g\cdot m\) for some \(g\in G\) is an equivalence relation.}

\item The \emph{stabilizer} (or \emph{isotropy group}, or \emph{little group}) of \(m\) is the subgroup of \(G\) of all elements that act trivially on \(m\), \textit{i.e.}\ \(g\cdot m=m\). The \(G\)-action is called \emph{free} if the stabilizer of every point in \(M\) is given by the identity of \(G\). The \(G\)-action is called \emph{locally free} if the stabilizer of every point is \emph{discrete}. The \emph{fixed point set} \(F\subseteq M\), is the set of all points that are stabilized by the entire \(G\).

\item Morphisms of \(G\)-spaces are called \emph{\(G\)-equivariant} functions. \(f:M\to N\) is \(G\)-equivariant if \(f(g\cdot m) = g\cdot f(m)\), for every \(m\in M, g\in G\).
\end{enumerate}
\end{defn}

In the following we will not care much about distinguishing between left and right actions, and assume all \(G\)-actions are from the left, unless otherwise stated. For the first part of the discussion it is not needed, but we are going to assume \(M\) and \(G\) to be at least topological manifolds, and then specialize to the case of smooth manifolds, since these are the most common structures arising in physics. Since, as we said above, we are interested in identifying those elements in \(M\) that are equivalent up to a \lq\lq symmetry\rq\rq\ transformation by \(G\), the first candidate for the \(G\)-equivariant cohomology of \(M\) could be simply the cohomology of the orbit space \(M/G\),
\begin{equation}
H_G^*(M) := H^*\left(\faktor{M}{G}\right) .
\end{equation}
This definition has the problem that, if  the \(G\)-action is not free and has fixed points on \(M\), the orbit space is singular: in the neighborhood of those fixed points there is no well-defined notion of dimensionality. This kind of singular quotient spaces are called \emph{orbifolds}. 

\begin{ex}[Some group actions and orbit spaces]\label{ex:group-actions} 
\hspace{1em}
\begin{enumerate}[label=(\roman*)]
\item Let us consider the Euclidean space \(\mathbb{R}^n\). \(O(n)\) rotations (and reflections) act naturally on it, with the only fixed point being the origin. Any point  but the origin identifies a direction in the Euclidean space, and thus is stabilized by the subgroup of \(n-1\) rotations, \(O(n-1)\). For example we see that, without considering parity transformations, the standard \(SO(2)\)-action is free on \(\mathbb{R}^2\setminus \lbrace 0\rbrace\).
%In the case of the plane \(\mathbb{R}^2\), the \(O(2)\)-action is thus free on \(\mathbb{R}^2\setminus \lbrace 0\rbrace\). 
If we bring translations into the game, considering the Euclidean space as an affine space acted upon by \(ISO(n)=\mathbb{R}^n\rtimes O(n)\), then the stabilizer of any point is the entire \(O(n)\), since any point can be considered an origin after translation. So the action of \(ISO(n)\) is neither free nor locally free, but has no fixed points on the entire \(\mathbb{R}^n\).

Considering only rotations,  the orbit space \(\mathbb{R}^n/O(n)\) is the space of points identified up to their angular coordinates, that is an half-line starting from the origin, \(\mathbb{R}^n/O(n)\cong [0,\infty)\). This is not a manifold, since the interval is closed on the left, giving a \lq\lq singularity\rq\rq\ on the original fixed point of the action.

Notice that, by embedding \(\mathbb{S}^{n-1}\) in \(\mathbb{R}^{n}\), an \(O(n)\)-action descends on it, and the orbit of any point of \(\mathbb{S}^{n-1}\) is the sphere itself. The orbit of any point can be seen as the quotient of \(O(n)\) by the stabilizer of that point, so that one has
\[ \faktor{O(n)}{O(n-1)} \cong \mathbb{S}^{n-1}.\]
This quotient describes the common situation of spontaneous symmetry braking inside \(O(n)\)-models, in Statistical Mechanics.

\item One can always consider circle actions on the spheres \(\mathbb{S}^n\). Starting with \(n=1\), and considering the circle as embedded in the \(\mathbb{C}\)-plane, \(U(1)\) acts on itself by multiplication: \(e^{i\varphi}\mapsto e^{ia}e^{i\varphi}\) for some \(a\in [0,2\pi)\). This action has clearly no fixed points, and it is also free. Thus the quotient is well defined, giving simply \(\mathbb{S}^1/U(1) \cong pt\).

The \(U(1)\)-action on the 2-sphere is already more interesting. Rotations around a given axis fix two points on \(\mathbb{S}^2\), that we identify with the North and the South poles. If we exclude the poles the resulting space is homeomorphic to a cylinder, the \(U(1)\)-action becomes free and indeed we have that \(\mathbb{S}^1\times (0,1)\to (0,1)\) is a trivial principal \(U(1)\)-bundle. But considering the poles, the quotient space is singular since \(\mathbb{S}^2/U(1) \cong [0,1]\). This is an elementary example of an orbifold.
\begin{figure}[ht]
\centering
\includegraphics[width=.40\textwidth]{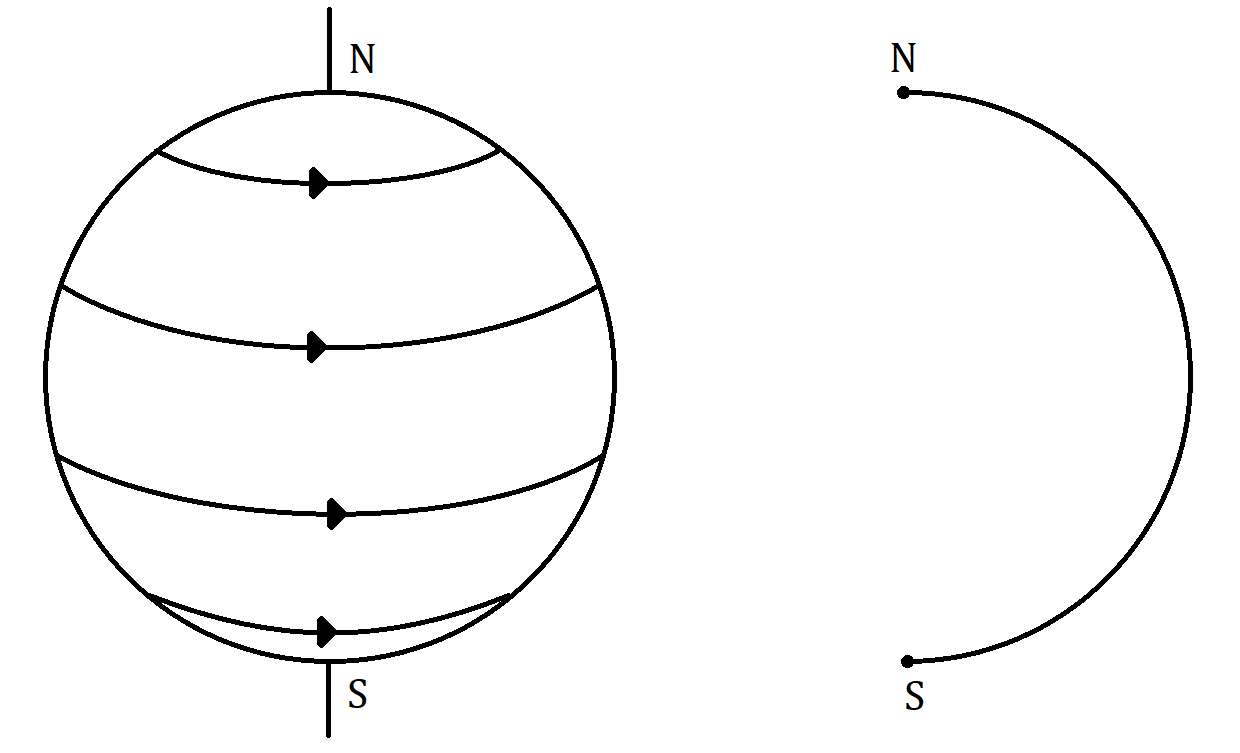}
\caption{The circle acting on the 2-sphere, the orbits being the parallels. The orbit space \(\mathbb{S}^2/U(1)\) is a meridian, homeomorphic to the interval \([0,1]\).}
\label{fig:sphere-rotation}
\end{figure}

Let us consider also the case \(n=3\). The 3-sphere \(\mathbb{S}^3\cong SO(2)\) can be parametrized by a pair of complex numbers \((z_1,z_2)\) such that \(|z_1|^2 + |z_2|^2 = 1\). The circle then acts naturally by diagonal multiplication: \((z_1,z_2)\mapsto (e^{ia}z_1,e^{ia}z_2)\) for some \(e^{ia}\in U(1)\). This action is clearly free, since the two coordinates \(z_i\) cannot be simultaneously zero on the sphere, thus the quotient is well defined, giving the 3-sphere the structure of a principal \(U(1)\)-bundle known as the \emph{Hopf bundle}. The equivalence classes \([z_1,z_2]\in \mathbb{S}^3/U(1)\) describe, by definition, points on the complex projective line \(\mathbb{C}P^1\cong \mathbb{S}^2\), that is isomorphic to the Riemann 2-sphere. The Hopf bundle can be thus seen as \(\mathbb{S}^3 \to \mathbb{S}^2\), with typical fiber \(\mathbb{S}^1\). Clearly this is not a trivial bundle, since \(\mathbb{S}^3 \neq \mathbb{S}^2\times \mathbb{S}^1\).

The above case generalizes to any odd-dimensional sphere \(\mathbb{S}^{2n+1}\), since they all can be embedded in complex spaces \(\mathbb{C}^{n+1}\). The circle acts always by diagonal multiplication, and the resulting action is free. The bundles \(\mathbb{S}^{2n+1}\to \mathbb{S}^{2n+1}/U(1)\cong \mathbb{C}P^{n}\) are all principal \(U(1)\)-bundles over the complex projective spaces \(\mathbb{C}P^n\).

\item The last case we mention is the possible \(U(1)\)-action on a torus \(T=(\mathbb{S}^1)^2\), by rotations along one of the two factors. This is the only possible free action on a closed surface, giving the well defined quotient \(T/U(1)\cong \mathbb{S}^1\). More examples can be found in \cite{audin}.
\end{enumerate}
\end{ex}

The example above showed that also in very simple cases singularities can appear in quotient spaces, so that one cannot define cohomology in a smooth way using the powerful de Rham theorem. It is thus more convenient to set up a definition of equivariant cohomology that automatically avoids this problem.  
This more clever definition is given by the \emph{Borel construction} for the \(G\)-space \(M\).
\begin{defn} Considering a \(G\)-space \(M\), its associated \emph{Borel construction}, or \emph{homotopy quotient}, is \[
M_G := \faktor{(M\times EG)}{G} \equiv M\times_G EG \]
where \(EG\) is some contractible space on which \(G\) acts freely, called the \emph{universal bundle of \(G\)} (see Appendix \ref{app:univ-bundles} for the precise definition).
%\footnote{A \emph{weakly contractible} space is a topological space whose homotopy groups are all trivial. Clearly any contractible space is weakly contractible. It is a fact that every CW complex that is weakly contractible is also contractible \cite{hatcher}. In the following of the section we implicitly work in the category of CW complexes, but we sweep this subtlety under the carpet. See \cite{tu-equiv_cohom} for all the mathematical details.} 
The \emph{\(G\)-equivariant cohomology} of \(M\) is then defined as\footnote{From now on we always consider cohomologies with coefficients in \(\mathbb{R}\), unless otherwise stated.} 
\[ H_G^*(M) := H^*(M_G) .\]
\end{defn}

We assume the action on the product \(M\times EG\) to be the diagonal action. Notice that, since \(G\) acts freely on \(EG\), the action on the product is automatically free. Indeed, if in the worst case \(p\in M\) is a fixed point, for every point \(e\in EG\),  \(g\cdot (p,e) = (p,g\cdot e)\neq (p,e)\). This means that the homotopy quotient defines a smooth manifold, and we can hope for a generalization of  de Rham's theorem, allowing to study this topological definition from its smooth structure in terms of something analogous to the differential forms on \(M\). We will discuss this result in the next section.

Since the space \(EG\) does not need to either exist or be unique a priori, one could think that the above definition contains some degree of arbitrariness,  so a natural question is: is equivariant cohomology well defined? The answer is of course yes, and the crucial fact allowing this stands in the contractibility of the space \(EG\). The arguments that lead to this conclusion are summarized in Appendix \ref{app:univ-bundles}, together with some examples of universal bundles. The important property that one has to keep in mind is that, intuitively, to get it acted freely by \(G\) and being contractible, one has to define it \lq\lq so big\rq\rq\ that for \emph{any} other principal \(G\)-bundle \(P\), there is a copy of \(P\) sitting inside \(EG\). This is why it is called \lq\lq universal\rq\rq.
Here we just notice that, if we assume a contractible free \(G\)-space \(EG\) to exist, we have an homotopy equivalence \(M\times EG \sim M\), that descends also to the homotopy quotient
\begin{equation}
\faktor{(M\times EG)}{G} \sim \faktor{M}{G} ,
\end{equation}
since one can show that \(M\times_G EG \to M/G\) is a fiber bundle with typical fiber \(EG\) \cite{tu-equiv_cohom}. From homotopy invariance of cohomology, we see that at least in the case in which \(G\) acts freely on \(M\) and \(M/G\) is well defined, the equivariant cohomology reduces to the naive definition above,
\begin{equation}
H_G^*(M) = H^*\left( M\times_G EG\right) \cong H^*\left( \faktor{M}{G}\right).
\end{equation}
Notice that the contractible space \(EG\) alone has a very simple cohomology. Indeed, for what we pointed out in Example \ref{ex:homology-contractibility}, it must be \(H^*(EG) \cong H^*(pt) \cong \mathbb{R}\) in degree zero. When we take the quotient, the base space \(BG:=EG/G\) can have a less trivial cohomology. This space is called \emph{classifying space} of the Lie group \(G\). When, for example, the \(G\)-action on \(M\) is trivial (all points are fixed points), the homotopy quotient is just \(M\times_G EG\cong M\times (EG/G) = M\times BG\), and in this case we have\footnote{This is an application of the so-called \emph{K\"{u}nneth theorem} \cite{hatcher}.}
\begin{equation}
\label{eq:eq-coho-trivial-action}
H^*_G(M) = H^*(M\times BG) \cong H^*(M)\otimes H^*(BG),
\end{equation}
so that the homotopy quotient by a trivial action does not bring any further information to the cohomology of \(M\) but for tensoring it with the cohomology of the classifying space. In Section \ref{sec:cartan} we will see that the latter can be described in general by a very simple algebraic model, while here we carry on the example of the case \(G=U(1)\).

\begin{ex}[A few \(U(1)\)-equivariant cohomologies]\label{ex:equivariant-cohomologies}
To search for a suitable principal \(U(1)\)-bundle whose total space is contractible, we can first notice from Example \ref{ex:group-actions} that we already described a class of principal \(U(1)\)-bundles, \(\mathbb{S}^{2n+1}\to \mathbb{C}P^n\), whose total spaces are the odd-dimensional spheres. The bad news is that any of these total spaces are contractible, but this problem can be solved considering the limiting case \(n\to\infty\), since it turns out that \(\mathbb{S}^\infty = \bigcup_n S^{2n+1} \equiv \bigcup_n S^n\) is contractible \cite{tu-equiv_cohom}!\footnote{Notice that any sphere \(\mathbb{S}^n\) can be embedded as the equator of \(\mathbb{S}^{n+1}\). Thus there is a sequence of inclusions \(\mathbb{S}^1\subset \mathbb{S}^3 \subset \cdots\) as well as \(\mathbb{C}P^1\subset \mathbb{C}P^3 \cdots\), and the circle action is compatible with the inclusion. Thus, in the limit, a free circle action induces on \(\mathbb{S}^\infty\).} Thus the universal bundle for \(U(1)\) can be chosen to be \(EU(1)=\mathbb{S}^\infty\), and the classifying space \(BU(1)=\mathbb{C}P^\infty\). Being infinite-dimensional, they are strictly speaking not manifolds, but \(\mathbb{S}^\infty\to \mathbb{C}P^\infty\) is still a topological bundle, and this is enough for the definition of an homotopy quotient. 
\begin{enumerate}[label=(\roman*)]
\item In Example \ref{ex:homology-contractibility} we quoted the resulting cohomology ring of the complex projective planes, \(H^*(\mathbb{C}P^n) \cong \mathbb{R}[\phi]/\phi^{n+1}\) and \(H^*(\mathbb{C}P^\infty)= H^*(BU(1)) \cong \mathbb{R}[\phi]\), with \(\phi\) in degree 2. Thus the \(U(1)\)-equivariant cohomology of any space \(M\) on which \(U(1)\) acts trivially is, from \eqref{eq:eq-coho-trivial-action}, \[H^*_{U(1)}(M)=H^*(M)\otimes \mathbb{R}[\phi].\] In particular if \(M\) is contractible, \(H^*_{U(1)}(M)=H^*_{U(1)}(pt) = \mathbb{R}[\phi]\).

\item The opposite case is the one of a free action, for example the circle acting on itself. As we pointed out above, \(\mathbb{S}^1/U(1)\cong pt\), so simply \(H^*_{U(1)}(\mathbb{S}^1)\cong \mathbb{R}\).

\item The last example that we mention is the case of \(U(1)\) acting on \(\mathbb{S}^2\). The equivariant cohomology \(H^*_{U(1)}(\mathbb{S}^2)\) is non-trivial a priori, as we remarked above, and it can be calculated easily for example using spectral sequences. We will not enter in the detail of the calculation but only describe the result. Consider first the standard cohomology of the 2-sphere that we already saw in various examples, being  \(H^*(\mathbb{S}^2) \cong \mathbb{R}\oplus \mathbb{R}y\), where we explicitly wrote a generator \(y\) for the term in degree 2, that can be identified in the de Rham model by a volume form \(y=[\omega], \omega\in \Omega^2(\mathbb{S}^2)\). It turns out that its equivariant version can be obtained simply by tensoring with the polynomial ring \(H^*(BU(1)) = \mathbb{R}[\phi]\),
\[ H^*_{U(1)}(\mathbb{S}^2) \cong \mathbb{R}[\phi]\oplus \mathbb{R}[\phi]y,\]
although the generator \(y\) has now a different interpretation, that we will give in terms of an equivariant version of the de Rham model in the next sections.\footnote{In terms of differential forms, \(\omega\) will have to be \emph{equivariantly extended} in the Cartan model, as discussed at the end of Section \ref{sec:cartan}.}
This equivariant cohomology actually can be given a ring structure, defining the multiplication \(y\cdot y = a \phi y + b \phi^2\) for some constants \(a,b\). It turns out \cite{tu-equiv_cohom} that the correct constants are \(a=1, b=0\), making \[
\faktor{\mathbb{R}[y,\phi]}{(y^2-\phi^2)} \to H_{U(1)}^*(\mathbb{S}^2)
\]
into a ring isomorphism, where the denominator stands for the ideal generated by the expression \((y^2-\phi^2)\) in \(\mathbb{R}[y,\phi]\). Notice that the cohomology groups \(H^n_{U(1)}(\mathbb{S}^2)\) are now non-empty in every even-degree (while in odd-degree they are all trivial), even when \(n\) is bigger than the dimension of the sphere! This intuitively matches the fact that the quotient \(\mathbb{S}^2/U(1)\) is singular, and thus simple dimensionality arguments do not make sense anymore at the fixed points.
\end{enumerate}
\end{ex}

\section{The Weil model and equivariant de Rham's theorem}
\label{sec:Weil-model}

From now on, we specialize the equivariant cohomological theory to \(G\) being a Lie group with Lie algebra \(\mathfrak{g}\),\footnote{Some of what follows is only rigorous if \(G\) is compact, but the formal discussion can be applied generically.} and \(M\) being a smooth \(G\)-manifold. We saw that  de Rham's theorem provides an \emph{algebraic model} for the singular cohomology (with real coefficients) of the smooth manifold \(M\), through the complex of differential forms. We now describe a way to obtain an algebraic model for the homotopy quotient \((M\times EG)/G\), the so-called \emph{Weil model} for the \(G\)-equivariant cohomology of \(M\). From the discussion of the last sections, it is already imaginable that this will contain in some way the de Rham complex of \(M\), but modifying it through a somewhat \lq\lq trivial\rq\rq\ extension, in the sense of the triviality of the cohomology of \(EG\). This is thus the most natural model that is connected to the topological definition of the last section, but we will see that it is also overly complicated. In fact, in the next section we will describe a simpler but equivalent way to obtain the same equivariant cohomology, the \emph{Cartan model}, that is more intuitive from the differential geometry point of view, and that we will use to generalize the theory of integration to the equivariant setting. This is what we often use in physics for practical calculations.
%This will lead to the main subject of this work, the so called \emph{Atiyah-Bott-Berline-Vergne localization theorem}.

Before defining the Weil model, we notice that, in presence of a \(G\)-action, the de Rham complex \((\Omega(M),d)\) of differential forms on \(M\) has more structure than being a dg algebra. In fact, if \(\rho:G\to\text{Diff}(M)\) is the \(G\)-action, this induces an infinitesimal action of the Lie algebra \(\mathfrak{g}\) on any tensor space via the Lie algebra homomorphism\footnote{Usually, in physics conventions, in the action of the exponential map one collects a factor of \(i\) at the exponent, in order to consider the Lie algebra element Hermitian for the most commonly considered group actions. Left and right actions should be taken with different signs at the exponent.}
\begin{equation}
\label{eq:fundam-vf}
\begin{aligned}
\mathfrak{g} &\to \Gamma(TM) \\
X &\mapsto \underline{X} := \left.\frac{d}{dt}\right|_{t=0} \rho\left(e^{-tX}\right)^*
\end{aligned}
\end{equation}
that defines for any \(X\in\mathfrak{g}\) the corresponding \emph{fundamental vector field} \(\underline{X}\in \Gamma(TM)\). Then \(\mathfrak{g}\) acts infinitesimally on \(\Omega(M)\) via the \emph{Lie derivative} and the \emph{interior multiplication} with respect to the fundamental vector fields,
\begin{equation}
\begin{aligned}
\mathcal{L}_X (\alpha) &:= \mathcal{L}_{\underline{X}} (\alpha) \\
\iota_X \alpha &:= \iota_{\underline{X}} \alpha
\end{aligned} 
 \qquad \text{for}\ X\in \mathfrak{g}, \alpha\in \Omega(M) ,
\end{equation}
with the additional property (\emph{Cartan's magic formula})
\begin{equation}
\mathcal{L}_X = d\circ \iota_X + \iota_X \circ d .
\end{equation}
This makes \(\Omega(M)\) into a so-called \emph{\(\mathfrak{g}\)-gd algebra}. In general, a \(\mathfrak{g}\)-gd algebra is defined as a differential graded algebra (cf. definition \ref{def:dga}) with two actions of \(\mathfrak{g}\), denoted by analogy as \(\iota\) and \(\mathcal{L}\), such that for any \(X\in\mathfrak{g}\)
\begin{enumerate}[label=(\roman*)]
\item \(\iota_X\) acts as an \emph{antiderivation} of degree \(-1\), satisfying \((\iota_X)^2= 0\);
\item \(\mathcal{L}_X\) acts as a derivation (of degree 0);
\item the Cartan's magic formula holds: \(\mathcal{L}_X = d\circ \iota_X + \iota_X \circ d\).
\end{enumerate}
Morphisms of \(\mathfrak{g}\)-dg algebras are naturally defined as maps between \(\mathfrak{g}\)-dg algebras that commute with all the above stated operations. It is not difficult to show that \(G\)-equivariant maps of \(G\)-manifolds induce pull-backs of differential forms that preserve the \(\mathfrak{g}\)-dg algebra structure.

Now we can define the Weil model via an extension of \((\Omega(M),d)\) that preserves this new structure. We want this extension to be an algebraic analog of \(EG\), so its cohomology must be trivial, but carrying information about \(\mathfrak{g}\). To do this, we associate it to the characteristic differential structure of a generic principal \(G\)-bundle \(P\) (remember that any principal \(G\)-bundle sits inside \(EG\)): its connection 1-form \(A\in \Omega^1(P)\otimes \mathfrak{g}\) and the associated curvature \(F=dA + \frac{1}{2}[A\stackrel{\wedge}{,} A]\in \Omega^2(P) \otimes \mathfrak{g}\), satisfying the Bianchi identity \(dF = [F,A]\). We first notice that the connection 1-form and the curvature 2-form  can be seen as linear maps
\begin{equation}
\begin{aligned}
A: \mathfrak{g}^* &\to \Omega^1(P) \\ \eta &\mapsto A(\eta):= (\eta\circ A) ,
\end{aligned}  \qquad
\begin{aligned}
F: \mathfrak{g}^* &\to \Omega^2(P) \\ \eta &\mapsto F(\eta):=(\eta\circ F) .
\end{aligned}
\end{equation}
These maps can be extended multi-linearly to the whole \(\Omega(P)\), if we start from algebras constructed by \(\mathfrak{g}^*\) that respect the commutativity of 1- and 2-forms, respectively. This means that \(A\) has to \lq\lq eat\rq\rq\ an element of the (anticommutative) exterior algebra \(\bigwedge(\mathfrak{g}^*)\), while \(F\) has to \lq\lq eat\rq\rq\ an element of the (commutative) symmetric algebra \(S(\mathfrak{g}^*)\):
\begin{equation}
\begin{aligned}
A: \bigwedge(\mathfrak{g}^*) &\to \Omega(P) \\ \eta_1\wedge\cdots\wedge\eta_k &\mapsto A(\eta_1)\wedge\cdots\wedge A(\eta_k) ,
\end{aligned}  \qquad
\begin{aligned}
F: S(\mathfrak{g}^*) &\to \Omega(P) \\ \eta_1\cdots\eta_k &\mapsto F(\eta_1)\wedge\cdots\wedge F(\eta_k) .
\end{aligned}
\end{equation}
We can combine the two maps in the homomorphism of graded-algebras
\begin{equation}
\begin{aligned}
f: S(\mathfrak{g}^*) \otimes \bigwedge(\mathfrak{g}^*) &\to \Omega(P) \\
\eta \otimes \xi &\mapsto F(\eta)\wedge A(\xi) .
\end{aligned}
\end{equation}
This captures the fact that a connection of \(P\) could be \emph{defined} as a map \(S(\mathfrak{g}^*) \otimes \bigwedge(\mathfrak{g}^*) \to \Omega(P)\), and motivates the following definition.
\begin{defn} The \emph{Weil algebra} of \(\mathfrak{g}\) is the graded algebra \[
W(\mathfrak{g}) := S(\mathfrak{g}^*) \otimes \bigwedge(\mathfrak{g}^*) \]
and the map \(f: W(\mathfrak{g})\to \Omega(P)\) is called the \emph{Weil map}. We define the graded structure of \(W(\mathfrak{g})\) by assigning to the generators \(\lbrace \phi^a\rbrace\) of \(S(\mathfrak{g}^*)\) degree \(\textrm{deg}(\phi^a)=2\), and to the generators \(\lbrace \theta^a\rbrace\) of \(\bigwedge(\mathfrak{g}^*)\) degree \(\textrm{deg}(\theta^a)=1\).
\end{defn}

The generators \(\lbrace \phi^a, \theta^a\rbrace\) are two copies of a basis set for \(\mathfrak{g}^*\), but taken in different degrees. With respect to this \emph{graded} basis, the Weil algebra can also be written as \begin{equation}
W(\mathfrak{g}) = \bigwedge \left( \mathbb{R}[ \phi^1,\cdots,\phi^{\dim\mathfrak{g}}] \oplus \mathbb{R}[\theta^1,\cdots,\theta^{\dim{\mathfrak{g}}}] \right) ,
\end{equation}
and a generic element will be expanded as
\begin{equation}
\label{tu-3}
\alpha = \alpha_0 + \alpha_a\theta^a + \frac{1}{2}\alpha_{ab}\theta^a\theta^b + \cdots + \alpha^{(top)} \theta^1\theta^2\cdots\theta^{\dim\mathfrak{g}} \quad \text{with} \quad \alpha_I \in \mathbb{R}[\phi^1,\cdots,\phi^{\dim\mathfrak{g}}] ,
\end{equation}
since higher order terms vanish by the anticommutativity of the \(\theta\)'s. Here we suppressed tensor and wedge products to simplify the notation, as we will often do in the following.
On this basis, the Weil map projects simply the connection and the curvature on the given Lie algebra components,
\begin{equation}
f(\theta^a) = \theta^a\circ A = A^a , \qquad f(\phi^a) = \phi^a\circ \Omega = \Omega^a .
\end{equation}
The Weil algebra is the central object to define an algebraic model for \(EG\). We need to define a \(\mathfrak{g}\)-dg algebra structure on it to properly take its cohomology, but this is naturally done requiring the Weil map to be a morphism of \(\mathfrak{g}\)-dg algebras. This means introducing a differential \(d_W :W(\mathfrak{g})\to W(\mathfrak{g})\) and two \(\mathfrak{g}\)-actions \(\mathcal{L}, \iota\) such that the following diagram commutes for all the three operations separately,
%\arrow{d}{d_W, \mathcal{L}, \iota } \arrow{d}{d, \mathcal{L}, \iota}\arrow{r}{f}
\begin{equation}
\label{diagram-Weil}
\begin{tikzcd}
W(\mathfrak{g}) \arrow[d, "d_W\ \iota\ \mathcal{L}"] \arrow[r,"f"]  & \Omega(P) \arrow[d, "d\ \iota\ \mathcal{L}"]  \\
W(\mathfrak{g}) \arrow[r,"f"] & \Omega(P) .
\end{tikzcd}
\end{equation}
One can check that, defining the \emph{Weil differential} \(d_W\) on the generators as\footnote{In the second column we introduced \(\theta := \theta^i\otimes T_i\) and \(\phi:= \phi^i \otimes T_i\) in \(W(\mathfrak{g})\otimes \mathfrak{g}\), where \(\lbrace T_i\rbrace\) is a basis of \(\mathfrak{g}\) dual to the generators. Notice that this notation make the formulas independent on a choice of basis. Also, these are the objects that are really correspondent to the connection \(A\) and the curvature \(F\) on \(P\), respectively.}
\begin{equation}
\label{eq:weil-differential} 
\begin{aligned}
d_W \theta^a &= \phi^a - \frac{1}{2}f^a_{bc} \theta^b \theta^c \\
d_W \phi^a &= f^a_{bc} \phi^b \theta^c
\end{aligned} \quad \text{or} \quad 
\begin{aligned}
d_W \theta &= \phi - \frac{1}{2} [\theta \stackrel{\wedge}{,} \theta] \\
d_W \phi &= [\phi \stackrel{\wedge}{,} \theta]
\end{aligned}
\end{equation}
where \(f_{ab}^c\) are the structure constants of \(\mathfrak{g}\), it commutes with \(f\) giving correctly the definition of curvature and the Bianchi identity. Moreover, extending the differential on \(W(\mathfrak{g})\) as an antiderivation of degree \(+1\), it gives \(d_W^2 = 0\) (since \(d_W^2\) is a derivation, it is enough to check it on the generators). To be compatible with the properties of the connection and the curvature
\begin{equation}
\iota_X A =A(\underline{X}) = X , \qquad \iota_X F = 0  \qquad \forall X\in\mathfrak{g},
\end{equation}
the interior multiplication must be defined as 
\begin{equation}
\label{tu-interior-mult}
\begin{aligned}
\iota_X \theta^a &:= \theta^a(X) = X^a \\
\iota_X \phi^a &:= 0
\end{aligned} \quad \text{or} \quad 
\begin{aligned}
\iota_X \theta &:= \theta(X) = X \\
\iota_X \phi &:= 0
\end{aligned} 
\end{equation}
and extended as an antiderivation of degree \(-1\). Then the Lie derivative is simply defined via Cartan's magic formula. We finally have defined the Weil algebra as a \(\mathfrak{g}\)-dg algebra.

\begin{thm} \label{thm:Weil-cohomology}
 The cohomology of the Weil algebra is \[ 
H^0(W(\mathfrak{g}),d_W) \cong \mathbb{R} , \qquad H^k(W(\mathfrak{g}),d_W) \cong 0\ \text{for}\ k>0 .\]
\end{thm}
\begin{proof}
The full proof can be found in \cite{tu-equiv_cohom}. Schematically it follows the proof of the Poincaré lemma: one has to find an \emph{cochain homotopy}, \textit{i.e.}\ a map \(K:W(\mathfrak{g})\to W(\mathfrak{g})\) of degree -1, such that \([K,d_W]_+ = id\). Then any cocycle (\(d_W \alpha = 0\)) is also a coboundary, since \(\alpha = [K,d_W]_+ \alpha = d_W(K\alpha)\). This can be found for any degree \(k>0\). In degree zero \(W^0(\mathfrak{g})\cong \mathbb{R}\) by definition, so every element is a cocycle, and no one is a coboundary for degree reasons.
\end{proof}

\begin{ex}[Weil model for torus and circle actions] \label{ex:U(1)-Weil-model}
Consider the case of a compact Abelian group, \textit{i.e.}\ a torus \(T=U(1)^l\) for some \(l\). Remember that a possible purpose of the Weil algebra is to describe the connection and the curvature of any principal \(T\)-bundle, so we are somewhat analyzing the structure of \(l\) electromagnetic fields, from the point of view of the Lie algebra \(\mathfrak{t}\). Since the structure constants are all zero, the Weil differential \eqref{eq:weil-differential} and the \(\mathfrak{t}\)-actions \eqref{tu-interior-mult} on the generators \((\theta^a,\phi^a)\) simplify as \[
\begin{array}{ll}
d_W\theta^a = \phi^a , & d_W\phi^a =0 ,\\
\iota_b \theta^a = \delta^a_b, & \iota_b \phi^a = 0, \\
\mathcal{L}_b \theta^a = 0, & \mathcal{L}_b \phi^a=0 ,
\end{array}\]
where we denoted \(\iota_a\equiv \iota_{T_a}\)  and \(\mathcal{L}_a \equiv \mathcal{L}_{T_a}\), with \(\lbrace T_a\rbrace\) the basis of \(\mathfrak{t}\) dual to the generators of \(W(\mathfrak{t})\). We jump ahead a little and notice that the first line really resembles the structure of a \lq\lq supersymmetry\rq\rq\ transformation, with \(\phi^a\) being the \lq\lq bosonic partner\rq\rq\ of \(\theta^a\). The remaining non-Abelian piece of the generic case can be viewed as the action of a Chevalley-Eilemberg differential, so that \(d_W = d_{susy} + d_{CE}\).\footnote{Remember that the C-E differential is the one that appears in BRST quantization of gauge theories.} We will return to this point in Section \ref{sec:BRST-cohom-equiv}, after having introduced some technology about supergeometry and supersymmetry.

Let us simplify again and prove theorem \ref{thm:Weil-cohomology} for \(l=1\). In the case of \(U(1)\), the Lie algebra has only one generator \(T\cong i\), and the symmetric algebra is the algebra of polynomials in the indeterminate \(\phi\in \mathfrak{g}^*\), \(S(\mathfrak{g}^*) = \mathbb{R}[\phi]\), while the exterior algebra reduces to \(\bigwedge(\theta) = \mathbb{R}\oplus \mathbb{R}\theta\) by anticommutativity. The Weil algebra is thus \[
W(\mathfrak{u}(1)) = \mathbb{R}[\phi] \otimes \left( \mathbb{R}\oplus \mathbb{R}\theta\right) = \mathbb{R}[\phi] \oplus \mathbb{R}[\phi]\theta .
\]
The cohomology of \(W(\mathfrak{u}(1))^0 = \mathbb{R}\) in degree zero is as always trivial, since all constant numbers are closed, and none of them is exact, giving \(H^0(W(\mathfrak{u}(1)),d_W)\cong\mathbb{R}\). In degree 1 we have \(W(\mathfrak{u}(1))^1 = \mathbb{R}\theta\), thus no one element (but zero) is closed. This extends to any odd-degree, since \(W(\mathfrak{u}(1))^{2n+1} = \mathbb{R}\phi^n \theta\), and \(d_W(\phi^n \theta) = \phi^{n+1} \neq 0\). This means that \(H^{2n+1}(W(\mathfrak{u}(1)),d_W)\cong 0\). In degree 2, we have \(W(\mathfrak{u}(1))^2 = \mathbb{R}\phi\), so that any element is closed but also exact, since \(\phi =d_W\theta\). This extends to any even-degree, since \(W(\mathfrak{u}(1))^{2n} = \mathbb{R}\phi^n\), and \(\phi^n = d_W\theta \phi^{n-1} = d_W(\theta\phi^{n-1})\). Thus we have also \(H^{2n}(W(\mathfrak{u}(1)),d_W)\cong 0\), showing the triviality of the Weil algebra in the simplest case of a circle action. Almost the same direct computation can be carried out in the \(l\)-dimensional case.
\end{ex}

Theorem \ref{thm:Weil-cohomology} shows that we are in business: the Weil algebra is exactly an algebraic analog of the universal bundle \(EG\). Since the de Rham model for \(M\) is just \(\Omega(M)\), the product \(M\times EG\) can be modeled by the  complex \(W(\mathfrak{g})\otimes \Omega(M)\), since by the K\"{u}nneth formula \cite{hatcher} and de Rham's theorem
\begin{equation}
H^*(EG\times M) = H^*(EG)\otimes H^*(M) = H^*(W(\mathfrak{g}^*),d_W)\otimes H^*(\Omega(M),d).
\end{equation}
The differential and the \(\mathfrak{g}\)-actions are extended naturally on this complex as graded derivations, making it into a \(\mathfrak{g}\)-dg algebra too. Explicitly,
\begin{equation}
\label{tu-gactions-Weil}
\begin{aligned}
d_T &:= d_W \otimes 1 + 1 \otimes d , \\
\iota &\equiv \iota \otimes 1 + 1\otimes \iota , \\
\mathcal{L} &\equiv \mathcal{L}\otimes 1 + 1\otimes \mathcal{L}.
\end{aligned}
\end{equation}
A model for the homotopy quotient \(M_G\) can be guessed by the following argument. Since \(M_G\) is the base of the principal bundle \(EG\times G\to M_G\), differential forms on \(M_G\) identify the \emph{basic forms} on \(EG\times M\) (see Appendix \ref{app:diff_geo}), \textit{i.e.}\ those that are both \emph{\(G\)-invariant} and \emph{horizontal}. It is thus reasonable that the homotopy quotient can be modeled by the \emph{basic subcomplex} of the Weil model. Since the differential closes on the basic subcomplex, we are allowed to take its cohomology, giving the \(G\)-equivariant cohomology of \(M\). This is exactly the content of the \emph{equivariant de Rham's theorem}. A recent original proof of it can be found in \cite{tu-equiv_cohom}.
\begin{thm}[equivariant de Rham]
If \(G\) is a connected Lie group, and \(M\) is a \(G\)-manifold, \[
\boxed{
H^*_G(M) \cong H^*\left( \left( W(\mathfrak{g})\otimes \Omega(M) \right)_{bas}, d_T \right) }. \] 
\end{thm}

The equivariant de Rham's theorem is telling us that the \lq\lq correct\rq\rq\ differential complex that encodes the topology of the \(G\)-action on \(M\) is \emph{not} anymore the complex \(\Omega(M)\) of differential forms, but a modification of it through the presence of the Weil algebra. Remember always that, via the Weil map, \(W(\mathfrak{g})\) can be thought as in correspondence with the presence of a connection and a curvature on some principal \(G\)-bundle. This means that the right extension of the de Rham complex in presence of a \(G\)-action embeds somewhat the presence of a connection and a curvature with respect to \(G\). We can analyze as an example the simplest case of a \(U(1)\)-action. The unrestricted Weil model is (from Example \ref{ex:U(1)-Weil-model})
\begin{equation}
W(\mathfrak{u}(1))\otimes \Omega(M) = \Omega(M)[\phi] \oplus \Omega(M)[\phi]\theta ,
\end{equation}
thus any element can be written as 
\begin{equation}
\alpha = \alpha^{(0)} + \alpha^{(1)} \theta ,
\end{equation}
where \(\alpha^{(0)},\alpha^{(1)}\in \Omega(M)[\phi]\) are polynomials in \(\phi\) with differential forms as coefficients. The subcomplex of basic forms consists of those elements that satisfy both \(\iota_T\alpha = 0\) and \(\mathcal{L}_T\alpha = 0\), imposing the conditions
\begin{equation}
\alpha^{(1)} = -\iota_T \alpha^{(0)} ,\qquad \mathcal{L}_T \alpha^{(0)} = 0 .
\end{equation}
Thus any basic element can be written as \(\alpha = (1-\theta\iota_T) \sum_i  (\phi)^i \alpha^{(0)}_i \), where all the differential forms \(\alpha^{(0)}_i\in \Omega(M)^G\) must be \(G\)-invariant, and the basic subcomplex can be identified with the polynomials in \(\phi\) with invariant differential forms as coefficients. In the next section we will argue that this is not a special case, and that the Weil model can be simplified in general, producing another model for the same equivariant cohomology.

\section{The Cartan model}
\label{sec:cartan}

As we said at the beginning of the last section, the cohomology of the Weil complex is not the unique algebraic model for the \(G\)-equivariant cohomology of the \(G\)-manifold \(M\). Moreover, although very transparent, the Weil model seems overly complicated for differential geometric applications. In fact, the extreme simplicity of the basic subcomplex of the Weil algebra \(W(\mathfrak{g})_{bas}\) suggests that a simpler model for equivariant cohomology can be obtained simplifying this one. To see this, we analyze this basic subcomplex first. As we recalled at the end of the last section (for further details see Appendix \ref{app:diff_geo}), a basic element \(\alpha\in W(\mathfrak{g})_{bas}\) is both \emph{horizontal} and \emph{invariant}, \textit{i.e.}\
\begin{equation}
\iota_X \alpha = 0 = \mathcal{L}_X\alpha .
\end{equation} 
The horizontal condition means that we pick only the symmetric algebra inside \(W(\mathfrak{g})\), since by definition \(\iota_X \phi^a = 0\). Imposing also the \(G\)-invariance we have
\begin{equation}
W(\mathfrak{g})_{bas} \cong S(\mathfrak{g}^*)^G ,
\end{equation}
\textit{i.e.}\ the basic subcomplex is the algebra of Casimir invariants. It is easy to check that on this subcomplex \(d_W \cong 0\), so that
\begin{equation}
H^*(W(\mathfrak{g})_{bas}, d_W) =  H^*(S(\mathfrak{g}^*)^G, d_W) = S(\mathfrak{g}^*)^G \quad \text{in degree 0,}
\end{equation}
since every element is closed, and no one element can be exact. Moreover, from the equivariant de Rham's theorem \(H_G^*(pt)\cong H^*(W(\mathfrak{g})_{bas}, d_W)\), so the Casimir invariants are precisely the cohomology of the classifying space,
\begin{equation}
H^*(BG) \cong S(\mathfrak{g}^*)^G .
\end{equation}

Motivated by the above simplification, we can turn now to analyze the complete Weil model \((W(\mathfrak{g})\otimes \Omega(M))_{bas}\). Let us see concretely what it means to restrict the attention to a basic element \(\alpha \in (W(\mathfrak{g})\otimes \Omega(M))_{bas}\), starting from its expansion on a basis \eqref{tu-3}. Imposing the horizontality condition means, using the multi-index notation \(I=(a_1,\cdots,a_{|I|})\),
\begin{equation}
0= \iota_X\alpha = \iota_X\left(\alpha_0 + \frac{1}{|I|!}\alpha_I \theta^I \right) \quad \Rightarrow \quad 0= \iota_X\alpha_0 + \frac{1}{|I|!}(\iota_X \alpha_I)\theta^I + \frac{1}{|I|!}\alpha_I (\iota_X \theta^I) .
\end{equation}
Equating the terms of the same degree and taking \(X\) to be a basis element, one arrives at the condition on the various components,
\begin{equation}
\alpha_{a_1 \cdots a_{|I|}} = (-1)^{|I|} \iota_{a_1}\cdots\iota_{a_{|I|}} \alpha_0
\end{equation}
where we denoted \(\iota_k := \iota_{T_k}\), meaning that a horizontal element is fully determined by its first component \(\alpha_0\in S(\mathfrak{g}^*)\otimes \Omega(M)\), and it can be expressed as
\begin{equation}
\alpha = \left( \prod_{k=1}^{\dim\mathfrak{g}} (1-\theta^k\iota_k)\right) \alpha_0 .
\end{equation}
This comment, with some more checks (see again \cite{tu-equiv_cohom} for a complete proof), proves the following theorem, and extends the above discussion to the complete Weil model.
\begin{thm}[Mathai-Quillen isomorphism] 
\label{thm:Mathai-Quillen-iso} 
There is an isomorphism of \(\mathfrak{g}\)-dg algebras, called the \emph{Mathai-Quillen isomorphism} \cite{mathai-quillen} (or \emph{Cartan-Weil} in \cite{tu-equiv_cohom}), \[
\begin{aligned}
\varphi: (W(\mathfrak{g})\otimes \Omega(M))_{hor} &\to S(\mathfrak{g}^*)\otimes \Omega(M) \\
\alpha = \alpha_0 + \frac{1}{|I|!}\alpha_I \theta^I &\mapsto \alpha_0 \\
 \left( \prod_{k=1}^{\dim\mathfrak{g}} (1-\theta^k\iota_k)\right) \alpha_0 &\mapsfrom \alpha_0 .
\end{aligned}
\]
\end{thm}

The RHS of the isomorphism above inherits the \(\mathfrak{g}\)-actions and the differential from the Weil model on the LHS, making commutative the following diagram, similarly to \eqref{diagram-Weil},
\begin{equation}
\begin{tikzcd}
(W(\mathfrak{g})\otimes \Omega(M))_{hor} \arrow[d, "d_T\ \iota\ \mathcal{L}"] \arrow[r,"\varphi"]  & S(\mathfrak{g}^*)\otimes \Omega(M) \arrow[d, "d_C\ \iota\ \mathcal{L}"]  \\
(W(\mathfrak{g})\otimes \Omega(M))_{hor} \arrow[r,"\varphi"] & S(\mathfrak{g}^*)\otimes \Omega(M) .
\end{tikzcd}
\end{equation}
In particular, the new differential is called \emph{Cartan differential}, defined such that
\begin{equation}
d_C := \varphi \circ d_T \circ \varphi^{-1} .
\end{equation}
It is not difficult to get, directly from this definition, that it can be expressed more simply as
\begin{equation}
\label{eq:Cartan-differential}
d_C = 1\otimes d - \phi^k  \otimes \iota_k
\end{equation}
where \(d\) is the de Rham differential on \(\Omega(M)\). The two \(\mathfrak{g}\)-actions commute with \(\varphi\) without modification, so they agree with their behavior in the Weil model,
\begin{equation}
\mathcal{L}_X \phi^a = f^a_{bc} \phi^b X^c , \qquad \iota_X \phi^a =0 .
\end{equation}
Using the nilpotency of \(d\) and \(\iota\), and Cartan's magic formula, wee see that the Cartan differential on the horizontal subcomplex squares to a Lie derivative (an infinitesimal symmetry transformation)
\begin{equation}
d_C^2 = - \phi^k \otimes \mathcal{L}_k ,
\end{equation}
so when we restrict to the \(G\)-invariant subspace, \(d_C^2 \cong 0 \) on \(\left(S(\mathfrak{g}^*)\otimes \Omega(M)\right)^G\), as it should. Using the Mathai-Quillen isomorphism and the equivariant de Rham theorem we then have the fundamental result,
\begin{equation}
\boxed{
H_G^*(M) \cong H^*\left( \left(S(\mathfrak{g}^*)\otimes \Omega(M)\right)^G, d_C \right)
}
\end{equation}
that simplifies the algebraic model for the \(G\)-equivariant cohomology of \(M\).
\begin{defn} The \(\mathfrak{g}\)-dg algebra \[
\Omega_G(M) := \left(S(\mathfrak{g}^*)\otimes \Omega(M)\right)^G
\] with the Cartan differential \(d_C\) is called the \emph{Cartan model} for the equivariant cohomology of \(M\). Elements of \(\Omega_G(M)\) are called \emph{equivariant differential forms} on \(M\). The degree of an equivariant form is the total degree with respect to the generators of \(\Omega(M)\) (in degree 1), and the generators of \(S(\mathfrak{g^*})\) (in degree 2).
\end{defn}

Equivariant forms will be in the next chapters the the principal object of study. In the physical applications we are interested in, we will always search for an interpretation of the space of interest as a Cartan model with respect to the action of a symmetry group \(G\). The Cartan differential will be some object that squares to an infinitesimal symmetry, and on the subspace of \(G\)-invariant forms (or \lq\lq fields\rq\rq, in the following) it will define a \(G\)-equivariant cohomology. Cartan differentials arise in Field Theory as \emph{supersymmetry transformations}, that we will contextualize in Chapter \ref{cha:susy} and relate to equivariant cohomology in Chapter  \ref{cha:loc-susy}. This interpretation will be crucial in treating some of the most important objects in QM and QFT that arise as (infinite-dimensional) path integrals over the space of fields. In fact, we will see in the next chapter that integration of equivariant forms leads to powerful \emph{localization thorems}, that formally extended to the infinite-dimensional case greatly simplifying those integrals.

\begin{ex}[Cartan model for \(U(1)\)-equivariant cohomology]\label{ex:U(1)-cohom}
Until Chapter \ref{cha:non Abelian YM}, we will actually deal with the equivariant cohomology with respect to a circle action of \(G=U(1)\), or at most a torus action of \(U(1)^n\) for some \(n\). As we saw also for the Weil model, this greatly simplifies the problem, so we carry on that example also in the Cartan model for a \(U(1)\)-action. We recall from Example \ref{ex:U(1)-Weil-model} that in the Weil model \[
\iota_T \phi = \mathcal{L}_T \phi = d_W \phi = 0 ,
\]
\textit{i.e.}\ \(\phi\) is automatically also \(U(1)\)-invariant. The equivariant forms are thus \[
\Omega_{U(1)}(M) = (\mathbb{R}[\phi]\otimes \Omega(M))^{U(1)} \cong \Omega(M)^{U(1)}[\phi] ,
\]
so polynomials in \(\phi\) with \(U(1)\)-invariant forms as coefficients. The Cartan differential is, suppressing tensor products, \[
d_C = d - \phi\ \iota_T . \]
In this 1-dimensional case, the indeterminate \(\phi\) is just a spectator, and serves only to properly count the equivariant form-degree. This is important of course, but for many purposes it creates no confusion to suppress its presence. More precisely, we often \emph{localize} the algebra \(\Omega_G(M)\), substituting the \emph{indeterminate} \(\phi\) with a \emph{variable}, and setting it for example to \(\phi=-1\),\footnote{This could seem harmless, but it is definitely a non-trivial move. We are really able to do this without spoiling the resulting equivariant cohomology because (algebraic) localization commutes with taking cohomology. More details on this are reported in Appendix \ref{app:borel-loc}.} so that \[
d_C = d + \iota_T .\]
This differential squares to an infinitesimal symmetry generated by \(\underline{T}\), \(d_C^2 = \mathcal{L}_T\). Equivariant differential forms after this localization are just \(U(1)\)-invariant forms.

Often it is useful to generate equivariant forms from invariant differential forms in \(\Omega(M)\), for the purpose of integration for example. If \(\alpha\in \Omega^{2n}(M)\), an \emph{equivariant extension} of \(\alpha\) is \(\tilde{\alpha}\in \Omega_{U(1)}(M)\) such that \[
\tilde{\alpha} = \alpha + f_{(2n-2)} \phi + f_{(2n-4)}\phi^2 + \cdots \]
where any coefficient is an invariant form in \(\Omega(M)\).
As an example, we can take the circle acting on the 2-sphere \(\mathbb{S}^2\), via rotations around a chosen axis. If \(\theta\) is the polar coordinate and \(\varphi\) is the azimutal coordinate, \[
(\theta,\varphi)\left(e^{it}\cdot p \right) := (\theta(p), \varphi(p) + t) \qquad \text{for}\ p\in \mathbb{S}^2, e^{it}\in U(1), \]
so that the fundamental vector field is \(\underline{T} =  \frac{\partial}{\partial\varphi}\).\footnote{Recall that this action has two fixed points, at the North and the South pole.} Consider the canonical volume-form \(\omega = d\cos{(\theta)}\wedge d\varphi\). It is obviously closed, and also \(U(1)\)-invariant, since \(\mathcal{L}_T \omega = 0\). Aiming to the extension of \(\Omega(M)\) to \(\Omega_{U(1)}(M)\), we can find an \emph{equivariantly closed extension} of the volume form \(
\tilde{\omega} = \omega + f \phi \),
with \(f\in C^\infty(M)\) such that \(d_C \tilde{\omega} = 0\), so that it is closed in the \lq\lq correct\rq\rq\ complex. This imposes the equation \(df = \iota_T \omega\),\footnote{\(\omega\) is a \emph{symplectic} form on \(\mathbb{S}^2\), and the \(df = \iota_T \omega\) means that \(H:= -f\) is the \emph{Hamiltonian function} with respect to the \(U(1)\)-action on the sphere. We will deepen this point of view in the next chapter.} and so \(f = - \cos(\theta)\): \[
\tilde{\omega} = \omega - \cos(\theta)\phi .\]
\end{ex}

\section{The BRST model}
\label{sec:BRST-model}

In this section we mention the last popular model for equivariant cohomology: the so-called \emph{BRST model}, or sometimes \emph{intermediate model}. It is worth to mention it because we will see in the next chapter that it is (as the first name suggests) intimately related to the BRST method for gauge-fixing in the Hamiltonian formalism. Moreover, its complex is the one that arises naturally in Topological Field Theories (TFT), as we will mention in Chapter \ref{cha:non Abelian YM}. It is also important because it provides (as the second name suggests) an \lq\lq interpolation\rq\rq\ between the Weil and the Cartan models that we saw in the last sections, relating the latter more \lq\lq physical\rq\rq (or differential geometric) point of view with the former more \lq\lq topological\rq\rq\ one. 

As an algebra, the (unrestricted) complex of the BRST model is identical to that of the Weil model,
\begin{equation}
B := W(\mathfrak{g})\otimes \Omega(M) ,
\end{equation}
but with the new differential (compare to \eqref{eq:weil-differential} and \eqref{tu-gactions-Weil})
\begin{equation}
\label{eq:BRST-diff}
d_B = d_W \otimes 1 + 1 \otimes d + \theta^a \otimes \mathcal{L}_a - \phi^a \otimes \iota_a 
\end{equation}
that satisfies \(d_B^2 = 0\) on \(B\), and has the same trivial cohomology of the unrestricted Weil model.

The idea that brought to the construction of this model in \cite{kalkman-BRST}, was essentially to prove along the line we did in the last section the equivalence of the models, but from a slightly different point of view. In fact one can construct \(d_B\) using an algebra automorphism that carries the Weil model \((B,d_W)\) into the BRST model \((B,d_B)\), at the level of the unrestricted algebras. The restriction to the basic subcomplex gives then automatically the Cartan model. The automorphism is given by the map 
\begin{equation}
\varphi := e^{\theta^a \iota_a} \equiv \prod_a (1 + \theta^a \otimes \iota_a ) ,
\end{equation}
that looks very similar to the Mathai-Quillen isomorphism of theorem \ref{thm:Mathai-Quillen-iso}, but now is applied to the whole algebra and not only on the horizontal part. Analogously to the definition of the Cartan differential, \(d_B\) is got as \eqref{eq:BRST-diff} from the commutativity of the diagram
\begin{equation}
\begin{tikzcd}
B \arrow[d, "d_W"] \arrow[r,"\varphi"]  & B \arrow[d, "d_B"]  \\
B \arrow[r,"\varphi"] & B ,
\end{tikzcd}
\end{equation}
so that \(d_B = \varphi^{-1} \circ d_W \circ \varphi\), as well as the two \(\mathfrak{g}\)-actions. In particular, it results
\begin{equation}
\begin{aligned}
\iota^{(B)} &= \iota \otimes 1 \neq \iota^{(W)} , \\
\mathcal{L}^{(B)} &= \mathcal{L}\otimes 1 + 1\otimes \mathcal{L} = \mathcal{L}^{(W)},
\end{aligned}
\end{equation}
where we called \(\iota^{(W)}, \mathcal{L}^{(W)}\) the one defined in \eqref{tu-gactions-Weil}. Thus the BRST differential carries the same information of the Weil differential, giving the same trivial cohomology of the unrestricted Weil model,
\begin{equation}
H^*(B,d_B) \cong H^*(B,d_W) \cong H_{dR}(M)
\end{equation}
where the last equivalence follows from the triviality of the cohomology of the Weil algebra \(W(\mathfrak{g}^*)\). Of course, we have to restrict the the action of \(d_B\) to the basic subcomplex, \textit{i.e.} to the intersection with the kernels of \(\iota^{(B)}\) and \(\mathcal{L}^{(B)}\), to get a meaningful \(G\)-equivariant cohomology. This reproduces again the Cartan model, as expected.

This result shows that there is in fact a whole continuous family of \(\mathfrak{g}\)-dg algebras that give equivalent models for the \(G\)-equivariant cohomology of \(M\), because we can conjugate the Weil differential through the modified automorphism
\begin{equation}
\varphi_t :=  e^{t \theta^a \iota_a} \qquad \text{with} \ t\in\mathbb{R} .
\end{equation}
This produces, by conjugation, the family of differentials and \(\mathfrak{g}\)-actions on \(B\),
\begin{equation}
\begin{aligned}
d^{(t)} &= d_W \otimes 1 + 1 \otimes d + t\theta^a \otimes \mathcal{L}_a - t\phi^a \otimes \iota_a + \frac{1}{2}t(1-t)f_{ab}^c \theta^a\theta^b \otimes \iota_c , \\
\iota^{(t)} &= \iota\otimes 1 + (1-t)1\otimes \iota , \\
\mathcal{L}^{(t)} &= \mathcal{L}^{(W)} \quad \forall t .
\end{aligned}
\end{equation}
We see that for \(t=0\) we recover the Weil model, while for \(t=1\) we get the BRST model, as special cases. When restricted to the basic subcomplex, they all give the same equivariant cohomology.

%%%%%%%%%%%%%%%%%%%%%%%%%%%%%%%%%%%%%%%%%%%%%%%%% cap 2
\chapter{Localization theorems in finite-dimensional geometry}
\label{cha:loc theorems}

In this chapter we are going to introduce one of the most important results of the equivariant cohomology theory: the \emph{Atiyah-Bott-Berline-Vergne (ABBV) localization formula} for torus actions, discovered independently by Berline and Vergne \cite{berline-vergne-loc}, and by Atiyah and Bott \cite{atiyah-bott-localization}. For the most applications to QM and QFT, we will focus on the case of a circle action, and higher-dimensional generalizations will be postponed to Chapter \ref{cha:non Abelian YM}. This formula can be viewed as a generalization of an analogous result of Duistermaat and Heckman \cite{duistermaat-heckman}, that treats the special case in which the torus action is Hamiltonian on a symplectic manifold. We will expand on this point of view in the second part of the chapter, since this is the situation we are more commonly interested in when we treat dynamical systems in physics, at least at the classical level. The formal generalization of these formulas in the infinite-dimensional setting of QFT will be discussed in Chapter \ref{cha:loc-susy}. 

Since we are going to deal with integration of equivariant forms, we consider \(U(1)\)-equivariant cohomologies from the point of view of the Cartan model. The definition and notational conventions for integration of equivariant forms on a smooth \(G\)-manifold are reported in Appendix \ref{app:integr}, as well as an equivariant version of the Stokes' theorem, needed for the proof of the localization formulas that are presented in the following.

\section{Equivariant localization principle}
\label{sec:eq_loc_princ}

Let \(U(1)\) act (smoothly) on a compact oriented \(n\)-dimensional manifold \(M\) without boundaries,\footnote{If not specified a manifold is always \lq\lq without boundaries\rq\rq\ since, strictly speaking, manifolds with boundaries have to be defined in an appropriate separated way. In particular, near points at the boundary the manifold is locally homeomorphic not to an open set in \(\mathbb{R}^n\), but to an \emph{half-open} disk in \(\mathbb{R}^n\).} with fixed point set \(F\subseteq M\), and consider the integral of a generic \(U(1)\)-invariant top-form
\begin{equation}
\int_M \alpha \qquad \quad \text{with}\ \alpha\in \Omega^n(M)^{U(1)} .
\end{equation}
As we saw in Example \ref{ex:U(1)-cohom}, in some cases we can find an \emph{equivariantly closed extension} \(\tilde{\alpha}\in \Omega_{U(1)}(M)\) such that \(d_C\alpha = 0\), with
\begin{equation}
d_C = d + \iota_T
\end{equation}
and \(T\cong i\) being the generator of \(U(1)\).\footnote{Notice that we have localized the Cartan model and set \(\phi=-1\), as discussed in Example \ref{ex:U(1)-cohom}. This will be our standard convention up to Chapter \ref{cha:non Abelian YM}.} Then we can deform the integral without changing its value,
\begin{equation}
I[\tilde{\alpha}] := \int_M \tilde{\alpha} = \int_M \alpha
\end{equation}
since only the top-degree component \(\alpha\) is selected by integration. We are going to argue now that such integration of an equivariantly closed form is completely captured by its values at the fixed point locus \(F\), using two different arguments. The first is cleaner, the second less explicit but more common especially in the physics literature. We are going to need in both cases some preliminary facts, that we collect in the following lemma.
\begin{lemma}
\begin{enumerate}[label=(\roman*)]
\item If \(G\) is a compact Lie group, any smooth \(G\)-manifold \(M\) admits a \(G\)-invariant Riemannian metric. In other words, \(G\) acts via isometry on \(M\), and the fundamental vector field \(\underline{T}\) is a Killing vector field,\footnote{This follows from two facts: if a \(G\)-action on \(M\) is smooth and \emph{proper}, then \(M\) admits a \(G\)-invariant Riemannian structure \cite{Kankaanrinta-proper}; also, it is easy to prove that any smooth action of a \emph{compact} Lie group is proper.}
\[ \mathcal{L}_T g = 0 . \]
\item If \(G\) is a connected Lie group, then the fixed point locus is the zero locus of all the fundamental vector fields:\footnote{This is just reasonable, see \cite{tu-equiv_cohom} for a proof. Connectedness is required because we passed from the action of \(G\) to the action of \(\mathfrak{g}\) by the exponential map.} \[
F \cong \lbrace\left. p\in M \right| \underline{A}_p = 0 \quad \forall A\in \mathfrak{g}\rbrace . \]
\item For any point \(p\in M\), the stabilizer of \(p\) under the action of a Lie group \(G\) is a closed subgroup of \(G\).\footnote{By continuity of the action, every sequence inside the stabilizer of \(p\) converges inside the stabilizer.}
\end{enumerate}
\end{lemma}

\paragraph*{\(1^{st}\) argument: Poincaré lemma}  

For simplicity, suppose that \(F\) contains only isolated fixed points. From lemma (i), we can pick any \(U(1)\)-invariant metric on \(M\), and  define through it \emph{open balls} of radius \(\epsilon\) \(B(p,\epsilon)\) around any fixed point \(p\in F\). Then \(U(1)\) acts without fixed points on the complement
\begin{equation}
\tilde{M}(\epsilon) := M \setminus \bigcup_{p\in F} B(p,\epsilon) ,
\end{equation}
that is a manifold \emph{with boundaries}, them being the union of the surfaces of the balls at every fixed point (oriented in the opposite direction to the usual one). From lemma (iii), the stabilizer of any point in \(\tilde{M}\) is a closed subgroup of \(U(1)\), but it cannot be \(U(1)\) since we excluded the fixed points, so it is discrete.\footnote{The closed subgroups of \(U(1)\) are \(U(1)\) and the finite cyclic groups \(\lbrace 1\rbrace, \mathbb{Z}/n\) with \(n\in \mathbb{Z}\).} This means that the \(U(1)\)-action on \(\tilde{M}\) is locally free.

We would like to find an equivariant version of the Poincaré lemma on \(\tilde{M}\), where the action is locally free. This means finding a map \(K: \Omega(\tilde{M})^{U(1)} \to \Omega(\tilde{M})^{U(1)}\) of odd-degree such that \([d_C, K ]_+ = id\). If we are able to find such a map, then any equivariantly closed form \(\eta \in \Omega(\tilde{M})^{U(1)}\) is also equivariantly exact,
\begin{equation}
\eta = (Kd_C + d_C K)\eta = K(d_C\eta) + d_C(K\eta) = d_C(K\eta).
\end{equation}
We can define the map \(K\) by multiplication with respect to an equivariant form \(\xi\in \Omega(\tilde{M})^{U(1)}\) of pure odd-degree such that \(d_C\xi = 1\), since 
\begin{equation}
[\xi,d_C]_+ = \xi d_C + (d_C \xi) + (-1)^{\textrm{deg}(\xi)} \xi d_C = 1.
\end{equation}
This form can be defined using again a \(U(1)\)-invariant metric on \(M\), that we call \(g\). We define the following 1-form away from the fixed point set, where \(\underline{T}= 0\),
\begin{equation}
\beta := \frac{1}{g(\underline{T},\underline{T})}g(\underline{T},\cdot)
\end{equation}
and notice that it is \(U(1)\)-invariant by invariance of \(g\), and \(\iota_T \beta = 1\), so that the action of the Cartan differential on it gives \(d_C \beta = d\beta + 1\). Then the odd-degree form \(\xi\) can be defined as
\begin{equation}
\label{eq:inverse-of-localiz-form}
\xi := \beta (d_C \beta)^{-1} = \beta \left( 1 + d\beta\right)^{-1} = \beta \sum_{i=0}^{n-1} (-1)^i (d\beta)^i .
\end{equation}
The inverse of the form \((1+ d\beta)\) can be guessed pretending that \(d\beta\) is a number, and using the Taylor expansion \[ 
(1+z)^{-1} = \sum_{i=0}^\infty (-1)^i z^i .\] 
In the case of forms, the sum at the RHS stops at finite order, since by degree reasons \((d\beta)^i=0\) for \(i>(n/2)\). It is easy to check that \((d_C\beta)^{-1}(d_C\beta)=1\),  \(d_C\xi = 1\), and \(\textrm{deg}(\xi)\) is odd.

Now we know that any equivariantly closed form in \(\Omega(\tilde{M})^{U(1)}\) is also equivariantly exact, so we can simplify the integral \(I[\alpha]\) of an equivariantly closed form \(\alpha\) using an equivariant version of Stokes' theorem (see Appendix \ref{app:integr}):
\begin{equation}
\int_{\tilde{M}} \alpha = \int_{\tilde{M}} d_C(\xi\alpha) = \int_{\partial\tilde{M}} \xi\alpha .
\end{equation}
Taking the limit \(\epsilon\to 0\), the domain of integration on the LHS covers all \(M\), and the integral over the boundary on the RHS reduces to a sum of integrals over the boundaries of \(n\)-spheres centered at each fixed point \(p\in F\) (since \(\partial M = \emptyset\)). Thus the integral of an equivariantly closed form \lq\lq localizes\rq\rq\ as a sum over the fixed points of the \(U(1)\)-action,
\begin{equation}
I[\alpha] = \int_M \alpha = \lim_{\epsilon \to 0} \int_{\tilde{M}(\epsilon)} \alpha = \sum_{p\in F} \lim_{\epsilon \to 0}  \left(-\int_{S^{n-2}_{\epsilon}(p)} (\xi\alpha)\right) = \sum_{p\in F} c_p
\end{equation}
for some contributions \(c_p\) at each fixed point. The precise form of these contributions will be discussed in the next section.

\paragraph*{\(2^{nd}\) argument: localization principle}

The second argument for the localization of the equivariant integral is less explicit, but more direct. Also, it is closer to the approach we will use in the infinite-dimensional context of supersymmetric QFT.

Again, we start from the integral \(I[\alpha]\) of an equivariantly closed form \(\alpha\in \Omega(M)^{U(1)}\). The basic idea is to take advantage of the equivariant cohomological nature of the integral over \(M\): this depends really on the cohomology class of the integrand, not on the particular representative. So we can deform the integral staying in the same class in a way that simplifies its evaluation, without changing the final result. To do this, we pick a positive definite \(U(1)\)-invariant 1-form \(\beta\) on \(M\), and define the new integral
\begin{equation}
I_t[\alpha] := \int_M \alpha e^{-t d_C \beta}
\end{equation}
with  \(t\in \mathbb{R}\). It is again an integral of an equivariantly closed form,
\begin{equation}
d_C \left( \alpha e^{-t d_C\beta} \right) = (d_C\alpha) e^{-t d_C\beta} -t \alpha (d_C^2 \beta)e^{-t d_C\beta} = 0 ,
\end{equation}
since \(d_C^2 = \mathcal{L}_T\) and \(\beta\) is \(U(1)\)-invariant. To show that this integral is equivalent to \(I[\alpha]\), we show that it is independent on the parameter \(t\):
\begin{equation}
\begin{aligned}
\frac{d}{dt} I_t[\alpha] &= \int_M \alpha(- d_C \beta) e^{-t d_C \beta} \\
&= - \int_M d_C\left( \alpha\beta e^{-t d_C \beta} \right) \qquad &\text{(integration by parts)} \\
&= 0 \qquad &\text{(equivariant Stokes' theorem)}. 
\end{aligned}
\end{equation}
Noticing that \(I[\alpha] = I_{t=0}[\alpha]\), from the \(t\)-independence  it follows that \(I[\alpha] = I_{t}[\alpha]\) for every value of the parameter.

We showed that the deformation via the exponential \(e^{-t d_C\beta}\) does not change the equivariant cohomology class of the integrand, so we are free to compute the integral for any value of the parameter. In particular, in the limit \(t\to \infty\), we see that the only contributions come from the zero locus of the exponential. This gives the \lq\lq localization formula\rq\rq
\begin{equation}
\label{eq:loc-principle-U1}
\int_M \alpha = \lim_{t\to\infty} \int_M \alpha e^{-t d_C \beta} ,
\end{equation}
that will be the starting point for all the applications of the equivariant localization principle of the next chapters, also in the infinite-dimensional case in which \(M\) describes generically the \lq\lq space of fields\rq\rq\ of a given QFT.

The 1-form \(\beta\) is usually called \lq\lq localization 1-form\rq\rq. Notice that choosing different localization 1-forms produces different practical localization schemes, but at the end of the computation they must all agree on the final result! In particular, by lemma (i) we can pick a \(U(1)\)-invariant Riemannian metric \(g\), and choose the 1-form as
\begin{equation}
\beta := g(\underline{T},\cdot) .
\end{equation}
This makes it positive definite and produces the same localization scheme of the first argument, since its zeros coincide with the zeros of the fundamental vector field \(\underline{T}\) and thus with the fixed point locus \(F\) of the circle action, by lemma (ii).

\section{The ABBV localization formula for Abelian actions}
\label{sec:ABBV}

Here we state the celebrated result by Atiyah-Bott and Berline-Vergne, about the localization formulas for circle and torus actions. The rationale of the last section showed that the equivariant cohomology of the manifold \(M\) is encoded in the fixed point set \(F\) of the symmetry action, but left us with the evaluation of an integral over the fixed point set. We show the result of this integration here, and we are going to give an argument for the proof in the next chapter, with some tools from supergeometry. That proof is different from the original ones in \cite{atiyah-bott-localization, berline-vergne-loc}, but will introduce a method that can be easily generalized to functional integrals.

To warm up, we consider first the simple case of isolated fixed point set \(F\subseteq M\), and a \(U(1)\)-action. Notice that, at any fixed point \(p\in F\), the circle action gives a representation of \(U(1)\) on the tangent space, since for any \(\psi \in U(1)\)
\begin{equation}
(\psi \cdot )_* : T_p M \to T_{\psi\cdot p}M\equiv T_pM ,
\end{equation}
so \((\psi \cdot)_* \in GL(T_p M)\). Since \(T_p M\) is finite dimensional, it can be decomposed in irreducible representations of \(U(1)\),
\begin{equation}
\label{abbv-irrepr}
T_p M \cong V_1 \oplus \cdots \oplus V_n .
\end{equation}
The circle has to act faithfully on \(T_p M\), since if there was  \(v\in T_p M\) such that \((\psi\cdot)_* v = v\), then the whole curve \(\exp(tv)=\exp(t(\psi\cdot)_* v) = \psi \cdot \exp(tv)\) would be fixed by \(U(1)\), thus \(p\) would not be isolated. Recall that the irreducible representations of \(U(1)\) are complex 1-dimensional, and are labeled by integers,
\begin{equation}
\psi= e^{ia} \in U(1) , \qquad \rho_m(\psi) := e^{ima} \quad \text{with} \ m\in\mathbb{Z}.
\end{equation}
This means that the irreducible representations in \eqref{abbv-irrepr} are all non-trivial (of real dimension 2), and that \(\dim(M)=2n\). In other words, if a circle action on a manifold \(M\) has isolated fixed points, \(M\) must be even-dimensional. 
Excluding the trivial representation with \(m=0\), the tangent spaces at the fixed points are thus labeled by a set of integers,
\begin{equation}
\label{eq:circle-representation}
T_p M \cong V_{m_1} \oplus \cdots \oplus V_{m_n}
\end{equation}
where \((m_1,\cdots,m_n)\in \mathbb{Z}^n\) are called the \emph{exponents} of the circle action at \(p\in F\).\footnote{In terms of the Lie algebra representation, every exponent \(m\) coincide with the \emph{weight} of the single generator of \(U(1)\) in the fundamental representation.} They can be regarded as maps \(m_i:F\to\mathbb{Z}\). We formulate now a simplified version of the localization theorem in term of this local data. The proof of this can be found in \cite{tu-equiv_cohom}.
\begin{thm}[Localization for circle actions]\label{thm:ABBV-circle}
Let \(U(1)\) act on a compact oriented manifold \(M\) of dimension \(\dim(M)=2n\), with isolated fixed point locus \(F\). If \(m_1,\cdots, m_n :F\to \mathbb{Z}\) are the exponents of the circle action, and \[
\alpha = \alpha^{(2n)} + \alpha^{(2n-2)}\phi + \alpha^{(2n-4)}\phi^4 + \cdots + \alpha^{(0)} \]
is an equivariant top-form in \(\Omega_{U(1)}(M)\) such that \(d_C\alpha=0\), then \[ \boxed{
\int_M \alpha^{(2n)} = \int_M \alpha = (2\pi)^n \sum_{p\in F} \frac{\alpha^{(0)}(p)}{m_1(p)\cdots m_n(p)}  }\]
where the last component \(\alpha^{(0)}\in C^\infty(M)\).
\end{thm}

\begin{ex}[Localization on the 2-sphere]\label{ex:S2-ABBV} 
Let us consider again the case of the height function \(H:\mathbb{S}^2 \to \mathbb{R}\) such that, in spherical coordinates \((\theta,\varphi)\), \(H(\theta,\varphi) := \cos(\theta)\). In Example \ref{ex:U(1)-cohom} we related this function to the equivariantly closed extension of the volume form on the 2-sphere, \[
\tilde{\omega} = \omega + H . \]
We can use the last localization theorem to compute integrals involving this \lq\lq Hamiltonian\rq\rq\ function on \(\mathbb{S}^2\). The 2-sphere has two isolated fixed points at the poles, and only one exponent \(m:F\to \mathbb{Z}\). It is not difficult to see that the exponent of the action at the fixed points is \(m(N)=1\) at the North pole, and \(m(S)=-1\) at the South pole (the sign comes from the orientation of the charts).

We can check the theorem with two instructive integrals. The first is simply the area of the sphere, \textit{i.e.}\ the integral of \(\omega\). Using the theorem we easily get the correct result,
\begin{equation*}
\int_{\mathbb{S}^2}\omega = \int_{\mathbb{S}^2} (\omega + H) =  2\pi \sum_{p\in \lbrace N,S\rbrace} \frac{H(p)}{m(p)} = 2\pi \left( \frac{\cos(0)}{1} +  \frac{\cos(\pi)}{-1} \right) = 4\pi.
\end{equation*}
The second integral is the \lq\lq partition function\rq\rq on the sphere,
\begin{equation*}
Z(t) := \int_{\mathbb{S}^2} \omega e^{itH} = \frac{1}{it}\int_{\mathbb{S}^2} e^{it(H+\omega)}
\end{equation*}
where the second equality comes from degree arguments. This is the integral of an equivariantly closed form, since \(d_C e^{it(H+\omega)}\propto d_C \tilde{\omega}=0\), whose \(C^\infty(\mathbb{S}^2)\) component is given by \(e^{itH}\). Using the localization theorem we get \[ 
Z =  \frac{1}{it}2\pi \left( \frac{e^{it\cos(0)}}{1} + \frac{e^{it\cos(\pi)}}{-1} \right) = 4\pi \frac{\sin(t)}{t} \]
matching the result from the \lq\lq semiclassical\rq\rq\ saddle-point approximation \eqref{intro-height-function}.
\end{ex}

\bigskip
We now get to the main theorem, considering a more generic torus action with higher dimensional fixed point locus on \(M\).
\begin{thm}[Atiyah-Bott \cite{atiyah-bott-localization}, Berline-Vergne \cite{berline-vergne-loc}] \label{thm:ABBV} 
Let the torus \(T=U(1)^l\) of dimension \(l\) act on a compact oriented \(d\)-dimensional manifold \(M\), with fixed point locus \(F\). If \(\alpha\in \Omega_T(M)\) is an equivariantly closed form, \textit{i.e.}\ \(d_C\alpha = 0\), and \(i: F\hookrightarrow M\) is the inclusion map, then \[ \boxed{
\int_M \alpha = \int_F \frac{i^*\alpha}{\left. e_T(R)\right|_N}
}\]
where \(\left. e_T(R)\right|_N\) is the \emph{T-equivariant Euler class} of the normal bundle of \(F\) in \(M\).
\end{thm}

This is the localization formula as originally presented for a torus action and fixed point locus \(F\), that is generically an embedded (regular) submanifold of \(M\). The \emph{normal bundle} to \(F\) can be regarded as
\begin{equation}
T N =  \faktor{TM}{i_* TF} ,
\end{equation}
where the quotient is taken pointwise at any \(p\in F\), so that the tangent bundle of \(M\) is split as \(TM = i_* TF \oplus TN\).
The finite sum is replaced by an integral over \(F\), and the zero-degree component of \(\alpha\) is replaced by the component with the correct dimensionality, that matches \(\dim(F)\), by  pulling-back \(\alpha\) on \(F\). The product of the exponents at the denominator is represented in general by the equivariant Euler class of the normal bundle,
\begin{equation}
\left. e_T (R)\right|_N =\mathrm{Pf}_N\left(\frac{R^T}{2\pi}\right) = \mathrm{Pf}_N\left(\frac{R+\mu}{2\pi}\right) ,
\end{equation}
where the pfaffian is taken over the coordinates that span the normal bundle \(TN\), \(R\) is the curvature of an invariant Riemannian metric on \(M\), \(\mu:\mathfrak{t}\to \Omega^0(M;\mathfrak{gl}(d))\) is the \lq\lq moment map\rq\rq\ that makes \(R^T\) an equivariant extension of the Riemannian curvature in the Cartan model (see Appendix \ref{app:char-classes}). 

As an example, let us apply the ABBV localization formula in the case of a discrete fixed point set \(F\), so that we can recover at least the more readable version of theorem \ref{thm:ABBV-circle}. The normal bundle in this case is the whole tangent bundle and, since \(F\) is 0-dimensional, the restriction of the equivariant curvature \(R^T\) to \(F\) makes only its \(\Omega^0\) component contribute, so \(\mathrm{Pf}(R^T)=\phi^a \otimes \mathrm{Pf}(\mu_a)\). At an isolated fixed point \(p\), as we said before, the tangent space \(T_p\) is a representation space for the torus action. Since the torus is Abelian, analogously to the above discussion this representation can be decomposed as the sum of 2-dimensional \emph{weight spaces} \cite{brocker-book, guillemin-sternberg-book}, 
\begin{equation}
\label{ABBV-weigth-decomposition-tangent-space}
T_pM \cong \bigoplus_{i=1}^{d/2} V_{v_i}.
\end{equation}
In Section \ref{sec:ABBVproof} we will see that the moment map at an isolated fixed point encodes exactly these weights, being the representation \(\mu(p):\mathfrak{t}\to \mathfrak{gl}(d)\cong \mathrm{End}(T_pM)\). The equivariant Euler class computes exactly the product of the weights,
\begin{equation}
\label{eq:euler-weights}
e_T(R)_p = \frac{1}{(2\pi)^{d/2}}\prod_i v_i = \frac{1}{(2\pi)^{d/2}}\prod_i \phi^a \otimes v_i(T_a),
\end{equation}
where \(T_a\) are the generators of \(T\). This recovers the formula for the circle action in theorem \ref{thm:ABBV-circle}, where the exponents play the role of the weights for the single generator of \(\mathbb{S}^1\).

Notice that, as it is clear from the above example, in the generic \(l\)-dimensional case it is not so convenient to forget about the generators \(\lbrace \phi^a\rbrace\) of \(\mathfrak{t}^*\), and the ABBV localization formula should be thought as an equivalence of elements in \(H^*_T(pt) = \mathbb{R}[\phi^1,\cdots,\phi^l]\). The LHS is clearly polynomial in \(\phi^a\), so has to be the RHS. Since in the latter both the numerator and the denominator are polynomials in \(\phi^a\), some simplification has to occur in the rational expression to give a polynomial as the final answer.

\begin{rmk}
We anticipate that in QFT the pfaffian in the definition of the Euler class is usually realized in terms of a Gaussian integral over Grassmann (anticommuting) variables, as we will see in detail in Section \ref{sec:ABBVproof}. These \lq\lq fermionic\rq\rq\ Gaussian integrals arise naturally as \lq\lq 1-loop determinants\rq\rq\ from some saddle-point (semi-classical) approximation technique to the partition function of the theory, for example. In general, the differential form \(\alpha\) will be an \lq\lq observable\rq\rq\ of the QFT, and the equivariantly closeness condition will be interpreted as it being \lq\lq supersymmetric\rq\rq. The localization locus \(F\) will be then the fixed point set of a symmetry group that is the \lq\lq square\rq\rq\ of this supersymmetry (as \(d_C^2 \propto \mathcal{L}_{T}\) schematically), so a Poincaré symmetry or a gauge symmetry. The integral then localizes onto the \lq\lq moduli space\rq\rq\ of gauge-invariant (or BPS) field configurations. In the context of Hamiltonian mechanics, the gauge symmetry can be one generated by the dynamics of the theory itself, and in this case the path integral localizes onto the classical solutions of the equations of motion. The ABBV formula thus gives a systematic way to understand in which cases the semi-classical approximation results to be exact. We will reexamine this point of view in the next section in the context of finite-dimensional Hamiltonian mechanics, while in Chapter \ref{cha:loc-susy} we will describe the infinite-dimensional case of QM and QFT, giving some examples of the ABBV localization formula at work.
\end{rmk}

\section{Equivariant cohomology on symplectic manifolds}
\label{sec:e-cohom-symplectic}

As we remarked at the beginning of the chapter, the localization formulas of the last section can be seen as generalizing a similar result showed by Duistermaat and Heckman \cite{duistermaat-heckman} in the context of Hamiltonian group actions on symplectic manifolds. This special case is of fundamental importance in physics, because this is the context in which classical Hamiltonian mechanics is constructed. In some special cases also the quantum theory can be formally given  such a structure, and thus some results from symplectic geometry can be extended to QM and QFT in general. We begin this section by quickly recalling some basic concepts about symplectic and Hamiltonian geometry, then we will describe how this can be seen as a special case of equivariant cohomology theory from the point of view of the localization formulas.

\subsection{Pills of symplectic geometry}

The notion of \emph{phase space} can be constructed in a basis-independent way in differential geometry through the definition of \emph{symplectic manifold}. We suggest for example \cite{marsden-ratiu, audin, dasilva} for a complete introduction to the subject.

\begin{defn} A \emph{symplectic manifold} is a pair \((M,\omega)\), where \(M\) is a \(2n\)-dimensional smooth manifold, and \(\omega\) is a \emph{symplectic form} on \(M\):
\begin{enumerate}[label=(\roman*)]
\item \(\omega\in \Omega^2(M)\);
\item \(d\omega = 0\);
\item \(\omega\) is non-degenerate.
\end{enumerate}
\end{defn}

The fact that \(M\) is even-dimensional is not really a requirement but a consequence of its symplectic structure. This is because any skew-symmetric bilinear map on a \(d\)-dimensional vector space can be represented in a suitable basis by the matrix
\begin{equation}
\label{symplectic-1}
\left( \begin{array}{c|cc}
\mathbf{0}_{k} & \mathbf{0} & \mathbf{0} \\
\hline
\mathbf{0} & \mathbf{0} & -\mathds{1}_{n} \\
\mathbf{0} & \mathds{1}_{n} & \mathbf{0} 
\end{array} \right)
\end{equation}
with \(2n+k=d\). To be non degenerate, it must be \(k=0\). The symplectic form is a skew-symmetric bilinear form on \(T_pM\) at any point \(p\in M\), so the even-dimensionality of \(M\) follows from its non-degeneracy. On manifolds, a stronger result than the above one holds: the so-called \emph{Darboux theorem}. It states that, for every point \(p\in M\), there exists an entire open neighborhood \(U_p\subseteq M\) and a coordinate system \(x:U_p \to \mathbb{R}^{2n}\) with respect to which \(\omega_{\mu\nu} = \omega(\partial_\mu,\partial_\nu)\) has the canonical form \eqref{symplectic-1}, with \(k=0\). The coordinates \(x\) are called \emph{Darboux coordinates}.\footnote{This means that all symplectic manifolds look locally as the prototype \(\mathbb{R}^{2n}\) with \(\omega = \sum_{i=1}^n dx^i\wedge dx^{i+n}\). This is a very strong property,  compared for example with the Riemannian case.} Notice that from the non-degeneracy of \(\omega\) we have a canonical choice for the volume form on \(M\), the so-called \emph{Liouville volume form}
\begin{equation}
\mbox{vol} := \frac{\omega^n}{n!} = \mathrm{Pf}||\omega^{(x)}_{\mu\nu}|| d^{2n} x = dp_1\wedge dp_2\wedge \cdots d p_n \wedge dq^1\wedge dq^2 \wedge \cdots \wedge dq^n ,
\end{equation}
where \((q^\mu,p_\mu)_{\mu=1,\cdots,n}\) are Darboux coordinates. The closeness of \(\omega\) implies that in some cases there can be a 1-form \(\theta \in \Omega^1(M)\) such that
\begin{equation}
d\theta = \omega .
\end{equation}
Such a 1-form, if it exists, is called \emph{symplectic potential}. In practice, sometimes it is useful to \emph{locally} define a symplectic potential even if \(\omega\) is not globally integrable. Isomorphisms of symplectic manifolds are called \emph{symplectomorphisms} or \emph{canonical transformations}, defined as  diffeomorphisms  that preserve the symplectic structure via pull-back. 

The standard example of a symplectic manifold is exactly the \emph{phase space} associated to some \(n\)-dimensional \emph{configuration space} \(Q\), \textit{i.e.}\ its cotangent bundle \(M:= T^* Q\). A point \(p\in Q\) represents the \lq\lq generalized position\rq\rq\ of the system with coordinates \(q(p)=(q^\mu(p))\) with \(\mu=1,\cdots,n\), and a point \(p\in T^* Q\) represents the \lq\lq generalized momentum\rq\rq, with coordinates \(\xi(p) := (q^\mu\circ \pi(p), \iota_\mu (p)) \equiv (q^\mu, p_\mu)\), where \(\pi:TQ\to Q\) is the projection and \(\iota_\mu \equiv \iota_{\partial/\partial q^\mu}\). The cotangent bundle has a canonical integrable symplectic form. In fact, the symplectic potential is the so-called \emph{tautological 1-form} given by the pull-back of the projection map, \(\theta:= \pi^* \in \Omega^1(T^*Q)\). In Darboux coordinates, at a point \(p\in T^*Q\),
\begin{equation}
\theta_p = \pi^*(p) = p_\mu dq^\mu
\end{equation}
where we denoted \(dq^\mu \equiv d(q\circ\pi)^\mu = d\xi^\mu\) with \(\mu=1\cdots,n\), as 1-forms on the cotangent bundle. The canonical symplectic form is then just \(\omega = d\theta\), and in Darboux coordinates
\begin{equation}
\label{eq:symplectic-darboux}
\omega = dp_\mu \wedge dq^\mu
\end{equation}
where again we simplified the notation setting \(dp_\mu \equiv d\xi^\mu\) for \(\mu=n+1,\cdots,2n\). Thus the canonical coordinates on the cotangent bundle are Darboux coordinates. One can show that canonical symplectic structures over diffeomorphic manifolds are \lq\lq canonically compatible\rq\rq, \textit{i.e.}\ if \(\phi:Q_1\to Q_2\) is a diffeomorphism, there is a lift of it as a symplectomorphism between \((T^* Q_1,\omega_1)\) and \((T^* Q_2, \omega_2)\). If we take \(Q_1 = Q_2\), this means that there is a group homomorphism
\begin{equation}
\text{Diff}(Q) \to \text{Symp}(T^*Q,\omega).
\end{equation}
This example showed that symplectic manifolds are the right generalization of the concept of phase space in a fully covariant setting. It is thus common to call functions on a symplectic manifold \emph{observables}.

Let us return to a generic symplectic manifold \((M,\omega)\). Giving to it some additional structure, it is possible to define on it \emph{dynamics} and \emph{symmetries} in the sense of classical mechanics. Naturally, we call symmetry of \((M,\omega)\)  a diffeomorphism \(\phi:M\to M\) that preserves the symplectic structure, \(\phi^*\omega = \omega\), that is a symplectomorphism. At the infinitesimal level, a diffeomorphism can be generated by the flow of a vector field \(X\in \Gamma(M)\), and the symmetry condition is rephrased to
\begin{equation}
\mathcal{L}_X \omega = 0.
\end{equation}
Such a vector field is called \emph{symplectic vector field}. It is easy to realize that a vector field is symplectic if and only if \(\iota_X\omega = \omega(X,\cdot)\) is closed, by Cartan's magic formula. More special vector fields are those for which \(\iota_X\omega\) is exact, so that it exists an observable \(f\in C^\infty(M)\) such that
\begin{equation}
df = - \iota_X \omega ,
\end{equation}
where the minus sign is conventional. The vector field \(X\) is called \emph{Hamiltonian vector field} associated to the observable \(f\). In components, 
\begin{equation}
\label{eq:ham-vf}
\partial_\mu f = \omega_{\mu\nu}X^\nu \qquad \text{or} \qquad X^\mu = \omega^{\mu\nu}\partial_\nu f ,
\end{equation}
where \(\omega^{\mu\nu}\) is the \lq\lq inverse\rq\rq\ of the symplectic form. Of course Hamiltonian vector fields are symplectic, and the flow of the Hamiltonian vector field \(X\) preserves the value of the Hamiltonian function \(f\), since \(\mathcal{L}_X(f) = X(f) = df(X) = \omega(X,X) = 0\). The flow of the Hamiltonian vector field is regarded as the \lq\lq time-evolution\rq\rq\ over the generalized phase space \(M\), generated by the observable \(f\).
\begin{defn} \label{def:hamiltonian-system} 
An \emph{Hamiltonian (or dynamical) system} is a tuple \((M,\omega, H)\), where \((M,\omega)\) is a symplectic manifold and \(H\in C^\infty(M)\) an observable called \emph{Hamiltonian}. 
The \emph{time-evolution} of points \(p\in M\) is defined by the flow of the Hamiltonian vector field \(X_H\) of \(H\), \[
p(t) := \gamma^H_p(t) \]
where \(\gamma^H_p\) is the integral curve of \(X_H\) with \(\gamma^H_p(0)=p\). In particular, the evolution of an observable \(f\in C^\infty(M)\) is regulated by the \emph{equation of motion} \[
\dot{f}(p) := (f\circ\gamma^H_p)'(0) \equiv \left. \mathcal{L}_{X_H}(f)\right|_p .\] \end{defn}

The equation of motion can be rewritten in a more usual way introducing the \emph{Poisson brackets} \(\lbrace\cdot,\cdot\rbrace:C^\infty(M)\times C^\infty(M)\to C^\infty(M)\) such that \(\lbrace f,g\rbrace := \omega(X_g,X_f)\), where \(X_f, X_g\) are the Hamiltonian vector fields of \(f\) and \(g\), respectively. In a chart and with respect to Darboux coordinates \((q^\mu,p_\mu)\) on \(M\), by the Darboux theorem the Poisson brackets take the usual form
\begin{equation}
\lbrace f,g\rbrace = \frac{\partial f}{\partial q^\mu}\frac{\partial g}{\partial p_\mu} - \frac{\partial g}{\partial q^\mu}\frac{\partial f}{\partial p_\mu} .
\end{equation}
With this definition we can write
\begin{equation}
\dot{f} = - \lbrace H, f\rbrace , \qquad \dot{q}^\mu = \frac{\partial H}{\partial p_\mu} , \qquad \dot{p}_\mu = - \frac{\partial H}{\partial q^\mu} ,
\end{equation}
recovering the Hamilton's equations for the Darboux coordinates. 
The Poisson brackets are anti-symmetric and satisfy the Jacobi identity, so this turns \((C^\infty(M),\lbrace\cdot,\cdot\rbrace)\) into a Lie algebra,\footnote{In fact this is a \emph{Poisson algebra}, \textit{i.e.}\ a Lie algebra whose brackets act as a derivation.} and one can check that there is a Lie algebra homomorphism
\begin{equation}
\begin{aligned}
(C^\infty(M),\lbrace\cdot,\cdot\rbrace) &\to (\text{Hamiltonian v.f.}, [\cdot,\cdot]) \\
f &\mapsto X_f ,
\end{aligned}
\end{equation}
where we also already used the fact that Hamiltonian vector fields form a Lie subalgebra with respect to the standard commutator on \(\Gamma(TM)\).

\medskip
We just reviewed that the concept of symmetry in symplectic geometry is correlated with the concept of dynamics on the symplectic manifold. The next fact that we need is to connect this formalism to the equivariant cohomology one, identifying these symmetries as generated by a \emph{group action} on \(M\). In particular, we would like to identify the Lie subalgebra of Hamiltonian vector fields as the Lie algebra of a Lie group that acts on the symplectic manifold. We can start thus the discussion of symmetry by declaring that \(M\) is a \(G\)-manifold with respect to a Lie group \(G\) of Lie algebra \(\mathfrak{g}\). Denoting  the \(G\)-action as \(\rho\), this is called \emph{symplectic} if it makes \(G\) act by symplectomorphisms on \((M,\omega)\), \textit{i.e.}\
\begin{equation}
\rho : G \to \text{Symp}(M,\omega) .
\end{equation}
We can characterize again infinitesimally this action by saying that \(\mathfrak{g}\) acts on \(\Omega(M)\) via symplectic vector fields: if \(A\in \mathfrak{g}\), the corresponding fundamental vector field preserves the symplectic structure, \(\mathcal{L}_A\omega = 0\). We are interested in the special case analogous to the one above, in which not only a fundamental vector field is symplectic, but it is also Hamiltonian. This forces a generalization of the concept of Hamiltonian function, because now there are more than one independent fundamental vector fields to take into account, if \(\dim(\mathfrak{g})>1\).
\begin{defn} The \(G\)-action \(\rho:G\to \text{Symp}(M,\omega)\) on the symplectic manifold \((M,\omega)\) is said to be an \emph{Hamiltonian action} if every fundamental vector field is Hamiltonian. In particular, there exists a \(\mathfrak{g}^*\)-valued function \(\mu\in \mathfrak{g}^*\otimes C^\infty(M)\) such that:
\begin{enumerate}[label=(\roman*)]
\item For every \(A\in \mathfrak{g}\), \(\mu(A)\equiv \mu_A\in C^\infty(M)\) is the Hamiltonian function with respect to \(\underline{A}\),
\[ d\mu_A = -\iota_A \omega = \omega(\cdot,\underline{A}) .\]
\item It is \(G\)-equivariant with respect to the canonical (co)adjoint action of \(G\) on \(\mathfrak{g}\ (\mathfrak{g}^*)\),\footnote{If, for every \(g\in G\), \(Ad_g:G\to G\) is the action by conjugation, the adjoint action \(Ad_{*g}\) on \(\mathfrak{g}\) is the push-forward of \(Ad_g\), while the coadjoint action \(Ad^*_g\) on \(\mathfrak{g}^*\) is the pull-back of \(Ad_{g^{-1}}\). If \(g=\exp(tA)\) for some \(A\in \mathfrak{g}\), differentiating one gets the infinitesimal actions of \(\mathfrak{g}\) on \(\mathfrak{g}\) and \(\mathfrak{g}^*\), \(ad_{A}(B) = [A,B]\) and \(ad^*_A(\eta) := \eta([\cdot,A])\). } so for any \(g\in G\)
\[ \mu \circ Ad_{*g} = \rho_g^*\circ \mu \qquad \text{or} \qquad  Ad^*_g \circ \mu  =  \mu \circ \rho_g ,\]
where in the first equation \(\mu\) is considered as \(\mathfrak{g}\xrightarrow{\mu} C^\infty(M)\), in the second one as \(M \xrightarrow{\mu} \mathfrak{g}^*\). If \(G\) is connected, this is equivalent to requiring \(\mu : \mathfrak{g}\to C^\infty(M)\) to be a Lie algebra anti-homomorphism with respect to the Poisson brackets, \[
\mu_{[A,B]} = \lbrace \mu_B, \mu_A\rbrace \qquad \forall A,B\in \mathfrak{g} .\]
\end{enumerate}
The map \(\mu\) is called \emph{moment map}, and \((M,\omega,G,\mu)\) is called \emph{Hamiltonian \(G\)-space}.
\end{defn}

In general the job of the moment map is to collect all the \lq\lq Hamiltonians\rq\rq\ with respect to which the system can flow. There are \(\dim(G)\) independent of them, one for every generator. In the 1-dimensional case, where \(G=U(1)\) (or its non-compact counterpart \(G=\mathbb{R}\)), the moment map produces only one independent Hamiltonian, \(\mu_T \equiv H\), and the above definition reduces to the Hamiltonian system \((M,\omega,H)\) of definition \ref{def:hamiltonian-system}. Notice that for any Hamiltonian structure we build on \((M,\omega)\), its flow preserves the symplectic form and thus the canonical Liouville volume form \(\omega^n/n!\). This is the content of the so-called \emph{Liouville theorem}.

%\begin{ex}[QFT] COMPLETARE
%\end{ex}

%\clearpage
%Liouville measure and D-H theorem for its ''Fourier transform'': stationary phase approx and localization. (Assumption: hamiltonian is a Morse function [szabo])
%
%Kirwan's thm: only Morse functions for which the stationary phase approximation can be exact are those which have only even Morse indices \(\lambda(p)\). This gives a relation between equivariant localization and the topology of the phase space of interest: if the manifold M has non-trivial cohomology groups of odd dimension, then the stationary phase series diverges for any Morse function define on M and in particular the Duistermaat-Heckman localization formula for such phase spaces can never give the exact result for Z(T). In this way, Kirwan’s theorem rules out a large number of dynamical systems for which the stationary phase approximation could be exact in terms of the topology of the underlying phase space where the dynamical system lives!

\subsection{Equivariant cohomology for Hamiltonian systems}
\label{subsec:eq-cohom-symplectic}

We can first generalize what we noticed in examples \ref{ex:U(1)-cohom} and \ref{ex:S2-ABBV}, in the case of a circle action on a symplectic manifold \((M,\omega)\). In the above examples the manifold was the 2-sphere \(\mathbb{S}^2\) and the symplectic form was the canonical volume form. Rephrased in terms of symplectic geometry, the existence of an \emph{equivariantly closed extension} \(\tilde{\omega}\) of the symplectic form is the condition of \(U(1)\) acting in an Hamiltonian way, since
\begin{equation}
d_C \tilde{\omega} = d_C(\omega + H) = \iota_T \omega + dH = 0
\end{equation}
is satisfied if and only if \(dH=-\iota_T \omega\). This is readily generalizable to the multidimentional case, so that we can describe the Hamiltonian \(G\)-space \((M,\omega,G,\mu)\) and its classical mechanics in equivariant cohomological terms. In fact, we can always find an equivariantly closed extension of the symplectic form in \(\Omega_G(M)\),
\begin{equation}
\tilde{\omega} := 1\otimes \omega - \phi^a \otimes \mu_a
\end{equation}
where \(\mu_a \equiv \mu_{T_a}\in C^\infty(M)\) and \(T_a\) are the dual basis elements with respect to the generators \(\phi^a\) of \(S(\mathfrak{g}^*)\). It is straightforward to check that \(\tilde{\omega}\in \Omega_G(M)\) is indeed \(G\)-invariant, and closed with respect to \(d_C\) thanks to the Hamiltonian property of the \(G\)-action, \(d\mu_a = -\iota_a \omega\).

In the language of \(G\)-equivariant bundles (see Appendix \ref{app:char-classes}), the symplectic structure on \(M\) can be seen as the presence of a principal \(U(1)\)-bundle \(P\to M\) whose connection 1-form is the symplectic potential \(\theta\) (that has not always a global trivialization on \(M\)), and whose curvature is the symplectic 2-form \(\omega = d\theta\) (that instead transforms covariantly on \(M\)). \(G\) acts symplectically if also \(\theta\) is \(G\)-invariant, 
\begin{equation}
\mathcal{L}_X \theta = 0 \quad \Rightarrow \quad \mathcal{L}_X \omega= 0 \qquad \forall X\in \mathfrak{g} ,
\end{equation}
so that \(P\to M\) is a \(G\)-equivariant bundle. Thus, this equivariant extension to the curvature \(\omega\) is the same as in \cite{bott-tu-classes,berline-vergne-loc}.

We return for a moment to the symplectic geometric interpretation, to describe the results of Duistermaat and Heckman related to the localization formulas that we described in the last section. In \cite{duistermaat-heckman} they proved an important property of the Liouville measure in the presence of an Hamiltonian action by a torus \(T\) on \((M,\omega)\). Namely, defining a measure on \(\mathfrak{g}^*\) as the \emph{push-forward} of the Liouville measure,
\begin{equation}
\mu_*\left(\frac{\omega^n}{n!}\right)(U) = \int_{\mu^{-1}(U)} \frac{\omega^n}{n!} \qquad \forall U\subseteq M \ \text{measurable},
\end{equation}
they proved that \(\mu_*\left(\omega^n/n!\right)\) is a \emph{piecewise polynomial function}.\footnote{To be more precise, denoting \(\mu_*\left(\omega^n/n!\right) = f d\xi\) with \(d\xi\) the standard Lebesgue measure on \(\mathfrak{g}^* \cong \mathbb{R}^{\dim\mathfrak{g}}\), the function \(f\) is piecewise polynomial.} This, and an application of the stationary phase approximation showed a localization formula for the oscillatory integral
\begin{equation}
\int_M \frac{\omega^n}{n!} \exp{\left(i\mu_X\right)}
\end{equation}
for every \(X\in \mathfrak{t}:=Lie(T)\) with non-null weight at every fixed point of the \(T\)-action. This can be viewed as the Fourier transform of the Liouville measure, or as the \emph{partition function} of a 0-dimensional QFT with target space \(M\). Let the fixed point locus \(F\) be the union of compact connected symplectic manifolds \(M_k \hookrightarrow M\) of even codimension \(2n_k\), and denote \((m_{kl})_{l=1,\cdots,n_k}\) the weights of the \(T\)-action at a tangent space of a fixed point \(p\in M_k\).\footnote{The components of the fixed point set being symplectic is not an assumption, but a consequence of the Hamiltonian action. See \cite{audin}, proposition IV.1.3.} Then the \emph{Duistermaat-Heckman (DH) localization formula} is
\begin{equation} \boxed{
\int_M \frac{\omega^n}{n!} \exp{\left(i\mu_X\right)} = \sum_k \frac{\mbox{vol}(M_k) \exp{(i\mu_X(M_k))}}{\prod_l^{n_k} (m_{kl}(X)/2\pi)}
}
\end{equation}
where \(\mu_X(M_k)\) denotes the common value of \(\mu_X\) at every point in \(M_k\). 

In equivariant cohomological terms, we can see the above result as a localization formula for the integral of an equivariantly closed form. In fact, if we fix the \(U(1)\) symmetry subgroup generated by \(X\in \mathfrak{t}\), and consider the Cartan model defined by the differential \(d_C = d + i\iota_X\), the LHS can be rewritten as
\begin{equation}
\label{hamiltonian-1}
\int_M \frac{\omega^n}{n!}\exp(i\mu_X) = \int_M \exp(\omega + i\mu_X) ,
\end{equation}
analogously to what we did in Example \ref{ex:S2-ABBV}, and this is clearly the integral of an equivariantly closed form with respect to the differential \(d_C\). To see the correspondence with the ABBV formula, let us examine the case of a circle action and discrete fixed point locus \(F\). In this case we have only one Hamiltonian function \(H:= \mu_X\), the weights are just the exponents \(m_1,\cdots,m_n\) of the circle action, and the sum over \(k\) runs over the isolated fixed points. The DH formula thus recovers exactly the localization formula of theorem \ref{thm:ABBV-circle}:
\begin{equation}
\int_M \exp(\omega+iH) = (2\pi)^n \sum_{p\in F} \frac{e^{iH(p)}}{m_1(p)\cdots m_n(p)} .
\end{equation}
As we remarked in \eqref{eq:euler-weights}, the denominator can be expressed as the equivariant Euler class of the normal bundle to \(F\) (that is just the tangent bundle since \(F\) is 0-dimensional), recovering the DH formula as a special case of the ABBV localization formula for torus actions. See also \cite{szabo} for an explicit correspondence between the two. We wish only to remark again that, especially in the context of Hamiltonian mechanics, this localization formula can be seen as the result of an \lq\lq exact\rq\rq\ saddle-point approximation on the partition function \eqref{hamiltonian-1}. This is the point of view we are going to take in the next chapters, when we are going to discuss the generalization of this formula to higher-dimensional QFT, where the integral of the partition function is turned into an infinite-dimensional \emph{path integral}.

To see the correspondence with the saddle-point approximation, we recall that the isolated fixed points of the \(U(1)\)-action are those in which \(\underline{X}=0\), so \(dH = 0\), and thus they are the critical points of the Hamiltonian. We need to assume that the function \(H\) is \emph{Morse}, so that these fixed points are non-degenerate, \textit{i.e.}\ the Hessian \(\mathrm{Hess}_{p_0}(H)_{\mu\nu} = \partial_\mu\partial_\nu H(p_0)\) at a given \(p_0\in F\) has non-null determinant.\footnote{This subject was in fact firstly connected with Morse theory by Witten in \cite{Witten-susy_morse}, where localization is applied in the context of supersymmetric QM to prove Morse inequalities. We do not need to deepen this point of view for what follows, but a discussion about Morse theory and its connection with the DH formula can be found in \cite{szabo}, and references therein.} This Hessian can be expressed in terms of the exponents \(m_k(p_0)\) via an equivariant version of the Darboux theorem \cite{guillemin-sternberg, audin}: at any fixed point we can choose Darboux coordinates in which the symplectic form takes its canonical form \eqref{eq:symplectic-darboux}, and moreover the action of the fundamental vector field at that tangent space decomposes as in \eqref{eq:circle-representation}. The latter can then be expressed as \(n\) canonical rotations of the type 
\begin{equation}
\underline{X} = \sum_{\mu=1}^n im_\mu\left( q^\mu \frac{\partial}{\partial p_\mu} - p_\mu \frac{\partial}{\partial q^\mu}\right)
\end{equation} 
with different weights \(m_k\). By the general form of the Hamilton's  equations \eqref{eq:ham-vf}, this means that the Hamiltonian near the isolated fixed point \(p_0\in F\) can be expanded as
\begin{equation}
H(x) = H(p_0) + \frac{1}{2}\sum_{\mu=1}^{n} im_\mu(p_0) \left( p_\mu(x)^2 + q^\mu(x)^2 \right) + \cdots
\end{equation}
Plugging this expansion into the oscillatory integral, we get the saddle-point approximation
\begin{equation}
\begin{aligned}
\int_M d^np d^nq\ e^{iH(p,q)} &\approx \sum_{p_0\in F} e^{iH(p_0)} \prod_{\mu=1}^n \left( \int dp\ e^{-\frac{m_\mu(p_0)}{2} p^2} \int dq\ e^{-\frac{m_\mu(p_0)}{2} q^2} \right) \\
&\approx \sum_{p_0\in F} e^{iH(p_0)}\frac{(2\pi)^{n}}{\prod_\mu m_\mu(p_0)} 
\end{aligned}
\end{equation}
that, again, is exactly the result of the localization formula above. This motivates in the context of Hamiltonian mechanics, and generalization to infinite-dimensional case, that the denominators appearing in these formulas are exactly the \lq\lq 1-loop determinants\rq\rq\ of a would-be semiclassical approximation to the partition function. More aspects of the equivariant theory in contact with symplectic geometry can be found in \cite{guillemin-sternberg-book}.

%%%%%%%%%%%%%%%%%%%%%%%%%%%%%%%%%%%%%%%%%%%%%%%%% cap 3
\chapter{Supergeometry and supersymmetry}
\label{cha:susy}

\section{Gradings and superspaces}
\label{sec:superspaces}
We give some definitions concerning \emph{graded} spaces and \emph{super}-spaces, that are useful for many applications of the localization theorems in physics. In particular, we will see how to translate the problem of integration of differential forms in the context of supergeometry, and how this is useful to prove the ABBV localization formula for circle actions. Also, in the next chapter we will apply this theorem to path integrals in QM and QFT, where the coordinates over which one integrates are those of a \lq\lq field space\rq\rq\ over a given manifold. To construct a suitable Cartan model over this kind of spaces, it is necessary to introduce a graded structure, that physically means to distinguish between \emph{bosonic} and \emph{fermionic} fields, and some operation that acts as a \lq\lq Cartan differential\rq\rq\ transforming one type of field into the other. These structures arise in the context of \emph{supersymmetric field theories}, or in the context of \emph{topological field theories}, and the differentials here are called \emph{supersymmetry} transformations or \emph{BRST} transformations. The precise mathematics behind this is a great subject and we do not seek to be complete here, we just give some of the basic concepts that are necessary to understand what follows. For a more extensive review of the subject, we suggest for example \cite{cattaneo-supergeometry}.

\subsection{Definitions}

To understand the concept of a supermanifold, we need first to recall the linearized case. We already introduced a \emph{graded module} or \emph{graded algebra} \(V\) over a ring \(R\), that is a collection of \(R\)-modules \(\lbrace V_n\rbrace_{n\in\mathbb{Z}}\) such that \(V = \bigoplus_{n\in\mathbb{Z}} V_n\). If \(V\) is an algebra, it must also satisfy \(V_n V_m \subseteq V_{n+m}\). An element \(a\in V_n\) for some \(n\) is called \emph{homogeneous} of \emph{degree} \(\mathrm{deg}(a) \equiv |a| := n\). We now can specialize to the case of \emph{super}- vector spaces (or modules) and \emph{super}-algebras.
\begin{defn} A \emph{super vector space} is a \(\mathbb{Z}_2\)-graded vector space \(V=V_0 \oplus V_1\) where \(V_0,V_1\) are vector spaces. Its \emph{dimension} as a super vector space is defined as \(\dim{V} := \left( \dim{V_0} | \dim{V_1}\right)\). A \emph{superalgebra} is a super vector space \(V\) with the product satisfying \[ 
V_0 V_0 \subseteq V_0 \ ; \quad V_0 V_1 \subseteq V_1 \ ; \quad V_1 V_1 \subseteq V_0 .\]
A \emph{Lie superalgebra} is a superalgebra where the product \( [\cdot,\cdot]:V\times V\to V\), called \emph{Lie superbracket}, satisfies also \[
\begin{aligned}
&[a,b] = -(-1)^{|a||b|}[b,a] & (\mathrm{supercommutativity}), \\
&(-1)^{|a||c|}[a,[b,c]] + (-1)^{|b||c|}[c,[a,b]] + (-1)^{|a||b|}[b,[c,a]] =0 & (\mathrm{super\ Jacobi\ identity}).
\end{aligned}
\]
\end{defn} 

\begin{defn} The \emph{k-shift} of a graded vector space \(V\) is the graded vector space \(V[k]\) such that \((V[k])_n = V_{n+k}\) \(\forall n\in \mathbb{Z}\). \end{defn}

A few remarks are in order. First, it is clear that every graded vector space has naturally also the structure a super vector space, if we split its grading according to \lq\lq parity\rq\rq:
\begin{equation}
\label{eq:Z-Z2-grading}
V_{even} := \bigoplus_{n\in 2\mathbb{Z}} V_n \ , \qquad V_{odd} := \bigoplus_{n\in 2\mathbb{Z}+1} V_n .
\end{equation}
In physics, the \(\mathbb{Z}\)-grading occurs on the \lq\lq field space\rq\rq\ as the so-called \emph{ghost number}, while the \(\mathbb{Z}_2\)-grading with respect to parity distinguish between \emph{bosonic} and \emph{fermionic} coordinates. Second, we notice that every vector space \(V\) can be considered as a (trivial) super vector space, if we think of it as \(V=V\oplus 0\) in even degree or \(V[1]=0\oplus V\) in odd degree. Notice that the even/odd parts of a super vector space can be considered as eigenspaces of an automorphism \(P:V\to V\) such that \(P^2=id_V\). In this sense, a super vector space is a pair \((V,P)\) made by a vector space and the given automorphism \(P\).

Morphisms of graded vector spaces are graded linear maps, \textit{i.e.} grading preserving maps:
\begin{defn} A \emph{graded linear map} \(f\) between graded vector spaces \(V\) and \(W\) is a collection of linear maps \(\lbrace f_k : V_k \to W_k \rbrace_{k\in\mathbb{Z}}\). A linear map of \emph{k-degree} is a graded linear map \(f:V\to W[k]\). \end{defn}

Now we can turn to the non-linear case and consider \emph{supermanifolds}.\footnote{Historically two (apparently) different concepts of \emph{supermanifolds} and \emph{graded manifolds} were firstly developed. They both aimed to generalizing the mathematics of manifolds to a non-commutative setting, following different approaches. Eventually it was proven in \cite{batchelor} that their definitions are equivalent.} Locally, they can be thought as extensions of a manifold via \lq\lq anticommuting coordinates\rq\rq: if we take an open set \(U\subset \mathbb{R}^n\) and a set of coordinates \(\lbrace x^\mu:U\to\mathbb{R}\rbrace_{\mu=1,\cdots,m}\), we can consider a set of additional coordinates \(\lbrace \theta^i\rbrace_{i=1,\cdots,n}\) with the algebraic properties
\begin{equation}
\label{eq:graded-coordinates}
x^\mu \theta^i = \theta^i x^\mu \ , \qquad \theta^i \theta^j = - \theta^j \theta^i .
\end{equation}
The anticommuting \(\lbrace \theta^i\rbrace\) can be thought as generators of \(\bigwedge (V^*)\) for some vector space \(V\), and the product between them and coordinates of \(C^\infty(U)\) is then interpreted as a tensor product in \( C^\infty(U) \otimes \bigwedge (V^*)=:C^\infty (U\times V[1])\).\footnote{Being generators of an exterior algebra, \(\theta^i\) are called \lq\lq Grassmann-odd\rq\rq\ coordinates, while \(x^\mu\) are called \lq\lq Grassmann-even\rq\rq\ consequently. This terminology is commonly inherited by every graded object (vector fields, forms, etc.) on the supermanifold.}  If we then patch together different open sets we get globally a manifold structure, with a modified atlas made by a \emph{graded} ring of local functions \(C^\infty (U\times V[1])\). More formally, we define:
\begin{defn} A (smooth) \emph{supermanifold} \(SM\) of dimension \((m|n)\) is a pair \((M,\mathcal{A})\), where \(M\) is a \(C^\infty\)-manifold of dimension \(m\), and \(\mathcal{A}\) is a sheaf of \(\mathbb{R}\)-superalgebras  that makes \(SM\) locally isomorphic to \[
\left( U, C^\infty(U)\otimes \bigwedge(V^*) \right) \]
for some \(U\subseteq \mathbb{R}^m\) open and some vector space \(V\) of  finite dimension \(\dim{V}=n\). \(M\) is called the \emph{body} of \(SM\) and \(\mathcal{A}\) is called the \emph{structure sheaf} (or, sometimes, \lq\lq soul\rq\rq) of \(SM\).\footnote{Note that also a regular d-dimensional smooth manifold can be viewed as a pair \((M,\mathcal{O}_M)\) composed by a topological space \(M\) (Hausdorff and paracompact) with a structure sheaf of local functions \(\mathcal{O}_M:\ \mathcal{O}_M(U)=C^\infty(U)\) for every \(U\subseteq M\) open, such that locally every \(U\) is isomorphic to a subset of \(\mathbb{R}^d\).}
\end{defn}

We just associated to a real manifold \(M\) a graded-commutative algebra \(C^\infty(SM)\) of functions over \(SM\).
%, that are usually called (scalar) \emph{superfields}. 
Locally in a patch \(U\subseteq M\), this matches the idea above of having coordinate systems as tuples \((x^\mu,\theta^i)_{\stackrel{\mu = 1,\cdots,m}{i=1,\cdots,n}}\) with the property (\ref{eq:graded-coordinates}). In particular, any local function \(\Phi\in \mathcal{A}(U)\) can be trivialized  with respect to the graded basis of \(\bigwedge(V^*)\):
\begin{equation}
\label{eq:superfield}
\Phi(x,\theta) = \Phi^{(0)}(x) + \Phi_i^{(1)}(x)\theta^i + \Phi_{ij}^{(2)}\theta^i\wedge\theta^j +\cdots + \Phi^{(n)}_{i_1,\cdots,i_n} \varepsilon^{i_1 \cdots i_n} \theta^{1}  \wedge \cdots \wedge \theta^{n}
\end{equation}
where \(\Phi^{(l)}_{i_1\cdots i_l}\in C^\infty(U)\) \(\forall l\in\lbrace 0,\cdots,n\rbrace\). The restriction to the zero-th degree \(\epsilon: \mathcal{A}\to C^\infty_M\) such that \(\epsilon(\Phi) :=\Phi^{(0)}\) is usually called the \emph{evaluation map}.
\begin{ex} To every super vector space \(V = V_0 \oplus V_1\) we can associate the supermanifold \[
\hat{V} = \left( V_0 , C^\infty(V_0) \otimes \Lambda(V_1^*)\right) \cong \mathbb{R}^{\dim(V_0)|\dim(V_1)}.\]
More generally, to every vector bundle \(E\to M\) with sections \(\Gamma(E)\) we can associate the \emph{odd vector bundle}, denoted \(\Pi E\) or \(E[1]\), that is the supermanifold with body \(M\) and structure sheaf \(\mathcal{A}=\Gamma\left(\bigwedge E^*\right)\).  The \emph{odd tangent bundle} \(\Pi TM\) is the supermanifold with \(\mathcal{A}=\Gamma\left( \bigwedge T^*M\right)\), \textit{i.e.}\ globally the functions here are the differential forms on \(M\), \(C^\infty(\Pi TM) = \Omega(M)\). Coordinates on \(\Pi TM\) are just \((x^\mu , dx^\mu)_{\mu=1\cdots,m}\), exactly as the coordinates on the tangent bundle \(TM\), but now we consider them as generators of a graded algebra. 
\end{ex}

Morphisms of supermanifolds can be given in terms of local morphisms of superalgebras, that respect compatibility between different patches. In particular, a morphism \((f,f^{\#}):SM\to SN\) is a pair such that \(f:M\to N\) is a diffeomorphism, and for every \(U\subseteq M\) there is a morphism of superalgebras \(f_U^{\#}:\mathcal{A}_M(U)\to \mathcal{A}_N(f(U))\) that respects \(f^{\#}_V \circ \mathrm{res}_{U,V} = \mathrm{res}_{f(U),f(V)} \circ f^{\#}_U\), where \(\mathrm{res}_{U,V}\) is the restriction to a subset \(V\subseteq U\). In less fancy words, if \(\dim{SM}=(m|p)\) and \(\dim{SN}=(n|q)\), a local coordinate system \((x,\theta)\) in \(SM\) is mapped through \(n\) functions \(y^\nu = y^\nu(x,\theta)\) and \(q\) functions \(\varphi^j = \varphi^j(x,\theta)\) to a coordinate system \((y,\varphi)\) of \(SN\).

A vector field \(X\) on a supermanifold \(SM\), or \emph{supervector field}, is a derivation on \(C^\infty (SM)\). Locally, considering \(U\subseteq M\) open and \(\mathcal{A}(U)=C^\infty(U)\otimes \bigwedge(V^*)\), it can be expressed with respect to a coordinate system \((x,\theta)\) as
\begin{equation}
\label{eq:supervector-field}
X =X^\mu (x,\theta) \frac{\partial}{\partial x^\mu} + X^i (x,\theta) \frac{\partial}{\partial \theta^i} ,
\end{equation}
where \((\partial / \partial x^\mu)\) acts as the corresponding vector field in \(\Gamma(TM)\) on the \(C^\infty(U)\) components and acts trivially on the odd coordinates \(\theta^i\); \((\partial / \partial \theta^i)\) acts  trivially on \(C^\infty(U)\), and as an \emph{interior multiplication} by the dual basis vector \(u_i \in V\): \(\frac{\partial}{\partial \theta^i}\theta^j := \theta^j(u_i) = \delta^j_i\). \(X^\mu , X^i\) are local sections in \(\mathcal{A}(U)\). Supervectors on \(SM\) form the tangent bundle \(TSM\). Notice that, in particular \((\partial / \partial x^\mu)\) preserves the grading of an homogeneous function, \textit{i.e.} it is a derivation of degree 0, while \((\partial / \partial \theta^i)\) shifts the grading by -1.
\begin{defn} A \emph{graded vector field} of degree \(k\) on \(SM\) is a graded linear map \(X: C^\infty(SM) \to C^\infty(SM)[k]\) that satisfies the graded Leibniz rule: \[ 
X(\phi \psi) = X(\phi)\psi + (-1)^{k|\phi|}\phi X(\psi)  \]
for any \(\phi,\psi\in \mathcal{A}\) of pure degree.
The \emph{graded commutator} between graded vector fields \(X,Y\) is defined as \[ 
[X,Y] := X\circ Y - (-1)^{|X||Y|}Y\circ X . \]\end{defn}
From what we said above, partial derivatives \((\partial/\partial x^\mu)\) with respect to even coordinates commute between each other, while \((\partial/\partial \theta^i)\) anticommute, being respectively graded vector fields of degree 0 and -1. Then from (\ref{eq:supervector-field}) and (\ref{eq:superfield}) we see that any supervector field can be decomposed with respect to the \(\mathbb{Z}_2\)-grading given by the parity, as the sum \(X=X_{(0)}+X_{(1)}\) of an even (\emph{bosonic}) and an odd (\emph{fermionic}) vector field. This makes \(\left( \Gamma(TSM),[\cdot,\cdot]\right)\) into a Lie superalgebra. The \emph{value} at a point \(p\in M\) of a supervector field \(X\in\Gamma(TSM)\) is defined through the evaluation map:
\begin{equation}
X_p (\Phi) := \epsilon_p (X(\Phi)) = \left[ X^\mu (x,\theta) \frac{\partial \Phi}{\partial x^\mu} + X^i (x,\theta) \frac{\partial \Phi}{\partial \theta^i} \right]_{\stackrel{x=x(p)}{\theta=0}} .
\end{equation}
Clearly a super vector field \(X\) is not determined by its values at points, since the evaluation map throws away all the dependence on the Grassmann-odd coordinates \(\theta^i\) in the coefficient functions \(X^\mu, X^i\). This means that at every point \(p\in M\), \(T_p SM\) is a super vector space generated by the symbols \(\partial / \partial x^\mu, \partial/\partial \theta^i\) of opposite degrees, with \emph{real} coefficients. We collect this result in the following proposition.

\begin{prop}
\label{prop:isom-supervectorbundles}
Let \(SM=(M,\mathcal{A})\) be a supermanifold such that for every chart \(U\subseteq M\) \(\mathcal{A} = C^\infty(U) \otimes \Lambda(V^*)\). Then at every point \(p\in M\), \(T_p SM \cong T_p M \oplus V[1]\) as real super vector spaces. For an odd vector bundle \(\Pi E\) this specializes as \(T_p \Pi E \cong T_p M \oplus E_p[1]\).
\end{prop}

%This induces an isomorphism of vector spaces \((T_p SM)_0 \cong T_p M\) at any point \(p\in M\). If the supermanifold is constructed as the odd vector bundle corresponding to a certain vector bundle \(E\to M\), we have the following:
%\begin{prop}
%\label{prop:isom-supervectorbundles}
%Consider the supermanifold \(\Pi E\) related to the vector bundle \(E\to M\). For any \(p\in M\), there is an isomorphism of \(\mathbb{Z}_2\)-graded vector spaces \(T_p \Pi E \cong T_p M \oplus E_p[1]\).
%\end{prop}
%\begin{proof}
%The isomorphism of the even parts \((T_p\Pi E )_0 \cong T_p M\) is given by the evaluation map. The odd parts are canonically identified by the interior multiplication described above: for any vector \(v\in E_p\) there is a derivation at \(p\), \(\iota_v\), and vice versa. In coordinates, \(u_i \leftrightarrow (\partial/\partial \theta^i)\).
%\end{proof}

Notice that one has always really both a \(\mathbb{Z}\) and a \(\mathbb{Z}_2\) grading of functions and supervector fields, analogously to the remark (\ref{eq:Z-Z2-grading}). As already mentioned, in field theory and in particular in the BRST formalism, the first one is called \emph{ghost number}, while the second one is the distinction between \emph{bosonic} and \emph{fermionic} degrees of freedom in the theory. The physical (\textit{i.e.} gauge-invariant) combinations are those of ghost number zero.

From the point of view of equivariant cohomology, it is very useful to relate the algebraic models we saw in Chapter \ref{cha:equivariant cohomology} to these graded manifold structures. In particular, in field theory we interpret the graded complex of fields as a Cartan model, with a suitable graded equivariant differential given in terms of supersymmetry or BRST transformations. Generically, a \emph{differential} in supergeometry can be interpreted as a special supervector field on a supermanifold:
\begin{defn} A \emph{cohomological vector field} \(Q\) on a supermanifold \(SM\) is a graded supervector field of degree +1 satisfying \[
[Q,Q]=0 . \] \end{defn}
It is immediate that any cohomological vector field corresponds to a \emph{differential} on the algebra of functions \(C^\infty(SM)\), since being it of degree +1, \(Q\circ Q= (1/2)[Q,Q]=0\). For example, consider the de Rham differential \(d:\Omega(M)\to \Omega(M)\) on a regular smooth manifold \(M\). It corresponds to the cohomological vector field on the odd tangent bundle \(\Pi TM\) given in local coordinates by
\begin{equation}
d = \theta^\mu \frac{\partial}{\partial x^\mu} ,
\end{equation}
where now \(\theta^\mu \equiv dx^\mu\) are odd coordinate functions on \(\Pi TM\). Similarly, the \emph{interior multiplication} \(\iota_X :\Omega(M)\to \Omega(M)\) with respect to some vector field \(X\in \Gamma(TM)\) is a nilpotent supervector field on \(\Pi TM\) of degree -1,
\begin{equation}
\iota_X = X^\mu \frac{\partial}{\partial \theta^\mu} .
\end{equation}

\subsection{Integration}

In the following we will use this graded machinery to translate the problem of integration of differential forms \(\Omega(M)\) on a smooth manifold \(M\), to an integration over the related odd tangent bundle \(\Pi TM\). For this, we need to consider differential forms on a supermanifold \(SM\). If \(SM\) has local coordinates \((x^\mu,\theta^i)\), we can locally form an algebra generated by the 1-forms \(dx^\mu , d\theta^i\), where now \(d\) is the de Rham differential on \(SM\), acting as a cohomological vector field on \(\Pi TSM\):
\begin{equation}
d = dx^\mu \partial_{x^\mu} + d\theta^i \partial_{\theta^i} .
\end{equation}
The odd tangent bundle \(\Pi TSM\) has thus coordinates \((x^\mu,\theta^i,dx^\mu,d\theta^i)\), where now \(dx^\mu\) is odd whereas \(d\theta^i\) is even.\footnote{To be more precise, the algebra of functions locally generated by \((x^\mu,\theta^i,dx^\mu,d\theta^i)\) on \(\Pi TSM\) has \emph{bi-grading}, \emph{i.e.}\ it inherits a \(\mathbb{Z}_2\) grading from the original supermanifold \(SM\) and a \(\mathbb{Z}\) grading from the action of the de Rham differential (the form-degree). The coordinates have thus bi-degrees \[
x^\mu : (\text{even},0), \qquad \theta^i : (\text{odd},0), \qquad
dx^\mu : (\text{even},1), \qquad d\theta^i : (\text{odd},1), \]
which result in a total even degree for  \(d\theta^i\) and a total odd degree for \(dx^\mu\).} From this point of view, the de Rham differential acts as a \lq\lq supersymmetry\rq\rq\ transformation:
\begin{equation}
d: \left\lbrace \begin{aligned} x^\mu &\mapsto dx^\mu \\
\theta^i &\mapsto d\theta^i \end{aligned} \right.
\end{equation} 
exchanging bosonic coordinates with fermionic coordinates. Since the odd 1-forms \(d\theta^i\) are commuting elements, it is not possible to construct, at least in the usual sense, a form of \lq\lq top degree\rq\rq on \(SM\). We will thus interpret integration over the odd coordinates by the purely algebraic rules of Berezin integration for Grassmann variables:
\begin{equation}
\label{eq:berezin-rules}
\int d\theta^i \theta^i = 1 , \qquad \int d\theta^i 1 = 0 ,
\end{equation}
and such that Fubini's theorem holds for multidimensional integrals. We see that symbolically
\begin{equation}
\int d\theta^i \leftrightarrow \frac{\partial}{\partial \theta^i} ,
\end{equation}
and in particular \(\int d\theta^i \frac{\partial}{\partial \theta^i}\Phi(\theta^i) =0\) always holds. We will use the important property of the Berezin integral:
\begin{equation}
\label{eq:gaussian-grassman}
\int d^n\theta\ e^{-\theta^i A_{ij} \theta^j} = \mathrm{Pf}(A)
\end{equation}
where \(d^n\theta \equiv d\theta^1 d\theta^2 \cdots d\theta^n\), to be compared with the usual Gaussian integral for real variables
\begin{equation}
\label{eq:gaussian-real}
\int d^nx\ e^{-x^i A_{ij} x^j} = \frac{\pi^{n/2}}{\sqrt{\det(A)}}.
\end{equation}
Notice that under an homogeneous change of coordinates \(\varphi^i = B^i_j \theta^j\), the \lq\lq measure\rq\rq\ shifts as  \(\int d^n \theta \to \det{B}\int d^n \varphi\), such that  the Gaussian integral (\ref{eq:gaussian-grassman}) is invariant under similarity transformations. For a mathematically refined theory of superintegration, we suggest looking at \cite{witten-superintegration}.

Concerning the integration of functions on the odd tangent bundle \(\Pi TM\), we notice how, making use of the Berezin rules (\ref{eq:berezin-rules}), this is nothing but a reinterpretation of the usual integrals of differential forms on \(M\). If \(\dim{M}=d\), the integral of the form \(\omega = \sum_i \omega^{(i)}\), where \(\omega^{(i)}\in \Omega^i(M)\) is
\begin{equation}
\int_M \omega = \int_M \omega^{(d)} = \int_M d^dx\ \omega^{(d)}(x) ,
\end{equation}
selecting the top-form of degree \(d\). If we consider the same form \(\omega\in C^\infty(\Pi TM)\), its trivialization in coordinates \((x,\theta)\) is
\begin{equation}
\omega(x,\theta) = \sum_i \omega^{(i)}_{\mu_1 \cdots \mu_i} (x) \theta^{\mu_1}\cdots\theta^{\mu_i} .
\end{equation}
Now the Berezin integration over \(d^d\theta\) selects just the term with the right number of \(\theta\)'s, giving
\begin{equation}
\label{eq:integral-odd-tangent-bundle}
\int_{\Pi TM} d^d x d^d\theta \ \omega(x,\theta) = \int_M d^d x\ \omega^{(d)}(x) \int d^d\theta\ \theta^d\theta^{d-1}\cdots\theta^1 = \int_M d^dx\ \omega^{(d)}(x) .
\end{equation}

\section{Supergeometric proof of ABBV formula for a circle action}
\label{sec:ABBVproof}

We give now a proof of the ABBV integration formula for a \(U(1)\)-action, starting from the expression (\ref{eq:loc-principle-U1}). The \lq\lq localization 1-form\rq\rq\ is chosen as
\begin{equation}
\beta := g(\underline{T},\cdot) 
\end{equation}
where \(g\) is a \(U(1)\)-invariant metric and \(\underline{T}\) is the fundamental vector field corresponding to the generator \(T\in \mathfrak{u}(1)\). In local coordinates, the action of the Cartan differential \(d_C = d + \iota_T \) on this 1-form can be written as
\begin{equation}
\begin{split}
d_C \beta &= B_{\mu\nu}(x) dx^\mu dx^\nu + g_{\mu\nu}(x) T^\mu (x) T^\nu (x) \\
B_{\mu\nu} &= (\nabla_\mu T)_\nu - (\nabla_\nu T)_\mu
\end{split}
\end{equation}
where, again, we suppressed the \(S(\mathfrak{u}(1)^*)\) generator setting \(\phi=-1\).

We use the result of the last section to rewrite \eqref{eq:loc-principle-U1} as an integral over the odd tangent bundle \(\Pi TM\), identifying the odd coordinates \(\theta^\mu \equiv dx^\mu\):
\begin{equation}
\label{proofABBV1}
\begin{split}
I[\alpha] &= \int_M \alpha = \lim_{t\to\infty}\int_M \alpha e^{-td_C \beta} \\
&= \lim_{t\to\infty}\int_{\Pi TM} d^dx d^d\theta\ \alpha(x,\theta)\exp{\left\lbrace -t B_{\mu\nu}(x)\theta^\mu \theta^\nu - t g_{\mu\nu}(x)T^\mu(x)T^\nu(x)\right\rbrace} ,
\end{split}
\end{equation}
where the equivariantly closed form \(\alpha\) is the sum of \(U(1)\)-invariant differential forms in \(\Omega(M)^{U(1)}\) suppressing the \(S(\mathfrak{u}(1)^*)\) generator,
\begin{equation}
\alpha(x,\theta) = \sum_i \alpha^{(i)}_{\mu_1 \cdots \mu_i} (x) \theta^{\mu_1}\cdots\theta^{\mu_i} ,
\end{equation}
such that \(d_C\alpha = \left( \theta^\mu \partial_{x^\mu} + T^\mu \partial_{\theta^\mu} \right)\alpha = 0\). Using the Gaussian integrals \eqref{eq:gaussian-grassman} and \eqref{eq:gaussian-real}, we have the following delta-function representations for Grassmann-even and Grassmann-odd variables
\begin{equation}
\begin{aligned}
\delta^{(n)}(y) &= \lim_{t\to\infty} \left(\frac{t}{\pi}\right)^{n/2}\sqrt{\det{A}}\ e^{-ty^\mu A_{\mu\nu} y^\nu} \\
\delta^{(n)}(\eta) &= \lim_{t\to\infty} t^{-n/2}\frac{1}{\mathrm{Pf}A}\ e^{-t A_{\mu\nu}\eta^\mu \eta^\nu}
\end{aligned}
\end{equation}
where the limits are understood in the weak sense. Multiplying and dividing by \((t^{n/2})\) in \eqref{proofABBV1}, and using the delta-representation we rewrite the integral as
\begin{equation}
I[\alpha] =  \pi^{d/2} \int_{\Pi TM} d^dx d^d\theta\ \alpha(x,\theta) \frac{\mathrm{Pf}B(x)}{\sqrt{\det{g}(x)}} \delta^{(d)}(T(x))\delta^{(d)}(\theta) .
\end{equation}
The delta function on the odd coordinates simply puts \(\theta^\mu =0\), that is analogous to selecting the top-degree form in \eqref{proofABBV1}, so that it remains \(\alpha(x,0)\equiv \alpha^{(0)}(x)\in C^\infty(M)\). The delta function on the even coordinates instead selects the values at which \(T^\mu(x)=0\), that corresponds to the fixed point set \(F\hookrightarrow M\) of the \(U(1)\)-action. Suppose this fixed point set to be of dimension 0, \textit{i.e.}\ composed by isolated points in \(M\). If this is the case, we can simply separate the integral \(\int d^dx\) in a sum of integrals, each of which domain \(\mathcal{D}(p): p \in F\) contains one and only one of those fixed points, and in each of them apply the delta function
\begin{equation}
\label{proofABBV2}
\begin{split}
I[\alpha] &=  \pi^{d/2} \sum_{p\in F} \int_{\mathcal{D}(p)} d^dx\ \alpha^{(0)}(x) \frac{\mathrm{Pf}B(x)}{\sqrt{\det{g}(x)}} \delta^{(d)}(T(x)) \\
&= \pi^{d/2} \sum_{p\in F} \frac{\alpha^{(0)}(p)}{|\det{dT}|(p)} \frac{\mathrm{Pf}B(p)}{\sqrt{\det{g}(p)}} .
\end{split}
\end{equation}
Here  the factor \(|\det{dT}|(p)\) is the Jacobian from the change of variable \(y^\mu := T^\mu (x)\). At any point \(p\in F\) we have \((\nabla_\mu T)_\nu (p) = \partial_\mu T_\nu (p)=\partial_\mu T^\rho(p) g_{\rho\mu}(p)\) since \(T^\mu (p)=0\), so the pfaffian in the numerator becomes 
\begin{equation}
\mathrm{Pf}B(p) = \sqrt{2^d}\ \mathrm{Pf}{dT}(p)\sqrt{\det{g}(p)} ,
\end{equation}
and we get the result
\begin{equation}
I[\alpha] = (2\pi)^{d/2} \sum_{p\in F} \frac{\alpha^{(0)}(p)}{\mathrm{Pf}dT(p)} .
\end{equation}

Notice that, at \(p\in F\), the operator \(dT(p)=\partial_\mu T^\nu \theta^\mu \otimes \partial_\nu = -[\underline{T},\cdot]\) coincide up to a sign with the infinitesimal action  \(\mathcal{L}_T\) of \(T\in \mathfrak{u}(1)\) on the tangent space \(T_pM\), just because here \(T^\mu(p)=0\). If we consider a continuous fixed point set, so that \(F\) is a regular submanifold of \(M\), it is not possible to use the delta functions like in (\ref{proofABBV2}), but we can consider the decomposition as a disjoint union \(M=F\sqcup N\). Points far away from \(F\)  give a zero contribution to \(I[\alpha]\) in the limit \(t\to \infty\), so we can consider a neighborhood of \(F\) and split here the tangent bundle as \(TM\cong i_* TF \oplus TN\)  where \(TN\) is the \emph{normal bundle} to \(F\) in \(M\) (\(i\) is the inclusion map). Consequently, in this neighborhood we can split the coordinates \((x^\mu,\theta^\mu)\) on the odd tangent bundle in tangent and normal to \(F\), and rescale the normal components as \(1/\sqrt{t}\),
\begin{equation}
x^\mu = x^\mu_0 + \frac{x^\mu_\perp}{\sqrt{t}}, \qquad \theta^\mu = \theta_0^\mu + \frac{\theta_\perp^\mu}{\sqrt{t}} .
\end{equation}
The measure simply splits as \(d^dx d^d\theta = d^nx_0 d^{\hat{n}}x_\perp d^n\theta_0 d^{\hat{n}}\theta_\perp\), where \(n+\hat{n}=d\), thanks to the Berezin integration rules \eqref{eq:berezin-rules}. Expanding \(B_{\mu\nu}(x), g_{\mu\nu}(x), T^\mu(x)\) around the \(x_0\) components, and taking the limit \(t\to \infty\), the integral becomes \cite{niemi-equiv-morse}
\begin{equation}
\begin{aligned}
I[\alpha] = \int d^nx_0 d^n\theta_0\ \alpha(x_0,\theta_0) &\int d^{\hat{n}}x_\perp d^{\hat{n}}\theta_\perp\ \\ 
& \exp{\left\lbrace -B_{\mu\sigma}(x_0)\left(B_\nu^\sigma(x_0) + R^\sigma_{\nu\lambda\rho}(x_0)\theta_0^\lambda\theta_0^\rho\right)x_\perp^\mu x_\perp^\nu - B_{\mu\nu}(x_0) \theta_\perp^\mu \theta_\perp^\nu \right\rbrace} 
 \end{aligned}
\end{equation}
where \(R^\sigma_{\nu\lambda\rho}(x_0)\) is the curvature relative to the metric \(g\). The integrals over the normal coordinates are Gaussian, giving the exact \lq\lq saddle point\rq\rq\ contribution
\begin{equation}
\frac{1}{\mathrm{Pf}_N \left( \frac{R+B}{2\pi} \right) (x_0)} = \frac{1}{ \left. e_T(R) \right|_N} ,
\end{equation}
where \(\left. e_T(R)\right|_N\) is the \(U(1)\)-\emph{equivariant Euler class} of the normal bundle to \(F\). We notice that this matches the definition in Appendix \ref{app:char-classes}, since \(B:= \nabla T\), seen as an element of the adjoint bundle \(\Omega^0(M;\mathfrak{gl}(n))\), is a moment map for the Riemannian curvature \(R\) satisfying (as can be checked by direct computation)
\begin{equation}
\nabla B^\sigma_\nu = -\iota_T R^\sigma_\nu = R^\sigma_\nu(\cdot,\underline{T}) ,
\end{equation}
where the covariant derivative acts as in the adjoint bundle with respect to the Levi-Civita connection, \(\nabla = d + [\Gamma,\cdot]\). As the final piece, the Berezin integration selects the component of \(\alpha\) of degree \(\dim{F}\) evaluated on \(F\), that is just the pull-back \(i^*\alpha\) along the inclusion map. Summarizing, we are left with
\begin{equation}
I[\alpha] = \int_{F} \frac{i^*\alpha}{\left. e_T(R)\right|_N}
\end{equation}
as presented in Section \ref{sec:ABBV}, for the case of a circle action. 

Notice that in Section \ref{sec:ABBV} we called the moment map \(B^\sigma_\nu(p) = \mu_{T}(p)^\sigma_\mu \in \mathrm{End}(T_pM)\) at a fixed point \(p\). If the fixed point is isolated, the tangent space \(T_pM\) splits, as in \eqref{ABBV-weigth-decomposition-tangent-space}, as the direct sum of the weight spaces of the \(U(1)\)-representation, and the vector field \(\underline{T}\) acts as a rotation in any subspace. On a suitable coordinate basis
\begin{equation}
\underline{T} = \sum_i v_i \left( x^i\frac{\partial}{\partial y^i} - y^i\frac{\partial}{\partial x^i}\right),
\end{equation}
so that the moment map block-diagonalizes as
\begin{equation}
B^\sigma_\nu = \partial_\nu T^\sigma = \left( 
\begin{array}{ccccc}
0 & -v_1 & \cdots & & \\
v_1 & 0  & & & \\
\cdots & & 0 & -v_2 & \\
 & & v_2 & 0 & \\
 & & & & \ddots
\end{array}\right).
\end{equation}
Taking the pfaffian then one gets exactly the product of the weights (or the \lq\lq exponents\rq\rq) \(v_i\) of the circle action, so that the Euler class results 
\begin{equation}
e_T(R) = (2\pi)^{-\dim(M)/2} \prod_i v_i .
\end{equation}
This recovers \eqref{eq:euler-weights} for the case of a circle action. See \cite{atiyah-bott-localization} for the result in presence of a torus action.

\section{Introduction to Poincaré-supersymmetry}
\label{sec:susy}

We have seen in the last section how it is useful to translate the integration problem of a differential form on the manifold \(M\), into an integration over the supermanifold \(\Pi TM\). Here the differential forms \(\Omega(M)\) are seen as the \emph{graded} ring of functions over \(\Pi TM\), and the differential \[
d_C = d + \iota_T \]
being the sum of two graded derivations\footnote{Both \(d\) and \(\iota_T\) are nilpotent supervector fields on \(\Pi TM\). Without suppressing the degree-2 generator \(\phi\) of \(S(\mathfrak{u(1)}^*)\), it is apparent that their sum \(d_C\) is a  cohomological vector field, \textit{i.e.}\ a good differential of degree +1, on the subspace of \(U(1)\)-invariant forms.}
of degree \(\pm 1\), can be viewed as an infinitesimal \emph{supersymmetry} transformation mapping odd-degree (\emph{fermionic}) forms to even-degree (\emph{bosonic}) forms. In field theory, we already mentioned that the presence of a supersymmetry on the relevant complex of fields often arises in two different ways: 
\begin{itemize}
\item The differential \(d_C\) is represented in the physical model by a BRST-like supercharge, introduced because of some gauge freedom. In this case, the complex of fields (analogously to the graded ring of functions \(\Omega(M)\)) is the BRST complex, and the grading is referred to as \emph{ghost number}. Physical states of the quantum field theory are then created by fields of 0-degree, \textit{i.e.}\ the functions on \(M\). We could refer to this type of supersymmetry as a \lq\lq hidden\rq\rq\ one, coming from the original internal gauge symmetry of the model. In Hamiltonian systems, as we shall see in the next chapter, this gauge freedom can be simply associated to the Hamiltonian flow.

\item The original theory could also be explicitly endowed with a supersymmetry. In this case the base space has a supermanifold structure, and the grading of the field complex follows. This is the case of QFT with Poincaré supersymmetry, where the action of a \emph{supercharge} as generator of the super-Poincaré algebra can be interpreted as an equivariant differential.
\end{itemize}
In this and the next sections we will review the geometric setup of Poincaré supersymmetry and its generalization to curved spacetimes, interesting case for practical applications of the localization technique in QFT.

As QFT (with global Poincaré symmetry) is formulated on Minkowski spacetime\footnote{We are really interested in both the Lorentzian and the Euclidean case, so we will express both the Minkowski and Euclidean spaces as \(\mathbb{R}^d\) without stressing of the signature in the notation. The choice of metric will be clear from the context.}
\begin{equation}
\label{eq:spacetime}
\mathbb{R}^d \cong ISO(\mathbb{R}^d)/O(d) ,
\end{equation}
we can formulate a Poincaré-supersymmetric theory on a super-extension of this space, coming from a given super-extension of the Poincaré group \(ISO(\mathbb{R}^d)\). We then first introduce the super-Poincaré groups starting from the super-extension of their algebras.

\subsection{Super-Poincaré algebra and superspace}

\begin{defn} \label{defn:susy-algebra}
A super-Poincaré algebra is the extension of the Poincaré algebra \(\mathfrak{iso}(d)\cong \mathbb{R}^d \oplus \mathfrak{so}(d)\) as a Lie superalgebra, via a given \emph{real} spin representation space \(S\) of \(Spin(d)\) taken in odd-degree:
\begin{equation}
\mathfrak{siso}_S (d) \cong \mathfrak{iso}(d) \oplus S[1].
\end{equation}
The super Lie bracket are extended on \(S[1]\) through the \emph{symmetric} and Spin-equivariant bilinear form \(\Gamma:S\times S \to \mathbb{R}^d\):
\begin{equation}
[\Psi,\Phi] := 2\Gamma(\Psi,\Phi) = 2(\overline{\Psi}\gamma^\mu \Phi)P_\mu = 2(\Psi^T C\gamma^\mu \Phi)P_\mu
\end{equation}
where \(P_\mu\) are generators of the translation algebra \(\mathbb{R}^d\), \(C\) is the charge conjugation matrix, \(\overline{\Psi}\in S^*\) is the Dirac adjoint of \(\Psi\),\footnote{In Lorentzian signature \(\overline{\Psi}=\Psi^\dagger \beta\) with \(\beta_{ab}\equiv (\gamma^0)^a_{\ b}\), in Euclidean signature \(\overline{\Psi}=\Psi^\dagger\). This is due to hermitianity of the generators of the Euclidean algebra, unlike the Lorentzian case.} and \(\gamma^\mu\) are the generators of the Clifford algebra acting on \(S\). 
%We scaled the generators \(P_\mu\) of \(\mathbb{R}^d\) to have a factor of \(2\) in front, following the common conventions in physics literature. 
In a real representation, the Majorana condition \(\Psi^T C \stackrel{!}{=} \overline{\Psi}\) is satisfied. The other  brackets involving \(S\) are defined by the natural action of \(\mathfrak{so}(d)\) on it, and by the trivial action on \(\mathbb{R}^d\): 
\begin{align}
\lambda\in \mathfrak{so}(d): \quad &[\lambda, \Psi] := \frac{i}{2} \lambda_{\mu\nu} \Sigma^{\mu\nu} (\Psi) ,  \\
& [\Psi ,\lambda] = -[\lambda ,\Psi] , \\
a\in \mathbb{R}^d : \quad &[a,\Psi] := 0 ,
\end{align}
where \(\Sigma^{\mu\nu} = \frac{i}{2}\gamma^{[\mu}\gamma^{\nu]}\) are the generators of the rotation algebra in the Spin representation \(S\). If the charge conjugation matrix is symmetric in the given representation,  we can further enlarge this superalgebra via a \lq\lq central extension\rq\rq, considering \(\mathfrak{siso}_S(d)\oplus \mathbb{R}\) with the extended brackets
\begin{align}
\Psi,\Phi\in S[1]: \quad &[\Psi,\Phi]:= \Gamma(\Psi,\Phi) + (\Psi^T C\Phi) ,\\
x\in \mathbb{R}, A\in \mathfrak{siso}_S(d): \quad &[x,A]:=0 .
\end{align} 
The super Jacobi identity is satisfied thanks to the Spin-equivariance of the spinor bilinear form:
\begin{equation}
\Gamma\left(e^{R^{(s)}}\Psi, e^{R^{(s)}}\Phi\right) = e^{R^{(v)}}\Gamma(\Psi,\Phi)
\end{equation}
where \(R^{(s)},R^{(v)}\) are the same element \(R\in\mathfrak{so}(d)\) in the spin and vector representations, respectively.
\end{defn}
Often the bracket structure of this superalgebra is given in terms of the generators. Regarding the odd part and picking a basis \(\left\lbrace Q_a\right\rbrace\) of \(S[1]\),
%\footnote{We follow the standard convention in the physics literature, such that an element of the Lie superalgebra is expanded as \(\epsilon = i\epsilon^a Q_a\).} 
their brackets are\footnote{We conventionally raise and lower spinor indices with the charge conjugation matrix: \[
(\gamma^\mu)_{ab} := C_{ac}(\gamma^\mu)^c_{\ b} , \qquad (\gamma^\mu)_{ab} = (\gamma^\mu)_{ba}.\]
Notice that the matrix \(C\) represents an inner product on \(S\), while the Clifford algebra generators \(\lbrace\gamma^\mu\rbrace\) act on \(S\) as endomorphisms, so the index structure follows. A review of classification of Clifford algebras, Spin groups and Majorana spinors can be found in \cite{Figueroa-majorana}.}
\begin{equation}
[Q_a , Q_b] = 2(\gamma^\mu)_{ab} P_\mu + C_{ab}.
\end{equation}
The generators \(\left\lbrace Q_a\right\rbrace\) are referred to as \emph{supercharges}. If the real spin representation on \(S\) is irreducible as a representation of the corresponding Clifford algebra, we have the minimal amount of supersymmetry and we refer to \(\mathfrak{siso}_S(d)\) as to an \(\mathcal{N}=1\) supersymmetry algebra. If instead the representation is reducible, then \(S=\bigoplus_{I=1}^{\mathcal{N}} S^{(I)}\) and we can split the basis of supercharges as \(\left\lbrace Q^I_a\right\rbrace\). This case is referred to as \emph{extended} supersymmetry. In this basis the gamma matrices block-diagonalize as \(\gamma^\mu \otimes \mathbb{I}\), with \(\gamma^\mu\) the (minimal) gamma matrices in every \(S^{(I)}\), and the central extension part separates as \(C\otimes Z\), with \(C\) being the (minimal) charge conjugation matrix in every \(S^{(I)}\) and \(Z\) a matrix of so-called \emph{central charges}. The odd part of the superalgebra then looks like
\begin{equation}
\label{eq:susy-algebra-ext}
[Q_a^I , Q_b^J] = 2(\gamma^\mu)_{ab} \delta^{IJ} P_\mu + C_{ab}Z^{IJ} .
\end{equation}
The matrix \(Z\) must be (anti)symmetric if \(C\) is (anti)symmetric. 
\begin{defn}
The subspace \(\mathfrak{st}_S(d) := \mathbb{R}^d \oplus S[1]\) is a Lie superalgebra itself (if there is no central extension), and can be referred to as the \emph{super-translation algebra}. 
\end{defn}
Even if \(\mathfrak{st}_S(d)\) it is not Abelian, it has the property 
\begin{equation}
[a,[b,c]] = 0 \qquad \forall a,b,c \in \mathfrak{st}_S(d)
\end{equation}
that can be easily checked by the definition. This means that the elements of the corresponding \emph{super-translation group} can be computed exactly using the exponential map and the BHC formula. This space can be identified with the \emph{super-spacetime}.
\begin{defn}
We can define the full super-Poncaré group as
%\footnote{Taking also into consideration all possible reflections, we should use the group \(Pin(d)\) as the double cover of \(O(d)\) for the extension of the full super-Poincaré or Euclidean group.}
\begin{equation}
SISO_S(d) = \exp{(\mathfrak{st}_S(d))} \rtimes Spin(d) ,
\end{equation}
so that we can identify the superspacetime with respect to the spin representation \(S\), analogously to (\ref{eq:spacetime}), as
\begin{equation}
S\mathbb{R}_S^d = SISO_S(d)/Spin(d) \cong \exp{(\mathfrak{st}_S(d))}.
\end{equation}
\end{defn}
As a supermanifold of dimension \((d|\dim(S))\), \(S\mathbb{R}_S^d\) is characterized by its sheaf of functions,
\begin{equation}
\mathcal{A} = C^{\infty}(\mathbb{R}^d) \otimes \bigwedge(S^*) .
\end{equation}
So, \(S\mathbb{R}_S^d\) is the odd vector bundle associated to the spinor bundle over \(\mathbb{R}^d\) of typical fiber \(S\). In particular, on \(S\mathbb{R}_S^d\) we have respectively even and odd coordinates \((x^\mu,\theta^a)\), with \(\mu=1,\cdots,d\) and \(a=1,\cdots,\dim(S)\). As a Lie group, we can get the group operation (the sum by supertranslation) from the exponentiation of its Lie superalgebra. Technically, to use the BHC formula
\begin{equation}
e^A e^B = e^{A+B+\frac{1}{2}[A,B]+\cdots}
\end{equation}
we would like to deal with a \emph{Lie algebra}, so we consider the group operation on coordinates functions instead of points on \(S\mathbb{R}_S^d\), taking the space
\begin{equation}
\left(\mathcal{A}\otimes \mathfrak{st}_S(d)\right)_{(0)} = \left(\mathcal{A}_{(0)}\otimes \mathbb{R}^d\right) \oplus \left( \mathcal{A}_{(1)} \otimes S[1]\right) .
\end{equation}
The Lie brackets on this space are inherited from those on \( \mathfrak{st}_S(d)\) and the (graded) multiplication of functions in \(\mathcal{A}\). The only non-zero ones come from couples of elements of \( \mathcal{A}_{(1)} \otimes S[1]\):
\begin{equation}
[f_1 \otimes \epsilon_1 , f_2 \otimes \epsilon_2 ] = -2 \Gamma(\epsilon_1,\epsilon_2) f_1 f_2 
\end{equation}
where the sign rule has been used since both \(f_1, f_2\) and \(\epsilon_1, \epsilon_2\) are odd.\footnote{We are being a little informal here, but this can be made more rigorous with the help of a construction called \emph{functor of points}. The important thing for us is that in this approach one can work with coordinate functions \(x,\theta\) instead of some would-be \lq\lq points\rq\rq\ on the supermanifold (a misleading concept since we know from the last section that a supermanifold is not a set). See \cite{varadarajan-susy_math} for more details.} We will use the combination (suppressing tensor products) \[
[\theta^a Q_a , \varphi^b Q_b] = -2\Gamma(Q_a,Q_b)\theta^a \varphi^b = -2(\gamma^\mu)_{ab}\theta^a \varphi^b P_\mu \equiv -2(\theta\gamma^\mu \varphi)P_\mu .
\]
This makes \(\left(\mathcal{A}\otimes \mathfrak{st}_S(d)\right)_{(0)}\) into a Lie algebra, by the antisymmetry of the product between odd functions. Representing via the exponential map \(
(x,\theta) \) as \( \exp\left\lbrace i(xP + \theta Q)\right\rbrace \),
where we suppressed also index contractions, we can finally use the BHC formula on this algebra to get the group operation on coordinates:
\begin{equation}
\label{eq:supertranslation_law}
\begin{aligned}
(x,\theta) \, (y,\varphi) &= \exp{\left\lbrace i(xP + \theta Q)\right\rbrace}\exp{\left\lbrace i(yP + \varphi Q)\right\rbrace} \\
&= \exp{\left\lbrace i \left(x P + \theta Q + y P + \varphi Q + \frac{i}{2}[\theta Q,\varphi Q]\right) \right\rbrace} \\
&= \exp{\left\lbrace i \left(x + y + i\theta\gamma \varphi \right)P + i \left(\theta + \varphi \right)Q \right\rbrace} \\
&= \left(x + y - i\theta\gamma \varphi , \theta + \varphi \right).
\end{aligned}
\end{equation}
%As usual, we can reformulate this group law as a natural action of the supertranslation group on superspace given in term of a \emph{left translation} \(l : S\mathbb{R}^d_S \times S\mathbb{R}^d_S \to S\mathbb{R}^d_S\), so that \((x,\theta)\) is obtained by a left translation via \(\exp\lbrace -i(xP+\theta Q)\rbrace\) of \((0,0)\).
If the odd dimension \(\dim{(S)}\) is zero, this reduces to a standard translation in \(\mathbb{R}^d\). The infinitesimal action of the superalgebra \(\mathfrak{st}_S(d)\) is defined as a Lie derivative with respect to the fundamental vector field representing a given element of the supertranslation algebra. For \(\Phi\in \mathcal{A}\), \(\epsilon = \epsilon^a Q_a\in (\mathcal{A}_{(1)} \otimes S[1])\)
\begin{equation}
\delta_\epsilon \Phi(x,\theta) := \mathcal{L}_{\underline{\epsilon}} (\Phi(x,\theta)) = \underline{\epsilon} (\Phi) = \epsilon^a \underline{Q}_a (\Phi(x,\theta)) ,
\end{equation}
where the odd vector field \(\underline{Q}_a\) is associated to the supercharge \(Q_a\) through the \emph{left} translation \eqref{eq:supertranslation_law}:
\begin{equation}
\begin{aligned}
\underline{Q}_a \Phi (x,\theta) &= \left. \frac{\partial}{\partial \varphi^a} \left( e^{-i\varphi Q}\right)^* \Phi (x,\theta) \right|_{\varphi=0} = \left. \frac{\partial}{\partial \varphi^a} \Phi\big( (0,-\varphi)(x,\theta) \big) \right|_{\varphi=0} = \\
&= \partial_\mu \Phi(x,\theta) (i\gamma^\mu \theta)_a + \partial_b\Phi(x,\theta) \delta^b_a .
\end{aligned}
\end{equation}
We recognize then
\begin{equation}
\underline{Q}_a = -\frac{\partial}{\partial \theta^a} + i (\theta\gamma^\mu)_a \frac{\partial}{\partial x^\mu} ,
\end{equation}
and one can check that \( [\underline{Q}_a, \underline{Q}_b] = 2 (\gamma^\mu)_{ab} \underline{P}_\mu\) with \(\underline{P}_\mu = -i \partial/\partial x^\mu\) is associated to the momentum generator.

The rest of the super-Poincaré group and its algebra acts naturally on superspace following the same type of arguments. In particular, the spin generators act in the vector and spin representations on the even and odd sectors, respectively. It is useful to introduce also the fundamental vector fields with respect to \emph{right} supertranslations on \(S\mathbb{R}^d_S\). These are called \emph{superderivatives} and are easily obtained from the law \eqref{eq:supertranslation_law} as\footnote{Remember that left and right actions correspond to opposite signs at the exponent in the definition of the fundamental vector fields.}
\begin{equation}
D_a := \frac{\partial}{\partial\theta^a} + i(\theta\gamma^\mu)_a \frac{\partial}{\partial x^\mu} .
\end{equation}
One can check that indeed they satisfy \([\underline{Q}_a, D_b]=0\) and \([D_a,D_b]= -2(\gamma^\mu)_{ab}P_\mu\), since right invariant and left invariant vector fields form anti-isomorphic algebras.

\subsection{Chiral superspace and superfields}

It is customary to construct supersymmetric field theories starting not from a \emph{real} spin representation, but from a \emph{complex} Dirac representation \(S\) endowed with a \emph{real structure}, \textit{i.e.} an antilinear map \(J:S\to S\) which is an involution (\(J^2=id_S\)).\footnote{\(J\) is the generalization of the \lq\lq complex conjugation\rq\rq\ operation on a \(\mathbb{C}\)-vector space.} In this case, diagonalizing \(J\) the representation splits as \[
S \cong S_\mathbb{R}\otimes \mathbb{C} \cong S^{(+)}\oplus S^{(-)} \]
where \(S_\mathbb{R}\cong S^{(+)}\cong iS^{(-)}\) are a real vector spaces. Choosing some basis, the matrix representing the real structure is \(J=C(\gamma^0)^T\) in Lorenzian signature, or \(J=C\) in Euclidean signature. The \emph{Majorana spinors} are the elements of \(S_\mathbb{R}\), that satisfy \(J(\Psi)=\Psi\), or in matrix notation \(\overline{\Psi}=\Psi^TC\). The subspace of Majorana spinors, taken as a real representation, would then give the super-extension of the last paragraph.

If instead we are interested in working with the whole complex representation \(S\), we are forced to introduce \emph{complexified} supersymmetry algebra and  superspace, and then impose constraints on the resulting objects to properly reduce their degrees of freedom a posteriori. In particular, complex spinors from \(S\) generating the supertranslations are taken satisfying the Majorana condition. This is always the case in QFT. As a paradigmatic example, we can take \(\mathcal{N}=1\) supersymmetry in (3+1)-dimensions, where \(\Psi\) is a Dirac spinor in \(S=\mathbb{C}^4 \cong \mathbb{R}^4 \otimes \mathbb{C}\).  

On the complex representation, we can chose a to work in the \emph{chiral basis} \(\lbrace Q_a, \tilde{Q}_{\dot{a}}\rbrace\), splitted in left- and right-handed Weyl spinors. The dotted and undotted indices now run between \(1,2\). Here the symmetric pairing \(\Gamma: S\times S \to \mathbb{C}^4\) is non-zero only on \(S^{(L/R)}\times S^{(R/L)}\), and and the restriction to the symmetrized subspace \(\Gamma: S^L \odot S^R \to \mathbb{C}^4\) is actually an isomorphism, so the relevant non-zero brackets are
\begin{equation}
\label{eq:susy-algebra-chiral}
[Q_a , \tilde{Q}_{\dot{b}}] = 2(\gamma^\mu)_{a\dot{b}} P_\mu  .
\end{equation}
On the chiral basis, \[
\gamma^\mu = \begin{pmatrix}
0 & \sigma^\mu \\
\bar{\sigma}^\mu & 0 \end{pmatrix} , \qquad
C = \begin{pmatrix}
\varepsilon & 0 \\
0 & -\varepsilon \end{pmatrix} , \qquad \sigma^\mu = (\mathbf{1},\sigma^i), \quad \overline{\sigma}^\mu = (\mathbf{1},-\sigma^i),\]
%\((\gamma^\mu)_{a\dot{b}}=\sigma^\mu_{a\dot{b}}\) and \((\gamma^\mu)_{\dot{a}b} = \overline{\sigma}^\mu_{\dot{a}b}\), 
where \((\sigma^i)_{i=1,2,3}\) are the Pauli matrices and \(\varepsilon_{ab}=\varepsilon_{\dot{a}\dot{b}}\) is the totally antisymmetric tensor.
The supercharges are represented by the odd vector fields
\begin{equation}
\label{eq:4dN1-generators}
\underline{Q}_a = -\frac{\partial}{\partial\theta^a} + i \tilde{\theta}^{\dot{b}}(\gamma^\mu)_{\dot{b}a}\frac{\partial}{\partial x^\mu} , \qquad \underline{\tilde{Q}}_{\dot{a}} = - \frac{\partial}{\partial\tilde{\theta}^{\dot{a}}} + i\theta^b (\gamma^\mu)_{b\dot{a}}\frac{\partial}{\partial \overline{x^\mu}} ,
\end{equation}
and the supersymmetry action on a superfield is
\begin{equation}
\label{eq:susyvar-chiral}
\delta_\epsilon \Phi = (\epsilon^a \underline{Q}_a + \tilde{\epsilon}^{\dot{a}} \underline{\tilde{Q}}_{\dot{a}}) \Phi 
\end{equation}
where \(\overline{\epsilon^a} = \tilde{\epsilon}^{\dot{a}}\), being  a Majorana spinor.
The superderivatives are
\begin{equation}
\label{eq:4dN1-superderivatives}
D_a = \frac{\partial}{\partial\theta^a} + i \tilde{\theta}^{\dot{b}}(\gamma^\mu)_{\dot{b}a}\frac{\partial}{\partial x^\mu} ,  \qquad \tilde{D}_{\dot{a}} = \frac{\partial}{\partial\tilde{\theta}^{\dot{a}}} + i\theta^b (\gamma^\mu)_{b\dot{a}}\frac{\partial}{\partial \overline{x^\mu}} .
\end{equation}
Here we split the odd coordinates \(\theta^a, \tilde{\theta}^{\dot{a}}\) according to the split of the supercharges. The reality constraint on the coordinates then reads
\begin{equation}
\overline{\theta^a} = \tilde{\theta}^{\dot{a}} , \qquad \overline{x^\mu}=x^\mu .
\end{equation}
%like if the tuple \((\theta,\tilde{\theta})\) was itself a Majorana spinor.

Since we are operating over \(\mathbb{C}\), the subspaces of the supertranslation algebra
\begin{equation}
\mathfrak{st}^{(L/R)} := \mathbb{C}^4 \oplus S^{(L/R)}[1]
\end{equation}
are both Lie superalgebras over \(\mathbb{C}\), and determines the corresponding complex Lie supergroups \(S\mathbb{C}^{(L/R)}\). Moreover these subalgebras are \emph{Abelian}, since \(\Gamma\) vanishes on \(S^{(L/R)}\times S^{(L/R)}\). 
\begin{defn} \(S\mathbb{C}^{(L)}\) is called \emph{chiral} superspace and \(S\mathbb{C}^{(R)}\) \emph{anti-chiral} superspace. \end{defn}

We can write the complexified superspace as
\begin{equation}
S\mathbb{C}^4_S \cong S\mathbb{C}^{(L)} \times_{\mathbb{C}^4} S\mathbb{C}^{(R)}
\end{equation}
where the \(\times_{\mathbb{C}^4}\) here denotes the fiber product with respect to the base \(\mathbb{C}^4\).\footnote{This is analogous to the pull-back bundle, not to be confused with an homotopy quotient.} The chiral and anti-chiral superspaces are those identified by the flows of the corresponding superderivatives, since they generates the chiral and anti-chiral part of the supertranslation algebra on \(\Gamma(TS\mathbb{C}_S^4)\).\footnote{Here \(\Gamma\) denotes the space of sections on the tangent bundle, \textit{i.e.} the vector fields on \(S\mathbb{C}_S^4\), not the spinor pairing.} We can easily find sets of \lq\lq holomorphic-like\rq\rq\ coordinates
\begin{equation}
y^\mu_{(\pm)} := x^\mu \pm i\tilde{\theta}^{\dot{a}}(\gamma^\mu)_{\dot{a}b} \theta^b , \qquad \varphi^a := \theta^a , \qquad \tilde{\varphi}^{\dot{a}} := \tilde{\theta}^{\dot{a}},
\end{equation}
where the superderivatives simplify as
\begin{equation}
D_a = \left\lbrace \begin{aligned} 
&\frac{\partial}{\partial \varphi^a}  \\
&\frac{\partial}{\partial \varphi^a} + 2i (\tilde{\varphi}\gamma^\mu)_a \frac{\partial}{\partial y^\mu_{(-)}}
\end{aligned} \right. , \qquad \tilde{D}_{\dot{a}} = \left\lbrace  \begin{aligned}
&\frac{\partial}{\partial \tilde{\varphi}^{\dot{a}}} + 2i (\varphi\gamma^\mu)_{\dot{a}} \frac{\partial}{\partial y^\mu_{(+)}} \\
&\frac{\partial}{\partial \tilde{\varphi}^{\dot{a}}} 
\end{aligned} \right. .
\end{equation}

Complexifying the superspace we doubled its real-dimension, and that leads to a sort of \lq\lq reducibility\rq\rq\ of the relevant physical objects, \emph{i.e.} the fields on superspace, or \emph{superfields}. The simplest kind of superfields in the complexified setting are complex \emph{even} maps \(\Phi: S\mathbb{C}^4 \to \mathbb{C}\), that is sections of a trivial \(\mathbb{C}\)-line bundle over \(S\mathbb{C}^4\). We can impose now some constraints on them in order to restore the correct number of degrees of freedom. This is usually done in two different ways: asking the superfields to depend only on the chiral (or anti-chiral) sector of the complexified superspace, or imposing a reality condition.

\begin{defn} 
\begin{enumerate}[label=(\roman*)]
\item  A \emph{chiral (anti-chiral) superfield} is a superfield \(\Phi\) such that \[
\tilde{D}_{\dot{a}} \Phi = 0 \qquad \left( D_a \Phi = 0 \right) .\]
\item A \emph{vector superfield} is a superfield \(V\) such that \[
V=V^\dagger . \]
\end{enumerate} 
Notice how the first one is a sort of (anti)holomorphicity condition with respect to the chiral/anti-chiral sectors of \(S\mathbb{C}^4_S\), spanned by the coordinates \(\varphi^a, \tilde{\varphi}^{\dot{a}}\). Moreover, we can see that the complex conjugate \(\Phi^\dagger\) of a chiral superfield \(\Phi\) is antichiral. Next we will see how these conditions reflect on the various component fields of the coordinate expansions of \(\Phi\) and \(V\).
\end{defn}

It is important to stress that in this particular case of \(\mathcal{N}=1\) in (3+1)-dimensions, the complex representation \(S\cong \mathbb{C}^4\) allows for a real structure and the presence of Majorana spinors, and also for a chiral decomposition into left- and right-handed parts. It does not exist though a common basis for the two decompositions, \textit{i.e.}\ \(S^{(\pm)}\neq S^{(L/R)}\). In other words, it is not possible to require \emph{both} the chirality and the Majorana conditions on spinors in 4-dimension, since Majorana spinors contain both  left- and right- handed components. This means that in a theory with some supersymmetry in 4-dimensions there will be the same number of left-handed and right-handed degrees of freedom.
In \(d=2\text{mod}8\) dimensions (in Lorentzian signature), instead, the minimal complexified spin representation \(S\) can be decomposed into Majorana-Weyl subrepresentations \(S=S_\mathbb{R}^{(+)}\oplus S_\mathbb{R}^{(-)}\), so that one can choose to work only with real left-handed spinors. In the case of extended supersymmetry, we can thus have in general a different number of left-handed and right-handed real supercharges, and \(S=(S_\mathbb{R})^{\mathcal{N}_+} \oplus (S_\mathbb{R})^{\mathcal{N}_-}\). This is denoted with \(\mathcal{N}=(\mathcal{N}_+,\mathcal{N}_-)\). When \(\mathcal{N}_+\) or \(\mathcal{N}_-\) is zero, the supersymmetry is called \emph{chiral}. For a review on spinors in different dimensions, Majorana and chirality conditions we refer to \cite{Figueroa-majorana}.

\subsection{Supersymmetric actions and component field expansion}
\label{sec:susy-actions}

Actions for supersymmetric field theories are constructed integrating over superspace combinations of superfields and their derivatives. For this purpose, it is useful to write the superfields in a so-called \lq\lq component expansion\rq\rq, with respect to the generators of \(\bigwedge(S^*)\). In this paragraph we continue with the example of \(\mathcal{N}=1\) in (3+1)-dimensions, but the construction is immediatly generalizable to other cases. Consider a complex superfield \(\Phi: S\mathbb{C}^4\to\mathbb{C}\), its trivialization on the set of coordinates \((x,\theta,\tilde{\theta})\) being
\begin{equation}
\label{eq:component_expansion_general}
\begin{aligned}
\Phi(x,\theta,\tilde{\theta}) = \phi(x) + v_\mu(x) (\tilde{\theta}\gamma^\mu\theta) + \psi_a(x) \theta^a + \tilde{\psi}_{\dot{a}}\tilde{\theta}^{\dot{a}} + F(x) \theta^{(2)} + \tilde{F}(x) \tilde{\theta}^{(2)} +& \\
+ \xi_a(x) \theta^a \tilde{\theta}^{(2)} + \tilde{\xi}_{\dot{a}}(x) \tilde{\theta}^{\dot{a}} \theta^{(2)} + D(x) \theta^{(2)} \tilde{\theta}^{(2)} &
\end{aligned}
\end{equation}
where the the wedge product between the \(\theta\)'s has been suppressed, and we agree they are anticommuting, \(\theta^{(2)} := \theta^2\theta^1\) and \(\tilde{\theta}^{(2)} := \tilde{\theta}^{\dot{2}}\tilde{\theta}^{\dot{1}}\). We used the isomorphism \(\Gamma:S^{L}\odot S^{R}\to \mathbb{C}^4\) to represent the component of degree (1,1) as \(v_{a\dot{a}} \mapsto v_\mu (\gamma^\mu)_{a\dot{a}}\), the reason for this will become clear shortly. The expansion stops at top-degree (here 4) for the anticommuting property of the exterior product.
%The functions \(\phi, v_\mu , \psi_a , \tilde{\psi}_{\dot{a}}, F, \tilde{F}, \xi_a , \tilde{\xi}_{\dot{a}}, D\) are \(C^\infty\)-functions on spacetime, and are called the \emph{component fields} of the superfield \(\Phi\).  
In order for \(\Phi\) to be an even (scalar) field, we must take the functions \(\phi, v_\mu, F, \tilde{F}, D\) to be even (commuting), and the functions \(\psi_a, \tilde{\psi}_{\dot{a}}, \xi_a, \tilde{\xi}_{\dot{a}}\) carrying a spinor index to be odd (anticommuting). These are called \emph{component fields} of the superfield \(\Phi\).\footnote{The possibility of writing down an expansion similar to \eqref{eq:component_expansion_general} with these properties could again be justified more rigorously thanks to the concept of \emph{functor of point}.}

\medskip
If we impose the chirality condition \(\tilde{D}_{\dot{a}}\Phi=0\), when expressed in the coordinates \((y_{(-)},\varphi,\tilde{\varphi})\) this simply requires the independence on \(\tilde{\varphi}\), so on these coordinates a chiral superfield can be written as
\begin{equation}
\Phi(y_{(-)},\varphi,\tilde{\varphi}) = \phi(y_{(-)}) + \psi_a(y_{(-)}) \varphi^a + F(y_{(-)}) \varphi^{(2)} .
\end{equation}
Taylor-expanding in the old coordinates, this is equivalent to 
\begin{equation}
\label{eq:chiral-superfield}
\Phi(x,\theta,\tilde{\theta}) = \phi(x) + \psi_a(x)\theta^a + F(x) \theta^{(2)} - i\partial_\mu \phi(x) (\tilde{\theta}\gamma^\mu \theta) + i(\tilde{\theta}\gamma^\mu \partial_\mu \psi(x)) \theta^{(2)}  -  \partial^2 \phi(x) \theta^{(2)} \tilde{\theta}^{(2)}
\end{equation}
where higher order terms are again automatically zero for degree reasons.\footnote{Here we used \(\theta^a\theta^b = \varepsilon^{ab} \theta^{(2)}\), and \((\tilde{\theta}\gamma^\mu \theta)(\tilde{\theta}\gamma^\mu \theta) = 2 \eta^{\mu\nu} \theta^{(2)}\tilde{\theta}^{(2)}\) with a \lq\lq mostly plus\rq\rq\ signature.} The \lq\lq irreducible\rq\rq\ chiral superfield has three non-zero field components: a complex scalar field \(\phi\), a left-handed Weyl spinor field \(\psi\) and another complex scalar field \(F\). 
One can work out the supersymmetry transformations of these component fields from the general rule (\ref{eq:susyvar-chiral}), and the result is 
\begin{equation}
\begin{aligned}
&\delta_\epsilon \phi =  \epsilon\psi \\
&\delta_\epsilon \psi = 2i (\tilde{\epsilon}\gamma^\mu) \partial_\mu \phi +  \epsilon F \\
&\delta_\epsilon F = -2i \tilde{\epsilon} \gamma^\mu \partial_\mu \psi
\end{aligned}
\end{equation}
where spinor contractions are implied, and spinor indices are lowered/raised via the charge conjugation matrix as usual. 

To construct a minimal Lagrangian density from the superfield \(\Phi\), we can look at its mass dimension: this is equal to its lowest component \(\phi\), that being a scalar field is \((d-2)/2=1\) in \(d=4\) dimensions. Since \(\psi\) is a spinor, it has dimension \((d-1)/2=3/2\), so the odd coordinates have always dimension -1/2. The highest component of any superfield has thus dimension two more than the superfield. This means that to construct a Lagrangian we have to take a quadratic expression in \(\Phi\). Since the action should be real, the simplest choice is \(\Phi^\dagger \Phi\). Its top-degree component is 
\begin{equation}
\int d^2\theta d^2\tilde{\theta}\ \Phi^\dagger \Phi =  4\left| \partial \phi\right|^2 + i \overline{\Psi}\slashed{\partial} \Psi + |F|^2 + \partial_\mu (\cdots)^\mu
\end{equation}
where \(\Psi\) is a Majorana 4-spinor, whose left-handed component is \(\psi\). Up to a total derivative, this is the Lagrangian of the free, massless \emph{Wess-Zumino model}:
\begin{equation}
S_{WZ,free}[\Phi] = \int_{S\mathbb{R}^4_{\mathbb{C}^4}} d^4x d^2\theta d^2\tilde{\theta}\ \Phi^\dagger \Phi = \int_{\mathbb{R}^4} d^4x \left( 4\left| \partial \phi\right|^2 + i \overline{\Psi}\slashed{\partial} \Psi + |F|^2 \right) .
\end{equation}
Since the field \(F\) appears without derivatives, its equation of motion is an algebraic equation. For this reason, it is called an \emph{auxiliary field}, and it is customary to substitute its on-shell value in the action. For this simple model, this means putting \(F=0\). This procedure makes in general the action to be supersymmetric only if the equations of motion (EoM) are imposed, and is often called \emph{on-shell} supersymmetry. The physical field content of an \(\mathcal{N}=1\) chiral superfield is thus the supersymmetry \emph{doublet} \((\phi, \psi)\). By CPT invariance, the theory must contain both the chiral field \(\Phi\) and its antichiral conjugate \(\Phi^\dagger\), so that the physical field content of a meaningful theory constructed from it is made by two real scalar fields \(\mathrm{Re}(\phi), \mathrm{Im}(\phi)\) and the Majorana spinor \(\Psi\).

\medskip
If we now start from a vector superfield \( V \), satisfying the reality condition \(V=V^\dagger \), we reach a different physical field content and supersymmetric Lagrangian. The component expansion in a chart \((x,\theta,\tilde{\theta}) \) is the following:
\begin{equation}
\begin{aligned}
V(x,\theta,\tilde{\theta}) = C(x) + \xi_a(x)\theta^a + \xi^\dagger_{\dot{a}}\tilde{\theta}^{\dot{a}} + v_\mu(x) (\tilde{\theta}\gamma^\mu \theta) + G(x) \theta^{(2)} + G^\dagger(x) \tilde{\theta}^{(2)} + \\
+ \eta_a(x) \theta^a \tilde{\theta}^{(2)} + \eta^\dagger_{\dot{a}}(x) \tilde{\theta}^{\dot{a}}\theta^{(2)} + E(x) \theta^{(2)} \tilde{\theta}^{(2)}
\end{aligned}
\end{equation}
where \(C,v_\mu , E\) are real fields. It is clear that this is the right type of superfield needed to describe (Abelian) gauge boson fields, represented here by \(v_\mu\). This is why \(V\) is called \emph{vector} superfield.

Notice that the real part of a chiral superfield is a special kind of vector superfield. In particular, its vector component is a derivative: if \(\Lambda\) is chiral, from (\ref{eq:chiral-superfield})
\begin{equation}
\Lambda + \Lambda^\dagger \supset i \partial_\mu (\phi - \phi^\dagger) (\tilde{\theta} \gamma^\mu \theta) .
\end{equation}
This suggest to interpret the transformation
\begin{equation}
\label{eq:Abelian-gauge-tr}
V \mapsto V + (\Lambda + \Lambda^\dagger)
\end{equation}
as the action of a \(U(1)\) internal gauge symmetry on superfields. In terms of component fields this gauge transformation reads
\begin{equation}
\label{eq:supergauge-transf}
\begin{array}{lll}
C\mapsto C+ (\phi+\phi^\dagger) & \xi_a \mapsto \xi_a + \psi_a & G\mapsto G+F \\
v_\mu \mapsto v_\mu - i\partial_\mu (\phi - \phi^\dagger) & \eta_a \mapsto \eta_a - i (\partial_\mu \psi^\dagger \gamma^\mu)_a & E \mapsto E -  \partial^2 (\phi+\phi^\dagger).
\end{array}
\end{equation}
We can notice two main things. The first is that the combinations
\begin{equation}
\begin{aligned}
&\lambda_a := \eta_a + i (\gamma^\mu \partial_\mu \xi^\dagger)_a \\
&D:= E + \partial^2 C
\end{aligned}
\end{equation}
are gauge invariants. The second is that, since \(C,G,\xi_a\) transform as shifts, we can chose a special gauge in which they vanish. This is called the \emph{Wess-Zumino (WZ) gauge}. Chosing a gauge of course breaks explicitly supersymmetry, but it is convenient for most of the calculations. In the WZ gauge, the vector superfield looks like
\begin{equation}
V = v_\mu (\tilde{\theta}\gamma^\mu \theta) + \lambda_a \theta^a \tilde{\theta}^{(2)} + \lambda^\dagger_{\dot{a}} \tilde{\theta}^{\dot{a}}\theta^{(2)} + D \theta^{(2)} \tilde{\theta}^{(2)} .
\end{equation}
Gauge transformations with immaginary scalar component (\(\phi +\phi^\dagger =0\)) preserve  the Wess-Zumino gauge and moreover induce on \(v_\mu\) the usual \(U(1)\) transformation of Abelian vector bosons. Indeed, if \(\alpha := 2\mathrm{Im}(\phi)\),
\begin{equation}
v_\mu \mapsto v_\mu + \partial_\mu \alpha .
\end{equation}

As for \(F\) in the case of chiral superfields, \(D\) is the top component field of  the vector superfield \(V\). It will have a purely algebraic equation of motion, so it can be considered as an auxiliary field. The physical field content of an \(\mathcal{N}=1\) vector superfield is thus the supersymmetry \emph{doublet} \((v_\mu, \lambda)\) composed by an Abelian gauge boson and a Majorana spinor, called the \emph{gaugino}.

A gauge-invariant supersymmetric action for the Abelian vector superfield can be given in terms of the spinorial superfields defined as
\begin{equation}
W_a := \frac{1}{2}\tilde{D}^2 D_a V , \qquad  \tilde{W}_{\dot{a}} := \frac{1}{2} D^2 \tilde{D}_{\dot{a}} V .
\end{equation}
\(W_a\)  (\(\tilde{W}_{\dot{a}}\)) is both chiral (antichiral) and gauge invariant, and moreover it satisfies the \lq\lq reality\rq\rq\ condition \(D^a W_a = \tilde{D}^{\dot{a}} \tilde{W}_{\dot{a}}\). Expanding the component fields in coordinates \((y_{(\pm)},\theta,\tilde{\theta})\), we have
\begin{equation}
\begin{aligned}
& W_a = \lambda_a + if_{\mu\nu}(\gamma^\mu\gamma^\nu)_{ab} \theta^b + D\varepsilon_{ab}\theta^b + 2i (\gamma^\mu \partial_\mu \lambda^\dagger )_a \theta^{(2)} \\
& \tilde{W}_{\dot{a}} = \lambda^\dagger_{\dot{a}} - i f_{\mu\nu} (\gamma^\mu \gamma^\nu)_{\dot{a}\dot{b}}\tilde{\theta}^{\dot{b}} + D \varepsilon_{\dot{a}\dot{b}} \tilde{\theta}^{\dot{b}}  - 2i (\gamma^\mu \partial_\mu \lambda)_{\dot{a}} \tilde{\theta}^{(2)}
\end{aligned}
\end{equation}
where \(f_{\mu\nu}:= (\partial_\mu v_\nu - \partial_\nu v_\nu)\) is the gauge invariant Abelian field-strength of \(v_\mu\). A gauge invariant action is then
\begin{equation}
\label{eq:susyYM4d-Abelian}
\begin{aligned}
S[V] &= \int d^4x \frac{1}{8}\left( \int d^2\theta\ W^a W_a + \int d^2\tilde{\theta}\ \tilde{W}^{\dot{a}} \tilde{W}_{\dot{a}} \right) \\
& = \int d^4x \left\lbrace  f_{\mu\nu}f^{\mu\nu} + i  \lambda \gamma^\mu \partial_\mu \lambda^\dagger + 2 D^2 \right\rbrace
\end{aligned}
\end{equation}
that is an \(\mathcal{N}=1\) supersymmetric extension of the Abelian Yang-Mills theory in 4-dimensions. The supersymmetry transformations of the component field, under which (\ref{eq:susyYM4d-Abelian}) is invariant can be obtained applying the supertranslation on \(V\) in WZ gauge. The result will not be in this gauge anymore, but can be translated back in WZ gauge applying an appropriate gauge transformation as (\ref{eq:supergauge-transf}). The result is 
\begin{equation}
\label{eq:susyvar-gauge}
\begin{aligned}
&\delta_\epsilon v_\mu = \frac{1}{2}\left(\tilde{\epsilon}\gamma_\mu \lambda - \lambda^\dagger \gamma_\mu \epsilon \right) \\
&\delta_\epsilon \lambda =  \frac{i}{2} f_{\mu\nu} \epsilon \gamma^{\mu}\gamma^{\nu} + D \epsilon \\
&\delta_\epsilon D = -i \left( \tilde{\epsilon} \gamma^\mu \partial_\mu \lambda + \epsilon \gamma^\mu \partial_\mu \lambda^\dagger \right).
\end{aligned}
\end{equation}
Notice that in this case the Super Yang-Mills (SYM) action remains supersymmetric even if we impose the EoM on the auxiliary field \(D\), setting \(D=0\) in both (\ref{eq:susyYM4d-Abelian}) and (\ref{eq:susyvar-gauge}). This is a special result, that holds in 4, 6 and 10 dimensions \cite{brink-SYM}.

\medskip
In a generic gauge theory with gauge group
%\footnote{In Yang-Mills theories on a base space \(M\), we refer to the \emph{gauge group} \(G\) as the one acting on right on the considered principal \(G-\)bundle \(P\to M\) defining the theory. With \emph{gauge group} sometimes is intended the group \(\mathcal{G}\) of diffeomorphism on \(P\) that induce a principal bundle map between \(P\to M\) and itself. We do not use this terminology here.
%\(\mathcal{G}\) naturally acts via pull-back, in an Hamiltonian way \cite{atiyah-bott-localization}, on the affine space of connection one forms on \(P\).}
\(G\), we consider a \(G-\)valued chiral multiplet \(\Phi\) which transforms under a gauge transformations as
\begin{equation}
\Phi \mapsto e^{\Lambda} \Phi
\end{equation}
where \(\Lambda\) is a \(\mathfrak{g}\)-valued chiral superfield. Now the combination \(\Phi^\dagger \Phi\) is not gauge invariant, so we introduce a \(\mathfrak{g}-\)valued vector superfield \(V\), transforming as
\begin{equation}
e^V \mapsto e^{\Lambda^\dagger} e^V e^{\Lambda}
\end{equation}
that reduces to the previous case (\ref{eq:Abelian-gauge-tr}) for Abelian \(G=U(1)\). The exponential of a superfield can be defined through its component field expansion, that stops at finite order for degree reasons:
\begin{equation}
e^V = 1 + v_\mu (\tilde{\theta}\gamma^\mu \theta) + \lambda_a \theta^a \tilde{\theta}^{(2)} + \lambda^\dagger_{\dot{a}} \tilde{\theta}^{\dot{a}}\theta^{(2)} + (D + 2 v_\mu v^\mu) \theta^{(2)} \tilde{\theta}^{(2)} .
\end{equation}
The kinetic term for the chiral superfield can be rewritten as a gauge invariant combination:
\begin{equation}
\int d^4x d^2\theta d^2\tilde{\theta}\ \Phi^\dagger e^V \Phi .
\end{equation}
Generalizing the supersymmetric field-strength \(W_a\) to the non-Abelian case as
\begin{equation}
W_a = \frac{1}{2} \tilde{D}^2 e^{-V} D_a e^V 
\end{equation}
we can write the  full matter-coupled gauge theory action:
\begin{equation}
\label{eq:susy-gauge-action}
S[V,\Phi] = \int d^4x \left\lbrace \int d^2\theta d^2\tilde{\theta}\ \Phi^\dagger e^V \Phi + \left[ \int d^2\theta\ \left( \frac{1}{4}\mathrm{Tr}W^aW_a + W(\Phi) \right) + c.c. \right] \right\rbrace
\end{equation}
where \(W(\Phi)\) is a holomorphic function of \(\Phi\) called \emph{superpotential}. The expansion in terms of component fields and the supersymmetry variations can be calculated with the same procedure we did in the other cases.\footnote{A more detailed treatment can be found in \cite{wess-bagger}, or \cite{figueroa-susylectures}.}

Notice that whenever the center of the Lie algebra \(\mathfrak{g}\) is non-trivial, \textit{i.e.}\ when there is a \(U(1)\) factor in \(G\), we could add another supersymmetric and gauge-invariant term to the action (\ref{eq:susy-gauge-action}). This is the so-called \emph{Fayet-Iliopoulos term}:
\begin{equation}
\int d^4x d^2\theta d^2\tilde{\theta}\ \xi(V) = \int d^4x\ \xi_A D^A
\end{equation}
where \(\xi=\xi_A \tilde{T}^A\) is a constant element in the dual of the center of \(\mathfrak{g}\).

\subsection{R-symmetry}

The subgroup of (outer) automorphisms of the supersymmetry group 
% \(SISO_S(d)\)
 which fixes the underlying Poincaré (Euclidean) group is called \emph{R-symmetry group}. At the level of the algebra, these are linear transformations that act only on the spin representation \(S\), leaving the brackets of two spinors unchanged. In the complexified case, when different chiral sectors are present, the R-symmetry acts differently on any sector.

For example, in the case of \(\mathcal{N}=1\) in (3+1)-dimensions, there is a \(U(1)_R\) R-symmetry group acting as
\begin{equation}
\label{eq:U(1)R-supercharges}
Q_a \mapsto e^{-i\alpha} Q_a , \qquad \tilde{Q}_{\dot{a}} \mapsto e^{i\alpha} \tilde{Q}_{\dot{a}} ,
\end{equation}
with \(\alpha\in\mathbb{R}\). This clearly leaves the brackets \([Q_a, \tilde{Q}_{\dot{a}}]\) invariant. The odd coordinates \(\theta^a\) on superspacetime, being elements of \(S^*\) transform as
\begin{equation}
\label{eq:U(1)R-coord}
\begin{array}{lcl}
\theta^a \mapsto e^{i\alpha}\theta^a & & \tilde{\theta}^{\dot{a}} \mapsto e^{-i\alpha}\tilde{\theta}^{\dot{a}} \\
d^2\theta \mapsto e^{-2i\alpha}d^2\theta & & d^2\tilde{\theta} \mapsto e^{2i\alpha}d^2\tilde{\theta},
\end{array}
\end{equation}
so that the volume element \(d^4\theta=d^2\theta d^2\tilde{\theta}\) is invariant under R-symmetry. This fixes the \emph{R-charge} of the superpotential \(W(\Phi)\) to be 2, if we want the action to be invariant under R-symmetry:
\begin{equation}
W(\Phi) \mapsto e^{2i\alpha} W(\Phi) .
\end{equation}
In principle we can chose the chiral superfield \(\Phi\) to have any R-charge \(r\), since the combination \(\Phi^\dagger \Phi\) is R-invariant. This, combined with (\ref{eq:U(1)R-coord}) means that the different field components in the chiral multiplet transform differently with respect to R-symmetry:
\begin{equation}
\phi \mapsto e^{ir\alpha}\phi , \quad \psi \mapsto e^{i(r-1)\alpha}\psi , \quad F \mapsto e^{i(r-2)\alpha}F .
\end{equation}
The vector superfield, being real is acted upon trivially by \(U(1)_R\). Its component fields are then forced to transform as
\begin{equation}
v_\mu \mapsto v_\mu , \quad \lambda \mapsto e^{i\alpha}\lambda , \quad D \mapsto D ,
\end{equation}
thus the gauge-invariant supersymmetric field-strength \(W_a\) has R-charge \(1\).

In general, if the spin representation is reducible and we have extended supersymmetry, the R-group is always compact. For \(S=(S_0)^{\mathcal{N}}\), where \(S_0\) is a real representation, it is of the type \(U(\mathcal{N})\), while for \(S= (S^{(+)})^{\mathcal{N}_+} \oplus (S^{(-)})^{\mathcal{N}_-}\), where \(S^{(\pm)}\) are the two real representations of different chirality, it is of the type \(U(\mathcal{N}_+)\times U(\mathcal{N}_-)\) \cite{varadarajan-susy_math}. Notice the isomorphism 
\begin{equation}
U(n) \cong (SU(n)\times U(1))/\mathbb{Z}_n
\end{equation}
\textit{i.e.}\ \(U(n)\) is an n-fold cover of \(SU(n)\times U(1)\). In particular, their Lie algebras are isomorphic. In terms of infinitesimal transformations then, the R-symmetry generators can be decomposed in one R-charge plus \(\mathcal{N}^2-1\) rotation generators. The supercharges are rotated into one another by
\begin{equation}
Q_a^I \mapsto e^{-i\alpha} \mathcal{U}^I_J Q_a^J , \quad \tilde{Q}_{\dot{a}}^I \mapsto e^{i\alpha} \mathcal{(U^\dagger)}^I_J \tilde{Q}_{\dot{a}}^J.
\end{equation}
In QFT, R-symmetry may or may not be present as a symmetry of the theory, and in many cases part of this symmetry may be broken by anomaly at quantum level.

\subsection{Supersymmetry multiplets}

A geometric analysis as the one carried out in the last subsections allows one to find the physical field content of a supersymmetric theory in every dimensions and for any degree of reducibility of the spin representation \(S\) that is used to extend the Poincaré algebra. Another systematic way to obtain the same result, from a more algebraic point of view, is to study the representation of the supersymmetry algebra \(\mathfrak{siso}_S(d)\), in analogy with the Wigner analysis of massive and massless representations of the Poincaré algebra. As the cases encountered above, this study leads to the presence of different \emph{supersymmetry multiplets} for different choices of spin and \(\mathcal{N}\). We will not present this here but refer for example to \cite{dhoker-susy-notes} for a comprehensive review, and list here some results for the multiplets at various \(\mathcal{N}\).

For \(1\leq \mathcal{N}\leq 4\) with spin less or equal to 1, the supersymmetry particle representations simply consists of spin 1 vector particles, spin 1/2 fermions and spin 0 scalars. In the supergeometric approach, these fields are interpreted as components of the same superfield, and thus transform one into another under the supersymmetry algebra. Let \(G\) be the gauge group, and \(\mathfrak{g}\) its Lie algebra. We are interested mainly in two types of multiplets. The first is the (massless) \emph{vector} or \emph{gauge multiplet}, which transforms under the adjoint representation of \(\mathfrak{g}\). For \(\mathcal{N} = 3, 4\), this is the only possible multiplet. It turns out that quantum field theories with \(\mathcal{N} = 3\) supersymmetries coincide with those with \(\mathcal{N} = 4\) in view of CPT invariance, thus we shall limit our discussion to the \(\mathcal{N} = 4\) theories.\footnote{To be more precise, it is possible to construct theories with genuine \(\mathcal{N}=3\) supersymmetry, but they lack of a Lagrangian description in terms of component fields.} For \(\mathcal{N} = 1,2\), we also have (possibly massive) matter multiplets: for \(\mathcal{N} = 1\), this is the \emph{chiral multiplet}, and for \(\mathcal{N} = 2\) this is the \emph{hypermultiplet}, both of which may transform under an arbitrary (unitary, and possibly reducible) representation of \(G\). 

In (3+1)-dimensions, the on-shell field content of these multiplets is:
\begin{itemize}
\item \(\mathcal{N} = 1\) \emph{gauge multiplet} \((A_\mu , \lambda)\): a gauge boson and a Majorana fermion, the gaugino.
\item \(\mathcal{N} = 1\) \emph{chiral multiplet} \((\phi , \psi)\): a complex scalar and a left-handed Weyl fermion.
\item \(\mathcal{N} = 2\) \emph{gauge multiplet} \((A_\mu , \lambda_{\pm}, \phi)\): \(\lambda_{\pm}\) form a Dirac spinor, and \(\phi\) is a complex \emph{gauge scalar}. Under the \(SU(2)_R\) symmetry, \(A_\mu\) and \(\phi\) are singlets, while \(\lambda_+ , \lambda_-\) transform as a doublet.
\item \(\mathcal{N} = 2\) \emph{hypermultiplet} \((\psi_+ , H, \psi_-)\): \(\psi_\pm\) form a Dirac spinor and \(H_\pm\) are complex scalars. Under the \(SU(2)_R\) symmetry, \(\psi_+\) and \(\psi_-\) transform as singlets, while \(H_+, H_-\) transform as a doublet.
\item \(\mathcal{N}=4\) \emph{gauge multiplet} \((A_\mu , \lambda^i, \Phi_A)\): \(\lambda^i\) , \(i=1,2,3,4\) are Weyl fermions (equivalents to two Dirac fermions), and \(\Phi_A\), \(A=1,\cdots,6\) are real scalars (equivalents to three complex scalars). Under the \(SU(4)_R\) symmetry\footnote{The R-symmetry group is actually \(SU(2)\times SU(2)\times U(1)\), as we will see in a practical application in the following.} the gauge field \(A_\mu\) is a singlet, the fermions \(\lambda^i\) transform in the fundamental representation \(\mathbf{4}\), the scalars \(\Phi_A\) transform in the rank-two antisymmetric representation \(\mathbf{6}\).
\end{itemize}

Even though in this thesis we do not work explicitly with gravity theories, we will see in the next section that the introduction of \emph{off-shell} supergravity is necessary in a possible approach to construct globally supersymmetric theories on curved base-spaces. For this purpose, it is useful to remind also the content of massless supersymmetry particle representations with helicity between 1 and 2. These are the \emph{gravitino multiplet} and the \emph{graviton multiplet} (or \emph{supergravity multiplet}, or \emph{metric multiplet}). In general the gravitino multiplet contains degrees of freedom with helicity less or equal than 3/2. Since in a theory without gravity one cannot accept particles with helicity greater than one,\footnote{This comes from the so-called \emph{Weinberg-Witten theorem} \cite{WWtheorem}.} that multiplet cannot appear in a supersymmetric theory if also a graviton, with helicity 2, does not appear. In (3+1)-dimensions, the field content of the relevant multiplets are:
\begin{itemize}
\item \(\mathcal{N}=1\) \emph{gravitino multiplet} \((\Phi_\mu, B_\mu)\): a helicity 3/2 fermionic particle and a vector boson. 

\item \(\mathcal{N}=1\) \emph{graviton multiplet} \((h_{\mu\nu}, \Psi_\mu)\): the \emph{graviton}, with helicity 2, and its supersymmetric partner the \emph{gravitino}, of helicity 3/2.

\item \(\mathcal{N}=2\) \emph{gravitino multiplet}: a spin 3/2 particle, two vectors and one Weyl fermion. 
\item \(\mathcal{N}=2\) \emph{graviton multiplet}: graviton, two gravitinos and a vector boson.
\end{itemize}

For \(\mathcal{N}>4\) it is not possible to avoid gravity since there do not exist representations with helicity smaller than 3/2. Hence, theories with \(\mathcal{N} > 4\) are all supergravity theories.

\subsection[Euclidean 3d N=2 supersymmetric gauge theories]{Euclidean 3d \(\mathcal{N}=2\) supersymmetric gauge theories} 
\label{subsec:3dN2}

As an example, which will be used in some applications of the localization principle in the next chapter, we can look at \(\mathcal{N}=2\) Euclidean supersymmetry in 3-dimensions. First, notice that the rotation algebra for 3d Euclidean space is \(\mathfrak{so}(3)\). The corresponding spin group is thus \(SU(2)\), whose fundamental representation \(\mathbf{2}\) does not admit a real structure. In fact here the charge conjugation can be taken as the totally antisymmetric symbol \(C_{ab}= \varepsilon_{ab}\), and the Majorana condition would be inconsistent:
\begin{equation}
\psi^T C = \psi^\dagger \Leftrightarrow \psi =0 .
\end{equation}
Thus we cannot construct an \(\mathcal{N}=1\) Euclidean supersymmetry algebra in 3 dimensions, in the sense of definition \eqref{defn:susy-algebra}. The problem can be cured considering a reducible spin  representation \(S\), where the spinors and the charge conjugation matrix can be split as 
\begin{equation}
\Psi = (\psi^{a}_I)^{a=1,2}_{I=1,\cdots ,\mathcal{N}},  \qquad \mathcal{C} = (\Omega_{IJ} C^{ab})^{a,b=1,2}_{I,J=1,\cdots ,\mathcal{N}},
\end{equation}
and the same reality condition \(\Psi^\dagger = \Psi^T \mathcal{C}\) now is consistent if also the matrix \(\Omega\) squares to \(-\mathds{1}\) and is anti-orthogonal: 
\begin{equation}
\Omega = - \Omega^T = - \Omega^{-1} .
\end{equation}

If we now fix  \(\mathcal{N}=2\), the resulting spinor representation is analogous to the one of \(\mathcal{N}=1\) in 4-dimensions, but now the two Weyl sectors are independent since they generates the two supersymmetries. To see this corrispondence, we can change basis of \(S=\mathbf{2}^{(1)} \oplus \mathbf{2}^{(2)}\) from the natural one in terms of the generators \(\lbrace Q^1_a, Q^2_a\rbrace\) to 
\begin{equation}
Q_a := \frac{1}{\sqrt{2}}(Q_a^1+iQ_a^2) , \qquad \tilde{Q}_a := \frac{1}{\sqrt{2}}(Q_a^1-iQ_a^2) .
\end{equation}
In this basis, using (\ref{eq:susy-algebra-ext}) the super Lie brackets become
\begin{equation}
\label{eq:susy-algebra-3d}
\begin{array}{lr}
\multicolumn{2}{l}{ [Q_a, \tilde{Q}_b] = 2(\gamma^\mu)_{ab}P_\mu + Z\varepsilon_{ab} } \\
\left[ Q_{a} , Q_{b} \right] = 0  &  [ \tilde{Q}_a , \tilde{Q}_b ] = 0 
\end{array}
\end{equation}
where \(Z\) is a constant central charge, and the gamma matrices in this representation can be chosen to be the Pauli matrices \(\gamma^\mu = \sigma^\mu\) for \(\mu=1,2,3\).\footnote{There is also another inequivalent representation of the Clifford algebra, as in any odd dimensions, in which \(\gamma^3 = -\sigma^3\). We chose the former one.} Note that the 4-dimensional \(Spin(3,1)\) Lorentz group breaks to \(SU(2)\times SU(2)_R\), where \(Spin(3)\cong SU(2)\) is the 3-dimensional Lorentz group, and the remaining \(SU(2)_R\) is an R-symmetry acting on the \(\mathcal{N}=2\) algebra. The generators \(Q_a\) and \(\tilde{Q}_a\) are represented in superspace by odd vector fields whose expressions are formally the same as in \eqref{eq:4dN1-generators}, and the \(\mathcal{N}=2\) supersymmetry variation of a superfield \(\Phi\) is 
\begin{equation}
\delta_{\epsilon,\eta} \Phi = (\epsilon^a Q_a + \eta^a \tilde{Q}_a) \Phi
\end{equation}
where now, as said before, \(\epsilon\) and \(\eta\) are two independent complex spinors.

If we want to construct a supersymmetric gauge theory in 3-dimensions, we consider the vector superfield, now expressed in WZ gauge as
\begin{equation}
V(x,\theta,\tilde{\theta}) = A_\mu (\tilde{\theta} \gamma^\mu \theta) + i\sigma \theta \tilde{\theta} + \lambda_a \theta^a \tilde{\theta}^{(2)} + \lambda^\dagger_a \tilde{\theta}^a \theta^{(2)} + D \theta^{(2)} \tilde{\theta}^{(2)} .
\end{equation}
The off-shell \(\mathcal{N}=2\) gauge multiplet is then composed by a gauge field \(A_\mu\), two real scalars \(\sigma, D\) and a 2-component complex spinor \(\lambda\). Notice that this is just the \emph{dimensional reduction} of the \(\mathcal{N}=1\) multiplet in 4-dimensions, with \(\sigma\) coming from the zero-th component of the gauge field in higher dimensions. The only difference with the 4-dimensional vector multiplet is that this zero-th component has been considered purely immaginary, \emph{i.e.} \(A_0 = i\sigma\) with real \(\sigma\). This ensures the kinetic term for \(\sigma\) to be positive definite and the path integral to converge, matching the would-be dimensional reduction from an Euclidean 4-dimensional theory. If the gauge group is \(G\), all fields are valued in its Lie algebra \(\mathfrak{g}\).

For what we are going to discuss in the next chapter, we now adopt the convention of \cite{itamar-wilson_loop_3dCS, Marino-locCS} for the supersymmetry variations of the vector superfield and the supersymmetric actions. Under a proper rescaling of the component fields and of the supercharges, one can work them out in an analogous way to which we did in the last sections, and get
\begin{equation}
\begin{aligned}
&\delta_{\epsilon,\eta} A_\mu = \frac{i}{2} (\eta^\dagger \gamma_\mu \lambda - \lambda^\dagger \gamma_\mu \epsilon) \\
&\delta_{\epsilon,\eta} \sigma = \frac{1}{2} (\eta^\dagger \lambda - \lambda^\dagger \epsilon) \\
&\delta_{\epsilon,\eta} D = \frac{i}{2} \left( \eta^\dagger \gamma^\mu D_\mu \lambda - (D_\mu \lambda^\dagger) \gamma^\mu \epsilon \right) - \frac{i}{2}\left( \eta^\dagger [\lambda, \sigma] - [\lambda^\dagger, \sigma] \epsilon \right) \\
&\delta_{\epsilon,\eta} \lambda = \left( -\frac{1}{2} \gamma^{\mu\nu} F_{\mu\nu} -D +i \gamma^\mu D_\mu \sigma\right) \epsilon \\
&\delta_{\epsilon,\eta} \lambda^\dagger = \eta^\dagger \left(- \frac{1}{2} \gamma^{\mu\nu} F_{\mu\nu} +D -i \gamma^\mu D_\mu \sigma\right) 
\end{aligned}
\end{equation}
where \(D_\mu = \partial_\mu + [A_\mu, \cdot]\) is the gauge-covariant derivative and \(\gamma^{\mu\nu}:= \frac{1}{2}[\gamma^\mu, \gamma^\nu]\). Up to some prefactors, they can be seen as a dimensional reduction of \eqref{eq:susyvar-gauge}.

We can consider two types of gauge supersymmetric actions constructed from the vector multiplet in 3 Euclidean dimensions: the Super Yang-Mills theory, that is a reduction of \eqref{eq:susy-gauge-action}, and the Super Chern-Simons (SCS) theory. In superspace, the former one is constructed in the same way as the 4-dimensional case from the spinorial superfield \(W_a\), while the SCS term is constructed as
\begin{equation}
S_{CS} = \int d^3x d^2\theta d^2\tilde{\theta}\ \frac{k}{4\pi} \left( \int_0^1 dt\   \mathrm{Tr}\left\lbrace V \tilde{D}^a e^{-tV} D_a e^{tV}\right\rbrace \right) .
\end{equation}
Integrating out the odd coordinates in superspace, these are given by \cite{itamar-wilson_loop_3dCS}:
\begin{align}
\label{eq:3dSYM}
&S_{YM} = \int d^3x \  \mathrm{Tr}\left\lbrace \frac{i}{2}\lambda^\dagger \gamma^\mu D_\mu \lambda + \frac{1}{4} F_{\mu\nu}F^{\mu\nu} + \frac{1}{2} D_\mu \sigma D^\mu \sigma + \frac{i}{2} \lambda^\dagger [\sigma, \lambda] + \frac{1}{2} D^2 \right\rbrace , \\
\label{eq:3dCS}
&S_{CS} = \frac{k}{4\pi} \int d^3x \ \mathrm{Tr}\left\lbrace  \varepsilon^{\mu\nu\rho} \left( A_\mu \partial_\nu A_\rho + \frac{2i}{3} A_\mu A_\nu A_\rho \right) - \lambda^\dagger \lambda + 2\sigma D \right\rbrace .
\end{align}

\subsection[Euclidean 4d N=4,2,2* gauge theories]{Euclidean 4d \(\mathcal{N}=4,2,2^*\) supersymmetric gauge theories}
\label{subsec:susy-4dN4}

We  describe here another example that will be useful in the next chapter, when we will apply the localization principle to supersymmetric QFT. The \(\mathcal{N}=4\) SYM theory on flat space can be derived via dimensional reduction of \(\mathcal{N}=1\) SYM in \((9+1)\) dimensions.\footnote{For convergence of the partition function, it would be nicer to start from the \((10,0)\) Euclidean signature. We follow the convention of \cite{pestun-article} and start from the \((9,1)\) one, Wick rotating a posteriori the path integral when needed, to match the would-be reduction from the \((10,0)\) theory.} The \(\mathcal{N}=2\) and \(\mathcal{N}=2^*\) theories can be derived as modification of the \(\mathcal{N}=4\) theory, as we will see later.

We start recalling the structure of the 10-dimensional Clifford algebra following the conventions of \cite{pestun-article}. This is independent from the choice of the signature, \(Cl(9,1)\cong Cl(1,9) \cong \mathrm{Mat}_{32}(\mathbb{R})\), it is real, and generated by the gamma matrices \((\gamma^M)_{M=0,1,\cdots,9}\) such that \[
\lbrace \gamma^M,\gamma^N\rbrace = 2\eta^{MN}  \]
where \(\eta\) is the 10-dimensional Minkowski metric, that we take with signature \((-,+,\cdots,+)\).
The fundamental representation of the spin group \(Spin(9,1)\hookrightarrow Cl(9,1)\) is then Majorana, and it is moreover reducible under chirality \cite{Figueroa-majorana} \(\gamma^{c}:=-i\gamma^0\gamma^1\cdots\gamma^9\) as \(Spin(9,1) = S^+ \oplus S^- \cong \mathrm{Mat}_{16}(\mathbb{R}) \oplus \mathrm{Mat}_{16}(\mathbb{R})\). Thus fundamental spinors are Majorana-Weyl, and have 16 real components.  In the chiral basis we denote
\begin{equation}
\begin{aligned}
&\gamma^M = \left(\begin{array}{cc}
0 & \tilde{\Gamma}^M \\ \Gamma^M & 0 \end{array} \right) \qquad \tilde{\Gamma}^M, \Gamma^M: S^{\pm} \to S^{\mp} \\
& \gamma^{MN} = \left(\begin{array}{cc}
\tilde{\Gamma}^{[M}\Gamma^{N]} & 0 \\ 0 & \Gamma^{[M}\tilde{\Gamma}^{N]} \end{array} \right) =: 
\left(\begin{array}{cc}
\Gamma^{MN} & 0 \\ 0 & \tilde{\Gamma}^{MN} \end{array} \right)
\end{aligned}
\end{equation}
where \(\Gamma^M, \tilde{\Gamma}^M\) act on the Majorana-Weyl subspaces, exchanging chirality, and are taken to be symmetric.\footnote{In the Euclidean signature, we would use \(\Gamma^M_E = \lbrace \Gamma^1,\cdots,\Gamma^9,i\Gamma^0\rbrace\).}

Let the gauge group \(G\) be a compact Lie group, and \(\mathfrak{g}\) its Lie algebra. The (on-shell) component field content of the gauge multiplet in 10 dimensions is of a gauge field, locally represented as \(A\in\Omega^1(\mathbb{R}^{9,1},\mathfrak{g})\), and a gaugino, a Mayorana-Weyl spinor \(\Psi:\mathbb{R}^{9,1}\to S^+ \otimes \mathfrak{g}\) with values in the Lie algebra \(\mathfrak{g}\). The field strength of the gauge field is locally represented by \(F=dA+[A,A]\), and the associated gauge-covariant derivative on \(\mathbb{R}^{9,1}\) is \(D_M = \partial_M + A_M\). The supersymmetry variations under the action of the 10-dimensional super-Poincaré algebra are
\begin{equation}
\label{pestun-1}
\begin{aligned}
\delta_\epsilon A_M &= \epsilon \Gamma_M \Psi \\
\delta_\epsilon \Psi &= \frac{1}{2}\Gamma^{MN}F_{MN}\epsilon
\end{aligned}
\end{equation}
where \(\epsilon\) is a Majorana-Weyl spinor, analogously to the on-shell version of \eqref{eq:susyvar-gauge} up to the chirality projection and conventional prefactors. The action functional for the \(\mathcal{N}=1\) 10-dimensional theory is \(S_{10d} = \int d^{10}x\ \mathcal{L}\), with Lagrangian
\begin{equation}
\label{pestun-2}
\mathcal{L} = \frac{1}{g_{YM}^2}\mathrm{Tr}\left( \frac{1}{2}F_{MN} F^{MN} - \Psi \Gamma^M D_M \Psi \right)
\end{equation}
where \(\mathrm{Tr}\) denotes a symmetric bilinear pairing in \(\mathfrak{g}\),\footnote{For semisimple \(\mathfrak{g}\), this is the Killing form as usual.} and \(g_{YM}\) is the Yang-Mills coupling constant. As we remarked in Section \ref{sec:susy-actions}, this action is exactly supersymmetric under \eqref{pestun-1} without the addition of auxiliary fields.

To get the Euclidean 4-dimensional theory, we perform dimensional reduction along the directions \(x^0,x^5,\cdots,x^9\), assuming independence of the fields on these coordinates. The fields split as
\begin{equation}
\label{pestun-3}
\begin{aligned}
A_M &\to \left((A_\mu)_{\mu=1,\cdots,4}, (\Phi_A)_{A=5,\cdots,9,0}\right) \\
\Psi &\to \left( \psi^L \ \chi^R \ \psi^R \ \chi^L\right)^T
\end{aligned}
\end{equation}
where \(\psi^{L/R}, \chi^{L/R}\) are four-component real chiral spinors. The spacetime symmetry group \(Spin(9,1)\) is broken to  \(Spin(4)\times Spin(5,1)^\mathcal{R}\hookrightarrow Spin(9,1)\), where \(Spin(4)\cong SU(2)_L\times SU(2)_R\) acts on the \(x^1,\cdots,x^4\) directions, and the R-symmetry group \(Spin(5,1)^\mathcal{R}\) rotates the other ones. It is often convenient to further break the R-symmetry group to \(Spin(4)^\mathcal{R}\times SO(1,1)^\mathcal{R}\hookrightarrow Spin(5,1)^\mathcal{R}\), where the first piece \(Spin(4)^\mathcal{R}\cong SU(2)_L^\mathcal{R}\times SU(2)_R^\mathcal{R}\) rotates the \(x^5,\cdots,x^8\) directions, and \(SO(1,1)^\mathcal{R}\) acts on the \(x^9,x^0\) ones. We thus consider the symmetry group
\begin{equation}
SU(2)_L\times SU(2)_R \times SU(2)_L^\mathcal{R}\times SU(2)_R^\mathcal{R} \times SO(1,1)^\mathcal{R}
\end{equation}
under which the fields behave as
\begin{itemize}
\item \(A_\mu\): vector of \(SU(2)_L\times SU(2)_R\), scalar under R-symmetry;
\item \((\Phi_I)_{I=4,\cdots,8}\): 4 scalars under \(SU(2)_L\times SU(2)_R\), vector of \(SU(2)_L^\mathcal{R}\times SU(2)_R^\mathcal{R}\), scalars under \(SO(1,1)^\mathcal{R}\);
\item \(\Phi_9,\Phi_0\): scalars under \(SU(2)_L\times SU(2)_R \times SU(2)_L^\mathcal{R}\times SU(2)_R^\mathcal{R}\), vector of \(SO(1,1)^\mathcal{R}\);
\item \(\psi^{L/R}\): \( \left(\frac{1}{2},0 \right) / \left(0, \frac{1}{2} \right)\) of \(SU(2)_L\times SU(2)_R\), \(\left(\frac{1}{2},0 \right)\) of \(SU(2)_L^\mathcal{R}\times SU(2)_R^\mathcal{R}\), \(+/-\) of \(SO(1,1)^\mathcal{R}\);
\item \(\chi^{L/R}\): \( \left(\frac{1}{2},0 \right) / \left(0, \frac{1}{2} \right)\) of \(SU(2)_L\times SU(2)_R\), \(\left(0,\frac{1}{2} \right)\) of \(SU(2)_L^\mathcal{R}\times SU(2)_R^\mathcal{R}\), \(-/+\) of \(SO(1,1)^\mathcal{R}\);	
\end{itemize}
here we denoted \(+,-\) the inequivalent Majorana-Weyl representations of \(SO(1,1)^\mathcal{R}\), seen as a subgroup of \(Cl(1,1)\cong \mathrm{Mat}_2(\mathbb{R})\).

The above decomposition of \(Spin(9,1)\) into four subrepresentations of \(Spin(4)\), rotated into each other by the R-symmetry group, gives the \(\mathcal{N}=4\) supersymmetry algebra on \(\mathbb{R}^4\). The supersymmetry variations of the reduced component fields are given by \eqref{pestun-1}, read in terms of the splitting \eqref{pestun-3},
\begin{equation}
\label{eq:susyvar-N4}
\begin{aligned}
\delta_\epsilon A_\mu &= \epsilon \Gamma_\mu \Psi \\
\delta_\epsilon \Phi_A &= \epsilon \Gamma_A \Psi \\
\delta_\epsilon \Psi &= \frac{1}{2}\left(\Gamma^{\mu\nu}F_{\mu\nu} + \Gamma^{AB}[\Phi_A,\Phi_B] + \Gamma^{\mu A}D_\mu \Phi_A \right)\epsilon.
\end{aligned}
\end{equation}
The action of the 4d \(\mathcal{N}=4\) SYM theory is \(S^{\mathcal{N}=4} = \int d^4x\ \mathcal{L}\) with the Lagrangian obtained by the reduction of \eqref{pestun-2}. More explicitly,
\begin{equation}
S^{\mathcal{N}=4} = \int d^4x \frac{1}{g_{YM}^2}\mathrm{Tr}\left(
\frac{1}{2}F_{\mu\nu}F^{\mu\nu} + (D_\mu \Phi_A)^2 - \Psi \Gamma^\mu D_\mu \Psi + \frac{1}{2}[\Phi_A,\Phi_B]^2 - \Psi \Gamma^A [\Phi_A,\Psi] \right) .
\end{equation}
Notice that, since contractions of \(A,B\) indices are done with a reduced Minkowski metric, upon dimensional reduction from the Lorentzian theory the scalar \(\Phi_0\) has a negative kinetic term. Analogously to the last section, we consider it to be purely immaginary, \textit{i.e.}\ \(\Phi_0 =: i\Phi_0^E\) with \(\Phi_0^E\) real. This makes the path integral match with the would-be reduction from the Euclidean \((10,0)\)-dimensional theory. 

The \(\mathcal{N}=4\) algebra closes \emph{on-shell}. In fact, it can be obtained from \eqref{eq:susyvar-N4} that
\begin{equation}
\label{pestun-flatalgebra}
\delta_\epsilon^2  = \frac{1}{2}[\delta_\epsilon,\delta_\epsilon] = - \mathcal{L}_v - G_\Phi
\end{equation}
up to the imposition of the EoM for \(\Psi\), \(\Gamma^M D_M\Psi=0\). Here \(v^M:= \epsilon\Gamma^M \epsilon\), \(\mathcal{L}_v\) is the Lie derivative (the action of the translation algebra) with respect to \(v\sim v^\mu \partial_\mu\), and \(G_\Phi\) is an infinitesimal gauge transformation with respect to \(\Phi:=A_M v^M\). A famous non-renormalization theorem by Seiberg \cite{seiberg-nonrenormalization} states that  the \(\mathcal{N}=4\) theory is actually \emph{superconformal}, \textit{i.e.}\ it has a larger supersymmetry algebra that squares to the \emph{conformal algebra}, whose generators are the Poincaré generators plus the generators of dilatations and special conformal transformations.\footnote{More precisely, the theorem states that the beta function of \(g_{YM}\) is zero non-perturbatively. This means that the theory is fully scale invariant at quantum level.} In fact, one can see that \(S^{\mathcal{N}=4}\) is classically invariant under supersymmetry variations with respect to the non-constant spinor
\begin{equation}
\label{pestun-10}
\epsilon = \hat{\epsilon}_s + x^\mu \Gamma_\mu \hat{\epsilon}_c
\end{equation}
where \(\hat{\epsilon}_s, \hat{\epsilon}_c\) are constant spinors parametrizing supertranslations and superconformal transformations. This enlarged supersymmetry algebra closes now on the superconformal algebra,
\begin{equation}
\label{pestun-7}
\delta_\epsilon^2 = - \mathcal{L}_v - G_{\Phi} - R - \Omega
\end{equation}
where  \(R\) is a \(Spin(5,1)^\mathcal{R}\) rotation, acting on scalars as \((R\cdot \Phi)_A = R_A^B\Phi_B\), and on spinors as \((R\cdot \Psi) = \frac{1}{4}R_{AB} \Gamma^{AB} \Psi\), where \(R_{AB} = 2 \epsilon \tilde{\Gamma}_{AB} \tilde{\epsilon}\). \(\Omega\) is an infinitesimal dilatation with respect to the parameter \(2(\tilde{\epsilon}\epsilon)\), acting on the gauge field trivially, on scalars as \(\Omega \cdot\Phi = -2(\tilde{\epsilon}\epsilon)\Phi\) and on spinors as \(\Omega\cdot\Psi = -3(\tilde{\epsilon}\epsilon) \Psi\). This new bosonic transformations are clearly symmetries of \(S^{\mathcal{N}=4}\).

Now we can restrict the attention to an \(\mathcal{N}=2\) subalgebra, considering the variations with respect to Majorana-Weyl spinors of the form 
\begin{equation}
\label{pestun-N2spinor}
\epsilon = \left( \epsilon^L \ 0 \ \epsilon^R \ 0 \right)^T 
\end{equation}
so in the subrepresentation \(\left( \left(\frac{1}{2},0 \right)\oplus \left(0,\frac{1}{2} \right) \right)\oplus \left(\frac{1}{2},0 \right)^\mathcal{R}\oplus (+ \oplus -)^\mathcal{R}\), the eigenspace of \(\Gamma^{5678}\) with eigenvalue +1. With respect to these supersymmetry variations, the gauge multiplet further splits  in
\begin{itemize}
\item \((A_\mu, \Phi_9, \Phi_0, \psi^L, \psi^R)\): the \(\mathcal{N}=2\) vector multiplet;
\item \((\Phi_I, \chi^L, \chi^R)\): the \(\mathcal{N}=2\) hypermultiplet, with value in the adjoint representation of \(G\).
\end{itemize}
These two multiplets are completely disentangled in the free theory limit \(g_{YM}^2\to 0\).\footnote{Working out the restricted supersymmetry variations taking into account the splitting of the gaugino, some non-linear term, coupling the fermionic sectors of the two multiplets, survives because of the gauge interaction. In the free theory limit, after the rescaling \(A_M \mapsto g_{YM}A_M, \Psi\mapsto g_{YM}\Psi\), these terms go to zero.} The same Lagrangian thus equivalently describes an \(\mathcal{N}=2\) matter-coupled gauge theory. It is also possible to insert a mass for the hypermultiplet, breaking explicitly the conformal invariance, and obtain the so-called \(\mathcal{N}=2^*\) theory. Since the fields of the vector multiplet are all scalars under \(SU(2)_R^{\mathcal{R}}\), and the hypermultiplet fields are all in the \(\frac{1}{2}\) representation, 
 these mass terms can at most rotate the hypermultiplet content with an  \(SU(2)_R^{\mathcal{R}}\) transformation. Thus replacing \( D_0\Phi_I \mapsto [\Phi_0,\Phi_I] + M_I^J\Phi_J\) and \(D_0\Psi \mapsto [\Phi_0,\Psi] + \frac{1}{4}M_{IJ}\Gamma^{IJ}\Psi\), where \((M^I_J)\) represents an \(SU(2)_R^{\mathcal{R}}\) rotation in the vector representation, one obtains the mass terms for the \(\Phi_I\) and \(\chi\) fields. Notice that \(\delta_\epsilon^2\) gets a contribution from the Lie derivative with respect to \(v^0\partial_0\cong 0 \mapsto v^0 M \), so that in the \(2^*\) theory
\begin{equation}
\begin{aligned}
\delta_\epsilon^2 \Phi_I &\mapsto (\delta_\epsilon \Phi_I)_{\mathcal{N}=2} - v^0 M_I^J \Phi_J \\
\delta_\epsilon^2 \chi &\mapsto (\delta_\epsilon \chi)_{\mathcal{N}=2} - \frac{1}{4} v^0 M_{IJ}\Gamma^{IJ} \chi .
\end{aligned}
\end{equation} 
In the limits of infinite or zero mass, the pure \(\mathcal{N}=2\) or \(\mathcal{N}=4\) theory is recovered. Notice that, since we argued that \(\Phi_0\) should be integrated over purely immaginary values for the convergence of the path integral, also \((M_{IJ})\) should be taken purely immaginary.

\section{From flat to curved space}
\label{sec:susy-curved}

Recently, localization theory has been extensively used in the framework of quantum field theories with rigid super-Poincaré symmetry, to compute exactly partition functions or expectation values of certain supersymmetric observables, when the theory is formulated on a \emph{curved} compact manifold. This cures the corresponding partition functions from infrared divergences making the path integral better defined, and is consistent with the requirement of periodic boundary conditions on the fields, that allows to generalize properly the Cartan model on the infinite dimensional field space. We will come back to this last point in the next chapter, when we will study circle localization of path integrals, while we close this chapter reviewing the idea behind some common approaches used to formulate rigid supersymmetry on curved space. 
%This is in principle a controversial task since, as we will see, curved space would imply local supersymmetry (or, if one wishes, the other way around), but anyway it is sometimes possible to describe a rigidly supersymmetric theory on a given compact smooth manifold \(\mathcal{M}\) with metric \(g\).

Following the approach of the last section, we would have to understand what does it mean to have supersymmetry on a generic metric manifold \((\mathcal{M},g)\) (Riemannian or pseudo-Riemannian) of dimension \(\dim{\mathcal{M}}=d\) from a geometric point of view. The supersymmetry of flat space was constructed as a super-extension of Minkowski (or Euclidean) space \(\mathbb{R}^d\), starting from a super-extension of the Lie algebra of its isometry group, the Poincaré group. Now in general the Poincaré group is not an isometry group for \(\mathcal{M}\), so the super Poincaré algebra \(\mathfrak{siso}_S(d)\) with respect to some (real or Majorana) spin representation \(S\) cannot be fully interpreted as a \lq\lq supersymmetry\rq\rq\ algebra for the space at hand. We can nonetheless associate in some way this algebra to a suitable super-extension of \(\mathcal{M}\), and then ask what part of it can be preserved as a supersymmetry of this supermanifold. We follow \cite{Alekseevsky-killing-sp} for this geometric introduction.

Since we want to work with spinors, we assume that \(\mathcal{M}\) admits a spin-structure. In particular, it exists a (real) spinor bundle \(S\to\mathcal{M}\) associated to the spin-structure, with structure sheaf \(\mathcal{S}: \mathcal{S}(U) =\Gamma (U,S), \forall U\subset \mathcal{M}\) open. Analogously to the flat superspace of the last section, we make now a super-extension of \(\mathcal{M}\) through this spinor bundle considering the \emph{odd spinor bundle} \(S\mathcal{M}_S := \Pi S\), with body \(\mathcal{M}\) and structure sheaf \(\bigwedge \mathcal{S}^*: \bigwedge\mathcal{S}^*(U) = C^\infty (\mathcal{U}) \otimes \bigwedge(S_0^*), \forall U\subset \mathcal{M}\) open, where \(S_0\) is the typical fiber of \(S\). From proposition \ref{prop:isom-supervectorbundles}, for any \(p\in \mathcal{M}\), there is an isomorphism of \(\mathbb{Z}_2\)-graded vector spaces \(T_p S\mathcal{M}_S \cong T_p\mathcal{M} \oplus S_p[1]\).

Now, the vector bundle \(V:= T\mathcal{M} \oplus S\) over \(\mathcal{M}\) carries the canonical spin-connection induced by the Levi-Civita connection of the manifold \((\mathcal{M},g)\). Assume that we can pick a parallel non-degenerate \(Spin(d)\)-invariant bilinear form \(\beta\) on \(S\) with respect to this connection.\footnote{This is always true if \(\mathcal{M}\) is simply-connected.} We can think of the \(Spin(d)\)-invariant bilinear form \(\tilde{g}=g+\beta\) as a (pseudo-)Riemannian metric on the supermanifold \(S\mathcal{M}_S\). Moreover, associated to the bilinear form \(\beta\) we have the map \(\Gamma: S^2 \to T\mathcal{M}\), that is a point-wise generalization of the usual symmetric and Spin-equivariant bilinear form for a Spin representation \(S_0 \cong S_p, \forall p\in\mathcal{M}\). This means that we can consider the bundle
\begin{equation}
\mathfrak{p}(V) := \mathfrak{spin}(d) \oplus V 
\end{equation}
as a \emph{bundle of super Poincaré algebras} over \(\mathcal{M}\), with the bracket structure extended through \(\Gamma\).

Having found how to (point-wise) set up the super Poincaré algebra on top of the supermanifold \(S\mathcal{M}_S\) constructed from \((\mathcal{M},g)\), we wish to establish which section of the super Poincaré bundle \(\mathfrak{p}(V)\) produces a suitable generalization of \lq\lq super-isometry\rq\rq\ for \(S\mathcal{M}_S\). In particular, we pay attention to which sections of \(S\), as the odd  subbundle of \(\mathfrak{p}(V)\), generates \lq\lq supersymmetries\rq\rq\ of the generalized metric \(\tilde{g}\). This problem was analyzed in \cite{Alekseevsky-killing-sp}, and connected to the problem of finding solution to the so called \emph{Killing spinor equation} for a section \(\psi\) of \(S\to \mathcal{M}\).
\begin{defn} A section \(\psi\) of the spinor bundle \(S\to \mathcal{M}\) is called a \emph{twistor spinor} (or \emph{conformal Killing spinor}) if it exists another section \(\phi\) such that, for any vector field \(X\in \Gamma(T\mathcal{M})\),
\begin{equation}
\label{eq:killing-sp-eq}
\nabla_X \psi = X \cdot \phi
\end{equation}
where \(X\cdot \phi = X^\mu \gamma_\mu \phi\) is the \emph{Clifford multiplication}. If in particular \(\phi = \lambda \psi\), for some constant \(\lambda\), the spinor \(\psi\) is called \emph{Killing spinor}.
\end{defn}
The equation (\ref{eq:killing-sp-eq}) is called \emph{twistor} or \emph{Killing spinor equation}. Note that \eqref{eq:killing-sp-eq} directly implies \(\phi =\pm (1/\dim(\mathcal{M})) \slashed{\nabla}\psi\), where \(\slashed{\nabla}:=\gamma^\mu \nabla_\mu\) is the Dirac operator, and the sign depends on conventions. The twistor spinor equation is thus equivalently written as
\begin{equation}
\nabla_X \psi = \pm \frac{1}{\dim{(\mathcal{M})}} X \cdot \slashed{\nabla}\psi .
\end{equation}
This characterizes the Killing spinors as those twistor spinors that satisfies also the Dirac equation \(\slashed{\nabla}\psi = m \psi\) for some constant \(m\). The main result proved in \cite{Alekseevsky-killing-sp} is stated in the following theorem.
\begin{thm}
\label{thm:superisometries}
Consider the supermanifold \(S\mathcal{M}_S\) with the bilinear form \(\tilde{g}=g+\beta\), and a section \(\psi\) of \(S\). The odd vector field \(X_\psi\) associated to \(\psi\) is a Killing vector field of \((S\mathcal{M}_S, \tilde{g})\) if and only if \(\psi\) is a twistor spinor.
\end{thm}
Here the Killing vector condition on the supermanifold is a conceptually straightforward generalization of the usual concept of Killing vector fields on a smooth manifold. It can be natually stated in terms of superframe fields. We refer to the above cited article for the details. Notice that, in particular, Killing spinors generate infinitesimal isometries of the supermanifold \(S\mathcal{M}_S\), and thus are good candidates to describe the \lq\lq preserved\rq\rq\ supersymmetries of the odd part of the super Poincaré algebra, when this is associated to the generic curved manifold \(\mathcal{M}\) in the way we saw above. See also \cite{baum-killing-sp} for a review on Killing spinors in (pseudo-)Riemannian geometry.

\medskip
From the QFT point of view, it is possible to derive a (generalized) Killing spinor equation, describing the preserved supercharges on the curved space, from a dynamical approach. This idea is based on a procedure also valid in the non-supersymmetric setting, when one aims to deform a certain QFT to redefine it on a generic curved manifold. In this case, one couples the theory to \emph{background} gravity, letting the metric fluctuate.\footnote{Since the metric can fluctuate and the field theory is defined locally, there is no harm in principle in considering different topologies of the base manifold, like requiring it to be compact.} Then the gravitational sector is decoupled from the rest of the theory taking the gravitational constant \(G_N \to 0\), while the metric is linearized around the chosen \emph{off-shell} configuration \(g=\eta + h\) and the higher order corrections disappear in the limit of weak gravitational interaction. It is important that we do not constrain the gravitational field to satisfy the equation of motion, since it is considered as a background (classical) field. The same idea applies when the theory is defined in the supersymmetric setting: in this case to preserve supersymmetry one has to couple to background \emph{supergravity (SUGRA)}. The resulting field theory will contain then more fields belonging to the so-called \emph{supergravity multiplet}. This time, taking the limit \(G_N\to 0\) we fix all the background supergravity multiplet to an allowed off-shell configuration. Note that in particular, the auxiliary fields are not eliminated in terms of the other fields using their equations of motion. If then there are supergravity transformations that leave the given background invariant, we say that the corresponding rigid supercharges are \emph{preserved} on this background. This procedure was systematically introduced in \cite{festuccia-susy-curved}, then many cases and classification were made in different dimensions and with different degree of supersymmetriy (see for example \cite{Zaffaroni-susy-curved-holography,Closset-susy-curved3d,Dumitrescu-curved,Kehagias-susy-curved}).

\subsection{Coupling to background SUGRA} %\mbox{} \\ \mbox{} \\

Suppose we have a supersymmetric field theory formulated on flat space specified by its Lagrangian \(\mathcal{L}^{(0)}\), whose variation under supersymmetry is a total derivative:\footnote{For simplicity, we consider now the formulation on the even space \(\mathbb{R}^d\) or \(\mathcal{M}\), at the level of component fields of the given supersymmetry multiplets. The supersymmetry variation of these component fields are those coming from the action of the odd supertranslations in superspace.}
\begin{equation}
\delta\mathcal{L}^{(0)} = *d*(\cdots) = \partial_\mu (\cdots)^\mu
\end{equation}
We can introduce supergravity by requiring the action of the super-Poincaré group to be \emph{local}, employing the usual gauge principle and minimal coupling or Noether procedure. 

In the non-supersymmetric setting, this would mean to introduce a gauge symmetry under local coordinate transformations, realized via diffeomorphisms on \(\mathcal{M}\). The Noether current associated to such infinitesimal transformations is the \emph{energy-momentum tensor} \(T^{\mu\nu}\), that we take to be symmetric.\footnote{In general this will not be a symmetric tensor, but there always exists a suitable modification that makes it symmetric, and moreover equivalent to the Hilbert definition of energy momentum tensor as a source of gravitational field. This is the \emph{Belinfante–Rosenfeld tensor} 
\begin{equation}
\tilde{T}^{\mu\nu} := T^{\mu\nu} + \frac{1}{2}\nabla_\lambda ( S^{\mu\nu\lambda} + S^{\nu\mu\lambda} - S^{\lambda\nu\mu})
\end{equation}
where \(S^\mu_{\nu\lambda}\) is the spin part of the Lorentz generators in a given spin representation satisfying \(\nabla_\mu S^\mu_{\nu\lambda} = T_{\nu\mu} - T_{\mu\nu}\), and \(\nabla\) is an appropriate torsion-free spin-covariant derivative induced from the metric (see for example \cite{diFrancesco-CFT}).} The minimal coupling procedure then requires to modify the Lagrangian,
\begin{equation}
\mathcal{L}' = \mathcal{L}^{(0)} + h_{\mu\nu} T^{\mu\nu} + O(h^2)
\end{equation}
where \(h_{\mu\nu}\) is regarded as a variation of the metric from the flat space values \(\eta_{\mu\nu}\), and \(O(h^2)\) are seagull non-linear terms that can be fixed requiring the gauge invariance of \(\mathcal{L}'\). The resulting non-linear coupling is obtainable substituting in the original theory 
\begin{equation}
d \mapsto \nabla , \qquad \eta \mapsto g=\eta + h,
\end{equation}
where \(\nabla\) is the gauge-covariant derivative with respect to a connection \(\Gamma\), that we take as the Levi-Civita connection. The theory is now coupled to a gravitational (classical) \emph{background}. If we want to make the graviton field \(h\) dynamical, we can add an Hilbert-Einstein term to \(\mathcal{L}^{(0)}\),
\begin{equation}
\label{eq:H-Eaction}
\mathcal{L}_{HE} = -\frac{\sqrt{|\det{ g}|}}{2\kappa^2} Ric_g
\end{equation}
where \(Ric_g\) is the Ricci scalar associated to \(g\), and \(\kappa:= 1/M_p = \sqrt{8\pi G_N}\).

In the supersymmetric setting, gauging the super-Poincaré group leads to the introduction of more fields into the theory, since as we know they can be interpreted as components of superfields in superspace, and thus belong to supersymmetry multiplets. In particular, we have to introduce the \emph{graviton multiplet} composed by the metric \(g\), the gravitino \(\Psi\) and other (maybe auxiliary) fields. The particular field content depends on the number \(\mathcal{N}\) of supersymmetries, the dimensionality of the theory, the presence or absence of an R-symmetry and whether the theory is or not superconformal. Consequently, also the energy-momentum tensor will belong to a multiplet, the so-called \emph{supercurrent multiplet}, composed by \(T\), a \emph{supercurrent} \(J\) associated to the local invariance under (odd) supertranslations, and other fields. We can schematically perform the first steps of the Noether procedure to see how the components of these multiplets arise naturally.

We start from the odd part of the super-Poincaré algebra \( \mathfrak{iso}(\mathbb{R}^d)\oplus S[1]\), writing the infinitesimal variation of the Lagrangian in terms of the supercurrent:
\begin{equation}
\delta_\epsilon \mathcal{L}^{(0)} \equiv \epsilon \cdot \mathcal{L}^{(0)} = (\partial_\mu J^\mu ) \epsilon 
\end{equation}
where \(\epsilon\) is now a Majorana spinor field, \textit{i.e.} a section of the spinor bundle with fiber \(S\). The supercurrent \(J\) is an \(S\)-valued vector field, and the spinor contraction is done via the usual charge conjugation matrix. We couple this current to a gauge field \(\Psi\), to be identified with the gravitino, that is locally an \(S\)-valued 1-form such that, at linearized level,
\begin{equation}
\delta_\epsilon \Psi_\mu = \frac{1}{\kappa}\partial_\mu \epsilon
\end{equation}
where the constant \(\kappa\) is introduced for dimensional reasons.\footnote{If we canonically take mass dimensions of scalars to be \((d-2)/2\) and of spinors to be \((d-1)/2\), and since schematically \[
\delta_\epsilon (boson) = (fermion) \epsilon \]
than \([\epsilon] = -1/2\), so \(\kappa\) must be a dimensionfull parameter of mass dimension \([\kappa]=(2-d)/d\). We can so identify this constant as the gravitational constant previously defined.} Then we add a term to the Lagrangian:
\begin{equation}
\mathcal{L}' = \mathcal{L}^{(0)} + \kappa \Psi_\mu J^\mu .
\end{equation}
Now the variation of \(\mathcal{L}'\) is proportional to the variation of the current \(\delta_\epsilon J\). Since the supercurrent is a supersymmetry variation of the original Lagrangian, its variation will be proportional to the action of the translation generators \(P_\mu\):
\begin{equation}
\begin{aligned}
\delta_\epsilon^2 \mathcal{L}^{(0)} &=  \epsilon \cdot (\epsilon \cdot \mathcal{L}^{(0)}) = \partial_\mu (\delta_\epsilon J^\mu ) \epsilon  \\
&= (1/2) [\epsilon , \epsilon ] \cdot \mathcal{L}^{(0)} =  \overline{\epsilon}\gamma^\nu \epsilon (P_\nu \cdot \mathcal{L}^{(0)} ) =  (\overline{\epsilon}\gamma_\nu \epsilon ) \partial_\mu T^{\mu\nu} \\
\Rightarrow \delta_\epsilon J^\mu &= \overline{\epsilon} \gamma_\nu T^{\nu\mu}
\end{aligned}
\end{equation}
where in the second line we wrote the variation under the action of the translation generators in terms of the energy-momentum tensor \(T\). We try to restore the gauge-invariance of the Lagrangian by minimally coupling this new current to a new gauge field \(h\), that we identify as a metric variation, the \emph{graviton}
\begin{equation}
\mathcal{L}'' = \mathcal{L}^{(0)} + \kappa  \Psi_\mu J^\mu + h_{\mu\nu}T^{\mu\nu} ,
\end{equation}
and naturally requiring the supersymmetry variation of the graviton to be
\begin{equation}
\delta_\epsilon h_{\mu\nu} = \kappa \overline{\epsilon}\gamma_{(\mu} \Psi_{\nu)} ,
\end{equation}
making it the superpartner of the gravitino \(\Psi\). 

The Lagrangian \(\mathcal{L}''\) is again not supersymmetric, since the variation \(\delta_\epsilon T^{\mu\nu} \neq 0\) in general, so the Noether procedure is not terminated yet. It is not easy to complete this procedure in this way, but in principle repeating this 
passages we would introduce more linearly coupled currents and gauge fields that, motivated by supersymmetry, we expect to come from the SUGRA supermultiplets mentioned above. To ensure the supersymmetry of the full Lagrangian at non-linear level, as in non supersymmetric gauge theories, non-linear couplings could have to be introduced as well as non-linear terms in the supersymmetry variations. Summarizing, we expect the fully coupled Lagrangian to be schematically of the form
\begin{equation}
\label{eq:min-coupl-sugra-L}
\mathcal{L} = \mathcal{L}^{(0)} + \kappa J^\mu \Psi_\mu + h_{\mu\nu}T^{\mu\nu} + \sum_i \mathcal{B}^i \cdot \mathcal{J}^i + (\mathrm{seagull\ terms})
\end{equation} 
where \(\mathcal{B}\) is the multiplet of background gauge fields \((h,\Psi,\cdots)\), \(\mathcal{J}\) the supercurrent multiplet \((T, J, \cdots)\), and we referred to possible higher-order terms in the background fields as seagull terms. As already said, the particular field content of these multiplets is not unique, so we remain generic for the moment and refer to the next subsections for some examples.  We can absorb the terms proportional to \(h\) as in the non-supersymmetric case, by making the substitutions \( d \mapsto \nabla\) and \( \eta \mapsto g=\eta+h\). If we want to have a full gravitational theory, we can add a kinetic term for the source fields, and complete their supersymmetry variations with possible non-linear terms to ensure gauge invariance. Regarding the metric and the gravitino, the kinetic terms are given by the Hilbert-Einstein action (\ref{eq:H-Eaction}) and the \emph{Rarita-Schwinger} action
\begin{equation}
\mathcal{L}_{RS} = - \frac{\sqrt{|\det{g}|}}{2} \overline{\Psi}_\mu \gamma^{\mu\nu\rho}(\nabla_\nu \Psi)_\rho ,
\end{equation}
where \(\gamma^{\mu\nu\rho}:=\gamma^{[\mu}\gamma^\nu \gamma^{\rho]}\), and \(\nabla\) acts on spinors via the spin connection, \((\nabla_\mu \Psi)_\nu = \partial_\mu \Psi_\nu + \frac{1}{4}\omega_\mu^{ab} \gamma_{ab} \Psi_\nu - \Gamma_{\mu\nu}^\rho \Psi_\rho\). The supersymmetry variations will be generically
\begin{equation}
\begin{aligned}
&\delta_\epsilon h_{\mu\nu} = \kappa \overline{\epsilon} \left\lbrace \gamma_{(\mu}\Psi_{\nu)} + (\cdots)^F \right\rbrace \\
&\delta_\epsilon \Psi_\mu = \frac{1}{\kappa}\left\lbrace \nabla_\mu   + (\cdots)_\mu^B \right\rbrace \epsilon + O(\kappa\Psi^2\epsilon) 
\end{aligned}
\end{equation}
where we stressed that non-linear higher-oreder terms for the gravitino are \(\kappa\)-suppressed, and the ellipses in both cases collect  contributions from the other (fermionic or bosonic, respectively) fields of the supergravity multiplet. Notice that also the supersymmetry variations of the field content of the original \(\mathcal{L}^{(0)}\) get modified with respect to their flat-space version. Once one has the full supergravity theory, their transformation rules follow from the corresponding formulas in the appropriate matter-coupled off-shell supergravity.

We now consider the \emph{rigid limit} \(G_N\to 0\) (or \(\kappa \to 0\), or \(M_P \to \infty\)) together with the choice of a given background gravitational multiplet \(\mathcal{B}\) compatible with the original request \((\mathcal{M},g)\).\footnote{We stress that a rigid supersymmetric background is characterized by a full set of supergravity background fields, \textit{i.e.} specifying only the metric does not determine the background. In particular, there are distinct backgrounds that have the same metric but lead to different partition functions.} Since we think at this classical configuration as a VEV, we require all the fermion fields in the supergravity multiplet to vanish on this background. We also look for those supergravity transformations that leave this background invariant.\footnote{This can be interpreted as a superisometry requirement with respect to the graviton multiplet.} These requirements produce the following effects:
\begin{itemize}
\item Fermionic gravitational fields as well as the kinetic term for the bosonic gravitational sector do not contribute to the lagrangian:
\begin{equation}
\mathcal{L} = \left. \mathcal{L}^{(0)}\right|_{\substack{d\to\nabla \\ \eta\to g }}  + \sum_i \mathcal{B}_B^i \cdot \mathcal{J}_B^i + (\mathrm{seagull\ terms})_B
\end{equation}
At the same time, supersymmetry variations of the bosonic gravitational fields automatically vanish.

\item Requiring the supersymmetry of the background is then equivalent to
\begin{equation}
\delta_\epsilon B_F = 0 .
\end{equation}
In particular, this condition on the gravitino generates the \emph{generalized Killing spinor equation}
\begin{equation}
\delta_\epsilon \Psi_\mu = 0 \quad \Leftrightarrow \quad \nabla_\mu \epsilon = (\cdots)_\mu \epsilon
\end{equation}
where, again, ellipses stand for terms proportional to bosonic fields in the graviton multiplet. The solutions to this equation determine which sections of the spinor bundle \(S\to\mathcal{M}\) generates the preserved supersymmetry transformations on \(\mathcal{M}\).
\end{itemize}

\subsection{Supercurrent multiplets and metric multiplets}

As we wrote before, the field content of supercurrent multiplets depends on the general properties of the theory at hand. In \cite{Dumitrescu-supercurrents} it was given a definition from basic general requirements starting from superfields in superspace, that allows a classification by specializing to the various particular cases. It is shown that the most general supercurrent is a real superfield \(\mathcal{S}_{a\dot{a}}\) satisfying
\begin{equation}
\label{eq:s-multiplet}
\begin{array}{ll}
\multicolumn{2}{l}{\tilde{D}^{\dot{a}} \mathcal{S}_{a\dot{a}} = \chi_a + \mathcal{Y}_a} \\
\tilde{D}^{\dot{a}} \chi_a = 0  &   \tilde{D}^{\dot{a}} \chi^\dagger_{\dot{a}} = D^a \chi_a \\
D_{(a} \mathcal{Y}_{b)} = 0   &   \tilde{D}^2 \mathcal{Y}_a = 0 .
\end{array}
\end{equation}
Every supersymmetric theory has such an \(\mathcal{S}\)-multiplet, containing the stress energy tensor \(T\) and the supercurrent \(S\). They are the only component fields with spin larger than one, since they couple to the graviton and the gravitino in the metric multiplet, that respectively are the only component fields with spin higher than one in this multiplet.  We report here special examples in 4 and 3 dimensions, that can be derived solving the constraints (\ref{eq:s-multiplet}) in cases where additional conditions on the superfields \(\chi_a\) and \(\mathcal{Y}_a\) hold.

For \(\mathcal{N}=1\) in 4-dimensions we have three possible interesting special cases:
\begin{enumerate}

\item The majority of theories admit a reduction of the \(\mathcal{S}\)-multiplet into the so-called \emph{Ferrara-Zumino (FZ) multiplet}
\begin{equation}
\mathcal{J}^{FZ}_\mu \to \left( j_\mu , (S_\mu)_a, x, T_{\mu\nu} \right)
\end{equation}
where \(j_\mu\) is a vector field and \(x\) a complex scalar field.

\item If the theory has a \(U(1)_R\) symmetry, the \(\mathcal{S}\)-multiplet reduces to the so-called \emph{\(\mathcal{R}\)-multiplet}, whose lower degree component is the conserved R-current \(j_\mu^{(R)}\):
\begin{equation}
\mathcal{R}_\mu \to \left( j_\mu^{(R)}, (S_\mu)_a , T_{\mu\nu}, C_{\mu\nu} \right)	
\end{equation}
where \(C_{\mu\nu}\) are the components of a conserved 2-form current, the so-called \emph{brane current}.\footnote{In curved space and in presence of topological defects as strings (1-brane) or domain walls (2-brane), the supersymmetry algebra \eqref{eq:susy-algebra-chiral} can be modified by the presence of \emph{brane charges}, \[ \begin{aligned}
 \left[ Q_a , \tilde{Q}_{\dot{b}}\right] &= 2(\gamma^\mu)_{a\dot{b}} (P_\mu + Z_\mu) \\  [Q_a , Q_b] &= (\gamma^{\mu\nu})_{ab} \tilde{Z}_{\mu\nu}
\end{aligned} \] 
where \(Z_\mu, \tilde{Z}_{\mu\nu}\) are non-zero for strings and domain walls, respectively. The corresponding tensor currents are the \emph{brane currents} \(C_{\mu\nu}, \tilde{C}_{\mu\nu\rho}\), that are topologically conserved. See \cite{Ferrara-central-extension, Gorsky-central-extension, Dumitrescu-supercurrents} for more details.}

\item For a superconformal theory, the \(\mathcal{S}\)-multiplet decomposes into the smaller supercurrent
\begin{equation}
\mathcal{J}_\mu \to \left( j_\mu^{(R)}, (S_\mu)_a , T_{\mu\nu} \right)
\end{equation}
where \(j_\mu^{(R)}\) is a conserved superconformal \(U(1)_R\)-current.
\end{enumerate}
Both the FZ multiplet and the \(\mathcal{R}\)-multiplet contain 12+12 real degrees of freedom out of the initial 16+16 of the general \(\mathcal{S}\)-multiplet, while the superconformal multiplet is reduced to 8+8 real degrees of freedom. The FZ multiplet can be coupled to the so-called \lq\lq old minimal supergravity multiplet\rq\rq\ \cite{stelle-west-old-sugra}:
\begin{equation}
\mathcal{H}_\mu \to \left( 	b_\mu, (\Psi_\mu)_a, M, h_{\mu\nu} \right)
\end{equation}
where \(b_\mu\) is a genuine 1-form field (\textit{i.e.} non gauge), \(M\) is a complex scalar, \((\Psi_\mu)_a\) is the gravitino and \(h_{\mu\nu}\) is the graviton. The variation of the gravitino in this case is given by \cite{festuccia-susy-curved}
\begin{equation}
\delta_\epsilon \Psi_\mu = -2\nabla_\mu \epsilon + \frac{i}{3}\left( 	M \gamma_\mu + 2b_\mu +2b^\nu \gamma_{\mu\nu} \right) \epsilon
\end{equation}
that implies a generalized Killing spinor equation of the form
\begin{equation}
\nabla_\mu \epsilon =  \frac{i}{6}\left(M \gamma_\mu + 2b_\mu +2b^\nu \gamma_{\mu\nu} \right) \epsilon
\end{equation}
in the Majorana spinor \(\epsilon\), given a background multiplet. In theories with an R-symmetry, one can couple the \(\mathcal{R}\)-multiplet to the \lq\lq new minimal supergravity multiplet\rq\rq\cite{sohnius-west-new-sugra}:
\begin{equation}
\mathcal{H}^{(new)}_\mu \to \left( 	A^{(R)}_\mu, (\Psi_\mu)_a, h_{\mu\nu}, B_{\mu\nu} \right)
\end{equation}
where \(A^{(R)}_\mu\) is the Abelian gauge field associated to the \(U(1)_R\) symmetry, and \(B_{\mu\nu}\) is a 2-form gauge field that is often treated through its Hodge dual \(V^\mu := i(\star B)^\mu = (i/2) \varepsilon^{\mu\nu\rho\sigma}\partial_\nu B_{\rho\sigma}\). The variation of the gravitino in this case gives rise to the following Killing spinor equation, that in 2-component notation is \cite{festuccia-susy-curved}
\begin{equation}
\label{eq:Killingsp-4d}
\begin{aligned}
&\left(\nabla_\mu - i A^{(R)}_\mu \right)\epsilon_a = -iV_\mu \epsilon_a - iV^\nu (\gamma_{\mu\nu}\epsilon)_a  \\
&\left(\nabla_\mu + i A^{(R)}_\mu \right)\tilde{\epsilon}_{\dot{a}} = iV_\mu \tilde{\epsilon}_{\dot{a}} + iV^\nu (\gamma_{\mu\nu}\tilde{\epsilon})_{\dot{a}}
\end{aligned}
\end{equation}
where in the parenthesis on the LHS there is a gauge-covariant derivative with respect to the \(U(1)_R\) symmetry, that acts with opposite charge on the two chiral sector of the spin representation, see (\ref{eq:U(1)R-supercharges}).

The \(\mathcal{N}=2\) case in 3 Euclidean dimensions can be derived in superspace by dimensional reduction from the four dimensional case: the supercurrent is reduced to a three dimensional \(\mathcal{S}\)-multiplet with 12+12 real DoF, plus a real scalar superfield \(\hat{\mathcal{J}} = \mathcal{S}_0 \equiv \mathcal{S}_{a\dot{a}}(\sigma_0)^{a\dot{a}}\), that contains 4+4 real DoF. Again, there are special cases analogue of those above: a FZ multiplet, an \(\mathcal{R}\)-multiplet, and a superconformal multiplet. For example, the \(\mathcal{N}=2\) \(\mathcal{R}\)-multiplet in 3 dimensions has the field content 
\begin{equation}
\mathcal{R}_\mu \to \left( j_\mu^{(R)}, j_\mu^{(Z)}, J, (S_\mu)_a, (\tilde{S}_\mu)_a, T_{\mu\nu}  \right)
\end{equation}
where \(j_\mu^{(R)}\) is the conserved R-current, \(j_\mu^{(Z)}\) is the conserved central charge current and \(J\) is a scalar operator, that with the conserved supercurrents and enery-momentum tensor sum up to 8+8 real DoF. This multiplet couples to the tree dimensional \(\mathcal{N}=2\) new minimal supergravity multiplet 
\begin{equation}
\mathcal{H}^{(new)}_\mu \to \left( A^{(R)}_\mu, C_\mu, H, (\psi_\mu)_a, (\tilde{\psi}_\mu)_a, h_{\mu\nu}  \right)
\end{equation}
with the graviton, two gravitini, two gauge 1-forms \(A^{(R)}\) and \(C\), and a scalar \(H\). The 1-form \(C\) is often treated in terms of the vector field \(V^\mu := i(\star dC)^\mu = i\varepsilon^{\mu\nu\rho} \partial_\nu C_\rho\) that is Hodge dual to its field strength. Putting to zero the gravitini and their variations leads to the generalized Killing spinor equations \cite{Closset-susy-curved3d}
\begin{equation}
\label{eq:Killingsp-3d}
\begin{aligned}
&\left(\nabla_\mu - i A^{(R)}_\mu \right)\epsilon = - \left( \frac{1}{2}H\gamma_\mu  + iV_\mu + \frac{1}{2}\varepsilon_{\mu\nu\rho} V^\nu \gamma^\rho \right) \epsilon  \\
&\left(\nabla_\mu + i A^{(R)}_\mu \right)\tilde{\epsilon} = - \left( \frac{1}{2}H\gamma_\mu  - iV_\mu - \frac{1}{2}\varepsilon_{\mu\nu\rho} V^\nu \gamma^\rho \right) \tilde{\epsilon}
\end{aligned}
\end{equation}
where in this case the two spinors \(\epsilon, \tilde{\epsilon}\) have to be treated as independent. Notice that both equations (\ref{eq:Killingsp-4d}) and (\ref{eq:Killingsp-3d}) are linear in the 4 spinor components, so their solutions (if exist) span a vector space of dimension less or equal than 4.

\subsection[N=2 gauge theories on the round 3-sphere]{\(\mathcal{N}=2\) gauge theories on the round 3-sphere}
\label{subsec:3dN2-curved}

It was shown that, in general, solutions of the Killing condition (\ref{eq:Killingsp-4d}) in four dimensions exist if \((\mathcal{M},g)\) is an Hermitian manifold, \textit{i.e.} \(\mathcal{M}\) has an integrable complex structure and \(g\) is a compatible Hermitian metric. Analogously, the existence of a solution to  (\ref{eq:Killingsp-3d}) in three dimensions was shown to be equivalent to the manifold admitting a \emph{transversally holomorphic fibration}.\footnote{This is an odd-dimensional analogue to a complex structure. It means, roughly speaking, that \(\mathcal{M}\) is locally isomorphic to \(\mathbb{R}\times \mathbb{C}\), and its transition functions are holomorphic in the \(\mathbb{C}\)-sector.} If one is interested in the case of maximal number of Killing spinor solutions, a suitable integrability condition (see \cite{Closset-susy-curved3d}) gives
\begin{equation}
\begin{aligned}
&H = \text{const} , \qquad d(A^{(R)}-V)=0 , \qquad  g(V,V)=\text{const} , \\
&(\nabla_\mu V)_\nu = -iH\varepsilon_{\mu\nu\rho}V^\rho ,  \\
&R_{\mu\nu} = -V_\mu V_\nu + g_{\mu\nu}( g(V,V) +2H^2).
\end{aligned}
\end{equation}
In particular, if we take \(A^{(R)}=V=0\), then \(\mathcal{M}\) is of Einstein type and so it has constant sectional curvature. \(H^2\) is then interpreted as a cosmological constant, and \(\mathcal{M}\) can be \(\mathbb{S}^3, \mathbb{T}^3\) or \(\mathbb{H}^3\) if \(H\) is purely immaginary, zero or real. All of them are examples of maximally supersymmetric backgrounds in \(\mathcal{N}=2\), thus we have 2 solutions for \(\epsilon\) and 2 solutions for \(\tilde{\epsilon}\) to the Killing equations
\begin{equation}
\nabla_\mu \epsilon = -\frac{H}{2}\gamma_\mu \epsilon \ ; \qquad \nabla_\mu \tilde{\epsilon} = -\frac{H}{2}\gamma_\mu \tilde{\epsilon} .
\end{equation}
In particular, if we take \(H=-(i/l)\), this solutions are consistent with the \(\mathbb{S}^3\) round metric
\begin{equation}
\label{eq:round-metric-3sph}
g = l^2 \left( d\varphi_1 \otimes d\varphi_1 + \sin^2\varphi_1 \ d\varphi_2 \otimes d\varphi_2 + \sin^2\varphi_1 \sin^2\varphi_2 \ d\varphi_3 \otimes d\varphi_3 \right).
\end{equation}

The action of the supersymmetry algebra on the curved manifold can be derived  by taking the rigid limit of the appropriate algebra of supergravity transformations. In the 3-dimensional case, it can be derived by a \lq\lq twisted\rq\rq\ reduction of the \(\mathcal{N}=1\) supergravity in 4 dimensions. The 4-dimensional supersymmetry algebra realizes on the curved manifold as
\begin{equation}
\label{eq:4dimsusycurved}
[\delta_{\epsilon}, \delta_{\epsilon} ] \phi_{(r)} = [\epsilon, \epsilon] \cdot \phi_{(r)} = 2(\overline{\epsilon}\gamma^\mu \epsilon) P_\mu \cdot \phi_{(r)}
\end{equation}
where \(\epsilon\) is a Majorana Killing spinor, \(\phi_{(r)}\) is a generic field of R-charge \(r\), and the local action of the momentum operator is through the fully covariant derivative
\begin{equation}
P_\mu \to  -\left( i\nabla_\mu + rA^{(R)}_\mu \right) ,
\end{equation}
so that (\ref{eq:4dimsusycurved}) can be written as
\begin{equation}
[\delta_{\epsilon}, \delta_{\epsilon} ] \phi_{(r)} = -2i\left(\mathcal{L}_v - irA^{(R)}(v)\right)\phi_{(r)}
\end{equation}
where \(v^\mu := (\overline{\epsilon}\gamma^\mu \epsilon)\) is a Killing vector field thanks to \(\epsilon\) being a Killing spinor field. This is reduced to the 3-dimensional case, taking \(\epsilon, \tilde{\epsilon}\) now as independent 2-component Killing spinors and \(v^\mu := \tilde{\epsilon}\gamma^\mu \epsilon\) in 3 dimensions, as (see again \cite{Closset-susy-curved3d})
\begin{equation}
\begin{aligned}
&[\delta_{\tilde{\epsilon}}, \delta_{\epsilon}] \phi_{(r,z)} = -2i \left[ \mathcal{L}_v -iv^\mu\left( r(A^{(R)}_\mu - \frac{1}{2}
V_\mu) + zC_\mu  \right) + \tilde{\epsilon} \epsilon (z-rH)\right] \phi_{(r,z)}, \\
&[\delta_{\tilde{\epsilon}}, \delta_{\tilde{\epsilon}}] \phi_{(r,z)} = 0 ,\qquad [\delta_{\epsilon}, \delta_{\epsilon}] \phi_{(r,z)} = 0,
\end{aligned}
\end{equation}
where \(z\) is the charge associated to the action of the central charge \(Z\) in (\ref{eq:susy-algebra-3d}). For the 3-sphere of radius \(l=1\), this is simplified to
\begin{equation}
\begin{aligned}
&[\delta_{\tilde{\epsilon}}, \delta_{\epsilon}] \phi_{(r,z)} = -2i \left[ \mathcal{L}_v + \tilde{\epsilon} \epsilon (z+ ir)\right] \phi_{(r,z)} ,  \\
&[\delta_{\tilde{\epsilon}}, \delta_{\tilde{\epsilon}}] \phi_{(r,z)} = 0 , \qquad [\delta_{\epsilon}, \delta_{\epsilon}] \phi_{(r,z)} = 0 .
\end{aligned}
\end{equation}

We report the resulting supersymmetry variation for the 3-dimensional \(\mathcal{N}=2\) vector multiplet \((A_\mu, \sigma\), \( \lambda_a, \tilde{\lambda}_a, D)\), with respect to two Killing spinors \(\tilde{\epsilon},\epsilon\). This multiplet is uncharged under the action of R-symmetry and of the central charge \(Z\). Following conventions of \cite{Marino-locCS} and \cite{itamar-wilson_loop_3dCS},
\begin{equation}
\begin{aligned}
&\delta A_\mu = \frac{i}{2} (\tilde{\epsilon} \gamma_\mu \lambda - \tilde{\lambda} \gamma_\mu \epsilon) \\
&\delta \sigma = \frac{1}{2}( \tilde{\epsilon}\lambda - \tilde{\lambda}\epsilon ) \\
&\delta \lambda =  \left( - \frac{1}{2} F_{\mu\nu} \gamma^{\mu\nu} - D + i(D_\mu \sigma)\gamma^\mu + \frac{2i}{3} \sigma \gamma^\mu D_\mu \right) \epsilon \\
&\delta \tilde{\lambda} = \left( - \frac{1}{2} F_{\mu\nu} \gamma^{\mu\nu} + D - i(D_\mu \sigma)\gamma^\mu - \frac{2i}{3} \sigma \gamma^\mu D_\mu \right) \tilde{\epsilon} \\
&\delta D = -\frac{i}{2} \left( \tilde{\epsilon} \gamma^\mu D_\mu \lambda - (D_\mu \tilde{\lambda}) \gamma^\mu \epsilon \right) + \frac{i}{2}\left( [\tilde{\epsilon}\lambda, \sigma] - [\tilde{\lambda}\epsilon, \sigma] \right) - \frac{i}{6}\left( \tilde{\lambda} \gamma^\mu D_\mu \epsilon + (D_\mu \tilde{\epsilon}) \gamma^\mu \lambda \right)
\end{aligned}
\end{equation}
where now \(D_\mu = \nabla_\mu - iA_\mu\) is the gauge-covariant derivative. On the 3-sphere, the actions in \eqref{eq:3dSYM} and \eqref{eq:3dCS} acquire a factor \(\sqrt{g}\) in the measure,\footnote{The pure Chern-Simons term \(\left( A\wedge dA + \frac{2i}{3}A^3\right)\) is actually unmodified, being already a 3-form.} and the Super Yang-Mills Lagrangian gets modified to 
\begin{equation}
\mathcal{L}_{YM} = \mathrm{Tr}\left\lbrace \frac{i}{2}\tilde{\lambda} \gamma^\mu D_\mu \lambda + \frac{1}{4} F_{\mu\nu}F^{\mu\nu} + \frac{1}{2} D_\mu \sigma D^\mu \sigma + \frac{i}{2} \tilde{\lambda}[\sigma, \lambda] + \frac{1}{2}\left( D + \frac{\sigma}{l} \right)^2 - \frac{1}{4l}\tilde{\lambda}\lambda \right\rbrace \\
\end{equation} 
where we reinstated the radius \(l\), to see that indeed in the limit \(l\to \infty\) this becomes the standard Euclidean SYM theory in 3 dimensions. We note an important feature of this Lagrangian, that will be important for the application of the localization principle: this can be written as a supersymmetry variation, \textit{i.e.}
\begin{equation}
\tilde{\epsilon}\epsilon \mathcal{L}_{YM} = \delta_{\tilde{\epsilon}} \delta_{\epsilon} \mathrm{Tr}\left\lbrace \frac{1}{2}\tilde{\lambda}\lambda - 2D\sigma \right\rbrace .
\end{equation}
The SCS Lagrangian does not get modified on curved space, since the term depending on the gauge field is topological, and the other ones do not contain derivatives.

It is important to remark that, in general, unbroken supersymmetry is consistent only with Anti-de Sitter geometry (or, in Euclidean signature, hyperbolic geometry) \cite{festuccia-susy-curved}. An exception to this is given by those theories that possess a larger group of symmetries, the \emph{superconformal} group. This is an extension of the super-Poincaré group, to include also conformal transformations of the metric. In this case, supersymmetry can be consistent also on conformally flat backgrounds with positive scalar curvature, of which the \(n\)-spheres are an example. The \(\mathcal{N}=2\) SCS theory of above is an example of superconformal theory.

\subsection[N=4,2,2* gauge theories on the round 4-sphere]{\(\mathcal{N}=4,2,2^*\) gauge theories on the round 4-sphere}
\label{subsec:4dN4-curved}

We continue also the example of the \(\mathcal{N}=4\) 4-dimensional theory, understanding how it can be realized on a different background compatible with \(\mathbb{S}^4\), and what part of the supersymmetry algebra can be preserved on this background. As in Section \ref{subsec:susy-4dN4}, the \(\mathcal{N}=2\) and \(\mathcal{N}=2^*\) cases follow from modifications of the \(\mathcal{N}=4\) theory.

Using stereographic coordinates \(x^1,\cdots,x^4\), such that the North pole is located at \(x^\mu = 0\), the round metric of the 4-sphere of radius \(r\) looks explicitly as a conformal transformation of the flat Euclidean metric,
\begin{equation}
\label{pestun-metric}
g_{\mu\nu}^{(x)} = \delta_{\mu\nu}e^{2\Omega(x)}, \quad \mathrm{where} \ e^{2\Omega(x)} = \frac{1}{\left(1 + \frac{x^2}{4r^2} \right)^2}
\end{equation}
where \(x^2 = \sum_{\mu=1}^4 (x^\mu)^2\). As remarked at the end of the last section, the conformal flatness of \(\mathbb{S}^4\) allows us to deform the superconformal YM theory on it, provided we preserve the conformal symmetry. In order to do this, we modify the kinetic term of the scalars \((\Phi_A)_{A=5,\cdots,9,0}\) adding a conformal coupling to the curvature: \[
(\partial \Phi_A)^2 \to (\partial \Phi_A)^2 + \frac{R}{6}(\Phi_A)^2  \]
where \(R=\frac{12}{r^2}\) is the scalar curvature of the metric \(g\).\footnote{In \(d\)-dimensions, the conformal coupling to the curvature scalar is made adding a term \(\xi R(\Phi)^2\) for the scalar field of canonical mass dimension \(\frac{d-2}{2}\), with \(\xi=(d-2)/4(d-1)\) (see \cite{wald-book}, Appendix D). The scalar curvature of the \(d\)-sphere of radius \(r\) is \(R=d(d-1)/r^2\).} This ensures conformal invariance of the action on the 4-sphere,
\begin{equation}
\label{pestun-4}
S^{\mathcal{N}=4}_{\mathbb{S}^4} = \int_{\mathbb{S}^4} d^4x\sqrt{g}\ \frac{1}{g_{YM}^2} \mathrm{Tr}\left(\frac{1}{2} F_{MN}F^{MN} - \Psi\Gamma^M D_M \Psi + \frac{2}{r^2} \Phi_A \Phi^A \right)
\end{equation}
where the derivatives have been promoted to covariant derivatives with respect to the Levi-Civita connection of \(g\).

Now we have to understand which supersymmetries of the \(\mathcal{N}=4\) algebra can be preserved on the new curved background. From theorem \ref{thm:superisometries}, we know that a necessary condition for a section \(\epsilon\) of the Majorana-Weyl spinor bundle on \(\mathbb{S}^4\), to produce a superisometry for the new background, is that it satisfies the \emph{twistor spinor equation}, or \emph{conformal Killing equation}
\begin{equation}
\label{pestun-5}
\nabla_\mu \epsilon = \tilde{\Gamma}_\mu \tilde{\epsilon}
\end{equation}
for some other section \(\tilde{\epsilon}\). It can be checked that, to ensure supersymmetry of \eqref{pestun-4}, \(\tilde{\epsilon}\) must be also a twistor spinor satisfying
\begin{equation}
\label{pestun-6}
\nabla_\mu \tilde{\epsilon} = -\frac{1}{4r^2} \Gamma_\mu \epsilon ,
\end{equation}
and the variations \eqref{pestun-1} have to be modified as the superconformal transformations
\begin{equation}
\begin{aligned}
\delta_\epsilon A_M &= \epsilon \Gamma_M \Psi \\
\delta_\epsilon \Psi &= \frac{1}{2}\Gamma^{MN}F_{MN}\epsilon + \frac{1}{2} \Gamma^{\mu A}\Psi_A \nabla_\mu \epsilon .
\end{aligned}
\end{equation}
Since \(\mathbb{S}^4\) is conformally flat, the number of solutions to \eqref{pestun-5} is maximal and equal to \(2\dim{(S^\pm)}=32\) \cite{baum-killing-sp}, so the whole \(\mathcal{N}=4\) superconformal algebra is preserved.\footnote{The number of generators of the \(\mathcal{N}=4\) super-Euclidean algebra is \(\dim{(S^\pm)}=16\). The other 16 are the remaining generators of the superconformal algebra.} If one restricts the attention to the \(\mathcal{N}=2\) subalgebra, then half of the generators are preserved. If instead the \(\mathcal{N}=2^*\) theory is considered, the conformal symmetry is broken and only 8 supercharges are preserved on \(\mathbb{S}^4\). With the above modifications, the \(\mathcal{N}=4\) superconformal algebra closes again only on-shell: imposing the EoM for \(\Psi\), one gets (see Appendix of \cite{pestun-article} for the details of the computation)
\begin{equation}
\delta_\epsilon^2 = - \mathcal{L}_v - G_{\Phi} - R - \Omega
\end{equation}
as \eqref{pestun-7}.

In order to prepare the ground for the exploitation of the localization principle on supersymmetric gauge theories, we remark that, if we want to define correctly an equivariant structure with respect to (at least a \(U(1)\) subgroup of) the Poincaré group, we need at least an \(\mathcal{N}=1\) supersymmetry subalgebra to close properly (\textit{i.e.}\ off-shell). If this is the case, we can use the corresponding variation \(\delta_\epsilon\) as a Cartan differential with respect to this equivariant cohomology (we are going to justify better this in the next chapter). It is not possible to close off-shell the full \(\mathcal{N}=2\) algebra on the hypermultiplet, but fixing a conformal Killing spinor \(\epsilon\) satisfying \eqref{pestun-5} and \eqref{pestun-6}, it is possible to close the subalgebra generated by \(\delta_\epsilon\) only. To do this, one has to add auxiliary fields to match the number of off-shell bosonic/fermionic degrees of freedom of the theory \cite{berkovits-10dSYM-offshell}, analogously to what happens for example to the \(\mathcal{N}=1\) vector multiplet in 4-dimensions. In 10-dimensions, we have 16 real fermionic components, and \((10-1)\) real physical bosonic components, so we have to add 7 bosonic (scalar) fields \((K_i)_{i=1,\cdots,7}\). The modified action 
\begin{equation}
S^{\mathcal{N}=4}_{\mathbb{S}^4} = \int_{\mathbb{S}^4} d^4x\sqrt{g}\ \frac{1}{g_{YM}^2} \mathrm{Tr}\left( F_{MN}F^{MN} - \Psi\Gamma^M D_M \Psi + \frac{2}{r^2} \Phi_A \Phi^A - \sum_{i=1}^7 K_i K_i\right)
\end{equation}
is supersymmetric under the modified \(\mathcal{N}=4\) superconformal transformations
\begin{equation}
\label{pestun-susyvarfinal}
\begin{aligned}
\delta_\epsilon A_M &= \epsilon \Gamma_M \Psi \\
\delta_\epsilon \Psi &= \frac{1}{2}\Gamma^{MN}F_{MN}\epsilon + \frac{1}{2} \Gamma^{\mu A}\Psi_A \nabla_\mu \epsilon + \sum_{i=1}^7 K_i \nu_i \\
\delta_\epsilon K_i &= - \nu_i \Gamma^M D_M \Psi .
\end{aligned}
\end{equation}
Here \(\epsilon\) is a fixed conformal Killing spinor, and \((\nu_i)\) are seven spinors satisfying
\begin{equation}
\label{pestun-auxiliary-cond}
\begin{aligned}
&\epsilon \Gamma^M \nu_i = 0 \\
&(\epsilon \Gamma_M \epsilon)\tilde{\Gamma}^{M}_{ab} = 2\left( \sum_i (\nu_i)_a (\nu_i)_b + \epsilon_a \epsilon_b \right) \\
&\nu_i\Gamma^M \nu_j = \delta_{ij} \epsilon \Gamma^M\epsilon .
\end{aligned}
\end{equation}
To ensure convergence of the path integral, as we did with the scalar field \(\Phi_0\), we path integrate the new auxiliary scalars on purely immaginary values, \textit{i.e.}\ \(K_j =: iK_j^E\) with \(K_j^E\) real. For every fixed non-zero \(\epsilon\), there exist seven linearly independent \(\nu_i\) satisfying these constraints, up to an \(SO(7)\) internal rotation, ensuring the closure \eqref{pestun-7} off-shell. Although, if we want \(\delta_\epsilon\) to describe the equivariant cohomology of a subgroup of the Poincaré group (not the conformal one), we should restrict to those \(\epsilon\) that generates only translations and R-symmetries at most (up to unphysical gauge transformations). Thus the dilatation term in \eqref{pestun-7} must vanish, imposing the condition \((\epsilon\tilde{\epsilon})=0\) on the conformal Killing spinors.

\bigskip

We  describe now which modifications to the above discussion have to be made in order to describe the \(\mathcal{N}=2\) and \(\mathcal{N}=2^*\) theories. The pure \(\mathcal{N}=2\) is classically obtained by restricting to the \(\mathcal{N}=2\) supersymmetry algebra generated by \eqref{pestun-N2spinor} and putting all the fields of the hypermultiplet to zero. At quantum level, this theory breaks in general the conformal invariance, so it is equivalent to consider the \(\mathcal{N}=2^*\) with hypermultiplet masses introduced as at the end of Section \ref{subsec:susy-4dN4}, by \[ 
D_0\mapsto D_0 + M\]
where \(M\) is an \(SU(2)_R^\mathcal{R}\) mass matrix. The mass terms for the fermions break the \(SO(1,1)^{\mathcal{R}}\) R-symmetry, so we must restrict the superconformal algebra further to those \(\epsilon\) for which the corresponding piece of the R-symmetry in \eqref{pestun-7} vanish. This imposes \((\tilde{\epsilon}\Gamma^{09}\epsilon)=0\). Moreover, this deformed theory is not invariant under the \(\mathcal{N}=2\) supersymmetry, because of the non-triviality of the conformal Killing spinor. In fact, using the conformal Killing equations it results that
\begin{equation}
\label{pestun-9}
\delta_\epsilon \left( \frac{1}{2} F_{MN}F^{MN} - \Psi\Gamma^M D_M \Psi + \frac{2}{r^2} \Phi_A \Phi^A \right) = -4\Psi \Gamma^i \tilde{\Gamma}^0\tilde{\epsilon}M_i^j\Phi_j
\end{equation}
where \(i,j=5,\cdots,8\), up to a total derivative. If \(\epsilon,\tilde{\epsilon}\) are restricted to the the \(+1\) eigenspace of \(\Gamma^{5678}\), we write \(\tilde{\epsilon} = \frac{1}{2r} \Lambda \epsilon\), where \(\Lambda\) is a generator of \(SU(2)_L^{\mathcal{R}}\). Explicitly \(\Lambda = \frac{1}{4}\Gamma^{ij}R_{ij}\), with components \((R_{ij})\) normalized such that \(R_{ij}R^{ij}=4\).  Then, after some gamma matrix technology, \eqref{pestun-9} gives
\begin{equation}
\delta_\epsilon (\cdots) = \frac{1}{2r} (\Psi\Gamma_i\epsilon)R^{ik} M_k^j \Phi_j = \frac{1}{2r} (\delta_\epsilon \Phi_i)R^{ik} M_k^j \Phi_j .
\end{equation}
Hence, we can modify further the mass-deformed action to get invariance with respect to this subalgebra of the original superconformal algebra on \(\mathbb{S}^4\), adding the new term
\begin{equation}
-\frac{1}{4r}(R^{ki}M_k^j)\Phi_i\Phi_j .
\end{equation}
Finally, the action
\begin{equation}
\begin{aligned}
S^{\mathcal{N}=2^*}_{\mathbb{S}^4} = \int_{\mathbb{S}^4} d^4x\sqrt{g}\ \frac{1}{g_{YM}^2} \mathrm{Tr} \Biggl( F_{MN}F^{MN} - \Psi\Gamma^M D_M \Psi & + \frac{2}{r^2} \Phi_A \Phi^A - \\
&  - \frac{1}{4r}(R^{ki}M_k^j)\Phi_i\Phi_j - \sum_{i=1}^7 K_i K_i\Biggr)
\end{aligned}
\end{equation}
where \( D_0\Phi_i \mapsto [\Phi_0,\Phi_i] + M_i^j\Phi_j\) and \(D_0\Psi \mapsto [\Phi_0,\Psi] + \frac{1}{4}M_{ij}\Gamma^{ij}\Psi\), is invariant under the subalgebra generated by a fixed conformal Killing spinor satisfying the conditions
\begin{equation}
\label{pestun-cKseq-restricted}
\Gamma^{5678}\epsilon=\epsilon ,\qquad \nabla_\mu \epsilon = \frac{1}{8r} \Gamma_\mu \Gamma^{ij}R_{ij} \epsilon .
\end{equation}

\subsection{Trial and error method}
Another method that was extensively used in the physics literature to promote a supersymmetric theory on curved spaces is based on a trial and error procedure \cite{cremonesi}. Suppose to have a supersymmetric QFT formulated in terms of component fields on flat Minkowski (or Euclidean) space \(\mathbb{R}^d\), specified by the Lagrangian density \(\mathcal{L}^{(0)}\) invariant under the supersymmetry variation \(\delta^{(0)}\). The starting point of this approach is to simply \lq\lq covariantize\rq\rq\ the original theory, replacing the flat metric \(\eta\) to the desired metric \(g\) defined on \(\mathcal{M}\) and every derivative \(\partial_\mu\) with the appropriate Levi-Civita or spin covariant derivative \(\nabla_\mu\) corresponding to \(g\). The problem is that in general this does define the theory on the curved space, but it is not guaranteed that the supersymmetry survives:
\begin{equation}
\left[\delta^{(0)}\mathcal{L}^{(0)}\right]_{(\eta, d)\to(g,\nabla)} \neq \nabla_\mu (\cdots)^\mu .
\end{equation}
The idea then is to correct the action of the supersymmetry  and the Lagrangian with an expansion in powers of \(1/r\), where \(r\) is a characteristic length of \(\mathcal{M}\),\footnote{Being \(\mathcal{M}\) compact, we can take it as an embedding in \(\mathbb{R}^n\) for some \(n\), and scale the metric according to some characteristic length \(r\).}
\begin{equation}
\begin{aligned}
\delta &= \delta^{(0)} + \sum_{i\geq 1} \frac{1}{r^i} \delta^{(i)} \\
\mathcal{L} &= \mathcal{L}^{(0)} + \sum_{i\geq 1} \frac{1}{r^i} \mathcal{L}^{(i)}
\end{aligned}
\end{equation}
requiring order by order the symmetry of the Lagrangian \emph{and} the closure of the super-algebra. This \lq\lq trial and error\rq\rq\  terminates if one is able to ensure both conditions at some finite order in \(1/r\), even though a priori the series contains an infinite number of terms. This procedure has the quality to be simple and operational in principle, but can be very cumbersome in practice to apply.

\section{BRST cohomology and equivariant cohomology}
\label{sec:BRST-cohom-equiv}

In gauge theories, BRST cohomology is a useful device to provide an algebraic description of the path integral quantization procedure, and the renormalizability of non-Abelian Yang-Mills theory in 4 dimensions. This formalism makes use of Lie algebra cohomology, while the BRST model of Section \ref{sec:BRST-model} describes equivariant cohomology, that is what we use in topological or supersymmetric field theories. It is natural to ask whether there is a relation between these two cohomology theories, and in fact there is. It turns out that equivariant cohomology of a Lie algebra \(\mathfrak{g}\) is the same as a \lq\lq supersymmetrized\rq\rq\ Lie algebra cohomology of a corresponding graded Lie algebra \(\mathfrak{g[\epsilon]}:=\mathfrak{g}\otimes \bigwedge \epsilon\) \cite{Cordes-2dYMTFT}.

Let us first see how the Weil model
\begin{equation}
W(\mathfrak{g}) = S(\mathfrak{g}^*)\otimes \bigwedge(\mathfrak{g}^*)
\end{equation} 
for the equivariant cohomology of \(\mathfrak{g}\) can be seen in more supergeometric terms. Notice that the space \(S(\mathfrak{g}^*)\) may be identified with the (commutative) algebra of functions on the Lie algebra \(\mathfrak{g}\), and thus us we can see the Weil algebra \(W(\mathfrak{g}^*)\) as the space of functions on a supermanifold built from the tangent bundle of \(\mathfrak{g}\), that is exactly the odd tangent bundle \(\Pi T\mathfrak{g}\equiv \Pi \mathfrak{g}\).\footnote{Notice that since \(\mathfrak{g}\) is a vector space, \(T^*\mathfrak{g} \cong \mathfrak{g}^*\).} Denoting \(\lbrace\tilde{c}^i\rbrace\) and \(\lbrace c^i\rbrace\) the generators of \(S(\mathfrak{g}^*)\) and \(\bigwedge(\mathfrak{g}^*)\) respectively, indeed a function on this superspace is trivialized as
\begin{equation}
\Phi = \Phi^{(0)}(\tilde{c})  + \Phi_j^{(1)}(\tilde{c})c^j + \Phi_{jk}^{(2)}(\tilde{c})c^j c^k + \cdots
\end{equation}
Introducing generators \(\lbrace\tilde{b}_i\rbrace\) and \(\lbrace b_i\rbrace\) of \(\mathfrak{g}[1]\) and \(\mathfrak{g}\), such that\footnote{Sometimes the action of this generators is denoted as a graded bracket structure, like \([b_i,c^j]=[\tilde{b}_i,\tilde{c}^j]_+ = \delta^j_i\).}
\begin{equation}
\label{brst-eq-4}
b_i(c^j) := c^j(b_i) = \delta^j_i ,\qquad \tilde{b}_i(\tilde{c}^j) := c^j(b_i) = \delta^j_i ,
\end{equation}
the Weil differential (\ref{eq:weil-differential}) can be written as
\begin{equation}
\label{brst-eq-1}
d_W = \tilde{c}^i b_i + f^i_{jk} c^j\tilde{c}^k b_i - \frac{1}{2}f^i_{jk} c^j c^k \tilde{b}_i ,
\end{equation}
that is very reminiscent of the form of a \lq\lq BRST operator\rq\rq.

In Lie algebra cohomology, the Chevalley-Eilenberg differential on \(\bigwedge(\mathfrak{g}^*)\) is defined on 1-forms \(\alpha \in \mathfrak{g}^*\) as 
\begin{equation}
\delta\alpha = - \alpha\left([\cdot ,\cdot] \right) = \alpha_i \left( -\frac{1}{2} f^i_{jk} c^j c^k \right) ,
\end{equation}
and then extended as an antiderivation on the whole complex. If we consider a \(\mathfrak{g}\)-module \(V\), such as the target space of a given field theory of gauge group \(G\), with a representation \(\rho : \mathfrak{g}\to \mathrm{End}( V)\),\footnote{If \(V\) is a field space \(C^\infty(\mathcal{M})\) over some (super)manifold \(\mathcal{M}\), \(\mathfrak{g}\) acts as usual as a Lie derivative with respect to the fundamental vector field, \(\rho(X) = \mathcal{L}_{X}\).} then the CE differential is extended to \(\bigwedge(\mathfrak{g}^*) \otimes V\) as
\begin{equation}
\begin{aligned}
&\delta v(X) := \rho(X)v \qquad \forall v\in V, X\in \mathfrak{g} \\
&\delta (\alpha \otimes v) = \delta \alpha \otimes v + (-1)^{k} \alpha \otimes \delta v \qquad \forall v\in V, \alpha\in \bigwedge\nolimits^{\!k}(\mathfrak{g}^*).
\end{aligned}
\end{equation}
This, expressed with respect to a basis \(\lbrace b_i\rbrace\) of \(\mathfrak{g}\), coincide with the BRST operator
\begin{equation}
\label{brst-eq-2}
\delta  = c^i \rho(b_i) - \frac{1}{2}f^i_{jk} c^j c^k b_i
\end{equation}
that satisfies \(\delta ^2=0\). The \(c^i\) are \emph{ghosts}, while the \(b_i\) are \emph{anti-ghosts}. The zero-th cohomology group of the complex \(\bigwedge(\mathfrak{g}^*) \otimes V\) with respect to the differential (\ref{brst-eq-2}) contains those states that have \emph{ghost number} 0 and are \(\mathfrak{g}\)-invariant,
\begin{equation}
H^0(\mathfrak{g},V) \cong V^{\mathfrak{g}}
\end{equation}
so the interesting \lq\lq physical\rq\rq\ states.

We see that there is a difference between the BRST operator (\ref{brst-eq-2}) and the Weil differential (\ref{brst-eq-1}), but we can connect these differentials as follows. To the Lie algebra \(\mathfrak{g}\) we can associate the \emph{differential graded Lie algebra} \(\mathfrak{g}[\epsilon] := \mathfrak{g}\otimes \bigwedge \epsilon\). Here \(\epsilon\) is a single generator taken in odd degree, \(\mathrm{deg}(\epsilon) := -1\), while \(\mathrm{deg}(\mathfrak{g}):=0\). A differential \(\partial: \mathfrak{g}[\epsilon] \to \mathfrak{g}[\epsilon]\) is defined as
\begin{equation}
\partial \epsilon := 1\in \mathfrak{g}, \qquad \qquad \partial X := 0 \quad \forall X\in \mathfrak{g} .
\end{equation}
This superalgebra has generators \(b_i := b_i \otimes 1\) and \(\tilde{b}_i := b_i \otimes \epsilon\), and the superbracket structure coming from the Lie brackets on \(\mathfrak{g}\) and the (trivial) wedge product on \(\bigwedge\epsilon\):
\begin{equation}
\label{brst-eq-3}
\begin{aligned}
&[b_i,b_j] = f^k_{ij} b_k \\
&[b_i, \tilde{b}_j] = f^k_{ij} \tilde{b}_k \\
&[\tilde{b}_i,\tilde{b}_j] = 0
\end{aligned}
\end{equation}
making it into a Lie superalgebra. The differential on the generators is rewritten as
\begin{equation}
\partial b_i = 0, \qquad \quad \partial \tilde{b}_i = b_i .
\end{equation} 
To this \lq\lq supersymmetrized\rq\rq\ algebra we can associate the Lie algebra cohomology with respect to the complex \(\bigwedge(\mathfrak{g}[\epsilon]^*)\), that is generated by \(\lbrace c^i, \tilde{c}^i\rbrace\) of degrees \(\mathrm{deg}(c^i)=1, \mathrm{deg}(\tilde{c}^i)=2\), such that
\begin{equation}
c^i(b_j) = \tilde{c}^i(\tilde{b}_j) = \delta^i_j .
\end{equation}
The BRST differential for the \(\mathfrak{g}[\epsilon]\)-Lie algebra cohomology is naturally defined analogously to before as
\begin{equation}
Q\alpha := -\alpha\left( [\cdot,\cdot] \right)
\end{equation}
on 1-forms \(\alpha\in \mathfrak{g}[\epsilon]^*\). But now, because of the Lie algebra extension (\ref{brst-eq-3}), its expression in terms of the generators results
\begin{equation}
Q = -\frac{1}{2} f^k_{ij} c^i c^j b_k + f^k_{ij} c^i \tilde{c}^j \tilde{b}_k .
\end{equation}
Moreover, the dual \(\partial^*\) acts as
\begin{equation}
\partial^* = \tilde{c}^i b_i
\end{equation}
with \(b_i\) acting as in \eqref{brst-eq-4}. Then the total differential on this complex coincides with the Weil differential,
\begin{equation}
d_W \cong \partial^* + Q
\end{equation}
and we identify an isomorphism of dg algebras
\begin{equation}
\left(\bigwedge(\mathfrak{g}[\epsilon]^*), \partial^* + Q\right) \cong \left( W(\mathfrak{g}), d_W \right).
\end{equation}

If we bring into the game the \(\mathfrak{g}\)-module \(V\) as before, we can work with \(\Omega(V)\) as a \(\mathfrak{g}[\epsilon]\)-dg algebra, defining the \(\mathfrak{g}[\epsilon]\) action as\footnote{Again, if \(V=C^\infty(\mathcal{M})\), then \(\Omega(\mathcal{M})\) is the space we considered when we constructed the equivariant cohomology of a \(G\)-manifold.}
\begin{equation}
(X\otimes 1) \to \mathcal{L}_X , \qquad (X\otimes \epsilon) \to \iota_X .
\end{equation} 
On the complex \(\bigwedge(\mathfrak{g}[\epsilon]^*) \otimes \Omega(V)\), the total differential inherited from the Weil differential and the \(\mathfrak{g}[\epsilon]\)-action is
\begin{equation}
d_B := c^k \otimes \mathcal{L}_k + \tilde{c}^k \otimes \iota_k + Q \otimes 1 + \partial^* \otimes 1 + 1 \otimes d
\end{equation}
and it coincides with the one of the BRST model of equivariant cohomology \eqref{eq:BRST-diff}!

This demonstrates how the BRST quantization formalism and equivariant cohomology are intimately related, and suggests that indeed BRST symmetry operators are good candidates to represent equivariant differentials in QFT, that can be used to employ the localization principle in these kind of physical systems.

\chapter{Localization for circle actions in supersymmetric QFT}
\label{cha:loc-susy}

In this chapter we describe how the equivariant localization principle can be carried out in the infinite dimensional case of path integrals in QM or QFT. In this setting, the first object of interest is the \emph{partition function}
\begin{equation}
\label{eq:path-integral}
Z = \int_{\mathcal{F}} D\phi \ e^{iS[\phi]}
\end{equation}
where \(\mathcal{F}=\Gamma(M,E)\) is the space of fields, \textit{i.e.} sections of some fiber bundle \(E\to M\) with typical fiber (the \emph{target space}) \(V\) over the (Lorentzian) \(n\)-dimensional spacetime \(M\), and \(S\in C^\infty(\mathcal{F})\) is the action functional.\footnote{In the (common) case of a trivial bundle, this is equivalent to considering \(\mathcal{F}=C^\infty(M,V)\), \textit{i.e.}\ \(V\)-valued functions over \(M\). Often \(V\) is a vector space, otherwise the theory describes a so-called \emph{non-linear \(\sigma\)-model}. In supersymmetric field theories, \(V=\bigwedge(S^*)\) for some vector space \(S\), and the field space acquires a natural graded structure.} The fields are supposed to satisfy some prescribed boundary conditions on \(\partial M\). In the Riemannian case, the corresponding object has the form
\begin{equation}
\label{eq:path-integral-euclidean}
Z = \int_{\mathcal{F}} D\phi\ e^{-S_E[\phi]} 
\end{equation}
where we denoted \(S_E\) the Euclidean action. If the spacetime has the form \(M=\mathbb{R}_t \times \Sigma\), this last expression can be reached from the Lorentzian theory via \emph{Wick rotation} of the time direction \(t \mapsto \tau:=it\). If the \(\tau\) direction is compactified to a circle of length \(T\), we can interpret the Euclidean path integral as the canonical ensemble partition function describing the original QFT at a finite temperature \(1/T\). If needed, we are always free to set the length \(T\) of the circle to be very large, and find the zero temperature limit when \(T\to\infty\). Given an observable \(\mathcal{O}\in C^\infty(M)\), its \emph{expectation value} is given by
\begin{equation}
\label{eq:expectation-value}
\left\langle \mathcal{O}\right\rangle = \frac{1}{Z}\int_\mathcal{F} D\phi \ \mathcal{O}[\phi]e^{iS[\phi]} \quad \mathrm{or} \quad \left\langle \mathcal{O}\right\rangle_E = \frac{1}{Z}\int_\mathcal{F} D\phi \ \mathcal{O}[\phi]e^{-S_E[\phi]} .
\end{equation}

The path integral measure \(D\phi\) on the infinite dimensional space \(\mathcal{F}\) is not rigorously defined,\footnote{In fact, it does not exist in general.} but it is usually introduced as 
\begin{equation}
\int D\phi = \mathcal{N} \prod_{x\in M} \int_V d\phi(x)
\end{equation}
where \(\mathcal{N}\) is some (possibly infinite) multiplicative factor, and at every point \(x\in M\) we have a standard integral over the fiber \(V\). Notice that the infinite factors \(\mathcal{N}\) cancel in ratios in the computations of expectation values, so we can still make sense of such objects and formally manipulate them to obtain physical information. Another convergence issue comes with the prescription of boundary conditions in computing the action \(S[\phi]\). If \(M\) is non compact, this often requires a specific regularization,\footnote{For example, one can first assume that the spacetime just extends up to some large but finite typical lenght \(r_0\), and then send this value to infinity at the end of the calculations.} while taking compact spacetimes ensure better convergence properties.

Very few quantum systems have an exactly solvable path integral. When this functional integral method was introduced, the only examples where \eqref{eq:path-integral} could be directly evaluated were the free particle and the harmonic oscillator.\footnote{Later on, ad hoc methods for other particular systems were developed, like the solution of the Hydrogen atom by Duru and Kleinert \cite{duru-Hatom-path-integral}, and others.} Both these theories are quadratic in the fields and their derivatives, thus the partition function can be computed using the formal functional analog of the classical Gaussian integration formula
\begin{equation}
\int_{-\infty}^{\infty}  d^n x \ e^{-\frac{1}{2} x^k M_{kl} x^l + A_k x^k} = \frac{(2\pi)^{n/2}}{\sqrt{\det{M}}} e^{\frac{1}{2}A_i (M^{-1})^{ij}A_j}
\end{equation}
where \(M=[M_{ij}]\) is an \(n\times n\) non singular matrix. The analogue in field theory has \(n\to\infty\) and a functional determinant at the denominator, that must be properly regularized in order to be a meaningful convergent quantity (see any standard QFT book, like \cite{peskin-schroeder, srednicki}). 

In perturbative QFT, one almost never has to explicitly perform such a functional integration. Suppose that the action has the generic form 
\begin{equation}
S = S_0 + S_{int}
\end{equation}
where \(S_0\) is the free term containing up to quadratic powers of the fields and their derivatives, and the rest is collected in \(S_{int}\). Then the expectation value of an observable \(\mathcal{O}\), expressible as a combination of local fields, is computed expanding the exponential of the interacting part in Taylor series, and exploiting \emph{Wick's theorem}\footnote{Again, see any standard QFT book.} for the vacuum expectation values in the free theory:
\begin{equation}
\left\langle \mathcal{O}\right\rangle = \sum_{k} \frac{1}{k!} \left\langle (iS_{int})^k \mathcal{O}\right\rangle_0 .
\end{equation}

Another perturbative approach, especially useful to compute the effective action in a given theory, is the so-called \emph{background field method}, where the action \(S\) is expanded around a classical \lq\lq background\rq\rq element \(\overline{\phi}\in \mathcal{F}\),
\begin{equation}
S[\overline{\phi}+\eta] = S[\overline{\phi}]+ 
\int_M d^nx \left( \frac{\delta S}{\delta \phi(x)}\right)_{\overline{\phi}} \eta(x) + 
\frac{1}{2} \int_M d^nx d^ny \left(\frac{\delta^{(2)}S}{\delta \phi(x) \delta \phi(y)}\right)_{\overline{\phi}} \eta(x)\eta(y)  + \cdots
\end{equation}
If we chose \(\overline{\phi}\) to be a solution of the classical equation of motion \( \left(\frac{\delta S}{\delta \phi(x)}\right)_{\overline{\phi}} =0 \), the first order term disappears from the expansion. If we also neglect the terms of order higher than quadratic in \(\eta\), and substitute the resulting expression in \eqref{eq:path-integral} or \eqref{eq:path-integral-euclidean}, we get the equivalent of the saddle point approximation, or \lq\lq one-loop approximation\rq\rq\, of the partition function
\begin{equation}
\label{eq:saddle-point-path-int}
Z \approx e^{-S[\overline{\phi}]} Z_{1-loop}[\overline{\phi}] ,
\end{equation}
where 
\begin{equation}
\begin{aligned}
Z_{1-loop}[\phi] &:= \int_{\mathcal{F}} D\eta\ e^{-\frac{1}{2} \eta \cdot \Delta[\phi] \cdot \eta} \equiv \left[ \det{\left(\Delta[\phi] \right)}\right]^{-1/2}   \\
\Delta[\phi](x,y) &:= \left(\frac{\delta^{(2)}S}{\delta \phi(x) \delta \phi(y)}\right)_{\phi} ,
\end{aligned}
\end{equation}
and we denoted convolution products over \(M\) with \((\cdot)\) for brevity.

We are interested in those cases in which such a \lq\lq semiclassical\rq\rq\ approximation of the partition function turns out to give the exact result for the path integral in the full quantum theory. This is possible if some symmetry of the field theory, \textit{i.e.}\ acting on the space \(\mathcal{F}\), allows us to formally employ the equivariant localization principle and reduce the path integration domain from \(\mathcal{F}\) to a (possibly finite-dimensional) subspace. In the next part of the chapter we will see some cases in which it is possible to interpret \(\mathcal{F}\), or a suitable extension of it, as a Cartan model with some (super)symmetry operator acting as the Cartan differential. As we already anticipated, this is possible if \(\mathcal{F}\) has a graded structure that can both arise from the supersymmetry of the underlying spacetime, or can be introduced via an extension analogous to what happens in the BRST formalism.

We will first describe the application of the Duistermaat-Heckman theorem in the case of Hamiltonian QM, where the equivariant structure can be constructed from the symplectic structure of the underlying theory. Then we will be concerned with more general applications of the localization principle in supersymmetric QFT, where the super-Poincaré group action allows for an equivariant cohomological interpretation. In both frameworks, we present examples of localization under the action of a single supersymmetry, whose \lq\lq square\rq\rq\ generates a bosonic \(U(1)\) symmetry. Supersymmetric localization was recently applied to many cases of QFT on curved spacetimes, so we must consider those curved background that preserve at least one supersymmetry, as discussed in Section \ref{sec:susy-curved}.

\section{Localization principle in Hamiltonian QM}
\label{sec:QM}

We consider now, following \cite{szabo} and refernces therein, the path integral quantization of an Hamiltonian system \((M,\omega ,H)\), of the \(2n\)-dimensional phase space \(M\) with symplectic form \(\omega\), and an Hamiltonian function \( H \in C^\infty(M)\). This is simply QM viewed as a (0+1)-dimensional QFT over the base \lq\lq spacetime\rq\rq\ \(\mathbb{R}\) or \(\mathbb{S}^1\), that now is only \lq\lq time\rq\rq, and with target space \(M\), that physically represents the phase space of the system. The fields of the theory are the paths \(\gamma: \mathbb{R} (\mathbb{S}^1) \to M\), that means we consider a trivial total space \(E=\mathbb{R}(\mathbb{S}^1)\times M\). In principle the time axis can be chosen to be the real line (or an interval with some prescribed boundary conditions), or the circle (that corresponds to periodic boundary conditions), but we will soon see that it is much convenient technically to chose the latter possibility, so consider the \lq\lq loop space\rq\rq\ \( \mathcal{F} = C^\infty(\mathbb{S}^1,M)\). A field for us is so a closed curve \(\gamma:[0,T]\to M\) such that \(\gamma(0)=\gamma(T)\). Since we make the periodicity explicit in \(t\), we interpret this parameter as an \lq\lq Euclidean\rq\rq\ Wick-rotated time, so that \(T\) is the inverse temperature of the canonical ensemble. 

We set up now some differential geometric concept on the loop space that we are going to use in the following. If \(\lbrace x^\mu\rbrace\) are coordinates on \(M\), on the loop space we can choose an infinite set of coordinates \(\lbrace x^\mu (t) \rbrace\) for \(\mu\in{1,\cdots,2n}\) and \(t\in[0,T]\), such that for any \(\gamma\in \mathcal{F}\) \[
x^\mu(t)[\gamma] := x^\mu (\gamma(t)). \]
Using the standard rules of functional derivation, a vector field in \(\Gamma(T\mathcal{F})\) can be thus expressed locally with respect to these coordinates as
\begin{equation}
X = \int_0^T dt\ X^\mu(t) \frac{\delta}{\delta x^\mu(t)}
\end{equation}
where \(X^\mu(t)\) are functions over \(\mathcal{F}\), and \((\delta/\delta x^\mu(t))_\gamma\) is a basis element of the tangent space \(T_\gamma \mathcal{F}\) at \(\gamma\). Many other geometric objects can be lifted from \(M\) to \(\mathcal{F}\) following this philosophy. For example, for any function in \(C^\infty(M)\) as the Hamiltonian \(H\), we can define \(H(t)\in C^\infty(\mathcal{F})\) at a given time \(t\), such that \(H(t)[\gamma] := H(\gamma(t))\). The action functional instead is the function over the loop space such that
\begin{equation}
\begin{aligned}
S[\gamma] &= \int_0^T dt\ \left[ \dot{q}^a(t) p_a(t) - H(p(t),q(t)) \right] \\
&= \int_0^T dt\ \left[ \theta_{\gamma(t)}(\dot{\gamma}) - H(\gamma(t)) \right]
\end{aligned}
\end{equation}
where in the first line we expressed \(\gamma\) through its trivialization in Darboux coordinates \(\lbrace q^a,p_a\rbrace\) with \(a\in\lbrace 1,\cdots,n\rbrace\), and in the second line we expressed the same thing more covariantly using the (local) symplectic potential \(\theta\) of \(\omega\) and the velocity vector field \(\dot{\gamma}\) of the curve. Concerning differential forms, if we consider the basis set \(\lbrace \eta^\mu (t) := dx^\mu(t)\rbrace\), a \(k\)-degree element of \(\Omega(\mathcal{F})\) can be expressed locally as 
\begin{equation}
\alpha = \int_0^T dt_1 \cdots \int_0^T dt_k \ \frac{1}{k!} \alpha_{\mu_1\cdots \mu_k}(t_1,\cdots,t_k) \eta^{\mu_1}(t_1)\wedge \cdots \wedge \eta^{\mu_k}(t_k)
\end{equation}
and we recall that \(\Omega(\mathcal{F})\) can be considered as the space of functions over the \emph{super-loop space} \(\Pi T \mathcal{F}\), of coordinates \(\lbrace x^\mu(t), \eta^\mu(t)\rbrace\). The de Rham differential on the loop space can be expressed as the cohomological vector field on \(\Pi T \mathcal{F}\)
\begin{equation}
d_{\mathcal{F}} = \int_0^T dt\ \eta^\mu(t) \frac{\delta}{\delta x^\mu(t)} .
\end{equation}
Finally, it is natural to lift on the loop space the symplectic structure of \(M\), as well as a choice of Riemannian metric, as 
\begin{equation}
\label{QM:metric-symplecticform}
\begin{aligned}
\Omega &= \int_0^T dt \ \frac{1}{2}\omega_{\mu\nu}(t) \eta^\mu(t)\wedge \eta^\nu(t) \\
G &= \int_0^T dt \ g_{\mu\nu}(t) \eta^\mu(t)\otimes \eta^\nu(t) ,
\end{aligned}
\end{equation}
\textit{i.e.}\ \(\Omega_{\mu\nu}(t,t'):= \omega_{\mu\nu}(t)\delta(t-t')\) and \(G_{\mu\nu}:=g_{\mu\nu}(t)\delta(t-t')\). \(\Omega\) is closed under the loop space differential \(d_{\mathcal{F}}\).\footnote{Strictly speaking, the 2-form \(\Omega\) should be called \lq\lq pre-symplectic\rq\rq, since although it is certainly closed, it is not necessarily non-degenerate on the loop space.}

Considering the standard Liouville measure on \(M\)
\begin{equation}
\frac{\omega^n}{n!} = d^{2n}x\ \mathrm{Pf}(\omega(x)) = d^nq d^np ,
\end{equation}
we can write now the path integral measure for QM on the loop space as an infinite product of the Liouville one for any time \(t\in[0,T]\), and get
\begin{equation}
\int_{\mathcal{F}} \frac{\Omega^n}{n!} = \int_{\mathcal{F}} D^{2n}x\ \mathrm{Pf}(\Omega[x]) = \int_{\Pi T\mathcal{F}} D^{2n}x D^{2n}\eta\ \frac{\Omega^n}{n!} .
\end{equation}
Here in the last equality we rewrote the integral over \(\mathcal{F}\) as an integral over the super-loop space, analogously to \eqref{eq:integral-odd-tangent-bundle}. The path integral for the quantum partition function is thus
\begin{equation}
\begin{aligned}
Z(T) &= \int_\mathcal{F} D^{2n}x\ \mathrm{Pf}(\Omega[x]) e^{-S[x]} \\
&= \int_{\Pi T\mathcal{F}} D^{2n}x D^{2n}\eta\ \frac{\Omega^n[x,\eta]}{n!}  e^{-S[x]} \\
&=  \int_{\Pi T\mathcal{F}} D^{2n}x D^{2n}\eta\ e^{-(S[x]+\Omega[x,\eta])}
\end{aligned}
\end{equation}
where in the last line we exponentiated the loop space symplectic form, making explicit the formal analogy with the Duistermaat-Heckman setup. In particular, we associated to the \lq\lq classical\rq\rq\ Hamiltonian system \((M,\omega,H)\) an Hamiltonian system \((\mathcal{F},\Omega,S)\) on the loop space. Here the loop space Hamiltonian function \(S\) generates  an Hamiltonian \(U(1)\)-action on \(\mathcal{F}\), that can be used to formally apply the same equivariant localization principle as in the finite-dimensional case. As for the proof of the ABBV formula, we introduced a graded structure on field space formally rewriting the path integration on the super-loop space \(\Pi T\mathcal{F}\). This graded structure is now simply given by the form-degree on the extended field space \(\Omega(\mathcal{F})\).

Let now \(X_S\) be the Hamiltonian vector field associated to \(S\in C^\infty(\mathcal{F})\), such that \(d_{\mathcal{F}}S = -\iota_{X_S}\Omega\), or equivalently \(X_S = \Omega(\cdot,d_{\mathcal{F}}S)\). Explicitly, in coordinates \(\lbrace x^\mu(t)\rbrace\)
\begin{equation}
\begin{split}
X_S^\mu(t) &= \int_0^T dt'\ \Omega^{\mu\nu}(t,t')\frac{\delta S}{\delta x^\nu(t)} \\
&= \omega^{\mu\nu}(x(t))\bigl( \omega_{\nu\rho}(x(t))\dot{x}^\rho(t) - \partial_\nu H(x(t))  \bigr) \\
&= \dot{x}^\mu(t) - X_H^\mu(x(t))
\end{split}
\end{equation}
where \(\dot{x}(t)\) is the vector field on \(\mathcal{F}\) with components such that \(\dot{x}^\mu(t)[\gamma] := (x^\mu \circ \gamma)'(t) \equiv \dot{\gamma}^\mu(t)\). The flow of \(X_S\) defines the Hamiltonian \(U(1)\)-action on \(\mathcal{F}\) and the infinitesimal action of the Lie algebra \(\mathfrak{u}(1)\) on \(\Omega(\mathcal{F})\) through the Lie derivative \(\mathcal{L}_{X_S}\). The Cartan model for the \(U(1)\)-equivariant cohomology of \(\mathcal{F}\) is then defined by the space of equivariant differential forms
\begin{equation}
\Omega_{S}(\mathcal{F}) := \left(\mathbb{R}[\phi]\otimes \Omega(\mathcal{F})\right)^{U(1)} \cong \Omega(\mathcal{F})^{U(1)}[\phi]
\end{equation}
and the equivariant differential
\begin{equation}
\label{QM-Cartan-differential}
\begin{aligned}
Q_S &:= \mathds{1} \otimes d_{\mathcal{F}} - \phi \otimes \iota_{X_S} \equiv d_{\mathcal{F}} +  \iota_{X_S} \\
&= \int_0^T dt \bigl( \eta^\mu(t) + \dot{x}^\mu(t) - X^\mu_H(x(t)) \bigr) \frac{\delta}{\delta x^\mu(t)} ,
\end{aligned}
\end{equation}
where as usual we localized algebraically setting \(\phi=-1\) to ease the notation.  The square of this operator gives, after some simplifications
\begin{equation}
Q_S^2 = \int_0^T dt \left( \frac{d}{dt} - \left. \mathcal{L}_{X_H}\right|_{x(t)} \right)
\end{equation}
where the second term is the Lie derivative on \(M\) with respect to \(X_H\), lifted on \(\mathcal{F}\) at every value of \(t\). The first term, when evaluated on a field, gives only contributions from the values at \(t=0,T\), and so it vanishes thanks to the fact that we chose periodic boundary conditions! This means that, with this choice, the Cartan model on field space is completely determined by the lift of the \(U(1)\)-invariant forms on \(M\), for which \(\mathcal{L}_{X_H}\alpha =0\). Consistently, if we restrict to this subspace of \(\Omega(M)\) where the energy is preserved, \(Q_S \equiv Q_{\dot{x}} = d_{\mathcal{F}} +  \iota_{\dot{x}}\) acts as the \emph{supersymmetry} operator generating time-translations on the base \(\mathbb{S}^1\):
\begin{equation}
Q_{\dot{x}}^2 = \frac{1}{2} [Q_{\dot{x}},Q_{\dot{x}}] = \int_0^T dt \frac{d}{dt}
\end{equation}
resembling the supersymmetry algebra \eqref{eq:susy-algebra-ext} for \(\mathcal{N}=1\) and 1-dimensional spacetime. We will see in the next section that this is not just a coincidence, but we can relate this model to a supersymmetric version of QM. This restricted differential acts on coordinates of the super-loop space as
\begin{equation}
Q_{\dot{x}} x^\mu(t) = \eta^\mu(t), \qquad Q_{\dot{x}} \eta^\mu(t) = \dot{x}^\mu(t),
\end{equation}
while the full equivariant differential acts as
\begin{equation}
Q_S x^\mu(t) = \eta^\mu(t), \qquad Q_S \eta^\mu(t) = X_S^\mu(t) ,
\end{equation}
both exchanging \lq\lq bosonic\rq\rq\ with \lq\lq fermionic\rq\rq\ degrees of freedom.

We remark that we started from a standard (non supersymmetric) Hamiltonian theory on \(\mathcal{F}\), and from this we constructed a supersymmetric theory on \(\Pi T\mathcal{F}\), whose supersymmetry is encoded in the Hamiltonian symmetry (so in the symplectic structure) of the original theory. This \lq\lq hidden\rq\rq\ supersymmetry is thus interpretable, in the spirit of Topological Field Theory, as a BRST symmetry, and the differential \(Q_S\) as a \lq\lq BRST charge\rq\rq\, under which the augmented action \((S+\Omega)\) is supersymmetric:
\begin{equation}
Q_S (S+\Omega) = d_{\mathcal{F}}S + d_{\mathcal{F}}\Omega + \iota_{X_S}\Omega = d_{\mathcal{F}}S + 0 - d_{\mathcal{F}}S = 0.
\end{equation}
In other words, \((S+\Omega)\) is an \emph{equivariantly closed extension} of the symplectic 2-form \(\Omega\), analogously to the finite-dimensional Hamiltonian geometry discussed in Chapter \ref{sec:e-cohom-symplectic}. The same argument cannot be straightforwardly applied to any QFT, since in general we do not have a symplectic structure on the field space,\footnote{A symplectic structure can be induced from the action principle on the subspace of solutions of the classical EoM, but it does not lift on the whole field space in general.} but in the presence of a gauge symmetry we know that a BRST supersymmetry can be used to define the equivariant cohomology on the field space and exploit the localization principle. We will expand this a bit in the next chapter.

\medskip
It is now possible to mimic the procedure of Section \ref{sec:ABBVproof} to localize the supersymmetric path integral
\begin{equation}
Z(T) = \int_{\Pi T\mathcal{F}} D^{2n}x D^{2n}\eta \ e^{-(S[x]+\Omega [x, \eta ] )}
\end{equation}
seen as an integral of an equivariantly closed form. We modify the integral introducing the \lq\lq localizing action\rq\rq\ \(S_{loc}[x,\eta] :=  Q_S \Psi[x,\eta]\), with localization 1-form \(\Psi\in \Omega^1(\mathcal{F})^{U(1)}\), the so-called \lq\lq gauge-fixing fermion\rq\rq:
\begin{equation}
\label{QM-1}
Z_T(\lambda) = \int_{\Pi T\mathcal{F}} D^{2n}x D^{2n}\eta \ e^{-(S[x]+\Omega [x, \eta ] - \lambda Q_S \Psi[x,\eta])}
\end{equation}
where \(\lambda\in \mathbb{R}\) is a parameter. The resulting integrand is again explicitly equivariantly closed, and we can check that this path integral is formally independent on the parameter \(\lambda\). Indeed, after a shift \(\lambda \mapsto \lambda + \delta \lambda\), \eqref{QM-1} becomes
\begin{equation}
Z_T(\lambda + \delta \lambda) = \int_{\Pi T\mathcal{F}} D^{2n}x D^{2n}\eta \ e^{-(S[x]+\Omega [x, \eta ] - \lambda Q_S \Psi[x,\eta] - \delta \lambda Q_S \Psi[x,\eta])} ,
\end{equation}
and we can make a change of integration variables to absorb the resulting shift at the exponential. Since the exponential is supersymmetric, we change variables using a supersymmetry transformation:
\begin{equation}
\begin{split}
&x^\mu(t) \mapsto x'^\mu(t) := x^\mu(t) + \delta \lambda \Psi Q_S x^\mu(t) = x^\mu(t) + \delta \lambda \eta^\mu(t) \\
&\eta^\mu(t) \mapsto \eta'^\mu(t) := \eta^\mu(t) + \delta \lambda \Psi Q_S \eta^\mu(t) = \eta^\mu(t) + \delta \lambda X_S^\mu(t)
\end{split}
\end{equation}
so that the only change in the integral comes from the integration measure,
\begin{equation}
\begin{aligned}
D^{2n}x D^{2n}\eta \mapsto D^{2n}x' D^{2n}\eta' 
&= \mathrm{Sdet} \left[\begin{array}{cc}
\partial x'/\partial x & \partial x'/\partial \eta \\
\partial \eta'/\partial x & \partial \eta'/\partial \eta 
\end{array} \right] D^{2n}x D^{2n}\eta \\
&= e^{-\delta\lambda Q_S \Psi} D^{2n}x D^{2n}\eta .
\end{aligned}
\end{equation}
Putting all together,
\begin{equation}
\begin{aligned}
Z_T(\lambda + \delta \lambda) &= \int_{\Pi T\mathcal{F}} D^{2n}x' D^{2n}\eta' \ e^{-(S[x']+\Omega [x', \eta' ] - \lambda Q_S \Psi[x',\eta'] - \delta \lambda Q_S \Psi[x',\eta'])} \\
&= \int_{\Pi T\mathcal{F}} D^{2n}x D^{2n}\eta \ e^{-(S[x]+\Omega [x, \eta ] - \lambda Q_S \Psi[x,\eta])} = Z_T(\lambda) .
\end{aligned}
\end{equation}
The same property can be seen less rigorously  by exploiting some sort of (arguable) infinite-dimensional version of Stokes' theorem:
\begin{equation}
\begin{aligned}
\frac{d}{d\lambda}Z_T(\lambda) &= \int_{\Pi T\mathcal{F}} D^{2n}x D^{2n}\eta\ (Q_S \Psi) e^{-(S[x]+\Omega [x, \eta ] - \lambda Q_S \Psi[x,\eta])} \\
&= \int_{\Pi T\mathcal{F}} D^{2n}x D^{2n}\eta \left(Q_S \Psi e^{-(S[x]+\Omega [x, \eta ] - \lambda Q_S \Psi[x,\eta])}\right) =0,
\end{aligned}
\end{equation}
that holds if we assume the path integration measure to be non-anomalous under \(Q_S\).

Since the path integral \eqref{QM-1} is independent on the parameter, we can take the limit \(\lambda \to \infty\) and obtain the localization formula
\begin{equation}
\label{QM-2}
Z(T) = \lim_{\lambda\to\infty} \int_{\Pi T\mathcal{F}} D^{2n}x D^{2n}\eta \ e^{-(S[x]+\Omega [x, \eta ] - \lambda Q_S \Psi[x,\eta])}
\end{equation}
that \lq\lq localizes\rq\rq\ \(Z(T)\) onto the zero locus of \(S_{loc}\). Of course different choices of gauge-fixing fermion induce different final localization formulas for the path integral, but at the end they should all give the same result. We now present two different localization formulas derived from \eqref{QM-2} with different choices of localizing term.

\medskip
The fist canonical choice of gauge fixing fermion we can make mimics the same procedure we used in the finite-dimensional case. Under the same assumptions we made in Section \ref{sec:eq_loc_princ}, we pick a \(U(1)_H\)-invariant metric \(g\) on \(M\), and lift it to \(\mathcal{F}\) using \eqref{QM:metric-symplecticform}, so that the resulting \(G\) is \(U(1)_S\)-invariant: \(\mathcal{L}_S G = Q_S^2 G = 0\). Then the localization 1-form is taken to be
\begin{equation}
\Psi[x,\eta] := G(X_S,\cdot) = \int_0^T dt\ g_{\mu\nu}(x(t))\Bigl( \dot{x}^\mu(t) - X^\mu_H(x(t))\Bigr)\eta^\nu(t) ,
\end{equation}
so that the localization locus is the subspace where \(X_S = 0\), \textit{i.e.}\ the moduli space of classical trajectories \cite{niemi-equiv-morse}:
\begin{equation}
\mathcal{F}_S = \left\lbrace \gamma\in\mathcal{F} : \left( \frac{\delta S}{\delta x^\mu(t)}\right)_\gamma =0 \right\rbrace .
\end{equation}
If this space consists of isolated, non-degenerate trajectories, we can apply the non-degenerate version of the ABBV formula for a circle action, and get
\begin{equation}
Z(T) = \sum_{\gamma\in \mathcal{F_S}} \frac{\mathrm{Pf}\left[ \omega(\gamma(t))\right]}{\sqrt{\det{[dX_S[\gamma]/2\pi]}}} e^{-S[\gamma]}
\end{equation}
where the pfaffian and the determinant are understood in the functional sense, spanning both the phase space indices \(\mu\in\lbrace 1,\cdots,2n\rbrace\) and the time continuous index \(t\in[0,T]\). In general, for non isolated classical trajectories we have to decompose any \(\gamma\in\mathcal{F}\) near to the fixed point set, splitting the classical component and normal fluctuations, as we did in Section \ref{sec:ABBVproof}. Then, rescaling the normal fluctuations by \(\sqrt{\lambda}\) and thanks to the Berezin integration rules on the super-loop space, the same argument of the finite-dimensional case leads to
\begin{equation}
\label{QM-3}
Z(T) = \int_{\mathcal{F_S}} D^{2n}x\ \frac{\mathrm{Pf}\left[\omega(x(t))\right]}{\left.\mathrm{Pf}\left[\delta^\mu_\nu \partial_t -(B+R)^\mu_\nu (x(t))\right]\right|_{N\mathcal{F}_S}} e^{-S[x]}
\end{equation}
where \(B^\mu_\nu = g^{\mu\rho} (\nabla_{[\rho}X_{H})_{\nu]}\), while \(\nabla\) and \(R\) are the connection and curvature of the metric \(g\) on \(M\), evaluated on \(\mathcal{F}\) time-wise as usual. Notice that this localization scheme makes  the contribution from the classical configurations explicit, resembling the exactness of the saddle point approximation \eqref{eq:saddle-point-path-int}, with 1-loop determinant given by the pfaffian at the denominator. However, even if the integration domain has been reduced, one must still perform a difficult infinite-dimensional path integration to get the final answer, whose \(T\)-dependence for example looks definitely non-trivial from \eqref{QM-3}. 

We can consider another choice of localizing term to simplify the final result, setting the gauge-fixing fermion to
\begin{equation}
\Psi[x,\eta] := G(\dot{x},\cdot) = \int_0^T dt\ g_{\mu\nu}(x(t)) \dot{x}^\mu(t) \eta^\nu(t) .
\end{equation}
With this choice, the gauge-fixed action reads
\begin{equation}
\label{QM-action-NT}
\begin{aligned}
S&[x]+\Omega[x,\eta]+\lambda Q_S \Psi[x,\eta] = \\ 
&= \int_0^T dt\biggl( \dot{x}^\mu \theta_\mu - H + \frac{1}{2}\omega_{\mu\nu}\eta^\mu \eta^nu + \lambda \left( 
g_{\mu\nu,\sigma} \dot{x}^\mu \eta^\sigma \eta^\nu + \eta^\mu \partial_t( g_{\mu\nu} \eta^\nu) + g_{\mu\nu} \dot{x}^\mu \dot{x}^\nu - g_{\mu\nu}\dot{x}^\mu X_H^\nu \right) \biggr) \\
&= \int_0^T dt\biggl(
\lambda g_{\mu\nu} \dot{x}^\mu \dot{x}^\nu + \lambda  \eta^\mu \nabla_t \eta^\nu + \dot{x}^\mu \theta_\mu + \frac{1}{2}\omega_{\mu\nu} \eta^\mu \eta^\nu - H - \lambda g_{\mu\nu} \dot{x}^\mu X_H^\nu \biggr)
\end{aligned}
\end{equation}
where in the second line the time-covariant derivative acts as \(\nabla_t \eta^\nu = \partial_t \eta^\nu + \Gamma^\nu_{\rho\sigma}\dot{x}^\rho \eta^\sigma\), and the localization locus is the subset of constant loops 
\begin{equation}
\mathcal{F}_0 := \left\lbrace \gamma\in\mathcal{F} : \dot{\gamma}=0 \right\rbrace \cong M
\end{equation}
that is, points in \(M\). Splitting again the loops near this subspace in constant modes plus fluctuations, and rescaling the latter as we did before, the path integral is reduced to a finite-dimensional integral over \(M\), the Niemi-Tirkkonen localization formula \cite{niemi-tirkkonen}
\begin{equation}
Z(T) = \int_{\Pi TM} \sqrt{g}d^{2n}x d^{2n}\eta\ \frac{e^{-T(H(x)-\omega(x,\eta))}}{\sqrt{\det'{\left[ g_{\mu\nu}\partial_t - (B_{\mu\nu}+R_{\mu\nu}) \right]}}}
\end{equation}
where the prime on the determinant means it is taken over the normal fluctuation modes, excluding the constant ones, giving exactly the equivariant Euler form of the normal bundle to \(\mathcal{F}_0\). This formula is much more appealing since it contains no refernce to \(T\)-dependent submanifolds of \(\mathcal{F}\), and the evaluation of the action on constant modes simply gives the Hamiltonian at those points multiplied by \(T\). The functional determinant at the denominator requires a specific regularization: using the \(\zeta\)-function method, it can be simplified to
\begin{equation}
\frac{1}{\sqrt{\det'{\left[ g_{\mu\nu}\partial_t - (B_{\mu\nu}+R_{\mu\nu}) \right]}}} = \sqrt{\det{\left[\frac{\frac{T}{2}(B+R)_{\mu\nu}}{\sinh{\left(\frac{T}{2}(B+R)_{\mu\nu} \right)}} \right]}} = \hat{A}_H(TR)
\end{equation}
where in the last equality we rewrote, by definition, the determinant as the \emph{\(U(1)_H\)-equivariant Dirac \(\hat{A}\)-genus} of the curvature \(R\) up to a constant \(T\), that is the Dirac \(\hat{A}\)-genus of the equivariant extension of the curvature, \(R+B\) \cite{berline-heatkernels}. Note that the exponential can be rewritten as the \(U(1)_H\)-equivariant Chern character of the symplectic form,
\begin{equation}
e^{-H+\omega} = \mathrm{ch}_H(\omega)
\end{equation}
and so the partition function can be nicely rewritten as
\begin{equation}
\label{QM-NT-formula}
Z(T) = \int_M \mathrm{ch}_H(T\omega)\wedge \hat{A}_H(TR)
\end{equation}
in terms of equivariant characteristic classes of the phase space \(M\) with respect to the \(U(1)_H\) Hamiltonian group action, that are determined by the initial classical system. The only remnant of the quantum theory is in the, now very explicit, dependence on the inverse temperature \(T\). This form of the partition function emphasizes the fact that if we put \(H=0\), we end up with a \emph{topological} theory. In this case there are no propagating physical degrees of freedom, and the partition function only describes topological properties of the underlying phase space. In the next section we will report a non-trivial example of this kind.

\section{Localization and index theorems}
\label{sec:index}

A famous and important application of the localization principle to loop space path integrals is the supersymmetric derivation of the \emph{Atiyah-Singer index theorem} \cite{Friedan-AS_index_thm} \cite{alvarez-susyAS} \cite{hietamaki-susyQM-AS}. The theorem relates the \emph{analytical index} of an elliptic differential operator on a compact manifold to a topological invariant, connecting the local data associated to solutions of partial differential equations to global properties of the manifold. Supersymmetric localization allowed to prove this statement, in a new way with respect to the original proof, for different examples of classical differential operators. We describe here the application to the index of the Dirac operator acting on the spinor bundle \(S\) on an even-dimensional compact manifold \(M\), in presence of a gravitational and electromagnetic background, specified by the metric \(g\) and a \(U(1)\) connection 1-form \(A\) on a \(\mathbb{C}\)-line bundle \(L_{\mathbb{C}}\) over \(M\).\footnote{This can be extended to non-Abelian gauge groups as well, but for simplicity we report the Abelian case.}

The Dirac operator is defined as the fiber-wise Clifford action of the covariant derivative on the twisted spinor bundle \(TM\otimes S\otimes L_{\mathbb{C}}\) over \(M\):
\begin{equation}
i\slashed{\nabla} = i\gamma^\mu \left( \partial_\mu +\frac{1}{8} \omega_{\mu ij}[\gamma^i,\gamma^j] + iA_\mu\right)
\end{equation}
where \(\lbrace \gamma^\mu \rbrace\) are the gamma-matrices generating the Clifford algebra in the given spin representation, satisfying the anticommutation relation
\begin{equation}
\lbrace \gamma^\mu,\gamma^\nu\rbrace = 2g^{\mu\nu},
\end{equation}
and \(\omega\) is the spin connection related to the metric.\footnote{See Appendix \ref{app:spin-geom}.} The \lq\lq curved\rq\rq\ and \lq\lq flat\rq\rq\ indices are related through the vielbein \(e^i_\mu(x)\),
\begin{equation}
g_{\mu\nu}(x) = e^i_\mu(x) e^j_\nu(x) \eta_{ij}, \qquad  \gamma^i = e^i_\mu(x) \gamma^\mu(x),
\end{equation}
with \(\eta\) the flat metric. The analytical index of the Dirac operator is defined as \cite{atiyah-singer-thm}
\begin{equation}
\mbox{index}(i\slashed{\nabla}) := \dim\mathrm{Ker}(i\slashed{\nabla}) - \dim\mathrm{coKer}(i\slashed{\nabla}) = \dim\mathrm{Ker}(i\slashed{\nabla}) - \dim\mathrm{Ker}(i\slashed{\nabla}^\dagger).
\end{equation}
We have thus to find the number of zero-energy solutions of the Dirac equation 
\begin{equation}
i\slashed{\nabla}\Psi = E\Psi ,
\end{equation}
where \(\Psi\) is a Dirac spinor. In even dimensions, we can decompose the problem in the chiral basis of the spin representation, where 
\begin{equation}
\label{AS-1}
\gamma^i = \left(\begin{array}{cc}
0 & \sigma^i \\  \overline{\sigma}^i & 0
\end{array} \right) , \quad \gamma^c = \left(\begin{array}{cc}
\mathds{1} & 0 \\ 0 & -\mathds{1}
\end{array} \right) ,
\quad i\slashed{\nabla} = \left(\begin{array}{cc}
0 & D \\ D^\dagger & 0
\end{array} \right) , \quad \Psi = \left(\begin{array}{c}
\psi_- \\ \psi_+
\end{array} \right),
\end{equation}
and the chirality matrix is denoted by \(\gamma^c\). In this representation, we see that the index counts the number of zero-energy modes with positive chirality minus the number of zero-energy modes of negative chirality,
\begin{equation}
\mbox{index}(i\slashed{\nabla}) = \dim\mathrm{Ker}(D) - \dim\mathrm{Ker}(D^\dagger) .
\end{equation}

It is possible to give a path integral representation of this index, a key ingredient to apply the localization principle. In order to do that, we first prove that it can be rewritten as a \emph{Witten index},
\begin{equation}
\label{AS-3}
\mbox{index}(i\slashed{\nabla}) = \mathrm{Tr}_{\mathcal{H}} \left(\gamma^c e^{-T\Delta}\right)
\end{equation}
where \(\Delta := (i\slashed{\nabla})^2\) is the Shr\"{o}edinger operator (the covariant Laplacian) and the parameter \(T>0\) is a regulator for the operator trace, taken over the space \(\mathcal{H}\) of Dirac spinors, sections of the twisted spinor bundle over \(M\).
\begin{proof}[Proof of \eqref{AS-3}]
First, we notice that \((i\slashed{\nabla})\) is symmetric and elliptic, and since \(M\) is compact it is also essentially self-adjoint \cite{wolf-chern-Dirac_op}. Thus, it has a well defined spectrum \(\lbrace \Psi^E\rbrace\) that forms a basis of the function space at hand. The same modes diagonalize also the Schr\"odinger operator,
\begin{equation}
\Delta \Psi^E = E^2 \Psi^E
\end{equation}
so we can shift the attention to solutions of the Schr\"odinger equation with eigenvalue satisfying \(E^2=0\). It is easy to see that the Dirac operator anticommutes with the chirality matrix \(\gamma^c\), and so the Schr\"odinger operator commutes with it. Thus we can split the field space in complementary subspaces \(\mathcal{S}^E_{\pm}:=\left\lbrace \Psi\in\mathcal{H}: \Delta\Psi=E^2\Psi, \gamma^c \Psi = \pm \Psi\right\rbrace \), for every eigenvalue \(E^2\) and chirality \((\pm)\). For every non-zero energy, we establish an isomorphism \(\mathcal{S}^E_+ \cong \mathcal{S}^E_-\): using the fact that \((i\slashed{\nabla})\) and \(\gamma^c\) anticommute, starting from a solution with eigenvalue \(E^2\) and definite chirality \(\Psi_{\pm}\) we can construct another one with opposite chirality \((i\slashed{\nabla}\Psi_{\pm})\),
\begin{equation}
\begin{aligned}
&\Delta(i\slashed{\nabla}\Psi_{\pm}) = i\slashed{\nabla}i\slashed{\nabla}i\slashed{\nabla} \Psi_{\pm} = E^2 (i\slashed{\nabla}\Psi_{\pm}) \\
&\gamma^c(i\slashed{\nabla}\Psi_{\pm}) = -i\slashed{\nabla}\gamma^c\Psi_{\pm} = -(\pm)(i\slashed{\nabla}\Psi_{\pm}) .
\end{aligned}
\end{equation}
Thus the maps \(\phi_{\pm} : \mathcal{S}^E_{\pm} \to \mathcal{S}^E_{\mp}\) such that
\begin{equation}
\phi_{\pm}(\Psi) := \frac{i\slashed{\nabla}}{|E|}\Psi
\end{equation}
are both well defined and are right and left inverse of each other, giving the bijective correspondence for every \(|E|\neq 0\). This is the well known fact that particles and antiparticles are created in pairs with opposite energy and chirality. 

If we take the trace over the field space \(\mathcal{H}\), this is splitted into the sum of the traces over every subspace \(\mathcal{S}^E\) of definite energy squared. In every one of them the restricted Witten index gives
\begin{equation}
\mathrm{Tr}_{\mathcal{S}^E} \left(\gamma^c e^{-T\Delta}\right) = e^{-TE^2}\left( \mathrm{Tr}_{\mathcal{S}^E_+}(\mathds{1}) - \mathrm{Tr}_{\mathcal{S}^E_-}(\mathds{1})\right) = 0
\end{equation} 
for every \(E\neq 0\). So the whole trace is formally independent on \(T\), resolving on the subspace of zero-energy, where the number of chirality + and - eigenstates is different in general:
\begin{equation}\begin{aligned}
\mathrm{Tr}_{\mathcal{H}} \left(\gamma^c e^{-T\Delta}\right) &= \mathrm{Tr}_{E=0}(\gamma^c) \\
&= \#^{E=0}(\mathrm{chirality\ (+)\ modes}) - \#^{E=0}(\mathrm{chirality\ (-)\ modes}) . \end{aligned}
\end{equation}
\end{proof}

The Witten index representation \eqref{AS-3} and the chirality decomposition of the operators of interest \eqref{AS-1} permit to see the current problem as a \(\mathcal{N}=1\) supersymmetric QM on the manifold \(M\), identifying chirality +(-) spinors with bosonic(fermionic) states. Here, the supersymmetry algebra \eqref{eq:susy-algebra-chiral} is simply\footnote{Often this is called \(\mathcal{N}=1/2\) supersymmetry, leaving the name \(\mathcal{N}=1\) for the complexified algebra with generators \( Q,\tilde{Q}\), and imposition of Majorana condition.} (choosing an appropriate normalization)
\begin{equation}
[Q,Q] = 2H ,
\end{equation}
and corresponds to the Schr\"odinger operator above, if we make the following identifications:
\begin{equation}
\begin{aligned}
i\slashed{\nabla} \quad &\leftrightarrow\quad Q \\
\Delta = (i\slashed{\nabla})^2 \quad &\leftrightarrow\quad H = Q^2 .
\end{aligned}
\end{equation}
The chirality matrix \(\gamma^c\) is identified with the operator \((-1)^F\), where \(F\) is the fermion number operator, that assigns eigenvalue \(+1\) to bosonic states and \(-1\) to fermionic states. The Witten index representation in the quantum system is thus
\begin{equation}
\label{AS-2}
\begin{aligned}
\mbox{index}(i\slashed{\nabla}) &= \mathrm{Tr}\left( (-1)^F e^{-TH} \right) \\
&= n^{E=0}(\mathrm{bosons}) - n^{E=0}(\mathrm{fermions}) .
\end{aligned}
\end{equation}
The proof above, translated in terms of the quantum system, shows that in supersymmetric QM eigenstates of the Hamiltonian have non-negative energy, and are present in fermion-boson pairs for every non-zero energy.  Since \(Q^2\) is a positive-definite Hermitian operator, the zero modes \(|0\rangle\) of \(H\) are supersymmetric, \(Q|0\rangle = 0\) (they do not have supersymmetric partners). Thus the non-vanishing of the Witten index \eqref{AS-2} is a sufficient condition to ensure that there is at least one supersymmetric vacuum state available, whereas its vanishing is a necessary condition for spontaneous braking of supersymmetry by the vacuum.

The Witten index has a super-loop space path integral representation \cite{cecotti-functional}, so that we can rewrite \eqref{AS-2} as
\begin{equation}
\label{AS-path-integral}
\mbox{index}(i\slashed{\nabla}) = \int D^{2n}\phi D^{2n}\psi\ e^{-T S[\phi,\psi]}
\end{equation}
where \(S\) is the Euclidean action corresponding to the Hamiltonian \(H\) and the fields are defined on the unit circle. The appropriate supersymmetric theory which describes a spinning particle on a gravitational background is the 1-dimensional \emph{supersymmetric non-linear \(\sigma\) model}. The superspace formulation of this model considers the base space as a super extension of the 1-dimensional spacetime with coordinates \((t,\theta)\), and \(M\) as the target space with covariant derivative given by the Dirac operator. A superfield, trivialized with respect to coordinates \((t,\theta)\) and \((x^\mu)\) on \(M\) is then
\begin{equation}
\Phi^\mu (t,\theta) \equiv (x^\mu \circ \Phi)(t,\theta) = \phi^\mu(t) + \psi^\mu(t) \theta ,
\end{equation}
and supersymmetry transformations are given by the action of the odd vector field \(\underline{Q}=\partial_\theta + \theta \partial_t\), that in terms of component fields reads
\begin{equation}
\label{AS-5}
\delta \phi^\mu = \psi^\mu, \qquad \delta \psi^\mu = \dot{\phi}^\mu .
\end{equation}
Denoting the superderivative as \(D=-\partial_\theta + \theta \partial_t\), the action of the non-linear \(\sigma\) model coupled to the gauge field \(A\) can be given as
\begin{equation}
S[\Phi] = \int dt \int d\theta\ \frac{1}{2} g_{\Phi(t,\theta)}\left( D\Phi, \dot{\Phi}\right) + A(D\Phi)
\end{equation}
where \(D\Phi\) and \(\dot{\Phi}\) are thought as (super)vector fields on \(M\) such that, for any function \(f\in C^\infty(M)\), \(D\Phi(f):=D(f\circ\Phi)\) and \(\dot{\Phi}(f):=\partial_t (f\circ\Phi)\). Inserting the trivialization for the metric components \(g_{\mu\nu}(\Phi(t,\theta))=g_{\mu\nu}(\phi)+\theta\psi^\sigma(t) g_{\mu\nu,\sigma}(\phi)\), and the component expansion for \(\Phi\), the action is simplified to
\begin{equation}
\label{AS-4}
S[\phi,\psi] = \int dt \left( \frac{1}{2}g_{\mu\nu}(\phi)\dot{\phi}^\mu \dot{\phi}^\nu + \frac{1}{2}g_{\mu\nu}(\phi)\psi^\mu (\nabla_t \psi)^\nu + A_\mu(\phi)\dot{\phi}^\mu - \frac{1}{2}\psi^\mu F_{\mu\nu}\psi^\nu \right)
\end{equation}
where we suppressed the \(t\)-dependence, \(F_{\mu\nu} = \partial_{[\mu}A_{\nu]}\) are the components of the electromagnetic field strength, and the time-covariant derivative \(\nabla_t\) acts as \[
(\nabla_t V)^\sigma(\phi(t)) = \dot{V}^\sigma (\phi(t)) + \Gamma^\sigma_{\mu\nu}\dot{\phi}^\mu(t)V^\nu(\phi(t)) .\]

The action \eqref{AS-4} is formally equivalent to the model-independent action \eqref{QM-action-NT} of the last section, with \(T\) behaving like the localization parameter \(\lambda\), if we identify \(\theta\) with the electromagnetic potential \(A\) and \(\omega\) with the field strength \(F\), and if we set the Hamiltonian and its associated vector field \(H, X_H\) to zero. This means that we can give to it an equivariant cohomological interpretation on the super-loop space \(\Pi T\mathcal{F}\) over \(M\), with coordinates identified with \((\phi^\mu,\psi^\mu)\), as was pointed out first by Atiyah and Witten \cite{Atiyah-circular-symmetry}. Moreover, since the Hamiltonian vanishes, this action describes no propagating physical degrees of freedom, and thus the model is topological. Indeed, its value has to give the index of the Dirac operator, expected to be a topological quantity. To emphasize the equivariant cohomological nature of the model, we notice that the action functional can be split in the loop space (pre-)symplectic 2-form
\begin{equation}
\Omega[\phi,\psi] := \int dt\ \frac{1}{2}\psi^\mu \left( g_{\mu\nu}\nabla_t - F_{\mu\nu}\right) \psi^\nu
\end{equation}
and the loop space Hamiltonian
\begin{equation}
\mathcal{H} = \int dt\ \left(\frac{1}{2}g_{\mu\nu}\dot{\phi}^\mu \dot{\phi}^\nu + A_\mu \dot{\phi}^\mu \right).
\end{equation}
They satisfy \(d_{\mathcal{F}}\mathcal{H} = -\iota_{\dot{\phi}}\Omega\), so  the supersymmetry transformation \eqref{AS-5} is rewritten in terms of the Cartan differential \(Q_{\dot{\phi}} = d_{\mathcal{F}}+\iota_{\dot{\phi}}\), and \(Q_{\dot{\phi}}S = Q_{\dot{\phi}}(\mathcal{H}+\Omega) = 0\). Moreover, we can find a loop space symplectic potential \(\Sigma\) such that \(S\) is equivariantly exact:
\begin{equation}
\label{index-Q-exact-S}
\begin{aligned}
S[\phi,\psi] &= Q_{\dot{x}}\Sigma[\phi,\psi] \\
\mathrm{where}\quad \Sigma[\phi,\psi] &:= \int dt \left( g_{\mu\nu}(\phi)\dot{\phi}^\mu + A_\nu(\phi) \right) \psi^\nu .
\end{aligned}
\end{equation}
Notice that, since the Hamiltonian \(H\) vanishes, the localizing \(U(1)\) symmetry is the one generated by time-translation with respect to the base space \(\mathbb{S}^1\), an intrinsic property of the geometric structure that underlies the model.

We can now apply the Niemi-Tirkkonen formula, localizing the path integral \eqref{AS-path-integral} into the moduli space of constant loops, \textit{i.e.}\ as an integral over \(M\). The result is the same as equation \eqref{QM-NT-formula}, but now since the Hamiltonian vanishes, the Chern class and the Dirac \(\hat{A}\)-genus (see Appendix \ref{app:char-classes}) are not equivariantly extended by the presence of an Hamiltonian vector field, giving the topological formula
\begin{equation}
\mbox{index}(i\slashed{\nabla}) = \int_M \mbox{ch}(F)\wedge \hat{A}(R) .
\end{equation}
This is the result of the Atiyah-Singer index theorem for the Dirac operator on the twisted spinor bundle over \(M\). Similar applications of the localization principle to variations of the non-linear \(\sigma\) model give correct results for other classical complexes as well (de Rham, Dolbeault for example), in terms of different topological invariants \cite{alvarez-susyAS}.

\section{Equivariant structure of supersymmetric QFT and supersymmetric localization principle}
\sectionmark{Equivariant structure of supersymmetric QFT}
\label{sec:susy-loc}

In the last section we saw how to give an equivariant cohomological interpretation to a model exhibiting Poincaré-supersymmetry, in terms of the super-loop space symplectic structure introduced before. In the following we are interested in applying the same kind of supersymmetric localization principle to higher dimensional QFT on a, possibly curved, \(n\)-dimensional spacetime \(M\), where there is some preserved supersymmetry. 

In \cite{morozov-supersymplectic-QFT}\cite{palo} it is argued that any generic quantum field theory with at least an \(\mathcal{N}=1\) Poincaré supersymmetry admits a field space Hamiltonian (symplectic) structure and a corresponding \(U(1)\)-equivariant cohomology responsible for localization of the supersymmetric path integral. The key feature is an appropriate off-shell component field redefinition which defines a splitting of the fields into loop space \lq\lq coordinates\rq\rq\ and their associated \lq\lq differentials\rq\rq. In general, unlike the simplest case of the last section where bosonic fields were identified with coordinates and fermionic fields with 1-forms, loop space coordinates and 1-forms involve both bosonic and fermionic fields. It is proven that, within this field redefinition on the super-loop space, any supersymmetry charge \(Q\) can be identified with a Cartan differential
\begin{equation}
Q = d_{\mathcal{F}} + \iota_{X_+}
\end{equation}
analogously to \eqref{QM-Cartan-differential}, whose square generates translations in a given \lq\lq light-cone\rq\rq\ direction \(x_+\),
\begin{equation}
Q^2 = \mathcal{L}_{X_+} \sim \int_M \frac{\partial}{\partial x^+}
\end{equation}
that corresponds to the \(U(1)\) symmetry that can be used to exploit the localization principle. Taking the base spacetime to be compact in the light-cone direction, the periodic boundary conditions ensure \(Q^2=0\), analogously to the loop space assumption of the one dimensional case. Also, it is argued that the supersymmetric action can be generally split into the sum of a loop space scalar function and a (pre-)symplectic 2-form,
\begin{equation}
S_{susy} = \mathcal{H} + \Omega
\end{equation}
related by \(d_{\mathcal{F}}\mathcal{H} = - \iota_{X_+}\Omega\). Thus, the supersymmetry of the action can be seen in general as the \(U(1)\)-equivariant closeness required to the application of the localization principle, and the path integral localizes onto the locus of constant loops, \textit{i.e.}\ zero-modes of the fields. 

Even without entering in the details of this construction in terms of auxiliary fields redefinition, we feel now allowed to translate in full generality the circle localization principle in the framework of Poincaré-supersymmetric QFT. In the component-field description, we consider a rigid supersymmetric background over the given compact\footnote{This ensures the loop space interpretation of above.} spacetime \(M\) and a graded field space \(\mathcal{F}\) that plays the role of the super-loop space over \(M\), whose even-degree forms are bosonic fields and the odd-degree forms are fermionic fields. The (infinitesimal) supersymmetry action of a preserved supercharge \(Q\) plays the role of the Cartan differential \(d_{\mathcal{F}}\), squaring to a bosonic symmetry that corresponds to the (infinitesimal) action of a \(U(1)\) symmetry group, \(Q^2 \sim \mathcal{L}_X\) with \(X\) an \emph{even} vector field. The Cartan model for the \(U(1)\)-equivariant cohomology of \(\mathcal{F}\) is defined by the subcomplex of \(Q\)-invariant (or supersymmetric, or \lq\lq BPS\rq\rq) observables, where the supercharge squares to zero.

We consider a supersymmetric model specified by the (Euclidean) action functional \(S\in C^\infty(\mathcal{F})\) such that \(\delta_{Q}S = 0\), and a BPS observable \(\mathcal{O}\) such that \(\delta_{Q} \mathcal{O}=0\). Now the partition function \eqref{eq:path-integral-euclidean} and the expectation value \eqref{eq:expectation-value} are seen as integrations of equivariantly closed forms with respect to the differential \(Q\). The supersymmetric localization principle then tells us that we can modify the respective integrals adding an equivariantly exact localizing term to the action,
\begin{equation}
\lambda S_{loc}[\Phi] := \lambda \delta_{Q} \mathcal{V}[\Phi]
\end{equation}
where \(\mathcal{V}\in \Omega^1(\mathcal{F})^{U(1)}\) is a \(Q^2\)-invariant fermionic functional, the \lq\lq gauge-fixing fermion\rq\rq\ of Section \ref{sec:QM}, and \(\lambda\in\mathbb{R}\) is a parameter. Assuming the supersymmetry \(\delta_Q\) to be non anomalous, the partition function and the expectation value are not changed by this modification, \textit{i.e.}\ the the new integrand lies in the same equivariant cohomology class,
\begin{equation}
\begin{aligned}
&\frac{d}{d\lambda} Z(\lambda) = \int_{\mathcal{F}} D\Phi\ (-\delta_{Q} \mathcal{V}[\Phi]) e^{-(S+\lambda \delta_{Q} \mathcal{V})[\Phi]} = - \int_{\mathcal{F}} D\Phi\ \delta_{Q}\left( \mathcal{V}[\Phi] e^{-(S+\lambda \delta_{Q} \mathcal{V})[\Phi]}\right) = 0 \\
&\frac{d}{d\lambda} \langle \mathcal{O}\rangle_\lambda = -\frac{1}{Z(\lambda)}\left(\frac{d}{d\lambda}Z(\lambda)\right)\langle \mathcal{O}\rangle_\lambda + \frac{1}{Z(\lambda)}\int_{\mathcal{F}}D\Phi\ \delta_Q \left( \mathcal{V}[\Phi] \mathcal{O}[\Phi] e^{-(S+\lambda\delta_Q \mathcal{V})[\Phi]}\right) = 0
\end{aligned}
\end{equation}
by the same argument of Section \ref{sec:QM}. Assuming the bosonic part of \( \delta_Q \mathcal{V}\) to be positive-semidefinite, and using the \(\lambda\)-independence of the path integral, we can evaluate the partition function or the expectation value in the limit \(\lambda\to+\infty\), getting the localization formulas
\begin{equation}
\label{eq:susy-loc-limit}
Z = \lim_{\lambda\to\infty} \int_{\mathcal{F}}D\Phi\ e^{-(S+\lambda S_{loc})[\Phi]} ,\qquad \langle\mathcal{O}\rangle = \frac{1}{Z} \lim_{\lambda\to\infty} \int_{\mathcal{F}}D\Phi\ \mathcal{O}[\Phi]e^{-(S+\lambda S_{loc})[\Phi]} .
\end{equation}
The path integrals localize then onto the locus \(\mathcal{F}_0\) of saddle points of \(S_{loc}\). Following again the same argument of Section \ref{sec:ABBVproof} we can in fact expand the fields about these saddle point configurations, rescale the normal fluctuations as 
\begin{equation}
\Phi = \Phi_0 + \frac{1}{\sqrt{\lambda}} \tilde{\Phi} ,
\end{equation}
and the augmented action functional as
\begin{equation}
(S+\lambda S_{loc})[\Phi] = S[\Phi_0] + \frac{1}{2}\int_M d^nx \int_M d^ny \left(\frac{\delta^2 S_{loc}}{\delta \Phi(x) \delta \Phi(y)}\right)_{\Phi_0} \tilde{\Phi}(x) \tilde{\Phi}(y) + O(\lambda^{-1/2}).
\end{equation}
The functional measure on the normal sector \(D\tilde{\Phi}\) is not affected by the rescaling, since the supersymmetric model contains the same number of bosonic and fermionic physical component fields,\footnote{Note that this has to be true \emph{off-shell}, \textit{i.e.}\ without imposing any EoM.} and the corresponding Jacobians cancel by Berezin integration rules. The integral over this fluctuations is Gaussian and can be performed, giving the \lq\lq 1-loop determinant\rq\rq\ analogous to the equivariant Euler class that appeared in Theorem \ref{thm:ABBV}. The leftover integral corresponds to the saddle point formula \eqref{eq:saddle-point-path-int}, but as an exact equality:
\begin{equation}\begin{aligned}
Z &= \int_{\mathcal{F}_0} D\Phi_0\ e^{-S[\Phi_0]} Z_{1-loop}[\Phi_0]  \\
\langle\mathcal{O}\rangle &= \frac{1}{Z} \int_{\mathcal{F}_0} D\Phi_0\ \mathcal{O}[\Phi_0] e^{-S[\Phi_0]} Z_{1-loop}[\Phi_0]
\end{aligned}
\end{equation}
where 
\begin{equation}
Z_{1-loop}[\Phi_0] := \left( \mathrm{Sdet} \left[ \frac{\delta^2 S_{loc}}{\delta \Phi(x) \delta \Phi(y)} [\Phi_0] \right]\right)^{-1}
\end{equation}
and the super-determinant denotes collectively the result of Gaussian intergrations over bosonic or fermionic fields.

Although the choice of localizing term \(\mathcal{V}\) is arbitrary, and different choices give in principle different localization loci, the final result must be the same for every choice. At the end of Section \ref{sec:susy-curved} we remarked that the \(\mathcal{N}=2\) supersymmetric Yang-Mills Lagrangian on a 3-dimensional maximally supersymmetric background is \(Q\)-exact, and thus can be used as a localizing term for supersymmetric gauge theories on this type of 3-dimensional spacetimes. It turns out that also the \(\mathcal{N}=2\) matter (chiral) Lagrangian is \(Q\)-exact in three dimensions \cite{willet-loc3d}. A canonical choice of localizing action can be, schematically \cite{cremonesi}
\begin{equation}
\label{eq:localizing-term-canonical}
S_{loc}[\Phi] := \int_M \delta_{Q} \sum_f \left( (\delta_Q \Phi_f)^\dagger \Phi_f + \Phi_f^\dagger (\delta_Q \Phi_f^\dagger)^\dagger \right)
\end{equation}
where the sum runs over the fermionic fields of the theory. Its bosonic part is
\begin{equation}
\left. S_{loc}[\Phi]\right|_{bos} = \sum_f \left( |\delta_Q \Phi_f|^2 + |\delta_Q \Phi_f^\dagger|^2 \right) ,
\end{equation}
that is indeed positive semidefinite. The corresponding localization locus is the subcomplex of BPS configurations,
\begin{equation}
[\mathrm{fermions}]=0 ,\qquad \delta_Q [\mathrm{fermions}] = 0.
\end{equation}

\bigskip
Concluding this general and schematic discussion, there are a couple of remarks we wish to point out. Firstly, in the above formulas we always considered generic BPS observables that are expressed through (local or non-local) combinations of the fields. In other words, their quantum expectation values are defined as insertions in the path integral of corresponding (classical) functionals on the field space. Examples of local objects of this kind are correlation functions of fundamental fields. A famous class of non-local quantum operators that are expressible as classical functionals are the so-called \emph{Wilson loops}. In gauge theory with gauge group \(G\) and local gauge field \(A\), a Wilson loop in the representation \(	R\) of \(Lie(G)\) over the closed curve \(C:\mathbb{S}^1 \to M\) is defined by
\begin{equation}
W_R(C) := \frac{1}{\dim R} \mathrm{Tr}_R\left(\mathcal{P}\exp{i\oint_{\mathbb{S}^1} C^*(A)} \right) 
\end{equation} 
where the trace \(\mathrm{Tr}_R\) is taken in the given representation.\footnote{In the adjoint representation, this denotes an invariant inner product in \(\mathfrak{g}\), for example the Killing form for a semisimple Lie algebra.} This is gauge invariant, and represents physically the phase acquired by a charged probe particle in the representation \(R\) after a tour on the curve \(C\), in presence of the gauge potential \(A\). Mathematically, if \(R\) is the adjoint representation, the Wilson loop represents the parallel transport map between the fibers of the principal \(G\)-bundle defining the gauge theory. These operators have many interesting applications in physics: depending on the chosen curve \(C\) their expectation value can be interpreted as an order parameter for the confinement/deconfinement phase transitions in QCD or the Bremsstrahlung function for an accelerated particle \cite{Maldacena-wilson-loops,Correa-bremsstr,Aharony-deconfinement}. In the case of 3-dimensional Chern-Simons theory, they can be used to study topological invariants in knot theory \cite{witten-knot}. In supersymmetric theories, they are particularly relevant for tests of the AdS/CFT correspondence \cite{roadmap-wilson-loops}. In the next sections we will review some interesting cases in which expectation values of this type of operators can be evaluated exactly using the supersymmetric localization principle. There exists another class of interesting operators in the quantum theory, that cannot be expressed as classical functionals on the field space. These are the so-called \emph{disorder operators}, and their expectation values are defined by a restriction of the path integral to those field configurations which have prescribed boundary conditions around some artificial singularity introduced in spacetime. An example of these are the \emph{'t Hooft operators}, which introduce a Dirac monopole singularity along a path in a 4-dimensional space \cite{gomis-tHooft}. These kind of operators can be also studied non-perturbatively with the help of localization techniques \cite{Gomis-loc-tHooft}. For a great review of different examples of localization computations in supersymmetric QFT, see \cite{LocQFT}.

The second remark we wish to make is that, in presence of a gauge symmetry, the action functionals in the above formulas have to be understood as the quantum (\textit{i.e.}\ gauge-fixed) action in order to give meaning to the corresponding partition function. That is, one has to introduce Faddeev-Popov ghost fields in the theory and the associated BRST transformations \(\delta_{BRST}\). We have seen in Section \ref{sec:BRST-cohom-equiv} that it is always possible to see the BRST complex in terms of equivariant cohomology on the field space, so this supersymmetry transformation have to be incorporated in the equivariant structure of the supersymmetric theory. In this case, the field space acquires a \(\mathbb{Z}\)-grading corresponding to the ghost number, on top of the \(\mathbb{Z}_2\) one from supersymmetry, and the appropriate Cartan differential with respect to which the equivariant cohomological structure is defined is then the total supersymmetry variation
\(
Q = \delta_{susy} + \delta_{BRST} \).

\section[Localization of N=4,2,2* gauge theory on the 4-sphere]{Localization of \(\mathcal{N}=4,2,2^*\) gauge theory on the 4-sphere}

In this section we review, following the seminal work of Pestun \cite{pestun-article}, how to exploit the supersymmetric localization principle in \(\mathcal{N}=4\) Euclidean Super Yang-Mills theory on the four-sphere \(\mathbb{S}^4\). The \(\mathcal{N}=2\) and \(\mathcal{N}=2^*\) theories can be also treated with the same technique. In particular, it was possible to solve exactly the partition function of the theory and the expectation value of the Wilson loop defined by
\begin{equation}
W_R(C) := \frac{1}{\dim R}\mathrm{Tr}_R\left(\mathcal{P}\exp{i\oint_{C(\mathbb{S}^1)} (A_\mu \dot{C}^\mu + |\dot{C}| \Phi_0) dt} \right) 
\end{equation}
where \(C\) is a closed equatorial curve on \(\mathbb{S}^4\) of tangent vector \(\dot{C}\), and the scalar field \(\Phi_0\) is required by supersymmetry, as will become clear later. The localization procedure makes the path integral reduce to a finite-dimensional integral over the Lie algebra of the gauge group, a so-called \lq\lq matrix model\rq\rq.

\subsection{The action and the supersymmetric Wilson loop}

We will consider the theories revisited in Sections \ref{subsec:susy-4dN4} and \ref{subsec:4dN4-curved}. We report the action of the \(\mathcal{N}=2^*\) theory on the 4-sphere,
\begin{equation}
\begin{aligned}
S^{\mathcal{N}=2^*}_{\mathbb{S}^4} = \int_{\mathbb{S}^4} d^4x\sqrt{g}\ \frac{1}{g_{YM}^2} \mathrm{Tr}\left( F_{MN}F^{MN} - \Psi\Gamma^M D_M \Psi \right.& + \frac{2}{r^2} \Phi_A \Phi^A - \\
&\left. - \frac{1}{4r}(R^{ki}M_k^j)\Phi_i\Phi_j - \sum_{i=1}^7 K_i K_i\right)
\end{aligned}
\end{equation}
where \( D_0\Phi_i \mapsto [\Phi_0,\Phi_i] + M_i^j\Phi_j\) and \(D_0\Psi \mapsto [\Phi_0,\Psi] + \frac{1}{4}M_{ij}\Gamma^{ij}\Psi\), \(i,j=5,\cdots,8\). In the limit of zero mass \(M\) we get the \(\mathcal{N}=4\) YM theory, while in the limit of infinite mass the \(\mathcal{N}=2\) hypermultiplet decouples and the pure \(\mathcal{N}=2\) YM theory is recovered. This is invariant under the superconformal transformations \eqref{pestun-susyvarfinal} that we report here,
\begin{equation}
\label{pestun-variationscopy}
\begin{aligned}
\delta_\epsilon A_M &= \epsilon \Gamma_M \Psi \\
\delta_\epsilon \Psi &= \frac{1}{2}\Gamma^{MN}F_{MN}\epsilon + \frac{1}{2} \Gamma^{\mu A}\Psi_A \nabla_\mu \epsilon + \sum_{i=1}^7 K_i \nu_i \\
\delta_\epsilon K_i &= - \nu_i \Gamma^M D_M \Psi
\end{aligned}
\end{equation}
with \((\nu_i)_{i=1,\cdots,7}\) satisfying \eqref{pestun-auxiliary-cond}, and \(\epsilon\) being a conformal Killing spinor satisfying \eqref{pestun-5} and \eqref{pestun-6}. When the mass is non-zero, the Killing condition is restricted to \eqref{pestun-cKseq-restricted}, or equivalently
\begin{equation}
\label{pestun-15}
\tilde{\epsilon} = \frac{1}{2r}\Lambda \epsilon
\end{equation} 
where \(\Lambda\) is an \(SU(2)_L^R\) generator. The superconformal algebra closes schematically as
\begin{equation}
\delta_\epsilon^2 = -\mathcal{L}_v - G_\Phi - (R+M) - \Omega .
\end{equation}
To obtain a Poincaré-equivariant differential interpretation of this variation, we want \(\delta_\epsilon\) to generate rigid supersymmetry, \textit{i.e.}\ square only to the Poincaré algebra (plus R-symmetry, up to gauge transformations). Thus, to eliminate the dilatation contribution, we impose also the condition 
\begin{equation}
\label{pestun-11}
\epsilon\tilde{\epsilon} = 0 .
\end{equation}
If the mass is non-zero, the \(SU(1,1)^\mathcal{R}\) is broken, so also its contribution should be eliminated, imposing further the condition
\begin{equation}
\label{pestun-12}
\tilde{\epsilon}\Gamma^{09}\epsilon=0 .
\end{equation}

Solutions to \eqref{pestun-5} and \eqref{pestun-6} are easy to compute in the flat space limit \(r\to\infty\): here \(\nabla_\mu = \partial_\mu\), and \(\partial_\mu \tilde{\epsilon}=0\) imposes \(\tilde{\epsilon}(x) = \hat{\epsilon}_c\) constant. Thus, the conformal Killing spinor in flat space is just the one considered in \eqref{pestun-10},
\begin{equation}
\label{pestun-13}
\epsilon(x) = \hat{\epsilon}_s + x^\mu \Gamma_\mu  \hat{\epsilon}_c
\end{equation}
where the first constant term generates supertranslations, while the term linear in \(x\) generates superconformal transformations. The constant spinors \(\hat{\epsilon}_s, \hat{\epsilon}_c\) parametrize in general the space of solutions of the conformal Killing spinor equation. For a finite radius \(r\), using stereographic coordinates and the round metric \eqref{pestun-metric}, the covariant derivative acts as \(\nabla_\mu \epsilon = \left( \partial_\mu + \frac{1}{4}\omega_{ij\mu}\Gamma^{ij} \right)\epsilon\), where \(\omega\) is the spin connection
\begin{equation}
\omega^i_{j\mu} = \left( e^i_\mu e^\nu_j  - e_{j\mu}e^{i\nu} \right) \partial_\nu \Omega
\end{equation}
and \(e\) is the vielbein corresponding to the metric.\footnote{Here we use latin indices as \lq\lq flat\rq\rq\ indices and greek indices as \lq\lq curved\rq\rq\ indices, so that as \( g_{\mu\nu} = e^i_\mu e^j_\nu \delta_{ij}\).} The general solution in this coordinate system is
\begin{equation}
\epsilon(x) = \frac{1}{\sqrt{1+\frac{x^2}{4r^2}}}\left(\hat{\epsilon}_s + x^\mu\Gamma_\mu \hat{\epsilon}_c \right) \qquad \tilde{\epsilon}(x) = \frac{1}{\sqrt{1+\frac{x^2}{4r^2}}}\left(\hat{\epsilon}_c - \frac{x^\mu\Gamma_\mu}{4r^2} \hat{\epsilon}_s\right) 
\end{equation}
that indeed simplifies to \eqref{pestun-13} in the limit of infinite radius.
The conditions \eqref{pestun-11}, \eqref{pestun-12}, \eqref{pestun-15} are rewritten in terms of the constant spinors as 
\begin{equation}
\label{pestun-14}
\hat{\epsilon}_s\hat{\epsilon}_c = \hat{\epsilon}_s \Gamma^{09}\hat{\epsilon}_c = 0 \qquad \hat{\epsilon}_s \Gamma^{\mu}\hat{\epsilon}_s = \frac{1}{4r^2} \hat{\epsilon}_c \Gamma^{M}\hat{\epsilon}_c  \qquad \hat{\epsilon}_c = \frac{1}{2r} \Lambda \hat{\epsilon}_s .
\end{equation} 
The second condition is solved if the two constant spinors are taken to be chiral with respect to the 4-dimensional chirality operator \(\Gamma^{1234}\), so that both terms vanish automatically. In Pestun's conventions, they are chosen to have the same definite chirality and orthogonal to each other (to satisfy the first condition), so that \(\epsilon\) is chiral only at the North and South poles, where \(x^2=0,\infty\).\footnote{There cannot be chiral spinor fields on \(\mathbb{S}^4\) without zeros, because a chiral spinor defines an almost complex structure at each point, but \(\mathbb{S}^4\) has no almost complex structure. Since \(\mathbb{S}^4\) has constant scalar curvature, it can be proved that the conformal Killing condition on \(\epsilon\) actually implies that \(\epsilon\) is also a Killing spinor, \[
\nabla_\mu \epsilon = \mu \Gamma_\mu \epsilon \]
for some constant \(\mu\). This condition implies that the spinor is never zero, since it has constant norm. Thus \(\epsilon\) cannot be chiral. \cite{Zaffaroni-susy-curved-holography}}
%Under this restrictions, the spinors \(\epsilon(x)\) and \(\tilde{\epsilon}(x)\) can be rewritten as \(Spin(5)\) transformations of \(\epsilon(0)\) and \(\tilde{\epsilon}(0)\) respectively (see the article for the details). Explicitly, the gamma matrices are taken to be represented as
%\begin{equation}
%\Gamma^0 = \left(\begin{array}{cc}
%\mathds{1} & 0 \\ 0 & \mathds{1} \end{array} \right) \quad
%\Gamma^9 = \left(\begin{array}{cc}
%\mathds{1} & 0 \\ 0 & -\mathds{1} \end{array} \right) \quad
%\Gamma^i = \left(\begin{array}{cc}
%0 & E_i^T \\ E_i & 0 \end{array} \right) \quad
%\end{equation}
%where \((E_i)_{i=1,\cdots,8}\) are matrices representing left multiplication by octonions. In this representation, the constant spinor \(\hat{\epsilon}_s\) is taken as \((0_8,1,0_7)^T\).

The Wilson loop under consideration is of the type considered in \cite{Drukker_2007,Drukker_2008},
\begin{equation}
\label{pestun-wilson}
W_R(C) := \frac{1}{\dim R} \mathrm{Tr}_R\left(\mathcal{P}\exp{i\oint_C dt (A_\mu \dot{C}^\mu + |\dot{C}| \Phi_0 )} \right) 
\end{equation}
where \(C:[0,1]\to \mathbb{S}^4\) is an equatorial closed curve, parametrized in stereographic coordinates as \((x\circ C)(t) = 2r(\cos{(t)}, \sin{(t)},0,0)\), spanning a great circle of radius \(r\). Its tangent vector is \(\dot{C}(t) = 2r(-\sin(t),\cos(t),0,0)\), and the normalization \(|\dot{C}|=2r\) in front of \(\Phi_0\) is needed for the reparametrization invariance of the line integral. We argue now that this Wilson loop preserves some supersymmetry under the action of \(\delta_\epsilon\). In fact, its variation is proportional to
\begin{equation}
\delta_\epsilon W_R(C) \propto \epsilon \left( \Gamma_\mu \dot{C}^\mu + 2r \Gamma_0 \right) \Psi
\end{equation}
and for this to vanish for every value of the gaugino \(\Psi\), it must be that
\begin{equation}
\begin{aligned}
0 = &\epsilon \left( \Gamma_\mu \dot{C}^\mu + 2r \Gamma_0\right) \propto \left(\hat{\epsilon}_s + C^\mu\Gamma_\mu \hat{\epsilon}_c \right) \left( \Gamma_\mu \dot{C}^\mu + 2r \Gamma_0\right) \\
\Leftrightarrow  0 = &\sin(t)\left( -\hat{\epsilon}_s \Gamma_1 + 2r \hat{\epsilon}_c \Gamma_2 \Gamma_0 \right) + \cos(t) \left( \hat{\epsilon}_s \Gamma_2 + 2r \hat{\epsilon}_c \Gamma_1 \Gamma_0 \right) + \left( 2r  \hat{\epsilon}_c \Gamma_1\Gamma_2 + \hat{\epsilon}_s \Gamma_0\right)
\end{aligned}
\end{equation}
where we inserted the values for \(C^\mu, \dot{C}^\mu\) and simplified some trivial terms. For this to vanish at all \(t\),  the three parentheses have to vanish separately, giving the condition
\begin{equation}
\hat{\epsilon}_c = \frac{1}{2r} \Gamma_0 \Gamma_1\Gamma_2 \hat{\epsilon}_s .
\end{equation}
This condition halves the number of spinors that preserve the Wilson loop under supersymmetry, so this is called a \emph{1/2-BPS} operator.\footnote{This is just common terminology, that does not refer to any BPS condition between mass and central charges in the supersymmetry algebra (see \cite{dhoker-susy-notes}). It only means that the observable under consideration preserves half of the supercharges.} In the \(\mathcal{N}=4\) case, it preserves 16 supercharges.

If the mass of the hypermultiplet is turned on, the third condition in \eqref{pestun-14} has a non-zero solution for \(\hat{\epsilon}_s\) if \(\det(\Lambda - \Gamma_0 \Gamma_1\Gamma_2)=0\), that fixes \(\Lambda\) up to a sign.

\subsection{Quick localization argument}

Without considering the unphysical redundancy in field space given by the gauge symmetry of the theory, we can give a quick argument for the localization of the \(\mathcal{N}=4\) SYM, using the procedure outlined in Section \ref{sec:susy-loc}. We consider the \(U(1)\)-equivariant cohomology generated by the action of a fixed supersymmetry \(\delta_{\epsilon}\).\footnote{
\(\delta_\epsilon\) squares to the Poincaré algebra up to an gauge transformation, so we are really considering an \((U(1)\rtimes G)\)-equivariant cohomology, because \(\delta_\epsilon\)-closed equivariant forms are supersymmetric and gauge invariant observables. Ignoring the gauge fixing procedure, we are really not considering the complete field space, since the BRST procedure teaches us that in presence of a gauge symmetry this is automatically extended to include ghosts, that may contribute to the localization locus. It turns out that ghosts contribution is trivial, so the rough argument already gives the correct localization locus. Here we continue with this simplified procedure, and in the next section we are going to argue the above claim.
} 
Since \(\delta_\epsilon S^{\mathcal{N}=4}_{\mathbb{S}^4} = 0\) is equivariantly closed (off-shell) with respect to the variations \eqref{pestun-variationscopy}, we can perform the usual trick and add the localizing term 
\begin{equation}
\begin{aligned}
\lambda S_{loc} &:= \lambda \delta_\epsilon \mathcal{V} \\
\mathrm{where}\quad \mathcal{V}&:= \mathrm{Tr}\left(\Psi \overline{\delta_\epsilon \Psi} \right)
\end{aligned}
\end{equation}
where \(\overline{\delta_\epsilon \Psi}\) is defined by complex conjugation in the Euclidean signature,
\begin{equation}
\overline{\delta_\epsilon \Psi} = \frac{1}{2}\tilde{\Gamma}^{MN}F_{MN}\epsilon + \frac{1}{2} \Gamma^{\mu A}\Psi_A \nabla_\mu \epsilon - \sum_{i=1}^7 K_i \nu_i .
\end{equation}
The bosonic part of the localizing action is 
\begin{equation}
\left. S_{loc}\right|_{bos} = \mathrm{Tr}\left(\delta_\epsilon \Psi \overline{\delta_\epsilon \Psi} \right)
\end{equation}
that is positive semi-definite. The localization locus is then the subspace of fields such that
\begin{equation}
\label{pestun-8}
[\mathrm{fermions}]=0 ;\qquad \delta_\epsilon [\mathrm{fermions}]\ s.t.\ \left. S_{loc}\right|_{bos} = 0.
\end{equation}

The solution to \eqref{pestun-8} is found by inserting the relevant supersymmetry variation in \(\left. S_{loc}\right|_{bos}\), collecting the terms as a sum of positive semi-definite contributions and requiring them to vanish separately. Under the assumption of smooth gauge field, this is given by the field configurations such that, up to a gauge transformation (see \cite{pestun-article} for the details)
\begin{equation}
\label{pestun-locus}
\left\lbrace \begin{array}{ll}
A_\mu = 0 & \mu=1,\cdots,4 \\ 
\Phi_i = 0 & i=5,\cdots,9 \\
\Phi^E_0 = a\in \mathfrak{g} & \mathrm{constant} \\
K_i^E = -2(\nu_i\tilde{\epsilon})a & i=5,6,7 \\
K_I = 0 & I=1,\cdots,4
\end{array} \right. .
\end{equation}
So the physical sector of the theory localizes onto the zero-modes of \(\Phi_0\). If also singular gauge field configurations are allowed, \eqref{pestun-8} receives contributions from instanton solutions, where \(F_{\mu\nu} = 0\) everywhere except from the North or the South pole. These configurations can contribute non-trivially to the partition function. Computing the action on the smooth solutions one gets
\begin{equation}
S^{\mathcal{N}=4}_{\mathbb{S}^4}[a] = \frac{1}{g_{YM}^2}\int_{\mathbb{S}^4} d^4 x \sqrt{g} \mathrm{Tr}\left( \frac{2}{r^2}(\Phi^E_0)^2 + (K_i^E)^2 \right) = \frac{1}{g_{YM}^2} \mbox{vol}(\mathbb{S}^4) \frac{3}{r^2}\mathrm{Tr}\left(a^2\right) = \frac{8\pi^2 r^2}{g_{YM}^2}\mathrm{Tr}(a^2)
\end{equation}
where we used \(\mbox{vol}(\mathbb{S}^4) = \frac{8}{3}\pi^2 r^4\) and \((\nu_i\tilde{\epsilon})^2 = \frac{1}{4r^2}\). This last equation can be derived from the conditions \eqref{pestun-auxiliary-cond} and the form of the conformal Killing spinor \(\epsilon\). The action is given by constant field contributions, so the path integral is expected to be reduced to a finite-dimensional integral over the Lie algebra \(\mathfrak{g}\), of the form
\begin{equation}
Z \sim \int_{\mathfrak{g}} da\ e^{-\frac{8\pi^2 r^2}{g_{YM}^2}\mathrm{Tr}(a^2)} |Z_{inst}[a]|^2 Z_{1-loop}[a] .
\end{equation}
Here \(Z_{inst}\) is the instanton partition function, coming from the singular gauge field contributions of above. Since \(G\) is often considered to be a matrix group, this partition function is said to describe a \emph{matrix model}. The Wilson loop \eqref{pestun-wilson}, evaluated on this locus is given by
\begin{equation}
W_R(C) = \frac{1}{\dim R} \mathrm{Tr}_{R} e^{2\pi r a} .
\end{equation}
This type of matrix models can be approached by reducing the integration over \(\mathfrak{g}\) to an integration over its Cartan subalgebra \(\mathfrak{h}\) (we will discuss this better later, in Section \ref{subsec:MMN4}).\footnote{For finite-dimensional semi-simple complex Lie algebras, this is the maximal Abelian subalgebra. In general it is the maximal Lie subalgebra such that there exists a basis extension \(\mathfrak{h}\oplus span(e_\alpha)\cong \mathfrak{g}\), and it holds the eigenvalue equation \([h,e_\alpha] = \rho_\alpha(h)e_\alpha\) for any \(h\in\mathfrak{h}\) and a certain eigenvalue \(\rho_\alpha(h)\). \(\rho_\alpha:\mathfrak{h}\to\mathbb{C}\) are called \emph{roots} of \(\mathfrak{g}\).} Assuming the zero-mode \(a\in\mathfrak{h}\), we can conveniently rewrite the trace in the representation \(R\) as the sum of all the \emph{weights} \(\rho(a)\) of \(a\) in \(R\),\footnote{Analogously to the definition of root in the adjoint representation of \(\mathfrak{g}\), if the Lie algebra acts on the representation \(R\), the \emph{weight space} \(R_\rho\) of weight \(\rho:\mathfrak{h}\to\mathbb{C}\) is defined	 as the subspace of elements \(A\in R\) such that \(h\cdot A = \rho(h)A\).}
\begin{equation}
\label{pestun-weights}
W_R(C) = \frac{1}{\dim R} \sum_{\rho\in\Omega(R)}n(\rho) e^{2\pi r \rho(a)} ,
\end{equation}
where \(n(\rho)\) is the multiplicity of the weight \(\rho\), and \(\Omega(R)\) is the set of all the weights in the representation \(R\).

We finally notice that the same result for the localization locus works also for the \(\mathcal{N}=2\) and the \(\mathcal{N}=2^*\) theories, and both theories localize to the same matrix model. If the mass term for the hypermultiplet is considered, this can of course give a non-trivial contribution to the 1-loop determinant.

\subsection{The equivariant model}

As remarked at the end of Section \ref{sec:susy-loc}, in presence of a gauge symmetry the path integral has to be defined with respect to a gauge-fixed action. To do so, one has to enlarge the field space to include the appropriate Faddeev-Popov ghosts in a BRST complex with the differential \(\delta_{B}\). We consider then the total differential 
\begin{equation}
Q = \delta_\epsilon + \delta_{B}
\end{equation} 
where \(\epsilon\) is a fixed conformal Killing spinor that closes off-shell the superconformal algebra, so that the gauge-invariant SYM action is \(Q\)-closed. From the equivariant cohomology point of view, this operator is an equivariant differential with respect to the \( U(1)_{\epsilon} \rtimes G\) symmetry group acting on the enlarged field space. To gauge fix the path integral, following the BRST procedure with respect to the differential \(Q\), the action has to be extended as
\begin{equation}
S_{phys}[A,\Psi,K,ghosts] = S_{SYM}[A,\Psi,K] + Q\mathcal{O}_{g.f.}[A,ghosts]
\end{equation}
with a gauge-fixing fermion \(\mathcal{O}[A,ghosts]\). Upon path integration over ghosts, this new term has to give the gauge-fixing action and Fadee-Popov determinant. The localization principle is then exploited augmenting again the action with a \(Q\)-exact term, \(\lambda Q\mathcal{V}\) with again
\begin{equation}
\mathcal{V} := \mathrm{Tr}(\Psi \overline{\delta_\epsilon\Psi}).
\end{equation}
This effectively gives the same localization term of the previous paragraph, since \(\mathcal{V}\) is gauge-invariant.

The BRST-like complex considered in \cite{pestun-article} is given by the following ghost and auxiliary field extension. The ghost \(c\), anti-ghost \(\tilde{c}\) and standard Lagrange multiplier for the \(R_\xi\)-gauges \(b\) (\lq\lq Nakanishi-Lautrup\rq\rq\ field) are introduced, respectively odd, odd and even with respect to Grassmann parity. Since the path integral is expected to localize on zero-modes, constant fields \(c_0,\tilde{c}_0\) (odd) and \(a_0,\tilde{a}_0,b_0\) (even) are also introduced. On the original fields of the SYM theory, the BRST differential acts as a gauge transformation parametrized by \(c\). On the gauge field \(A_\mu\)
\begin{equation}
\delta_{B}A_\mu = -[c,D_\mu].
\end{equation}
On ghosts and zero-modes the BRST transformation is defined by
\begin{equation}
\begin{array}{llll}
\delta_{B} c = -a_0 - \frac{1}{2}[c,c] & \delta_{B}\tilde{c} = b & \delta_{B}\tilde{a}_0 = \tilde{c}_0 & \delta_{B}b_0 = c_0 \\
\delta_{B}a_0 = 0 & \delta_{B}b = [a_0,\tilde{c}] & \delta_{B}\tilde{c}_0  = [a_0,\tilde{a}_0] & \delta_{B}c_0 = [a_0,b_0]
\end{array}
\end{equation}
and its square generates a gauge transformation with respect to the (bosonic) constant field \(a_0\),
\begin{equation}
\delta_B^2 = [a_0, \cdot] .
\end{equation}

The supersymmetry complex, constituted by the original fields, is reparametrized with respect to the basis \(\lbrace \Gamma^M \epsilon, \nu^i\rbrace\), with \(M=1,\cdots,9; i=1,\cdots,7\), of the 10-dimensional Majorana-Weyl bundle over \(\mathbb{S}^4\). Expanding \(\Psi\) over such a basis we have
\begin{equation}
\Psi = \sum_{M=1}^9 \Psi_M (\Gamma^M\epsilon) + \sum_{i=1}^7 \Upsilon_i \nu^i
\end{equation}
and the superconformal transformations are rewritten as 
\begin{equation}
\begin{aligned}
&\delta_\epsilon A_M = \Psi_M \\
&\delta_\epsilon \Psi_M = -(\mathcal{L}_{v} + R + G_\Phi)A_M \\
&\delta_\epsilon \Upsilon_i = H_i \\
&\delta_\epsilon H_i = -(\mathcal{L}_{v} + R + G_\Phi) \Upsilon_i ,
\end{aligned}
\end{equation}
where 
\begin{equation}
H_i := K_i + 2(\nu_i \tilde{\epsilon})\Phi_0 + \frac{1}{2}F_{MN}\nu_i \Gamma^{MN} \epsilon + \frac{1}{2}\Phi_A \nu_i \Gamma^{\mu A} \nabla_\mu \epsilon .
\end{equation}
With this field redefinition, we see that the supersymmetry transformations can be schematized in the form
\begin{equation}
\delta_\epsilon X = X' \qquad \delta_\epsilon X' = [\phi + \epsilon, X] \qquad \left( \delta_\epsilon \phi = 0 \right)
\end{equation}
where \(\phi:=-\Phi=v^M A_M\), \([\phi,X']:= -G_{\Phi}X'\) denotes a gauge transformation, \([\epsilon, X']:= -(\mathcal{L}_v + R)X'\) denotes a Lorentz transformation. Here \(X=(A_M(x),\Upsilon_i(x))\), \(X'=(\Psi_M(x), H_i(x))\) are the coordinates in the super-loop space interpretation of the supersymmetric model, of opposite statistics. As we pointed out in Section \ref{sec:susy-loc}, we espect every Poincaré-supersymmetric theory to have such super-loop equivariant structure, and this is an example of the fact that in higher dimensional QFT the reparametrization of the fields necessary to make this apparent can be non-trivial. Indeed, the loop space coordinates \(X\) and the corresponding 1-forms \(X'\) mix the bosonic/fermionic field components of the original parametrization! 

\medskip
Combining the two complexes, and giving supersymmetry transformation properties to the ghost sector, the equivariant differential \(Q\) is taken to act as
\begin{equation}
\begin{array}{lll}
QX = X' - [c,X] & Qc = \phi - a_0 - \frac{1}{2}[c,c] & Q\tilde{c}=b  \\
QX' = [\phi + \epsilon, X] - [c,X'] & Q\phi = - [c,\phi+\epsilon] & Qb = [a_0 + \epsilon, \tilde{c}] \\
Q\tilde{a}_0 = \tilde{c}_0 &  Qb_0 = c_0 & \\
Q\tilde{c}_0 = [a_0,\tilde{c}_0] & Qc_0 = [a_0,b_0]. & 
\end{array}
\end{equation}
Moreover, \(Qa_0 = Q\epsilon = 0 \). This differential squares to a constant gauge transformation generated by \(a_0\) and the Lorentz transformation generated by \(\epsilon\),
\begin{equation}
Q^2 = [a_0 + \epsilon, \cdot ].
\end{equation}
Notice that to make explicit the super-loop structure when the combined complex is taken into account, one needs another non-trivial reparametrization of the fields,
\begin{equation}
\tilde{X}':= X' - [c,X] \qquad \tilde{\phi} := \phi - a_0 -\frac{1}{2}[c,c].
\end{equation}
This makes the tranformations look like
\begin{equation}
Q(\mbox{field}) = \mbox{field}' \qquad Q(\mbox{field}') = [a_0+\epsilon, \mbox{field}]
\end{equation}
and the new pairs of coordinate/1-form in the extended super-loop space are \((c,\tilde{\phi})\), \((\tilde{c}, b)\), \((\tilde{a}_0, \tilde{c}_0)\), \((b_0,c_0)\).

The gauge-fixing term considered for the extended quantum action is, schematically
\begin{equation}
S_{g.f.} = \int_{\mathbb{S}^4} Q\left( \tilde{c}\left( \nabla^\mu A_\mu + \frac{\xi_1}{2}b + b_0\right) - c\left( \tilde{a}_0 - \frac{\xi_2}{2}a_0\right) \right)
\end{equation}
where the bilinear product in \(\mathfrak{g}\) is suppressed in the notation, assuming contraction of Lie algebra indices. Upon integration of the auxiliary field, this term produces the usual gauge fixing term for the Lorentz gauge \(\nabla^\mu A_\mu = 0\), and the ghost term of the action. Moreover, the path integral is independent of the parameters \(\xi_1, \xi_2\) (we refer to \cite{pestun-article} for the proof). We finally claim that the localization principle for the gauge-fixed theory remains the same, with the additional condition of vanishing ghosts in the localization locus, and identifying the zero-mode of \(\Phi_0\) with \(a_0\). In fact from the gauge-fixing term, 
\begin{equation}
S_{g.f.}\supset - \int_{\mathbb{S}^4} \left( \phi - a_0 - \frac{1}{2}[c,c]\right) \tilde{a}_0
\end{equation}
and integrating over \(\tilde{a}_0\), we have the condition \(\phi=a_0+ \frac{1}{2}[c,c]\), that in the localization locus where \(c=0\) and \(\phi=-v^M A_M = \Phi_0\), becomes precisely \(\Phi_0 = a_0\).

\subsection{Localization formulas}

We stated that the path integral localizes (apart from instanton corrections) to the zero-modes of the bosonic constant \(a_0\in \mathfrak{g}\), that correspond to the zero-modes of \(\Phi_0\). For the same reason of the scalar field corresponding to the reduced time-direction of the (9,1)-theory, we integrate over immaginary \(a_0 = ia_0^E\), where \(a_0^E\) is real. The application of the localization principle is now straightforward in principle, although very cumbersome in practice. In particular, integrating out the Gaussian fluctuations around the localization locus, the arising one-loop determinant in the partition function results of the form
\begin{equation}
Z_{1-loop} = \left(\frac{\det{K_f}}{\det{K_b}} \right)^{1/2}
\end{equation}
where \(K_f, K_b\) are the kinetic operators acting on the fermionic and bosonic fluctuation modes after the usual expansion of \(Q\mathcal{V}\). This factor requires in general a regularization, and it has been computed for the \(\mathcal{N}=2, \mathcal{N}=2^*\) and \(\mathcal{N}=4\) theory, using an appropriate generalization of the Atiyah-Singer theorem seen in Section \ref{sec:index} applied to transversally elliptic operators. The instanton partition functions have been also simplified for the theories under consideration. We refer to  \cite{nekrasov-instantons,pestun-article,okuda-instantonsN4} for the explicit form of instanton contributions in the cases of \(\mathcal{N}=2,2^*\).
%To consider instanton contributions to the YM action, one cannot ignore the (up to now ignored indeed) topological \(\theta\)-term, and have to consider the action
%\begin{equation}
%S_{YM} \mapsto S_{YM} - \frac{i\theta}{8\pi} \mathrm{Tr}\int F\wedge F .
%\end{equation}
%For topologically trivial solutions this new term is zero, but in general 
%\begin{equation}
%\frac{1}{8\pi^2}\mathrm{Tr}\int F\wedge F = k\in\mathbb{Z}
%\end{equation}
%is the \emph{charge} (the second Chern class) of the curvature \(F\). For instantons of different charges, one has different contributions to \(Z_inst[a]\), and to obtain the final result one should sum all these contributions. The YM action computed on the instanton solution (we take the self-dual case \(\star F = F\)) is non-zero,
%\begin{equation}
%S_{YM} = \mathrm{Tr}\int 
%\end{equation}

For the maximally supersymmetric \(\mathcal{N}=4\) SYM theory, the results for the 1-loop determinant and the instanton partition function are of the very simple form
\begin{equation}
Z_{1-loop}^{\mathcal{N}=4} = 1 , \qquad Z_{inst}^{\mathcal{N}=4} = 1 ,
\end{equation}
so that the resulting localization formulas for the partition function and the expectation value of the supersymmetric Wilson loop presented before become
\begin{equation}
\label{pestun-MMN4}
\begin{aligned}
Z_{\mathbb{S}^4} &= \frac{1}{\mbox{vol}(G)}\int_{\mathfrak{g}} da\ e^{-\frac{8\pi^2 r^2}{g_{YM}^2}\mathrm{Tr}(a^2)} ,\\
\langle W_R(C) \rangle &= \frac{1}{\dim R} \frac{1}{Z\ \mbox{vol}(G)} \int_{\mathfrak{g}} da\ e^{-\frac{8\pi^2 r^2}{g_{YM}^2}\mathrm{Tr}(a^2)} \mathrm{Tr}_R\left( e^{2\pi r  a} \right) .
\end{aligned}
\end{equation}
This result proved a previous conjecture, based on a perturbative analysis by Erickson-Semenoff-Zarembo \cite{erickson-wilsonloop}. Their calculation for \(\langle W_R(C)\rangle\) with \(G=U(N)\) showed that the Feynman diagrams with internal vertices cancel up to order \(g^4N^2\), and that the sum of all ladder diagrams (planar diagrams with no internal vertices) exponentiate to a matrix model. The result of this exponentiation gives an expectation value that coincides with the strong-coupling prediction of the AdS/CFT correspondence for \(\mathcal{N}=4\) SYM,\footnote{This \lq\lq correspondence\rq\rq\ conjectures a duality between the \(\mathcal{N}=4\) SYM in 4 dimensions and type IIB superstring theory in an \(AdS_5\times \mathbb{S}^5\) background. In particular, when the parameters of the gauge theory are taken to be such that \(N\to\infty\) and \(g_{YM}^2 N \to \infty\) (namely, in the planar and strong ’t Hooft coupling limit), \(\mathcal{N}=4\) SYM is dual to classical type IIB supergravity on \(AdS_5 \times \mathbb{S}^5\) and the computation of the Wilson loop in this limit is mapped to the evaluation of a minimal surface in this space \cite{Maldacena-wilson-loops,Rey_2001}.}
%It relaes the 't Hooft limit \(N\to\infty\), where \(g_{YM}^2N\) (that can also large) is kept fixed, with the low energy, weakly coupled limit of the string theory, that can be effectively described by a (classical) supergravity theory. Perturbative computations in the supergravity theory should thus match the strongly-coupled regime of the SYM theory.}
thus they conjectured that the diagrams with vertices have to vanish at all orders. Later this conjecture was supported by Drukker-Gross \cite{drukker-wilsonloop}, and finally proven with the exact localization technique described above. 
%In the next section, we will see an example of actual computation of such a matrix model in a different supersymmetric theory.

We quote now the results for the 1-loop determinants in the \(\mathcal{N}=2,2^*\) theories. For this, it is useful to introduce the notation
\begin{equation}
\label{eq:notation-detR}
\begin{aligned}
\mbox{det}_R f(a) &:= \prod_\rho f(\rho(a)) \\
H(z) &:= e^{-(1+\gamma)z^2} \prod_{n=1}^\infty \left(1-\frac{z^2}{n^2} \right)^n \prod_{n=1}^\infty e^{z^2/n}
\end{aligned}
\end{equation}
with \(\rho\) running over the weights of \(R\) (if \(R=Ad\), the weights are the roots of \(\mathfrak{g}\)), \(\gamma\) being the Euler-Mascheroni constant. Let also be \(m^2:=\frac{1}{4} M_{ij}M^{ij}\), and recall that \(m\) (as well as \(a_0\)) should take immaginary values. From \cite{pestun-article} we have
\begin{align}
Z_{1-loop}^{\mathcal{N}=2^*}[a_0;M] &= \exp{\left(- r^2m^2 \left( (1+\gamma)-\sum_{n=1}^\infty\frac{1}{n} \right) \right)} 
\mbox{det}_{Ad}\left[ \frac{ H(ra_0)}{\left[ H(r(a_0+m)) H(r(a_0-m)) \right]^{-1/2}}\right] , \\
Z_{1-loop}^{\mathcal{N}=2, pure}[a_0] &= \mbox{det}_{Ad}H(ra_0) , \\
Z_{1-loop}^{\mathcal{N}=2, W}[a_0] &= \frac{\det_{Ad} H(ra_0)}{\det_W H(ra_0)} ,
\end{align}
where the first result is for the massive \(\mathcal{N}=2^*\) theory, the second one is derived putting \(m=0\) in the first line, and describes the pure \(\mathcal{N}=2\) SYM, the third  one is for the matter-coupled theory to a massles hypermultiplet in the representation \(W\). Notice that the exponential prefactor in the first line diverges, but is independent of \(a_0\), and thus simplifies in ratios during the computation of expectation values. Also, the third line holds literally if the (\(a_0\)-independent) divergent factors are the same for the vector and the hypermultiplet.

\subsection[The Matrix Model for N=4 SYM]{The Matrix Model for \(\mathcal{N}=4\) SYM}
\label{subsec:MMN4}

As an example, we include here an explicit computation for the Gaussian matrix model \eqref{pestun-MMN4} in the case of \(\mathcal{N}=4\) SYM \cite{gomis-tHooft,marino-MM,drukker-wilsonloop}. We will take in particular the case of the compact matrix group \(G=U(N)\) with the Wilson loop in the fundamental representation \(R=\mathbf{N}\), but first we analyze generically how to simplify such an integration over the Lie algebra \(\mathfrak{g}\). We normalize the invariant volume element \(da\) on \(\mathfrak{g}\) such that 
\begin{equation}
\label{pestun-measure}
\int_\mathfrak{g}da\ e^{-\frac{2}{\xi^2} \mathrm{Tr}(a^2)} = \left( \frac{\xi^2\pi}{2} \right)^{\dim(G)/2}
\end{equation}
for any parameter \(\xi\). In the \(U(N)\) case, \(\mathfrak{g}= \mathfrak{u}(N) = \lbrace\text{Hermitian}\ N\times N\ \text{matrices}\rbrace\), so
this means taking
\begin{equation}
da = 2^{N(N-1)/2} \prod_{i=1}^N da_{ii} \prod_{1\leq j<i\leq N} d\text{Re}(a_{ij}) d\text{Im}(a_{ij}) .
\end{equation}
Setting  \(\xi^2 := g_{YM}^2/(4\pi^2 r^2)\) we have
\begin{equation}
Z_{\mathbb{S}^4} = \frac{1}{\mbox{vol}(G)}\left( \frac{g_{YM}^2}{8\pi r^2} \right)^{\dim(G)/2} .
\end{equation}

To simplify the integration of the Wilson loop expectation value, we notice that the matrix model has a leftover gauge symmetry under constant gauge transformations, since both the measure and the traces are invariant under the adjoint action of \(G\). Thus we can \lq\lq gauge-fix\rq\rq\ the integrand to depend only on the Cartan subalgebra \(\mathfrak{h}\subset \mathfrak{g}\), setting
\begin{equation}
\label{MM-1}
a = Ad_{g*}(X)
\end{equation}
for some \(g\in G/T\) and \(X\in \mathfrak{h}\), with \(T\) being the maximal torus in \(G\) generated by \(\mathfrak{h}\). %\footnote{Here to be precise we need to assume \(\mathfrak{g}\) to be semi-simple.} 
There is more than one \(X\) related to \(a\) by conjugation, but they are related via the action of the \emph{Weyl group} of \(G\), that we call \(\mathcal{W}\). Taking this into account, after the gauge-fixing we can perform the integral over the orbits obtaining a volume factor
\begin{equation}
\frac{\mbox{vol}(G/T)}{|\mathcal{W}|}.
\end{equation}
The gauge-fixing can be done with the usual Faddeev-Popov (FP) procedure, that is inserting the unity decomposition
\begin{equation}
1=\int dg\ \Delta^2(X)\delta(F(a^{(g)})) ,
\end{equation}
where the delta-function fixes the condition \eqref{MM-1}, and the FP determinant is given by 
\begin{equation}
\Delta(X)^2 = \prod_\alpha |\alpha(X)| = \prod_{\alpha >0} \alpha(X) ,
\end{equation}
where \(\alpha:\mathfrak{h}\to\mathbb{C}\) are the roots of \(\mathfrak{g}\), and in the second equality we used that roots come in pairs \((\alpha,-\alpha)\). We can rewrite the expectation value of the circular Wilson loop as
\begin{equation}
\begin{aligned}
\langle W_R(C) \rangle &= \left( \frac{\xi^2\pi}{2} \right)^{\dim(G)/2} \frac{\mbox{vol}(G/T)}{|\mathcal{W}|\dim R} \int_\mathfrak{h} dX\ \Delta(X)^2 e^{-\frac{2}{\xi^2} \mathrm{Tr}(X^2)} \mathrm{Tr}_R\left( e^{2\pi r  X} \right) \\
&= \left( \frac{\xi^2\pi}{2} \right)^{\dim(G)/2} \frac{\mbox{vol}(G/T)}{|\mathcal{W}|\dim R} \sum_{\rho\in\Omega(R)} n(\rho) \int_\mathfrak{h} dX\ \Delta(X)^2 e^{-\frac{2}{\xi^2} \mathrm{Tr}(X^2)}  e^{2\pi r \rho( X)}
\end{aligned}
\end{equation}
where in the second line we used \eqref{pestun-weights}.

We specialize now to the case \(G=U(N)\),  \(\mathfrak{h}=\lbrace X=\mathrm{diag}(\lambda_1,\cdots,\lambda_N) | \lambda_i \in \mathbb{R}\rbrace\), and we take \(r=1\). The adjoint action  of \(U(N)\) is the conjugation \(a\mapsto g a g^\dagger\), so the FP determinant is defined by
\begin{equation}
1 = \int dg\ \Delta(X)^2 \prod_{ij} \delta((g X g^\dagger)_{ij}) ,
\end{equation}
that imposes the off-diagonal terms to vanish in the given gauge. Expressing \(g=e^M\) with \(M\in \mathfrak{u}(N)\),
\begin{equation}
\Delta(X)^2 = \prod_{ij} \det_{kl}\left| \frac{\delta (e^MXe^{-M})_{ij}}{\delta M_{kl}} \right| = \prod_{ij} \det_{kl}\left| \delta_{ki}\delta_{lj}( \lambda_j - \lambda_i) \right| = \prod_{i>j} (\lambda_i-\lambda_j)^2
\end{equation}
is the so called \emph{Vandermonde determinant}, that can be related to the following matrix 
\begin{equation}
\label{MM-vanderm}
\Delta(\lambda) = \det||\lambda_i^{j-1}|| = 
\det\left(\begin{array}{ccccc}
1 & \lambda_1 & \lambda_1^2 & \cdots & \lambda_1^{N-1} \\
1 & \lambda_2 & \lambda_2^2 & \cdots & \lambda_2^{N-1} \\
 & & \vdots & & \\
1 & \lambda_N & \lambda_N^2 & \cdots & \lambda_N^{N-1} \\
\end{array} \right) .
\end{equation}
The partition function can thus be expressed as
\begin{equation}
Z_{\mathbb{S}^4} = \frac{1}{N!} \frac{1}{(2\pi)^N} \int \left(\prod_i d\lambda_i\right) \left(\prod_{i>j} (\lambda_i-\lambda_j)^2\right) e^{-\frac{2}{\xi^2} \sum_i \lambda_i^2} ,
\end{equation}
where \(N!\) is the order of the Weyl group \(\mathcal{W}=S_N\) and \((2\pi)^N\) is the volume of the \(N\)-torus \(U(1)^N\),
while the Wilson loop in the fundamental representation inserts in the path integral a factor
\begin{equation}
\label{MM-5}
\frac{1}{N}\sum_{j=1}^N e^{2\pi  \lambda_j} .
\end{equation}
There are two main approaches to the evaluation of this matrix model and the computation of the Wilson loop expectation value, at least in the limit \(N\to\infty\).

\paragraph*{\(1^{st}\) method: saddle-point}

The first method that we present is based on a suitable saddle-point approximation in the large-\(N\) limit. To see the possibility for this interpretation, we rewrite the partition function as 
\begin{equation}
\label{MM-partition}
\begin{aligned}
Z_{\mathbb{S}^4} &= \frac{1}{N!} \int \prod_i \frac{d\lambda_i}{2\pi}\ e^{-N^2S_{eff}(\lambda)} \\
\text{with}\quad S_{eff}(\lambda) &:= \frac{8\pi^2}{t N}\sum_{i=1}^N \lambda_i^2 - \frac{2}{N^2}\sum_{i>j} \log|\lambda_i-\lambda_j| ,
\end{aligned}
\end{equation}
where \(t:= g_{YM}^2 N\) is the 't Hooft coupling constant. This can be viewed as an effective action of a zero-dimensional QFT describing \(N\) sites (the eigenvalues \(\lambda_i\)), where the first piece is a \lq\lq one-body\rq\rq\ harmonic potential, and the second one is a repulsive \lq\lq two-body\rq\rq\ interaction. Notice that every sum is roughly of order  \(\sim N\), so \(S_{eff}\sim O(1)\) in \(N\). The limit \(N\to\infty\), with \(t\) fixed, can be regarded as a semi-classical approximation (we could compare it to \lq\lq \(1/\sqrt{\hbar}\to \infty\)\rq\rq), and in that limit we can solve the integral using a saddle-point approximation.
The saddle points are those values of \(\lambda_i\) that solve the classical EoM
\begin{equation}
\label{MM-2}
0=\frac{\delta S_{eff}}{\delta \lambda_i} \qquad \Rightarrow \qquad 0 = \frac{16\pi^2}{tN} \lambda_i - \frac{2}{N^2}\sum_{j\neq i} \frac{1}{\lambda_i-\lambda_j} .
\end{equation}
In the large-\(N\) limit we can study this equation in the continuum approximation, assuming the eigenvalues \(\lambda_i\) to take values in a compact interval \(I=[a,b]\), so that the (normalized) \emph{eigenvalue distribution}
\begin{equation}
\rho(\lambda) = \frac{1}{N}\sum_{i=1}^N \delta(\lambda-\lambda_i)
\end{equation}
is regarded as a continuous function of compact support on \(I\). Then every sum can be replaced by an integration over the reals,
\begin{equation}
\frac{1}{N}\sum_{i=1}^N f(\lambda_i) \to \int d\lambda\ f(\lambda)\rho(\lambda) ,
\end{equation}
and \eqref{MM-2} becomes
\begin{equation}
\label{MM-3}
\frac{8\pi^2}{t} \lambda = \mathcal{P}\int \frac{\rho(\lambda')d\lambda'}{\lambda-\lambda'} ,
\end{equation}
where we took the principal value of the integral to avoid the pole at \(\lambda_i = \lambda_j\). This is an integral equation in \(\rho(\lambda)\), whose solution gives the distribution of the eigenvalues at the saddle-point locus of the partition function.

It is useful to introduce an auxiliary function on the complex plane, the \lq\lq resolvent\rq\rq\
\begin{equation}
\omega(z) := \int \frac{\rho(\lambda) d\lambda}{z-\lambda} ,
\end{equation}
that has three important properties for our purposes:
\begin{enumerate}[label=(\roman*)]
\item it is analytic on \(\mathbb{C}\setminus I\), since there are poles for \(z=\lambda\) when \(z\in I\);
\item thanks to the normalization of \(\rho\), asymptotically for \(|z|\to\infty\) it goes as \(\omega(z)\sim \frac{1}{z}\);
\item using the residue theorem and the delta-function representation
\begin{equation}
\frac{\epsilon}{z^2 + \epsilon^2} \xrightarrow{\epsilon\to 0^+} \pi \delta(z) ,
\end{equation}
it relates to the eigenvalue distribution by the discontinuity equation
\begin{equation}
\label{MM-4}
\rho(\lambda) = -\frac{1}{2\pi i}\lim_{\epsilon\to 0^+} \left[ \omega(\lambda+i\epsilon) - \omega(\lambda-i\epsilon) \right] .
\end{equation}
\end{enumerate}
Knowing the resolvent we can easily compute the eigenvalue distribution by this last property, so we rewrite the saddle-point equation in terms of it.
%and in fact we can rewrite equation \eqref{MM-3}, using the formula (\emph{Sokhotski–Plemelj theorem})
%\begin{equation}
%\label{MM-SP}
%\mathcal{P}\int \frac{f(z)}{z}dz = \lim_{\epsilon\to 0^+} \frac{1}{2}\left( \int \frac{f(z)}{z+i\epsilon}dz + \int \frac{f(z)}{z-i\epsilon}dz \right) .
%\end{equation}
%Breaking the principal value and inserting the definition of the resolvent, we have
%\begin{equation}
%-\frac{16\pi^2}{t} = \lim_{\epsilon\to 0^+} \left[ \omega(\lambda+i\epsilon) + \omega(\lambda-i\epsilon) \right] .
%\end{equation}
To compute \(\omega\), we can start again from \eqref{MM-3}, multiply by \(1/(\lambda-z)\) and integrate over \(\lambda\) with the usual measure \(\rho(\lambda)d\lambda\):
\begin{equation}
\frac{8\pi^2}{t} \int d\lambda\ \rho(\lambda)\frac{\lambda}{\lambda-z} = \int d\lambda \frac{\rho(\lambda)}{\lambda-z}\ \mathcal{P}\int d\lambda' \frac{\rho(\lambda')}{\lambda-\lambda'} .
\end{equation}
We can add \(\pm z\) at the numerator of the LHS, and use the formula  (\emph{Sokhotski–Plemelj theorem})
\begin{equation}
\label{MM-SP}
\mathcal{P}\int \frac{f(z)}{z}dz = \lim_{\epsilon\to 0^+} \frac{1}{2}\left( \int \frac{f(z)}{z+i\epsilon}dz + \int \frac{f(z)}{z-i\epsilon}dz \right) 
\end{equation}
to break the principal value on the RHS. Inserting the definition of the resolvent and using the residue theorem, this gives
\begin{equation}
\frac{8\pi^2}{t}-\frac{8\pi^2}{t}\lambda \omega(\lambda) = -\frac{1}{2}\omega(\lambda)^2 ,
\end{equation}
that is solved for
\begin{equation}
\omega(\lambda) = \frac{8\pi^2}{t}\left( \lambda \pm \sqrt{\lambda^2 - \frac{t}{4\pi^2}}\right) .
\end{equation}
In order to match the right asymptotic behavior \(\omega(z\to\infty)\sim 1/z\), we have to chose the minus sign. With this choice, we can compute the saddle-point eigenvalue distribution using the discontinuity equation \eqref{MM-4},
\begin{equation}
\begin{aligned}
\rho(\lambda) &= -\frac{1}{2\pi i}\frac{8\pi^2}{t} \lim_{\epsilon\to 0^+} \left[ \omega(\lambda+i\epsilon) - \omega(\lambda-i\epsilon) \right]\\
&= \frac{4\pi}{it} \lim_{\epsilon\to 0^+} \left[ \sqrt{\lambda^2-\frac{t}{4\pi^2}+2i\epsilon\lambda} - \sqrt{\lambda^2 - \frac{t}{4\pi^2} - 2i\epsilon\lambda} \right]  \\
&= \frac{4\pi}{it} \left(2\sqrt{\lambda^2-\frac{t}{4\pi^2}}\right) \\
&= \frac{8\pi}{t}\sqrt{\frac{t}{4\pi^2}- \lambda^2}
\end{aligned}
\end{equation}
where we used that the principal square root has a branch cut on the real line. This function is called \emph{Wigner semi-circle distribution}, it has support on the interval \(I=[-\sqrt{t}/2\pi, \sqrt{t}/2\pi]\), and here it is correctly normalized to 1.

Now that we have the saddle-point locus in terms of the eigenvalue distribution, we can compute the expectation value for the circular Wilson loop in the fundamental representation. Since the exponential factor \eqref{MM-5} is of order \(\sim N^0\), this does not contribute to the saddle-point equation in the \(N\to\infty\) limit. We can thus still use the Wigner distribution at zero-order in \(1/N^2\), and insert in the path integral the trace in the continuum limit,
\begin{equation}
\label{MM-9}
\begin{aligned}
\langle W_{\mathbf{N}}(C)\rangle &= \int d\lambda\ \langle \rho(\lambda)\rangle e^{2\pi \lambda} \\
&= \frac{8\pi}{t} \int_{-\sqrt{t}/2\pi}^{\sqrt{t}/2\pi} d\lambda\ e^{2\pi \lambda} \sqrt{\frac{t}{4\pi^2}-\lambda^2}  +  O\left(1/N^2\right) \\
&= \frac{2}{\sqrt{t}} I_1\left(\sqrt{t}\right) + O\left(1/N^2\right)
\end{aligned}
\end{equation}
where \(I_1(z)\) is a modified Bessel function of the first kind. In the weak and strong coupling limits \(t\gg,\ll 1\) the expectation value gives
\begin{align}
t\ll 1 : \qquad \langle W_{\mathbf{N}}(C)\rangle &\sim 1 + \frac{t^2}{8} + \frac{t^4}{192} + \cdots \\
t\gg 1 : \qquad \langle W_{\mathbf{N}}(C)\rangle &\sim \sqrt{\frac{2}{\pi}} t^{-3/4} e^{\sqrt{t}} ,
\end{align}
so it explodes in the strong coupling limit, with an \emph{essential} singularity.\footnote{Interestingly, the strong coupling limit can be checked independently using holography, where Wilson loops are given by minimal surfaces in AdS \cite{Maldacena-wilson-loops,Rey_2001}.}

\paragraph*{\(2^{nd}\) method: orthogonal polynomials}

Another technique to solve matrix models involve the use of orthogonal polynomials \cite{drukker-wilsonloop}. Our starting point is again the partition function,
\begin{equation}
Z = \frac{1}{N!}\int \prod_{i=1}^N \left( \frac{d\lambda_i}{2\pi} e^{-\frac{8\pi^2 N}{t}\lambda_i^2}\right) \Delta(\lambda)^2 .
\end{equation}
Introducing the \(L^2(\mathbb{R})\) measure
\begin{equation}
d\mu(x) := dx\ e^{-\frac{8\pi^2 N}{t}x^2} ,
\end{equation}
we can write the partition function as
\begin{equation}
\label{MM-6}
Z = \frac{1}{N!} \int \prod_{i=1}^N d\mu(\lambda_i) \Delta(\lambda)^2 . 
\end{equation}
Recalling that the Vandermonde determinant is evaluated from the matrix \eqref{MM-vanderm}, expressed in terms of the polynomials \(\lbrace 1,x,x^2,\cdots\rbrace\), we notice that we can equivalently express it in terms of another set of \emph{monic} polynomials,
\begin{equation}
p_k(x) = x^k + \sum_{j=0}^{k-1} a_j^{(k)}x^j
\end{equation}
since by elementary row operations
\begin{equation}
\Delta(\lambda) = \det||\lambda_i^{j-1}|| = \det||p_{j-1}(\lambda_i)||.
\end{equation}
It is useful to chose the set \(\lbrace p_k \rbrace_{k\geq 0}\) to be \emph{orthogonal} with respect to the matrix model measure,
\begin{equation}
\int d\mu(\lambda) p_n(\lambda)p_m(\lambda) = h_n \delta_{nm} 
\end{equation}
since the knowledge of this set, and in particular of the normalization constants \(h_n\), allows to compute the partition function. Writing the determinant as \[\Delta(\lambda) = \sum_{\sigma\in S_N} (-1)^{\mbox{sign}(\sigma)}\prod_{k=1}^N p_{\sigma(k)-1}(\lambda_k),\] then \eqref{MM-6} reduces to
\begin{equation}
Z = \prod_{k=0}^{N-1} h_k .
\end{equation}

In our case the matrix model is Gaussian, and the corresponding set of orthogonal polynomials are the \emph{Hermite polynomials},
\begin{equation}
H_n(x) := e^{x^2} \left(-\frac{d}{dx}\right)^n e^{-x^2} , \qquad \int_{-\infty}^{+\infty} dx\ e^{-x^2}H_n(x) H_m(x) = \delta_{nm} 2^n n! \sqrt{\pi}
\end{equation}
so, normalizing \(h_n=1\) and inserting the correct prefactors, we consider the set of \emph{orthonormal} polynomials with respect to the measure \(d\mu(\lambda)\)
\begin{equation}
P_n(\lambda) :=\sqrt{\sqrt{\frac{8\pi N}{t}} \frac{1}{2^n n!}}\ H_n\left(\frac{\sqrt{8\pi^2 N}}{t}\lambda\right) .
\end{equation}

The expectation value of any observable of the type \(\mathrm{Tr}(f(X)) = \sum_k f(\lambda_k)\) can be simplified as
\begin{equation}
\begin{aligned}
\langle \mathrm{Tr}f(X)\rangle &= \frac{1}{N!Z}\int \left(\prod_{i=1}^N d\mu(\lambda_i) \right)\Delta(\lambda)^2 \sum_{k=1}^N f(\lambda_k) \\
& \begin{aligned} = \frac{1}{N!} \sum_k \sum_{\sigma\in S_N} \int d\mu(\lambda_1)\ P_{\sigma(1)-1}(\lambda_1)^2 \cdots \int d\mu(\lambda_k)\ P_{\sigma(k)-1}(\lambda_k)^2f(\lambda_k) \cdots & \\
\cdots \int d\mu(\lambda_N)\ P_{\sigma(N)-1}(\lambda_N)^2 &
\end{aligned} \\
&= \sum_{j=0}^{N-1}\int d\mu(\lambda)\ P_j(\lambda)^2 f(\lambda) .
\end{aligned}
\end{equation}
Applying this formula to the expectation value of the circular Wilson loop in the fundamental representation we have
\begin{equation}
\label{MM-7}
\begin{aligned}
\langle W_{\mathbf{N}}(C)\rangle &= \frac{1}{N}\left\langle \mathrm{Tr}\exp(2\pi X)\right\rangle \\
&= \frac{1}{N}\sum_{j=0}^{N-1} \int d\lambda\ P_j(\lambda)^2e^{-\frac{8\pi^2 N}{t}\lambda^2 + 2\pi \lambda}.
\end{aligned}
\end{equation}
A useful formula to simplify this integral is
\begin{equation}
\label{MM-8}
\int_{-\infty}^{+\infty} dx\ H_n(x)^2 e^{-(x-c)^2} = 2^n n! \sqrt{\pi} L_n(-2c^2)
\end{equation}
where \(c\) is a constant and \(L_n(x)\) are the \emph{Laguerre polynomials}, satisfying the properties
\begin{align}
L_n^{(m)}(x) &= \frac{1}{n!} e^x x^m \left(\frac{d}{dx}\right)^n \left(e^{-x}x^{n+m}\right) , \\
L_n(x) &\equiv L_n^{(0)}(x) , \\
L_n^{(m+1)} (x) &= \sum_{j=0}^n L_j^{(m)} (x) ,\\
L_n^{(m)}(x) &= \sum_{k=0}^n \binom{n+m}{n-k} \frac{(-x)^k}{k!}. 
\end{align}
Substituting \eqref{MM-8} in \eqref{MM-7}, and expanding in series we have 
\begin{equation}
\begin{aligned}
\langle W_{\mathbf{N}}(C)\rangle &= \frac{1}{N} e^{c^2} L_{N-1}^{(1)}(-2c^2) \qquad \text{with}\ c:= \sqrt{\frac{t}{8N}} \\
&\begin{aligned} = &\frac{1}{N} \sum_{k=0}^\infty \frac{1}{k!}\left(\frac{t}{8N}\right)^k \sum_{j=0}^{N-1}\frac{N!}{(j+1)!(N-1-j)!}\frac{1}{j!} \left(\frac{t}{4N}\right)^j - \\
& - \frac{1}{N}\sum_{j=1}^{N-1} \frac{1}{j!(j+1)!}\left(\frac{t}{4}\right)^j \frac{j(j+1)}{2} + \frac{1}{N}\sum_{j=0}^{N-1} \frac{1}{j!(j+1)!}\left(\frac{t}{4}\right)^j\frac{1}{2} + O(1/N^2) \end{aligned} \\
&= \sum_{j=0}^{N-1} \frac{1}{j!(j+1)!}\left(\frac{t}{4}\right)^j + O\left(1/N^2\right)
\end{aligned}
\end{equation}
where we expanded the first terms with respect to powers of \(1/N\), and already noticed that for \(N\gg 1\) the odd-power terms cancel. We can thus examine the large-\(N\) limit, and inserting the definition of the modified Bessel function \(I_n(2x) = \sum_{k=0}^{\infty} \frac{x^{n+2k}}{k!(n+k)!}\) the expectation value gives
\begin{equation}
\langle W_{\mathbf{N}}(C)\rangle = \frac{2}{\sqrt{t}} I_1\left(\sqrt{t}\right) + O(1/N^2) ,
\end{equation}
matching the result obtained with the saddle point technique in \eqref{MM-9}. In general, the expansion is in powers of \(1/N^2\) rather than \(1/N\), as expected from the analogy \lq\lq \(N^2 \leftrightarrow 1/\hbar\)\rq\rq\ that we noticed in  \eqref{MM-partition}. Solutions to the matrix model for higher representations have also been found, see \cite{gomis-tHooft,okuda-trancanelli-spectral-curves,zarembo-ads/cft-loc}.

\section[Localization of N=2 Chern-Simons theory on the 3-sphere]{Localization of \(\mathcal{N}=2\) Chern-Simons theory on the 3-sphere}

In this section we review another example of supersymmetric localization applied to the computation of Wilson loop expectation values, in an \(\mathcal{N}=2\) matter-coupled Euclidean Super Chern-Simons (SCS) theory on the 3-sphere \(\mathbb{S}^3\). We follow the derivation of Kapustin-Willet-Yaakov \cite{itamar-wilson_loop_3dCS}, and  Mari\~{n}o \cite{Marino-locCS}, inspired in part by the work discussed in the previous section. We consider a generic compact Lie group \(G\) as the gauge group, with Lie algebra \(\mathfrak{g}\).

\subsection[Matter-coupled N=2 Euclidean SCS theory on the 3-sphere]{Matter-coupled \(\mathcal{N}=2\) Euclidean SCS theory on \(\mathbb{S}^3\)}
The case of \(\mathcal{N}=2\) Euclidean supersymmetry on \(\mathbb{S}^3\) was discussed as an example in Sections \ref{subsec:3dN2} and \ref{subsec:3dN2-curved} for the gauge sector. We report the action for the SCS theory
\begin{equation}
S_{CS} = \frac{k}{4\pi} \int_{\mathbb{S}^3} d^3x\sqrt{g} \ \mathrm{Tr}\left\lbrace \frac{ \varepsilon^{\mu\nu\rho}}{\sqrt{g}} \left( A_\mu \partial_\nu A_\rho + \frac{2i}{3} A_\mu A_\nu A_\rho \right) - \tilde{\lambda} \lambda + 2\sigma D \right\rbrace
\end{equation} 
and the supersymmetry variations, already considered in curved space
\begin{equation}
\label{itamar-susy-gauge}
\begin{aligned}
&\delta A_\mu = \frac{i}{2} (\tilde{\epsilon} \gamma_\mu \lambda - \tilde{\lambda} \gamma_\mu \epsilon) \\
&\delta \sigma = \frac{1}{2}( \tilde{\epsilon}\lambda - \tilde{\lambda}\epsilon ) \\
&\delta \lambda =  \left( - \frac{1}{2} F_{\mu\nu} \gamma^{\mu\nu} - D + i(D_\mu \sigma)\gamma^\mu + \frac{2i}{3} \sigma \gamma^\mu D_\mu \right) \epsilon \\
&\delta \tilde{\lambda} = \left( - \frac{1}{2} F_{\mu\nu} \gamma^{\mu\nu} + D - i(D_\mu \sigma)\gamma^\mu - \frac{2i}{3} \sigma \gamma^\mu D_\mu \right) \tilde{\epsilon} \\
&\delta D = -\frac{i}{2} \left( \tilde{\epsilon} \gamma^\mu D_\mu \lambda - (D_\mu \tilde{\lambda}) \gamma^\mu \epsilon \right) + \frac{i}{2}\left( [\tilde{\epsilon}\lambda, \sigma] - [\tilde{\lambda}\epsilon, \sigma] \right) - \frac{i}{6}\left( \tilde{\lambda} \gamma^\mu D_\mu \epsilon + (D_\mu \tilde{\epsilon}) \gamma^\mu \lambda \right) .
\end{aligned}
\end{equation}
where the \(D_\mu\) are gauge-covariant derivatives with respect to the metric and spin connection induced by the round metric \eqref{eq:round-metric-3sph}, that in stereographic coordinates \(x^{\mu=1,2,3}\) is given by
\begin{equation}
g_{\mu\nu} = e^{2\Omega(x)}\delta_{\mu\nu} \qquad e^{2\Omega(x)} = \left( 1+ \frac{x^2}{4r^2} \right)^{-2}
\end{equation}
with \(r\) being the radius of the embedding \(\mathbb{S}^3 \hookrightarrow \mathbb{R}^4\). We remark again that this supersymmetric action is actually superconformal, thus can preserve supersymmetry on this conformally flat background, even with positive scalar curvature. The new background preserves all the original \(\mathcal{N}=2\) algebra, generated by conformal Killing spinors \(\epsilon, \tilde{\epsilon}\), taken to satisfy
\begin{equation}
\nabla_\mu \epsilon = \frac{i}{2r}\gamma_\mu \epsilon , \qquad \nabla_\mu \tilde{\epsilon} = \frac{i}{2r}\gamma_\mu , \tilde{\epsilon}
\end{equation}
where every equation has two possible solutions.

We consider also coupling the theory to matter fields, adding them in chiral multiplets in a representation \(R\) of the gauge group, to preserve supersymmetry. The 3-dimensional \(\mathcal{N}=2\) chiral multiplet (or hypermultiplet) is, as for the gauge multiplet, given by dimensional reduction of the \(\mathcal{N}=1\) chiral multiplet in 4 dimensions: a complex scalar \(\phi\), a 2-component Dirac spinor\footnote{Recall that \(Spin(3)=SU(2)\) has no Majorana spinors. We consider the reduced 4-dimensional Majorana  spinor \(\psi\) as a 3-dimensional Dirac (complex) spinor, since they have the same number of real components.} \(\psi\) and an auxiliary complex scalar \(F\).
Every field comes with its complex conjugate from the corresponding anti-chiral multiplet. The supersymmetric action for the matter multiplet coupled to the gauge multiplet is given by
\begin{equation}
\label{itamar-matter-action}
S_m = \int_{\mathbb{S}^3} d^3x\sqrt{g}\ \left( 
D_\mu \tilde{\phi} D^\mu \phi + \frac{3}{4r^2}\tilde{\phi}\phi + i\tilde{\psi}\slashed{D}\psi + \tilde{F}F + \tilde{\phi}\sigma^2\phi + i\tilde{\phi}D\phi + i\tilde{\psi}\sigma\psi + i \tilde{\phi}\tilde{\lambda}\psi - i\tilde{\psi}\lambda\phi
\right)
\end{equation}
where the \(\mathfrak{g}\)-valued fields in the gauge multiplets act on the chiral multiplet in the representation \(R\). This is the \lq\lq covariantization\rq\rq\ of the flat space action for the matter multiplet (see for example \cite{gaiotto-CStheory}), with the addition of the conformal coupling of the scalar field to the curvature, \(\frac{3}{4r^2}\tilde{\phi}\phi\). The supersymmetry transformations for the chiral multiplet, with respect to the conformal Killing spinors \(\epsilon,\tilde{\epsilon}\), are
\begin{equation}
\begin{aligned}
\delta \phi &= \tilde{\epsilon}\psi \qquad \delta \tilde{\phi} = \tilde{\psi}\epsilon \\
\delta \psi &= (-i\gamma^\mu D_\mu \phi-i\sigma\phi)\epsilon - \frac{i}{3}\gamma^\mu (\nabla_\mu \epsilon)\phi + \tilde{\epsilon} F \\
\delta \tilde{\psi} &= \tilde{\epsilon}(i\gamma^\mu D_\mu \tilde{\phi} +i\sigma\tilde{\phi}) + \frac{i}{3} (\nabla_\mu \tilde{\epsilon}) \gamma^\mu \tilde{\phi}   + \epsilon \tilde{F} \\
\delta F &= \epsilon ( -i \gamma^\mu D_\mu \psi + i\lambda \phi +i\sigma\psi) \\
\delta \tilde{F} &= ( i  D_\mu \tilde{\psi} \gamma^\mu - i\tilde{\lambda} \tilde{\phi} +i\sigma\tilde{\psi} )\tilde{\epsilon} .
\end{aligned}
\end{equation}

The above variations generates a superconformal algebra that closes off-shell:
\begin{equation}
[ \delta_\epsilon, \delta_{\tilde{\epsilon}} ] = -i ( \mathcal{L}_{v} + G_\Lambda + R_\alpha + \Omega_f )
\end{equation}
where \(\mathcal{L}_v\) is the Lie derivative (translation) along the Killing vector field \(v = (\tilde{\epsilon}\gamma^\mu \epsilon)\partial_\mu\), acting on one forms as \(\mathcal{L}_v(A)_\mu = v^\nu \partial_\nu A_\mu + A_\nu \partial_\mu v^\nu\), and on spinors as \(\mathcal{L}_v \psi = \nabla_\nu \psi - \frac{1}{4} (\nabla_\mu v_\nu)\gamma^{\mu\nu} \psi\). \(G_\Lambda\) is a gauge transformation with respect to the parameter \(\Lambda:= A(v) + \sigma (\tilde{\epsilon}\epsilon)\). \(R_\alpha\) is a \(U(1)^\mathcal{R}\) R-symmetry transformation, and \(\Omega_f\) is a dilatation \cite{Marino-locCS}. The matter coupled action \(S_{CS} + S_m\) is known to be superconformal at quantum level, but one could also add a superpotential for the matter multiplet. This choice is restricted by the condition of unbroken superconformal symmetry both at classical and at quantum level, since the localization principle works only if the supersymmetry algebra closes off-shell. It turns out that the localization locus is at trivial configurations of the matter sector, thus the precise choice of superpotential does not influence the computation.

\subsection{The supersymmetric Wilson loop}

The Wilson loop under consideration, in the representation \(R\) of the gauge group, is defined as \cite{gaiotto-CStheory}
\begin{equation}
\label{itamar-wilson}
W_R(C) = \frac{1}{\dim{R}} \mathrm{Tr}_R \left( \mathcal{P}\exp{\oint_{C}dt\ (iA_\mu \dot{C}^\mu + \sigma)} \right)
\end{equation}
with \(C:\mathbb{S}^1\to \mathbb{S}^3\) a closed curve of tangent vector \(\dot{C}\), normalized such that \(|\dot{C}|=1\). In order to localize its expectation value, we have to consider those curves such that this operator preserves some supersymmetry on the 3-sphere. Its variation under \eqref{itamar-susy-gauge} is proportional to 
\begin{equation}
\delta W_R(C) \propto - \tilde{\epsilon}(\gamma_\mu \dot{C}^\mu + 1 )\lambda + \tilde{\lambda} (\gamma_\mu \dot{C}^\mu - 1)  \epsilon . 
\end{equation}
Imposing the vanishing of this expression for all gauginos, we get the following conditions on the conformal Killing spinors,
\begin{equation}
\label{itamar-1}
\tilde{\epsilon}(\gamma_\mu \dot{C}^\mu + 1 ) = 0  ,\qquad (\gamma_\mu \dot{C}^\mu - 1)  \epsilon = 0 .
\end{equation}
We have two more conditions on the conformal Killing spinors, thus the maximum number of solutions is reduced by half. The Wilson loop can at most be invariant under two of the four possible supersymmetry variations, and for that it is called \emph{1/2-BPS}.

We can find explicitly one family of supersymmetric Wilson loops and one supersymmetry variation with respect to which we are going to perform the localization procedure. In order to solve the conformal Killing equations and the conditions \eqref{itamar-1}, we chose explicitly an orthonormal basis and a corresponding vielbein on \(\mathbb{S}^3\). Since as a manifold \(\mathbb{S}^3 \cong SU(2)\), we can use Lie theory to describe the geometry on the 3-sphere. In particular, the vielbein can be chosen proportional to the Maureer-Cartan form \(\Theta\in T^* (SU(2))\otimes \mathfrak{su}(2)\),\footnote{Again, we use Roman letters as \lq\lq flat\rq\rq\ indices, and Greek letters as \lq\lq curved\rq\rq\ indices.}
\begin{equation}
e^i_{\mu} := \frac{r}{2} e^i(\Theta(\partial_\mu))
\end{equation}
where \(\lbrace e^i\rbrace\) is a basis of \(\mathfrak{su}(2)^*\), dual to a basis \(\lbrace T_i\rbrace\) of \(\mathfrak{su}(2)\).\footnote{Say, the standard basis given by the Pauli matrices, \(T_i := \sigma_i/\sqrt{2}\).} One can check that this vielbein is consistent with the round metric, giving \(g_{\mu\nu} = e_\mu^i e_\nu^j \delta_{ij}\) (see \cite{Marino-locCS}). Using this orthonormal basis, the spin connection components are
\begin{equation}
(\omega_\mu)_{ij} = \frac{1}{r}e^k_\mu \varepsilon_{ijk}
\end{equation}
where \(\varepsilon_{ijk}\) is the Levi-Civita symbol. In this basis the conformal Killing spinor equation for \(\epsilon\) looks particularly simple,
\begin{equation}
\left( \partial_\mu + \frac{1}{8}(\omega_\mu)_{ij}[\gamma^i,\gamma^j] \right)\epsilon = \frac{i}{2r}\gamma_\mu \epsilon \quad \Leftrightarrow \quad \partial_\mu \epsilon = 0
\end{equation}
where we used the commutator \([\gamma^i,\gamma^j] = 2i\left.\varepsilon^{ij}\right._k \gamma^k\). We see that the components of \(\epsilon\) are constants. The corresponding condition for the supersymmetry of the Wilson loop then requires \(\gamma_\mu \dot{C}^\mu\) to be constant too, as the components of the vector field \(\dot{C}^i\) in the orthonormal frame. This means that the Wilson loop has to describe grat circles on \(\mathbb{S}^3\). Following \cite{itamar-wilson_loop_3dCS}, we take \(\dot{C}\) parallel to one of the \(e^i\), say \(e^3\), and the conformal Killing spinor to satisfy
\begin{equation}
(\gamma_3 - 1)\epsilon = 0 .
\end{equation}
We will consider the one dimensional subalgebra generated by such restricted spinor, and put \(\tilde{\epsilon}=0\).

\subsection{Localization: gauge sector}

We focus now on the localization of the Chern-Simons path integral, without coupling to the matter multiplet. Ignoring the issue of gauge fixing, we would add to the action the localizing term \(t S_{loc} = t \delta \mathcal{V}\), with \(t\in \mathbb{R}^+\) a parameter, \(\delta\) being the supersymmetry transformation generated by the conformal Killing spinor \(\epsilon\) described in the last section, and \(\mathcal{V}\) some fermionic functional whose bosonic part is positive semi-definite. At the end of Section \ref{sec:susy-curved}, we pointed out that the Super Yang-Mills Lagrangian is an example of \(\delta\)-exact term, so we put
\begin{equation}
\begin{aligned}
S_{loc} := 2 S_{YM} = \int_{\mathbb{S}^3} d^3x \sqrt{g}\ \mathrm{Tr} \left( i\tilde{\lambda} \gamma^\mu D_\mu \lambda + \frac{1}{2} F_{\mu\nu}F^{\mu\nu} + D_\mu \sigma D^\mu \sigma + i \tilde{\lambda}[\sigma, \lambda] + \right.\\
+\left.  \left( D + \frac{\sigma}{r} \right)^2 - \frac{1}{2r}\tilde{\lambda}\lambda  \right)
\end{aligned}
\end{equation}
whose bosonic part is indeed positive semi-definite. This localizing term can be derived also from the functional \cite{itamar-wilson_loop_3dCS}
\begin{equation}
\mathcal{V} = \int_{\mathbb{S}^3} d^3x \sqrt{g}\ \mathrm{Tr}\left( (\delta \tilde{\lambda})\lambda \right)
\end{equation}
analogously to the one used in the previous chapter for the gauge multiplet. \(S_{YM}\) being supersymmetric means that \(\delta^2 = 0\) on \(\mathcal{V}\), making the localization principle applicable. As usual, the limit \(t\to\infty\) localizes the path integral on the configurations that make this term vanish: the terms involving bosonic fields are separately non-negative, while the gaugino and its conjugate have to vanish identically. Summarizing, the localization locus is given by 
\begin{equation}
\label{itamar-locus-gauge}
\left\lbrace \begin{aligned}
&\lambda = \tilde{\lambda} = 0 \\
&F = 0 \Rightarrow A = 0\ \text{(up  to a gauge transformation)} \\
&\sigma = a\in\mathfrak{g}\ (\mathrm{constant}) \\
&D = - \frac{1}{r} a
\end{aligned}\right.
\end{equation}

Keeping into account the gauge-fixing procedure (as we should), the ghost \(c\), anti-ghost \(\tilde{c}\) and Lagrange multiplier \(b\) are added to the theory, taking value in the Lie algebra \(\mathfrak{g}\), together with the BRST differential \(\delta_B\) that acts as
\begin{equation}
\delta_B X = - [c,X] \qquad \delta_B c = -\frac{1}{2}[c,c] \qquad \delta_B \tilde{c} = b \qquad \delta_B b = 0
\end{equation}
where \(X\) is any field in the original theory, acted by a gauge transformation parametrized by \(c\). The BRST differential is nilpotent, \(\delta_B^2 = 0\).  The total differential 
\begin{equation}
Q := \delta_\epsilon + \delta_B
\end{equation}
acts now as the equivariant differential for the \((U(1)\rtimes G)\)-equivariant cohomology in the BRST-augmented field space. The original CS action is automatically \(Q\)-closed since it is gauge invariant, so we can combine the localization principle with the gauge-fixing procedure adding to the Lagrangian the term
\begin{equation}
Q \left( ( \delta \tilde{\lambda} ) \lambda   - \tilde{c} \left( \frac{\xi}{2} b - \nabla^\mu A_\mu \right) \right)
\end{equation}
where we suppressed the Lie algebra bilinear \(\mathrm{Tr}\) for notational convenience. Since the first term is gauge invariant, \(\delta_B \left( ( \delta \tilde{\lambda} ) \lambda \right)=0\), this gives the same localization term as before. If \(\delta[\mbox{ghosts}]=0\) on the gauge-fixing subcomplex, the second term gives
\begin{equation}
\label{itamar-2}
Q \left( \tilde{c} \left( \frac{\xi}{2} b - \nabla^\mu A_\mu \right) \right) = \frac{\xi}{2}b^2 - b \nabla^\mu A_\mu + \tilde{c} \nabla^\mu D_\mu  c + \tilde{c}\nabla^\mu \delta A_\mu .
\end{equation}
The first two terms give, upon path integration over \(b\), the usual gauge-fixing Lagrangian in the \(R_{\xi}\)-gauge; the third term is the ghost Lagrangian. The fourth term \(\propto \left( \tilde{c}\nabla^\mu \tilde{\lambda}\gamma_\mu\right)\) does not change the partition function: if we see this term as a perturbation of the gauge-fixed action, all diagrams with insertion of \(\left( \tilde{c}\nabla^\mu \tilde{\lambda}\gamma_\mu\right)\) will vanish, since \(\tilde{c}\) is coupled only to \(c\) via the propagator but there are no vertices containing \(c\). In other words, the fermionic determinant arising from the path integration over ghosts is not changed by this term. The modified localizing term \eqref{itamar-2} is \(Q\)-closed: the old localizing term because of gauge invariance and supersymmetry, while the gauge-fixing and ghost terms follows by \(Q^2 A_\mu=0\) that is easy to check. After path integration over the auxiliary \(b\) the limit \(t\to\infty\) finally localizes the theory to the same locus \eqref{itamar-locus-gauge}, with ghosts put to zero.

Evaluating the classical action at the saddle point configuration, we get 
%[HERE IN THE ARTICLE THERE IS AN IMMAGINARY FACTOR THAT I DO NOT UNDERSTAND, AND I DID NOT PUT IN THE FORMULAS! ALSO IN THE FINAL RESULTS, IT SHOULD BE THERE I THINK. MAYBE, THE CS LAGRANGIAN IN THE KAPUSTIN ART. IS WRONG, ONE MUST WICK ROTATE \(\sigma\) TO \(i\sigma\) (since it was the 0-th component of \(A_\mu\) in 4-d) FOR CONVERGENCE? KAPUSTIN INTRODUCES AN i AT FORMULA (3.7), MARINO SAYS IT WORKS IN EUCLIDEAN SIGNATURE, BUT AT 2.28 IT DEFINES THE PATH INTEGRAL WITH THE IMMAGINARY EXPONENT. BUT IN THE MEANTIME, MARINO SEEMS TO WICK ROTATE INSIDE \(S_{YM}\) THE TERMS CORRESPONDING TO \(\sigma\) (eq 2.29).]
\begin{equation}
S_{CS}[a] = \frac{k}{4\pi}\int_{\mathbb{S}^3} d^3x\sqrt{g}\ \mathrm{Tr}\left( -\frac{2}{r} a^2\right) = -k\pi r^2 \mathrm{Tr}(a^2)
\end{equation}
where we used \(\mbox{vol}(\mathbb{S}^3) = 2\pi^2 r^3\). The supersymmetric Wilson loop observable \eqref{itamar-wilson} localizes to
\begin{equation}
W_R(C) = \frac{1}{\dim{R}}\mathrm{Tr}_R \left( e^{2\pi r a}\right)
\end{equation}
since the curve \(C\) is a great circle of radius \(r\). Integrating as usual the rescaled fluctuations above the localization configuration, and taking the limit \(t\to\infty\) as in \eqref{eq:susy-loc-limit}, the partition function and the Wilson loop expectation value are thus given by a finite-dimensional integral over \(\mathfrak{g}\) with Gaussian measure, the \lq\lq matrix model\rq\rq\,
\begin{equation}
\begin{aligned}
Z &= \int_\mathfrak{g} da\ e^{-k\pi r^2 \mathrm{Tr}(a^2)} Z_{1-loop}^g[a] \\
\langle W_R(C)\rangle &= \frac{1}{Z\dim{R}}\int_\mathfrak{g} da\ e^{-k\pi r^2 \mathrm{Tr}(a^2)} Z_{1-loop}^g[a]\mathrm{Tr}_R \left( e^{2\pi r a}\right) .
\end{aligned}
\end{equation}
As we pointed out in the last section, the integration over the Lie algebra \(\mathfrak{g}\) can be reduced over its Cartan subalgebra \(\mathfrak{h}\), exploiting the gauge invariance of the matrix model under the adjoint action of \(\mathfrak{g}\) itself. This for example means, in the case of a matrix gauge group, that we integrate over the diagonalized matrices \lq\lq fixing the gauge\rq\rq\ of the matrix model. The corresponding Faddeev-Popov determinant is also called \emph{Vandermonde determinant}, \begin{equation}
\label{itamar-vanderm}
\prod_\alpha \left( \rho_\alpha(a)\right)
\end{equation}
where the product runs over the roots of \(\mathfrak{g}\). There is left an overcounting given by the possible permutations of the roots, the action of the \emph{Weyl group} \(\mathcal{W}\) of \(\mathfrak{g}\), cured dividing by its order \(|\mathcal{W}|\). The path integrals are thus rewritten as
\begin{equation}
\begin{aligned}
Z &= \frac{1}{|\mathcal{W}|} \int_\mathfrak{\mathfrak{h}} da\ \prod_\alpha \left( \rho_\alpha(a)\right) e^{-k\pi r^2 \mathrm{Tr}(a^2)} Z_{1-loop}^g[a] \\
\langle W_R(C)\rangle &= \frac{1}{Z|\mathcal{W}|\dim{R}}\int_\mathfrak{h} da\ \prod_\alpha \left( \rho_\alpha(a)\right) e^{-k\pi r^2 \mathrm{Tr}(a^2)} Z_{1-loop}^g[a]\mathrm{Tr}_R \left( e^{2\pi r a}\right) .
\end{aligned}
\end{equation}

Here we summarize the computation of the 1-loop determinant from \cite{itamar-wilson_loop_3dCS}. For convenience, we put \(r=1\) and \(\xi=1\). Inserting the contribution of ghosts, the Lagrangian for the localizing term is given by (suppressing the \(\mathrm{Tr}\))
\begin{equation}
\mathcal{L}_{loc} = \frac{1}{2} F_{\mu\nu}F^{\mu\nu} + D_\mu \sigma D^\mu \sigma + \left( D + \sigma \right)^2 + i\tilde{\lambda} \slashed{D} \lambda +  i [\tilde{\lambda},\sigma] \lambda - \frac{1}{2}\tilde{\lambda}\lambda + \partial_\mu \tilde{c} D^\mu c - \frac{1}{2}b^2 + b\nabla^\mu A_\mu .
\end{equation}
Considering the limit \(t\to\infty\), we rescale as usual the fields around the configuration \eqref{itamar-locus-gauge}:
\begin{equation}
\sigma = a + \sigma'/\sqrt{t} , \qquad D = -a + D'/\sqrt{t} , \qquad X= X'/\sqrt{t} ,
\end{equation}
where \(X\) are all the fields without zero modes, and then rename \(\sigma'\to\sigma\), \(D'\to D\), \(X'\to X\). In the limit, only quadratic terms in the fluctuations survive,
\begin{equation}
\mathcal{L}_{loc} \sim \frac{1}{2}\partial_{[\mu}A_{\nu]}\partial^{[\mu}A^{\nu]} - [A_\mu,a]^2 + (\partial\sigma)^2 + (D+\sigma)^2 + i\tilde{\lambda}\slashed{\nabla}\lambda + i [\tilde{\lambda},a]\lambda - \frac{1}{2}\tilde{\lambda}\lambda + |\partial \tilde{c}|^2 - \frac{1}{2}b^2 + b\nabla^\mu A_\mu  .
\end{equation}
The resulting theory is free, and we can integrate it giving the corresponding 1-loop determinant. We will neglect all overall normalization constant from the Gaussian integrations. The integral over the auxiliary field \(b\) gives the gauge fixing term \(-\frac{1}{2}(\nabla^\mu A_\mu)^2\). The contribution from \(D\) is purely Gaussian and can be integrated out removing the corresponding term. The integration over \(\sigma\) gives a determinant \(\det{(\nabla^2)}^{-1/2}\), and the (Grassman) integral over the ghosts gives \(\det{(\nabla^2)}\). It is useful to separate the gauge field as (Helmolz-Hodge decomposition) \[
A_\mu = B_\mu + \partial_\mu \phi \]
with \(\phi\) scalar and \(B_\mu\) divergenceless, \(\nabla^\mu B_\mu = 0\). With this decomposition, the Lorentz gauge condition becomes \(\nabla^2\phi=0\), and we can integrate \(\phi\) giving a determinant \(\det{(\nabla^2)}^{-1/2}\), that cancels the above two other contributions. We are left with
\begin{equation}
-B_\mu \Delta B^\mu - [a,B_\mu]^2 + i\tilde{\lambda}\slashed{\nabla}\lambda + i [\tilde{\lambda}, a]\lambda - \frac{1}{2}\tilde{\lambda}\lambda
\end{equation}
where \(\Delta\) is the vector Laplacian. Now we use the fact that the path integral can be reduced over the Cartan subalgebra of \(\mathfrak{g}\), considering \(a\in\mathfrak{h}\), and 
\begin{equation}
B_\mu = B_\mu^{(\mathfrak{h})} + B_\mu^\alpha e_\alpha
\end{equation}
where \(B_\mu^{(\mathfrak{h})}\) is the component of \(B_\mu\) along \(\mathfrak{h}\), and similarly for the gaugino. This component does not enter in the Lie brackets with \(a\), so its contribution to the path integral is independent of \(a\), and we drop it. The remaining interesting terms are
\begin{equation}
\sum_\alpha \left( B_\mu^{-\alpha} (-\Delta + \rho_\alpha(a)^2) B_\mu^{\alpha}  + \tilde{\lambda}^{-\alpha}\left( i\slashed{\nabla} +i \rho_\alpha(a) - \frac{1}{2}\right) \lambda^\alpha \right)
\end{equation}
where the \(a\)-dependent kinetic terms are clearly identified, and the component fields appearing are real or complex valued scalars and spinors. The Gaussian integration over these fields lead to the determinant factors
\begin{equation}
Z^g_{1-loop}[a] = \prod_\alpha \frac{\det\left( i\slashed{\nabla} +i \rho_\alpha(a) - \frac{1}{2}\right)}{\det{\left(-\Delta+\rho_\alpha(a)^2\right)}^{1/2}} .
\end{equation}

Now, using the fact that the eigenvalues of the Laplacian on divergenceless vectors are \((l+1)^2\) with degeneracy \(2l(l+2)\), and the eigenvalues of \(i\slashed{\nabla}\) are \(\pm \left(l+\frac{1}{2}\right)\) with degeneracy \(l(l+1)\), where \(l\in\mathbb{Z}^+\), the corresponding determinants can be written as infinite products
\begin{equation}
\prod_\alpha \prod_{l=1}^\infty \frac{(l+i\rho_\alpha(a))^{l(l+1)} (-l-1+i\rho_\alpha(a))^{l(l+1)}}{((l+1)^2+\rho_\alpha(a)^2)^{l(l+2)}} = \prod_\alpha \prod_{l=1}^\infty \frac{(l+i\rho_\alpha(a))^{(l+1)}}{(l-i\rho_\alpha(a))^{(l-1)}}
\end{equation}
where the equality follows after some simplifications. Since roots come in pairs \((\rho_\alpha, -\rho_\alpha)\), taking the square of this one gets
\begin{equation}
\left(Z^g_{1-loop}[a]\right)^2 = \prod_\alpha \prod_{l=1}^\infty \frac{(l^2 + \rho_\alpha(a)^2)^{(l+1)}}{(l^2 + \rho_\alpha(a)^2)^{(l-1)}} = \prod_\alpha \prod_{l=1}^\infty \left( l^2 + \rho_\alpha(a)^2 \right)^2 .
\end{equation}
Collecting a factor \(l^4\) the product splits in the factorization formula for the hyperbolic sine,
\begin{equation}
\frac{\sinh(\pi z)}{\pi z} = \prod_{l=1}^\infty \left(1+\frac{z^2}{l^2}\right)
\end{equation}
and an \(a\)-independent divergent part that can be regularized with the zeta-function method,
\begin{equation}
\prod_{l=1}^\infty l^4 = e^{4\sum_{l=1}^\infty \log(l)} = e^{-4\zeta '(0)}= e^{2\log(2\pi)} .
\end{equation}
Up to an overall normalization constant, the \(a\)-dependence of the 1-loop determinant is finally given by
\begin{equation}
Z_{1-loop}^g[a] = \prod_\alpha \left( \frac{2\sinh(\pi \rho_\alpha(a))}{\pi \rho_\alpha(a)} \right)
\end{equation}
where we see cancellation between the denominator and the Vandermonde determinant \eqref{itamar-vanderm}.

Collecting the above results, the localization formulas for the partition function and the expectation value of the supersymmetric Wilson loop in the pure CS theory are
\begin{equation}
\label{itamar-localized-gauge}
\begin{aligned}
Z &\sim \frac{1}{|\mathcal{W}|} \int_\mathfrak{\mathfrak{h}} da\  e^{-k\pi \mathrm{Tr}(a^2)} \prod_\alpha \left( 2\sinh(\pi \rho_\alpha(a))\right) \\
\langle W_R(C)\rangle &= \frac{1}{Z|\mathcal{W}|\dim{R}}\int_\mathfrak{h} da\ e^{-k\pi  \mathrm{Tr}(a^2)} \mathrm{Tr}_R \left( e^{2\pi a}\right)\prod_\alpha \left( 2\sinh(\pi \rho_\alpha(a))\right) \\
&= \frac{1}{\dim{R}} \frac{\int_\mathfrak{h} da\ e^{-k\pi  \mathrm{Tr}(a^2)} \mathrm{Tr}_R \left( e^{2\pi a}\right)\prod_\alpha \left( 2\sinh(\pi \rho_\alpha(a))\right)}{\int_\mathfrak{\mathfrak{h}} da\  e^{-k\pi \mathrm{Tr}(a^2)} \prod_\alpha \left( 2\sinh(\pi \rho_\alpha(a))\right)} .
\end{aligned}
\end{equation}

\bigskip
These general localization formulas can be tested comparing their results for specific choices of \(G\) to perturbative calculations, for example. In the case of \(U(N)\) gauge group, the integral over the Cartan subalgebra is an integral over diagonal matrices \(a=\mathrm{diag}(\lambda_1,\cdots,\lambda_N)\), and the roots are given by \(\rho_{ij}(a) = \lambda_i-\lambda_j\) for \(i\neq j\). The Weyl group is \(S_N\), thus \(|\mathcal{W}|=N!\). If we take the Wilson loop in the fundamental representation, from \eqref{itamar-localized-gauge} we get
\begin{equation}
\begin{aligned}
Z &\sim \frac{1}{N!} \int \left( \prod_i d\lambda_i\ e^{-k\pi \lambda_i^2}\right) \prod_{i\neq j} 2\sinh(\pi(\lambda_i - \lambda_j)) , \\
\langle W_\mathbf{N}(C)\rangle &= \frac{1}{Z N! N}\int \left( \prod_i d\lambda_i\ e^{-k\pi \lambda_i^2}\right) \left(e^{2\pi\lambda_1}+\cdots+e^{2\pi\lambda_N} \right) \prod_{i\neq j}  2\sinh(\pi(\lambda_i - \lambda_j))  ,
\end{aligned}
\end{equation}
that are sums of Gaussian integrals, and can be computed exactly. The result for the Wilson loop expectation value is
\begin{equation}
\langle W_\mathbf{N}(C)\rangle = \frac{1}{N}e^{-Ni\pi/k}\frac{\sin\left( \frac{\pi N}{k} \right)}{\sin\left( \frac{\pi }{k} \right)} ,
\end{equation}
which is known as the exact result \cite{witten-knot}, up to the overall phase factor \(e^{-Ni\pi/k}\). This kind of phase factors arise in perturbative calculations in the so-called \emph{framing} of the Wilson loop. A perturbative calculation of the Wilson loop involves computations of correlators of the type \(\langle A_{\mu_1} (x_1) A_{\mu_2} (x_2) \cdots \rangle\), where \(x_1,x_2,\cdots\) are coordinates of points on the image of the curve \(C\). This contribution diverges when \(x_1=x_2\), so it is necessary to choose some regularization scheme to perform the computations. For example, considering the 2-point function \(\langle A_{\mu_1} (x) A_{\mu_2} (y)\rangle\), this clashing of points can be avoided requiring that \(y\) is integrated over a shifted curve \(C_f\) such that
\begin{equation}
C^\mu_{f}(\tau) = C^\mu(\tau) + \alpha\ n^\mu(\tau)
\end{equation}
where \(n\) is orthogonal to \(\dot{C}\). The choice of such an orthogonal component (frame) at every point on the curve is called framing. Even if at the end of the calculation one takes \(\alpha\to 0\), this procedure leaves a deformation-dependent
term, that in pure \(U(N)\) CS is 
\begin{equation}
e^{\frac{i\pi N}{k} \chi(C,C_{f})}
\end{equation}
where \(\chi(C,C_{f})\) is a topological invariant that takes integer values corresponding to the number of times the path \(C_f\) winds around \(C\). We see that localization produces an expectation value at framing -1 (see also \cite{Bianchi-framing-loc-CS} for a detailed discussion about framing).

\subsection{Localization: matter sector}

We turn now to the result for the localization of the matter-coupled theory. This is of course gauge invariant, so the equivariant differential acts effectively as \(Q\sim\delta\), since the ghost sector has been already considered in the previous paragraph. This means that, following the localization principle, we have to extend the matter action with a \(\delta\)-exact term. We are free to consider the canonical choice \eqref{eq:localizing-term-canonical} as in \cite{itamar-wilson_loop_3dCS}, or using the fact that \cite{Marino-locCS} the matter action \eqref{itamar-matter-action} is actually given by a supersymmetry variation, as the case of the YM action. This means that we can consider the localizing terms \[
tS_m + tS_{YM+ghosts}\]
or, schematically \[
t\int \delta \left( (\delta\psi)^\dagger \psi + \tilde{\psi} ( \delta \tilde{\psi})^\dagger\right) + tS_{YM+ghosts} \]
that have positive semi-definite bosonic parts. The second term in both choices is the one analyzed in the previous paragraph, and gives the same localization locus for the gauge and ghost sector, while both the first terms  vanishes for the field configurations
\begin{equation}
\psi = 0 , \qquad \phi = 0 , \qquad F=0 .
\end{equation}
This means that the classical action of the matter sector does not contribute to the partition function, but only in the 1-loop determinant. Expanding the fields around this configuration and scaling the fluctuations with the usual \(1/\sqrt{t}\) factor, we see that there are no couplings to the gauge sector fluctuations that survive in the \(t\to\infty\) limit, but only to the zero mode \(a\) of \(\sigma\). Thus the determinant factorizes as
\begin{equation}
Z_{1-loop}[a] = Z_{1-loop}^g[a]Z_{1-loop}^m[a].
\end{equation}
If matter is present in different copies of chiral multiplets, in maybe different representations of the gauge group, the determinant factorizes in the same way for each multiplet.

The determinant for the matter sector can be computed diagonalizing the the kinetic operators acting on the scalar and the fermion field, after having integrated out the auxiliary \(F\), and considering the path integration over the Cartan subalgebra with \(a\in\mathfrak{h}\). In particular, the relevant kinetic operators that have to be diagonalized are
\begin{equation}
K_b^{(\rho)} = \left(-\nabla^2 + \rho(a)^2 - i\rho(a) + \frac{3}{4}\right) , \qquad
K_f^{(\rho)} = \left(i\slashed{\nabla} + i\rho(a)\right) ,
\end{equation}
for the (complex) bosonic and fermionic parts, where \(a\) is regarded as acting on the representation \(R\) with weights \(\lbrace \rho\rbrace\). The eigenvalues of \(-\nabla^2\) are \(4j(j+1)\) with \(j=0,\frac{1}{2},\cdots\) with degeneracy \((2j+1)^2\), that we can rewrite as \(l(l+2)\) with degeneracy \((l+1)^2\) and \(l=0,1,\cdots\). The eigenvalues of \(i\slashed{\nabla}\) are \(\pm\left(l+\frac{1}{2}\right)\) with degeneracy \(l(l+1)\), with \(l=1,2,\cdots\). Thus the one loop determinant results, after a change of dummy index and some simplifications
\begin{equation}
\begin{aligned}
Z_{1-loop}^m[a] &= \prod_\rho \frac{\det(K_f)}{\det(K_b)} \\
&= \prod_\rho \prod_{l=1}^\infty \frac{\left( l+ \frac{1}{2} + i\rho(a)\right)^{l(l+1)} \left( l+\frac{1}{2}-i\rho(a) \right)^{l(l+1)}}{\left( l+ \frac{1}{2} + i\rho(a)\right)^{l^2} \left( l-\frac{1}{2}-i\rho(a) \right)^{l^2}} \\
&=  \prod_\rho \prod_{l=1}^\infty \left(\frac{ l+ \frac{1}{2} + i\rho(a)}{ l-\frac{1}{2}-i\rho(a) }\right)^l
\end{aligned}
\end{equation}
This product can be regularized using the zeta-function. We refer to \cite{Marino-locCS} for the details of the computation, and report here the result in the case the fields take value in a self-conjugate representation \(R\) of the gauge group:\footnote{For example, if \(R=S\oplus S^*\).}
\begin{equation}
Z_{1-loop}^m[a] = \prod_\rho \left( 2\cosh{(\pi \rho(a))} \right)^{-1/2}
\end{equation}
where now \(a\in R(\mathfrak{h})\) and \(\rho(a)\) is the weight of the Cartan element in the representation \(R\).

Summarizing, we have seen that the application of the supersymmetric localization principle to the matter-coupled SCS theory on \(\mathbb{S}^3\) reduces the path integral to a finite-dimensional integral describing a matrix model over the Lie algebra of the theory. Using the notation \eqref{eq:notation-detR}, the localization formulas for the partition function and the supersymmetric Wilson loop expectation value, with matter multiplets coming in self-conjugate representations \(R_1\oplus R_1^*, R_2\oplus R_2^*, \cdots\) are
\begin{equation}
\label{itamar-finalMM}
\begin{aligned}
Z &= \frac{1}{|\mathcal{W}|} \int_\mathfrak{\mathfrak{h}} da\  e^{-k\pi \mathrm{Tr}(a^2)} \frac{\det_{ad} 2\sinh(\pi a)}{\left(\det_{R_1}2\cosh{(\pi a)}\right)\left(\det_{R_2}2\cosh{(\pi a)}\right)\cdots} \\
\langle W_R(C)\rangle &= \frac{1}{Z|\mathcal{W}|\dim{R}}\int_\mathfrak{h} da\ e^{-k\pi  \mathrm{Tr}(a^2)} \mathrm{Tr}_R \left( e^{2\pi a}\right)\frac{\det_{ad} 2\sinh(\pi a)}{\left(\det_{R_1}2\cosh{(\pi a)}\right)\left(\det_{R_2}2\cosh{(\pi a)}\right)\cdots}
\end{aligned}
\end{equation}

\subsection{The ABJM matrix model}

ABJM theory is a special type of matter-coupled SCS theory in 3-dimensions constructed in \cite{ABJM_2008}, that has the interesting property to be dual under the AdS/CFT conjecture to a certain orbifold background in M-theory. It consists of two copies of \(\mathcal{N}=2\) SCS theory, each one with gauge group \(U(N)\), and opposite levels \(k, -k\). In addition, the are four matter (chiral and anti-chiral) supermultiplets \(\Phi_i, \tilde{\Phi}_i\), with \(i=1,2\), in the bi-fundamental representation of \(U(N)\times U(N)\), \((\mathbf{N},\bar{\mathbf{N}})\) and \((\bar{\mathbf{N}}, \mathbf{N})\). This field content can be represented as the \emph{quiver} in Fig.\ \ref{fig:ABJMquiver}.
\begin{figure}[ht]
\centering
\includegraphics[width=0.49\textwidth]{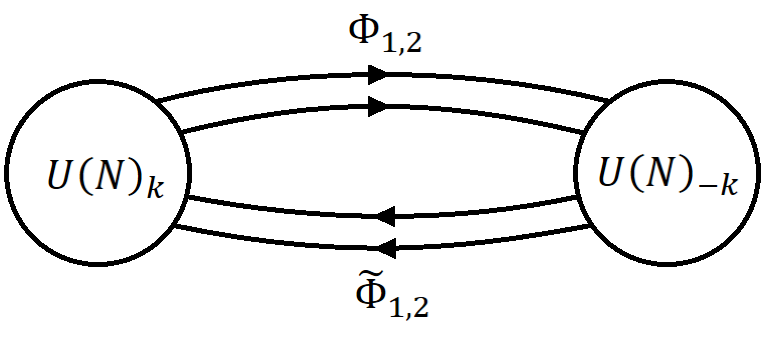}
\caption{The quiver for ABJM theory. The two nodes represent the gauge multiplets, with the convention of specifying the level of the CS term. The oriented links represent the matter multiplets in the bi-fundamental and anti-bi-fundamental representations.}
\label{fig:ABJMquiver}
\end{figure} 
The superpotential for the matter part is given by
\begin{equation}
W = \frac{4\pi}{k}(\Phi_1\tilde{\Phi}_1\Phi_2\tilde{\Phi}_2 - \Phi_1\tilde{\Phi}_2\Phi_2\tilde{\Phi}_1) ,
\end{equation}
and this structure actually enhance the supersymmetry of the resulting theory to \(\mathcal{N}=6\).\footnote{This is not apparent from the original action, but can be realized noticing that the superpotential has an \(SU(2)\times SU(2)\) symmetry that rotates separately the \(\Phi_i\) and the \(\tilde{\Phi}_i\). This, combined with the original \(SU(2)^\mathcal{R}\) symmetry of the theory, gives an \(SU(4)\cong Spin(6)\) symmetry that acts non-trivially on the supercharges. Thus the final theory has to have an enhanced \(\mathcal{N}=6\) supersymmetry.} If now \(a=\mathrm{diag}(\lambda_1,\cdots,\lambda_N,\hat{\lambda}_1,\cdots,\hat{\lambda}_N)\), the weights in the bi-fundamental representations are
\begin{equation}
\rho_{i,j}^{(N,\bar{N})}(a)= \lambda_i - \hat{\lambda}_j ,\qquad \rho_{i,j}^{(\bar{N},N)}(a)= \hat{\lambda}_j - \lambda_i .
\end{equation}

Plugging this information into \eqref{itamar-finalMM}, the partition function in this case localizes to the following matrix model,
\begin{equation}
Z \sim \frac{1}{N!N!} \int \left(\prod_i d\lambda_i d\hat{\lambda}_i\ e^{-k\pi(\lambda_i^2 - \hat{\lambda}_i^2)}\right) \frac{\prod_{i\neq j} \left( 2\sinh(\pi(\lambda_i-\lambda_j)) 2 \sinh(\pi(\hat{\lambda}_i - \hat{\lambda}_j)) \right)}{\prod_{i,j} \left( 2 \cosh(\pi(\lambda_i - \hat{\lambda}_j)) \right)} .
\end{equation}
The circular Wilson loop under consideration can be called now \emph{1/6 BPS} with respect to the enhanced supersymmetry of the model. Its expectation value in the fundamental representation is obtained by plugging a factor \((1/N)\sum_i e^{2\pi\lambda_i}\) as before. This matrix model cannot be solved exactly as in the case of the pure CS discussed above, but can be studied in the \(N\to\infty\) limit with the saddle-point technique showed in Section \ref{subsec:MMN4} \cite{Marino-locCS,marino-MM}. We also mention that, in this particular theory with enhanced \(\mathcal{N}=6\) supersymmetry, it was possible to construct a \emph{1/2 BPS} Wilson loop (so invariant under half of the \(\mathcal{N}=6\) supersymmetry algebra). The latter can be solved applying the same localization scheme that brings to the matrix model describing the 1/6 BPS Wilson loop presented above \cite{trancanelli-ABJM-wilson-loop,marino-ABJM}. A compact review introducing the state of the art on recent results about supersymmetric Wilson loops in ABJM and related theories can be found in \cite{roadmap-wilson-loops}.

%%%%%%%%%%%%%%%%%%%%%%%%%%%%%%%%%%%%%%%%%%%%%%%%% cap 5
\chapter{Non-Abelian localization and 2d YM theory}
\label{cha:non Abelian YM}

In this chapter we are going to summarize the result obtained mainly in \cite{witten-2dYM} by Witten. This was the first attempt in the physics literature of extending the equivariant localization formalism to possibly non-Abelian group actions. In that work, a modified definition of equivariant integration was defined, and this allowed for an extension of the same procedure discussed in Chapter \ref{cha:loc theorems} to show the localization property of integrals computed over spaces with generic symmetry group \(G\). This new formalism was applied to the study of 2-dimensional Yang-Mills (YM) theory over a Riemann surface, a relatively simple model from the physical point of view, but with a very rich underlying mathematical structure. In the following, we are going first to review the geometry of this special model, in connection with the symplectic geometry introduced in Section \ref{sec:e-cohom-symplectic}, as a motivation for the more mathematical discussion about the Witten's equivariant integration and non-Abelian localization principle that will follow. Next, we will review the ideas underlying the application of this new localization principle to the YM theory, and how this application results in a \lq\lq mapping\rq\rq\ between this model and a suitable topological theory, establishing the topological nature of the YM theory in the weak coupling limit. In the final section, we will summarize the interpretation given by the localization framework to the already existing solution for the partition function of this model.

As we pointed out in the Introduction, other generalizations of the Duistermaat-Heckman theorem to non-Abelian Hamiltonian systems also appeared in the mathematical literature, as the result obtained by Jeffrey and Kirwan in \cite{Jeffrey-nonabLoc}.
Other applications of this extended formalism followed, and Witten's approach was used for example more recently to describe Chern-Simons theories over a special class of 3-manifolds in \cite{witten-CS}.

\section{Prelude: moment maps and YM theory}
\label{sec:YM-prelude}

In the next section we are going to review  Witten's extension of the equivariant localization principle to possibly non-Abelian group actions, and a generalization of the DH formula in this direction. In \cite{witten-2dYM} this was applied to reinterpret the weak coupling limit of pure YM theory on a Riemann surface. This theory is exactly solvable, in the sense that its partition function  can be expressed in closed form, and its zero-coupling limit is known to describe a topological field theory. These features make 2-dimensional YM theory very appealing from the mathematical structure it carries, and make it possible to compare results or interpretations obtained via this \lq\lq new\rq\rq\ localization method with already existing solutions of the problem.

We are going to discuss more about the topological interpretation of 2d YM theory later, while in this section we review some results introduced by Atiyah and Bott \cite{atiyah-bott-YM} about the symplectic structure underlying this special QFT. This can be useful to contextualize the generic discussion of the next section, and it prepares the ground for the formal application of the non-Abelian localization principle.

We start by considering the partition function of YM theory on a compact orientable Riemannian manifold \(\Sigma\) of arbitrary dimension,
\begin{equation}
\begin{aligned}
Z(\epsilon) &= \frac{1}{\mathrm{vol}(\mathcal{G}(P))} \left(\frac{1}{2\pi\epsilon}\right)^{\dim(\mathcal{G})/2} \int_{\mathcal{A}(P)}DA\ e^{-S[A]} , \\
S[A] &= -\frac{1}{2\epsilon} \int_{\Sigma} \mathrm{Tr}(F^A\wedge\star F^A).
\end{aligned}
\end{equation}
Here \(\epsilon := g_{YM}^2\) is the square of the YM coupling constant. To describe the rest of the ingredients, let us recall the geometry underlying the gauge theory (to fill some of the details, see Appendix \ref{app:principal-bundles}). The dynamical field here is the \emph{connection} \(A\in \Omega(P;\mathfrak{g})\) on a principal \(G\)-bundle \(P\xrightarrow{\pi} \Sigma\), where \(G\) is a compact connected Lie group with Lie algebra \(\mathfrak{g}\). The path integral is thus taken over the space \(\mathcal{A}(P)\) of \(G\)-equivariant vertical 1-forms with values in \(\mathfrak{g}\), that is naturally an \emph{affine space} modeled on the infinite-dimensional vector space \(\mathfrak{a}\) of \(G\)-equivariant horizontal 1-forms with values in \(\mathfrak{g}\). This gives to \(\mathcal{A}(P)\) the structure of an infinite-dimensional manifold, whose tangent spaces are \(T_A\mathcal{A}(P) \cong \mathfrak{a} \cong \Omega^1(\Sigma; \mathrm{ad}(P))\), where we identified horizontal forms over \(P\) with forms over the base \(\Sigma\).\footnote{Recall that horizontality means essentially to have components only in the \lq\lq directions\rq\rq\ of the base space, and the \(G\)-equivariance ensures the right transformation behavior as forms valued in the \emph{adjoint bundle} \(\mathrm{ad}(P)\), the associated bundle to \(P\) that has \(\mathfrak{g}\) as typical fiber. Thus \(\mathfrak{a}\cong \Omega^1(\Sigma;\mathrm{ad}(P))\).} In other words, any vector field \(\alpha\in \Gamma(T\mathcal{A}(P))\) can be expanded locally as
\begin{equation}
\alpha = \alpha_\mu^a T_a\otimes dx^\mu , \qquad \alpha_\mu^a \in C^\infty(\Sigma\times \mathcal{A}(P)) ,
\end{equation}
with coefficients that depend on the point \(A\in \mathcal{A}(P)\) and \(p\in\Sigma\).
The curvature \(F^A = dA + \frac{1}{2}[A\stackrel{\wedge}{,}A]\) of the connection \(A\) is a horizontal 2-form over \(P\), so we can identify it as a 2-form on the adjoint bundle without loss of information, \(F^A \in \Omega^2(\Sigma;\mathrm{ad}(P))\). As such, it can be integrated as a differential form over \(\Sigma\). In the action \(S[A]\), \lq\lq \(\mathrm{Tr}\)\rq\rq\ represents a (negative definite) invariant inner product on \(\mathfrak{g}\), and \(\star\) is the Hodge dual operation, that is identified by the presence of a metric on \(\Sigma\).\footnote{The definition of the Hodge star is, implicitly, \(\alpha\wedge\star\beta = g^{-1}(\alpha,\beta)\omega\) for any \(\alpha,\beta\in \Omega^k(\Sigma)\). Here \(g^{-1}\) is the \lq\lq inverse\rq\rq\ metric on \(\Sigma\), that extends multi-linearly its action on every tangent space as \(g^{-1}(\alpha,\beta) = g^{\mu_1\nu_1}\cdots g^{\mu_k\nu_k} \alpha_{\mu_1\cdots\mu_k}\beta_{\nu_1\cdots\nu_k}\). \(\omega\) is a volume form (that can be induced by \(g\), for example). The Hodge star satisfies the property \(\star^2 \alpha = (-1)^{k(\dim(\Sigma)-k)}\alpha\).}

\(\mathcal{G}(P)\cong \Omega^0(\Sigma;\mathrm{Ad}(P)\) is the group of gauge transformations, that is locally equivalent to the space of \(G\)-valued functions over \(\Sigma\), and acts naturally on \(\mathcal{A}(P)\). If \(\phi\in Lie(\mathcal{G}(P))\cong \Omega^0(\Sigma; \mathrm{ad}(P))\) is an element of the Lie algebra of infinitesimal gauge transformations, its associated fundamental vector field at the point \(A\in \mathcal{A}(P)\) is 
\begin{equation}
\label{YM-fund-vf}
\underline{\phi}_A \equiv \delta_\phi A = \nabla^A\phi = d\phi + [A,\phi] .
\end{equation}

The path integral measure \(DA\) can be defined formally as the Riemannian measure induced by a metric on the affine space \(\mathcal{A}(P)\). The latter can be induced by the metrics on \(\Sigma\) and on \(\mathfrak{g}\), and defined pointwise in \(\mathcal{A}(P)\) as
\begin{equation}
\label{YM-metric}
(\alpha,\beta)_A := - \int_\Sigma \mathrm{Tr}(\alpha^A\wedge\star\beta^A)
\end{equation}
for every \(\alpha^A,\beta^A\in \Omega^1(\Sigma;\mathrm{ad}(P))\). With this definition, the YM action can be rewritten as
\begin{equation}
S[A] = \frac{1}{2\epsilon} (F,F)_A .
\end{equation}

We can now specialize the discussion to the case in which \(\dim(\Sigma)=2\), \textit{i.e.}\ the base space is a Riemann surface. It is a well-known fact in geometry that any Riemann surface is a \emph{K\"{a}hler manifold}: it admits a Riemannian metric \(g\), a symplectic form \(\omega\) (that can be a choice of volume form), and a complex structure \(J\) such that the compatibility condition \(g(\cdot,\cdot) = \omega(\cdot,J(\cdot))\) is satisfied.\footnote{A complex structure on a vector space \(V\) is an isomorphism \(J:V\to V\) such that \(J^2=-id_V\). It intuitively plays the role of \lq\lq multiplication by \(i\)\rq\rq\ when one considers the complexified \(V^{\mathbb{C}} := V\otimes \mathbb{C}\), allowing for a decomposition of \(V^{\mathbb{C}}\) in a holomorphic subspace (generated by the eigenvectors with eigenvalue \(+i\)) and anti-holomorphic subspace (generated by the eigenvectors with eigenvalue \(-i\)). A manifold \(M\) has  \emph{almost complex structure} if there is a tensor \(J\in \Gamma(T^1_1 M)\) that acts as a complex structure in every tangent space. If the holomorphic decomposition can be extended on an entire neighborhood of every point by a suitable choice of coordinates, \(M\) has \emph{complex structure}, and admits an atlas of holomorphic coordinates. Riemann surfaces can thus be thought as 2-dimensional real manifolds, or 1-dimensional complex manifolds.} This special property holds also for \(\mathcal{A}(P)\), since in addition to the metric \eqref{YM-metric} we can define the symplectic form \(\Omega\in \Omega^2(\mathcal{A}(P))\) such that
\begin{equation}
\Omega_A(\alpha,\beta) := - \int_\Sigma \mathrm{Tr}(\alpha^A\wedge\beta^A),
\end{equation}
and the complex structure on \(T\mathcal{A}(P)\) is provided by the Hodge duality, \(\star: \Omega^1(\Sigma;\mathrm{ad}(P)) \to \Omega^1(\Sigma;\mathrm{ad}(P))\) such that \(\star^2=-1\). Then the compatibility condition is immediately satisfied, since \((\cdot,\cdot) = \Omega(\cdot, \star(\cdot))\). The fact that \(\Omega\) is symplectic can be seen by noticing that, in any basis, it has constant components (\textit{i.e.}\ independent from \(A\in \mathcal{A}(P)\)):
\begin{equation}
\Omega_{ab}^{\mu\nu}(A) = \Omega_A(T_a\otimes dx^\mu , T_b\otimes dx^\nu) = - \mathrm{Tr}(T_a T_b) \varepsilon^{\mu\nu} \left(\int_\Sigma dx^1dx^2\right) \quad \in \mathbb{R}.
\end{equation}
The non-degeneracy follows from the non-degeneracy of \(\mathrm{Tr}\) and of \(\int_\Sigma\), and the skew-symmetry is obvious from the definition. Thus \(\mathcal{A}(P)\) is K\"{a}hler.

For our applications, we focus on the fact that \(\mathcal{A}(P)\) has now a canonical symplectic structure. It is natural to wonder if it possible to extend all the machinery that we introduced in Section \ref{sec:e-cohom-symplectic} also to this case, and in particular if the \(\mathcal{G}(P)\)-action on \(\mathcal{A}(P)\) results to be symplectic or Hamiltonian with respect to \(\Omega\). The answer was given by in \cite{atiyah-bott-YM}, and we state it in the following theorem.
\begin{thm}[Atiyah-Bott] In 2-dimensions, the group \(\mathcal{G}(P)\) of gauge transformations  acts in an Hamiltonian way on \(\mathcal{A}(P)\), with a moment map identified by the curvature \(F\).
\end{thm}
\begin{proof}
To see this, let us introduce the moment map as \(\mu : Lie(\mathcal{G}(P))\to C^\infty(\mathcal{A}(P))\) such that
\begin{equation}
\label{YM-moment-map}
\mu_\phi(A) := \langle F^A, \phi \rangle = - \int_\Sigma \mathrm{Tr}(F^A\phi) ,
\end{equation}
and check that the Hamiltonian property is satisfied. For every \(\alpha\in \Gamma(T\mathcal{A}(P))\) and \(\phi\in Lie(\mathcal{G}(P))\), we compute
\begin{align}
(\iota_\phi \Omega_A )(\alpha) &= \Omega_A(\underline{\phi}, \alpha) = -\int_\Sigma \mathrm{Tr}(\nabla^A\phi\wedge\alpha^A) = \int_\Sigma \mathrm{Tr}(\phi \nabla^A\alpha^A)  ,\\
\left. \delta \mu_\phi \right|_A(\alpha) &= - \int_\Sigma \mathrm{Tr}\left( F^{A+\alpha}\phi - F^A\phi\right) = - \int_\Sigma \mathrm{Tr}\left(\phi \nabla^A\alpha^A\right) ,
\end{align}
where \(\delta\) is the de Rham differential on \(\mathcal{A}(P)\), that acts in the usual sense of variational calculus. We see that \(\iota_\phi \Omega = - \delta\mu_\phi\), thus \(\mu\) provides a correct moment map for the \(\mathcal{G}(P)\)-action. If we identify \(Lie(\mathcal{G}(P))\) with \(Lie(\mathcal{G}(P))^*\) through the pairing \(\langle\cdot,\cdot\rangle\) introduced above, and regard the curvature \(F:\mathcal{A}(P)\to \Omega^2(\Sigma;\mathrm{ad}(P))\)  as an element of \(C^\infty(\mathcal{A}(P))\otimes Lie(\mathcal{G}(P))^*\), we can simply write that \(\mu \equiv F\).
\end{proof}

Another corollary of \(\mathcal{A}(P)\) being K\"{a}hler is that the path integral measure \(DA\) is formally equivalent to the Liouville measure induced from \(\Omega\), since by compatibility of the structures the latter is equivalent to the Riemannian measure induced by \((\cdot,\cdot)\). Since we are working on an infinite-dimensional space, we can write this measure formally as
\begin{equation}
DA = \exp(\Omega) ,
\end{equation}
as we did in \eqref{hamiltonian-1} but with \(n=\infty\). With this identification, we see that the path integral of the 2-dimensional YM theory acquires the very suggestive form
\begin{equation}
Z(\epsilon) \propto \int_{\mathcal{A}(P)} \exp\left(\Omega -\frac{1}{2\epsilon}(\mu,\mu)\right) .
\end{equation}
This path integral resembles very much an infinite-dimensional version of the type of integrals we treated when discussing the Duistermaat-Heckman localization formula in Section \ref{sec:e-cohom-symplectic}, but with the fundamental difference that now the exponent of the integrand is not the moment map, but its square. We will return to this point in the next section.

Here we notice that in the weak coupling limit \(\epsilon\to 0\) the path integral will receive contributions from the saddle points of the action \(S = \frac{1}{2\epsilon}(\mu,\mu)\), that is the space of solutions of the classical equations of motion \(\nabla^A\star F^A = 0\). Every one of these contributions brings roughly a term that decays as \(\sim \exp\left(-1/\epsilon\right)\) to the partition function, the main one being determined by the absolute minimum at \(\mu=0\), the subspace of flat connections \(\mu^{-1}(0)\subset\mathcal{A}(P)\). Eliminating the redundancy from the gauge freedom of the theory, the most interesting piece of the \emph{physical} field space, especially  in the weak coupling limit, is thus determined by the quotient
\begin{equation}
\mathcal{A}_0 := \faktor{\mu^{-1}(0)}{\mathcal{G}(P)},
\end{equation}
or in other words, when computing the path integral one is interested in the \(\mathcal{G}(P)\)-equivariant cohomology of \(\mu^{-1}(0)\), \(H_{\mathcal{G}}^*(\mu^{-1}(0)) \cong H^*(\mathcal{A}_0)\). The quotient \(\mathcal{A}_0\) is the \emph{moduli space of flat connections}. It turns out that this space has a nice interpretation in symplectic geometry in terms of \emph{symplectic reduction}. A theorem by Marsden-Weinstein-Meyer (MWM) \cite{marsden-weinstein,meyer,dasilva} in fact states that, in a generic Hamiltonian \(G\)-space \((M,\omega,G,\mu)\), if the zero-section of the moment map \(\mu^{-1}(0)\subset M\) is acted on \emph{freely} by \(G\), then the base space \(M_0 := \mu^{-1}(0)/G\) of the principal \(G\)-bundle \(\mu^{-1}(0)\xrightarrow{\pi} M_0\) is a symplectic manifold, with symplectic form \(\omega_0\in \Omega^2(M_0)\) satisfying 
\begin{equation}
\left.\omega\right|_{\mu^{-1}(0)} = \pi^*(\omega_0).
\end{equation}
In other words, the restriction of \(\omega\) to \(\mu^{-1}(0)\) is a basic form, completely determined by a symplectic form \(\omega_0\) on the base space. The space \((M_0,\omega_0)\) is called Marsden–Weinstein quotient, symplectic quotient or symplectic reduction of \(M\) by \(G\).\footnote{Symplectic reduction in classical mechanics on \(M=\mathbb{R}^{2n}\) occurs when one of the momenta is an integral of motion, \(0=\dot{p}_n=-\partial_n H\). In that case, one can solve the system in the reduced coordinates \((q^1,\cdots,q^{n-1},p_1,\cdots,p_{n-1})\) and then solve for the \(n^{th}\) coordinate separately. The MWM theorem essentially generalizes this process in a fully covariant setting.} Returning to the case of YM theory, this means that in the limit \(\epsilon\to 0\), when the path integral is reduced to \(\mathcal{A}_0\) by gauge fixing, the symplectic form can be reduced without loss of information on this base space. In the next section we will see that, applying localization, this is extended to the whole exponential.

\section{A localization formula for non-Abelian actions}
\label{sec:YM-localization-formula}

In the last section we found an Hamiltonian interpretation of the system \((\mathcal{A}(P),\Omega,\mathcal{G}(P),\mu\equiv F)\) for YM theory on a 2-dimensional Riemann surface \(\Sigma\). Here we would like to make contact with the DH formula, that we described for analogous systems in finite-dimensional geometry. We notice that the main differences with the case treated in Section \ref{sec:e-cohom-symplectic} are essentially two: \(\mathcal{G}(P)\) is non-Abelian in general for non-Abelian gauge groups \(G\), and the path integral is not in the form of an oscillatory integral of the DH type. Indeed, schematically we have
\begin{equation}
\label{YM-1}
\text{DH} \to \int \exp{\left(\omega+i\mu\right)} , \qquad\quad \text{YM}\to \int \exp{\left(\Omega-\frac{1}{2}|\mu|^2\right)} .
\end{equation}

In the following, we will describe the solution proposed in \cite{witten-2dYM} to generalize the DH formula to the non-Abelian case starting from the first integral in \eqref{YM-1}, and how this procedure can be used to recover the second one, of the YM type. We consider a generic Hamiltonian system \((M,\omega,G,\mu)\) with compact semisimple Lie group \(G\) of dimension \(\dim(G)=s\), and the associated Cartan model defined by the space \(\Omega_G(M) = (S(\mathfrak{g}^*)\otimes \Omega(M))^G\) of equivariant forms, on which we defined the action of the extended operators\footnote{We adopt Witten's conventions and substitute \(\phi^a\mapsto i\phi^a\) in the definition of the Cartan differential, analogously to the DH case of Section \ref{sec:e-cohom-symplectic}.}
\begin{equation}
\begin{aligned}
d_C &= 1\otimes d - i \phi^a \otimes \iota_a ,\\
\mathcal{L}_a &\equiv 1\otimes \mathcal{L}_a + \mathcal{L}_a\otimes 1 \qquad \text{with}\quad \mathcal{L}_a\phi^b = f^b_{ac}\phi^c ,\\
\iota_a &\equiv 1 \otimes \iota_a .
\end{aligned}
\end{equation}
An element \(\alpha\in \Omega_G(M)\) is an invariant polynomial in the generators \(\phi^a\) of \(S(\mathfrak{g}^*)\), with differential forms on \(M\) as coefficients. This means that integration over \(M\) provides a map in equivariant cohomology of the type
\begin{equation}
\int_M : H^*_G(M) \to S(\mathfrak{g}^*)^G ,
\end{equation}
or in other words that the integral of an equivariant form is in general a polynomial in the \(\phi^a\). This is not quite satisfactory, as we would like an integration that generalizes the standard de Rham case, giving a map \(H^*_G(M) \to \mathbb{R}\) (or \(\mathbb{C}\)). In the case of \(G=U(1)\) we often solved this problem by setting the unique generator \(\phi=-1\) (or \(\phi=i\) in this conventions), thus constructing a map in the localized cohomology, \(\int_M: H_{U(1)}^*(M)_\phi \to \mathbb{R}\).  Here in the non-Abelian case, the trick of algebraic localization is not so trivial in practice, and we avoid it.

An alternative and fruitful idea to saturate the \(\phi\)-dependence is to make them dynamical variables, and integrate over them too. Since \(\phi^a\) can be regarded as an Euclidean coordinate over \(\mathfrak{g}\), this means defining an integration over \(M\times\mathfrak{g}\). As a vector space, the Lie algebra has a natural measure \(d^s\phi\) that is unique up to a multiplicative factor. We fix that factor by choosing a (positive-definite) inner product \((\cdot,\cdot)\) on \(\mathfrak{g}^*\) and setting\footnote{The inner product on \(\mathfrak{g}^*\) is induced from an inner product on \(\mathfrak{g}\), and when we write \((\phi,\phi)\) we really mean \(\sum_a (\phi^a,\phi^a)\). This can be stated more formally defining \(\phi := \phi^a\otimes T_a\in \mathfrak{g}^*\otimes \mathfrak{g}\), and then letting act the inner-product on the \(T_a\)'s, normalized in order to produce a Kronecker delta. We avoid this cumbersome notation, since the action of the various inner products is always clear from the context.}
\begin{equation}
\int_\mathfrak{g} d^s\phi\ e^{-\frac{\epsilon}{2}(\phi,\phi)} = \left(\frac{2\pi}{\epsilon}\right)^{s/2} ,
\end{equation}
essentially as we did in \eqref{pestun-measure}. Since our goal is to integrate equivariant forms that have polynomial dependence on the \(\phi^a\), or at most expressions of the form \(\exp(\phi^a\otimes \mu_a)\) that have exponential dependence, integrating over \(\mathfrak{g}\) with the bare measure \(d^s\phi\) would produce possible divergences. To ensure convergence of these class of functions, the \emph{equivariant integration} is defined as \cite{witten-2dYM}
\begin{equation} \label{eq:equivariant-integral}\boxed{
\int_{M\times \mathfrak{g}} \alpha := \frac{1}{\mathrm{vol}(G)}\int_M \int_{\mathfrak{g}} \frac{d^s\phi}{(2\pi)^s} e^{-\frac{\epsilon}{2}(\phi,\phi) }\alpha 
}
\end{equation}
where \(\epsilon\) is inserted as a regulator. Notice that in general the limit \(\epsilon\to 0\) is not well-defined, for what we said above.

With this enhanced definition of equivariant integration of elements of \(\Omega_G(M)\), we can apply the equivariant localization principle to the present case. Let \(\alpha\in \Omega_G(M)\) be an equivariantly closed form, so that \(d_C\alpha = 0\), and choose an equivariant 1-form \(\beta\in\Omega_G^1(M)=\Omega^1(M)^G\). The latter is independent on \(\phi^a\), and plays the role of the localization 1-form. By the same arguments of Section \ref{sec:eq_loc_princ}, \(\alpha\) and \(\alpha e^{d_C\beta}\) are representatives of the same equivariant cohomology class in \(H^*_G(M)\), and we can deform the integral of \(\alpha\) as
\begin{equation}
\label{YM-2}
I[\alpha;\epsilon] := \int_{M\times \mathfrak{g}}\alpha = \int_{M\times \mathfrak{g}}\alpha e^{td_C\beta} \qquad \forall t\in \mathbb{R}.
\end{equation}
In particular, taking the limit \(t\to \infty\), this integral localizes on the critical point set of the localization 1-form \(\beta\). This can be seen simply by expanding the definition of equivariant integration from \eqref{YM-2},
\begin{equation}
\label{YM-3}
\begin{aligned}
I[\alpha;\epsilon] &= \frac{1}{\mathrm{vol}(G)}\int_M \int_\mathfrak{g} \frac{d^s\phi}{(2\pi)^s}\ \alpha \exp\left(-\frac{\epsilon}{2}(\phi,\phi) -it\phi^a(\iota_a\beta) + td\beta \right) \\
&=  \frac{1}{\mathrm{vol}(G)}\int_M \int_\mathfrak{g} \frac{d^s\phi}{(2\pi)^s}\ \alpha \exp\left(-\frac{\epsilon}{2}(\phi,\phi) -\frac{t^2}{2\epsilon} \sum_a (\iota_a\beta)^2 + td\beta \right) ,
\end{aligned}
\end{equation}
where in the second line we completed the square and shifted variable in the \(\phi\)-integral. Since the term \(td\beta\) gives a polynomial dependence on \(t\) (by degree reasons, it is expanded up to a finite order), the limit \(t\to\infty\) converges and makes the integral localize on the critical points of \(\iota_a\beta\). This shows the localization property of equivariant integrals in the non-Abelian setting. If for example we suppose \(\alpha\) to be independent on the \(\phi^a\), we can perform the Gaussian integration to further simplify \(I[\alpha;\epsilon]\),
\begin{equation}
I[\alpha;\epsilon] =  \frac{1}{\mathrm{vol}(G)(2\pi\epsilon)^{s/2}} \int_M \alpha \exp\left( -\frac{t^2}{2\epsilon} \sum_a (\iota_a\beta)^2 + td\beta \right) .
\end{equation}

\medskip
We now apply the above non-Abelian localization principle to the special case in which \(\alpha = \exp(\omega - i\phi^a\otimes\mu_a)\), \textit{i.e.}\ generalizing the DH formula of Section \ref{sec:e-cohom-symplectic}. We will suppress tensor products in the following, for notational convenience. First of all, it is straightforward to see that this  form is equivariantly closed, 
\begin{equation}
d_C e^{\omega - i\phi^a\mu_a} \propto d_C(\omega - i\phi^a\mu_a) = -i\phi^a\iota_a\omega - i\phi^a d\mu_a = -i\phi^a\iota_a\omega + i\phi^a\iota_a\omega =0.
\end{equation}
The DH oscillatory integral becomes, following the same steps of \eqref{YM-2} and \eqref{YM-3},
\begin{equation}
\label{YM-4}
\begin{aligned}
Z(\epsilon) &= \int_{M\times \mathfrak{g}} \frac{\omega^n}{n!}e^{-i\phi^a\mu_a} \\
&= \frac{1}{\mathrm{vol}(G)} \int_M \int_\mathfrak{g} \frac{d^s\phi}{(2\pi)^s}  \exp\left(\omega  -i\phi^a \mu_a -\frac{\epsilon}{2}(\phi,\phi) + t(d\beta - i\phi^a (\iota_a\beta)) \right) \\
&= \frac{1}{\mathrm{vol}(G)(2\pi\epsilon)^{s/2}} \int_M \exp\left(\omega -\frac{1}{2\epsilon} (\mu,\mu) - \frac{t^2}{2\epsilon}\sum_a (\iota_a\beta)^2 + \frac{it}{\epsilon} \sum_a \mu_a (\iota_a\beta)  \right) ,
\end{aligned}
\end{equation} 
and it is independent of \(t\). Specializing to the case \(t=0\), we get 
\begin{equation}
\label{YM-6}
\boxed{
\int_{M\times \mathfrak{g}} \frac{\omega^n}{n!}e^{-i\phi^a\mu_a} = \frac{1}{\mathrm{vol}(G)(2\pi\epsilon)^{s/2}} \int_M \exp\left(\omega -\frac{1}{2\epsilon} (\mu,\mu)\right)
},
\end{equation}
that shows the equivalence of the YM type partition function and the equivariant integral of the DH type!  If instead we take the limit \(t\to\infty\), we see that the integral localizes on the critical points of \((\iota_a\beta)\). With a smart choice of localization 1-form, we can show that this localization locus coincides with the critical point set of the function \(S:= \frac{1}{2}(\mu,\mu)\).
\begin{proof}
Since \((M,\omega)\) is symplectic, it admits an almost complex structure \(J\in \Gamma(T^1_1M)\) and a Riemannian metric \(G\in \Gamma(T^0_2 M)\) such that \(\omega(\cdot,J(\cdot))=G(\cdot,\cdot)\) (see \cite{dasilva}, proposition 12.6).\footnote{In the case of 2-dimensional YM theory, we recall that \(J\equiv\star\) is the Hodge duality operator.} We pick the localization 1-form \[
\beta := dS\circ J = J^\sigma_\nu \partial_\sigma S\ dx^\nu = J^\sigma_\nu \sum_a \mu_a (\partial_\sigma \mu_a)\ dx^\nu ,\]
and the localization condition \(\iota_a \beta = 0\). Now we use the compatible metric \(G\), that has components \(G_{\mu\nu}=\omega(\partial_\mu,J(\partial_nu)) = \omega_{\mu\sigma}J^\sigma_\nu\). We consider its \lq\lq inverse\rq\rq\ \(G^{-1}\) acting on \(T^*M\) with components \(G^{\mu\nu} = J^\mu_\sigma \omega^{\nu\sigma}\), where \(\omega^{\nu\sigma}\) are the components of the \lq\lq inverse\rq\rq\ symplectic form, and compute the norm of the 1-form \(dS\), \[
G^{-1}(dS,dS) = (\partial_\mu S J^\mu_\sigma) \omega^{\nu\sigma} \partial_\nu S = \beta_\sigma \sum_a \mu_a (\omega^{\nu\sigma}\partial_\nu \mu_a ) = - \sum_a \mu_a \beta_\sigma (T_a)^\sigma = - \sum_a \mu_a (\iota_a\beta) = 0 ,
\]
where we used the Hamiltonian equation \(d\mu_a = -\iota_a\omega\) and the localization condition \(\iota_a\beta=0\). By the non-degeneracy of \(G\), this condition is equivalent to \(dS = 0\), that precisely identifies the critical points of \(S\).
\end{proof}

Rephrasing the above result in the language of the last section, we just showed in general terms that  the 2-dimensional YM partition function localizes on the moduli space of solutions of the EoM, meaning that this  theory is essentially classical. We remark again that this localization locus consists of two qualitatively different types of points: those that minimize absolutely \(S\), that is \(\mu^{-1}(0)\subset M\), and the higher extrema with \(\mu\neq 0\). The former ones in the gauge theory are the flat connections, and they give the dominant contribution to the partition function. The latter ones decays exponentially in the limit \(\epsilon\to 0\) as \(\sim\exp(-S/\epsilon)\). In general thus the partition function can be written as a sum of terms coming from all these disconnected regions of \(M\),
\begin{equation}
Z(\epsilon) = \sum_n Z_n(\epsilon) .
\end{equation}
Let us consider the dominant piece \(Z_0(\epsilon)\) coming from \(\mu^{-1}(0)\), that we interpret in the gauge theory as the rough answer in the weak coupling limit, and that we can select by restricting the integration over \(M\) to a suitable neighborhood \(N\) of \(\mu^{-1}(0)\), 
\begin{equation}
Z_0(\epsilon) = \frac{1}{\mathrm{vol}(G)}\int_N \int_\mathfrak{g} \frac{d^s\phi}{(2\pi)^s} \exp\left( \omega - i\phi^a\mu_a - \frac{\epsilon}{2}(\phi,\phi) + td_C(dS\circ J) \right),
\end{equation}
where we inserted the localization 1-form such that the \(t\to\infty\) limit identifies the critical locus \(\mu^{-1}(0)\). Cohomological arguments show that, if \(G\) acts freely on \(\mu^{-1}(0)\), this integral retracts on the symplectic quotient \(M_0 := \mu^{-1}(0)/G\), giving
\begin{equation}
\label{YM-5}
Z_0(\epsilon) = \int_{M_0} \exp\left(\omega_0 + \epsilon\Theta \right)
\end{equation}
for some 4-form \(\Theta\in \Omega^4(M_0)\). In particular, we see that in the weak coupling limit \(\epsilon\to 0\), \(Z_0(0)\) gives the volume of the symplectic quotient \(M_0\).
\begin{proof}[Argument for \eqref{YM-5}]
The precise proof is technical and it can be found in \cite{witten-2dYM}, we only sketch the main instructive ideas here. The neighborhood \(N\) is chosen small enough to be preserved by the \(G\)-action, and represents the split with respect to the normal bundle we used in Section \ref{sec:ABBVproof}. Thus it retracts equivariantly onto \(\mu^{-1}(0)\), meaning that it is homotopic to \(\mu^{-1}(0)\) and that the homotopy commutes with the \(G\)-action.

First of all we recall what we noticed at the end of the last section: if \(G\) acts freely on \(\mu^{-1}(0)\) the MWM theorem tells us that the symplectic form retracts on the symplectic quotient \( M_0 := \mu^{-1}(0)/G \), so it does not contribute to the integration over the \lq\lq normal directions\rq\rq\ to \( M_0\) in \(N\). Here we are not considering a simple integration over \(N\), but an equivariant integration, that provides a map \(H^*_G(N)\to \mathbb{C}\). So in this case we consider the equivariantly closed extension  \(\tilde{\omega}=\omega-i\phi^a\mu_a\) as representative of a cohomology class \([\tilde{\omega}]\in H^2_G(M)\). When restricted over \(N\), this class is the pull-back of a cohomology class \([\omega_0]\in H^2(M_0)\), since \(H^*_G(N)\cong H^*_G(\mu^{-1}(0))\cong H^*(M_0)\), the first equivalence following from the retraction of \(N\) onto \(\mu^{-1}(0)\) and the second from the fact that the \(G\)-action is free on \(\mu^{-1}(0)\) (these properties were explained in Chapter \ref{cha:equivariant cohomology}). Thus we can substitute in the integral \(\tilde{\omega}\mapsto\omega_0\) without changing the final result.

The same kind of argument works for the term \( \frac{1}{2}(\phi,\phi)\). It is easy to check that this is both \(G\)-invariant and equivariantly closed, so it represents an element \( \left[ \frac{1}{2}(\phi,\phi)\right]\in H^4_G(M)\). When we restrict it to \(N\), as above, this class is the pull-back of some class \([\Theta]\in H^4(M_0)\), and we can make the substitution \(\frac{1}{2}(\phi,\phi) \mapsto \Theta\) in the integral without changing the final result.

Since both \(\omega_0\) and \(\Theta\) are standard differential forms over \(M_0\) and thus independent of \(\phi\), the integration over \(\mathfrak{g}\) goes along only with the remaining term \(d_C(dS\circ J)\). We already know that on \(\mu^{-1}(0)\subset N\) this term is zero, so one has to show that its integral over the normal directions to \(\mu^{-1}(0)\) in \(N\) produces a trivial factor of 1. In \cite{witten-2dYM} it is proven that \[
\frac{1}{\mathrm{vol}(G)}\int_{F}\int_\mathfrak{g}\frac{d^s\phi}{(2\pi)^s} \exp\left(t d_C (dS\circ J) \right) = 1,\]
where \(F\) is any fiber of the normal bundle to \(\mu^{-1}(0)\) in \(N\). From this \eqref{YM-5} follows.
\end{proof}

\begin{ex}[The height function on the 2-sphere, again]
Beside the main application of  Witten's localization principle to non-Abelian gauge theories, we try now to apply this new formalism to the old and simple example of the height function on the 2-sphere, to compare it with the results obtained in Chapter \ref{cha:loc theorems}. Setting \(G=U(1)\), \(M=\mathbb{S}^2\), \(\mu = \cos(\theta)\) and \(\omega = d\cos(\theta)\wedge d\varphi\), the equivariant integration \eqref{eq:equivariant-integral} of the DH oscillatory integral gives \[
\begin{aligned}
Z(\epsilon) &= \int_{\mathbb{S}^2\times \mathfrak{g}} \omega e^{-i\phi\mu} = \frac{1}{2\pi} \int_0^{2\pi}d\varphi \int_{-1}^{+1}d\cos\theta \int_{-\infty}^{+\infty} \frac{d\phi}{2\pi} \exp\left( -i\phi \cos\theta -\frac{\epsilon}{2}\phi^2\right) \\ 
&= \frac{1}{\sqrt{2\pi\epsilon}} \int_{-1}^{+1} dx\ e^{-\frac{x^2}{2\epsilon}} = 1 - 2 I(\epsilon) ,
\end{aligned}
\]
where \(I(\epsilon) := \int_{1}^{\infty} \frac{dx}{\sqrt{2\pi\epsilon}}\exp(-x^2/2\epsilon)\) is a trascendental error function. The three terms in the final result for \(Z(\epsilon)\)  (two of which are equal to \(-I(\epsilon)\)) correspond to the contributions of the extrema of \((\cos\theta)^2\): the two maxima at \(\theta=0,\pi\) contribute with \(-I(\epsilon)\) and the minimum at \(\theta=\pi/2\) contributes with \(1\). The latter is the dominant piece when \(\epsilon\to 0\), since \(I(0)=0\). We see that in general the modified equivariant integration of this new formalism gives an incredibly complicated answer, when compared to the simple result of Example \ref{ex:S2-ABBV} obtained via the usual equivariant localization principle.
\end{ex}

\begin{rmk} We notice that we could have   expressed equivalently the whole dissertation above in supergeometric language, since we discussed in Section \ref{sec:superspaces} that integration over \(M\) is equivalent to integration over \(\Pi TM\). In these terms, maybe more common in QFT, we can introduce coordinates \((x^\mu, \psi^\mu, \phi^a)\) over \(\Pi TM\times \mathfrak{g}\), where \(\psi^\mu := dx^\mu\) are Grassmann-odd, and interpret elements of \(\Omega_G(M)\) as elements of \(C^\infty(\Pi TM\times \mathfrak{g})^G\). An equivariant form \(\alpha\) is thus (locally) a \(G\)-invariant function of \((x,\psi,\phi)\). For example, the Cartan differential and the definition of equivariant integration become \begin{equation}
\label{YM-susy-language}
\begin{aligned}
d_C &= \psi^\mu \frac{\partial}{\partial x^\mu} - i\phi^a T_a^\mu \frac{\partial}{\partial \psi^\mu} ,\\
\int_{M\times \mathfrak{g}} \alpha &:= \frac{1}{\mathrm{vol}(G)}\int d^{2n}x d^{2n}\psi \frac{d^s\phi}{(2\pi)^s}\ \alpha(x,\psi,\phi)e^{-\frac{\epsilon}{2}(\phi,\phi)} .
\end{aligned}
\end{equation}
\end{rmk}

\section{\lq\lq Cohomological\rq\rq\ and \lq\lq physical\rq\rq\ YM theory}
\label{sec:YM-cohomologicalFT}

In this section we are going to review the relation between 2-dimensional YM theory that we described in Section \ref{sec:YM-prelude} and a topological field theory (TFT) that can be viewed as its \lq\lq cohomological\rq\rq\ counterpart. We can translate almost verbatim the general principles that we discussed in the last section, setting 
\begin{equation}
M \mapsto \mathcal{A}(P), \quad \omega \mapsto \Omega, \quad G \mapsto\mathcal{G}(P), \quad \mu \mapsto F,
\end{equation}
while we regard \(G\) as a compact connected Lie group that acts as the gauge group on the principal bundle \(P\to\Sigma\) over a Riemann surface \(\Sigma\). The moment map is formally equivalent to the curvature \(F\) if we identify \(Lie(\mathcal{G}(P)) \cong Lie(\mathcal{G}(P))^*\) through an inner product on \(\mathfrak{g}\), as in \eqref{YM-moment-map}.

The non-Abelian localization principle of the last section, if used in reverse, already showed that an equivalent way to express the standard YM theory is through a \lq\lq first-order formulation\rq\rq\
\begin{equation}
S[A,\phi] = -\int_\Sigma \mathrm{Tr}\left(i\phi F^A + \frac{\epsilon}{2} \phi\star\phi \right)
\end{equation}
where we consider \(\phi\in \Omega^0(\Sigma;\mathrm{ad}(P)) \cong Lie(\mathcal{G}(P))\), and \(F^A\in \Omega^2(\Sigma;\mathrm{ad}(P))\) is the curvature of \(A\).\footnote{More precisely, we should say that \(A^a_\mu\) and \(\phi^a\) are \emph{coordinates} functions on \(\mathcal{A}(P)\times Lie(\mathcal{G}(P))\), so they effectively are elements of \(C^\infty(\mathcal{A}(P))\otimes Lie(\mathcal{G}(P))^*\). This caveat will be logically important in the following, and it goes along with the \emph{functor of points} approach we used in the supergeometric discussion of Chapter \ref{cha:susy}.} This is essentially what  is written in \eqref{YM-6}, where on the LHS we have the first-order action (the \(\epsilon\) dependence is contained in the equivariant integration), and on the RHS we have the standard YM action \(S[A]=\frac{1}{2\epsilon}(F,F)_A\). The first-order formulation has the quality of showing very clearly the weak coupling limit behavior when \(\epsilon\to 0\), that is less obvious in the standard formulation. In this limit, the theory becomes topological, in the sense that the action does not depend on the metric anymore (the metric appears in the Hodge duality \(\star\)),
\begin{equation}
S_{\epsilon\to 0}[A,\phi] = -\int_\Sigma \mathrm{Tr}(i\phi F^A).
\end{equation}
This theory is called \lq\lq BF model\rq\rq, and it is the prototype of a TFT of Schwarz-type. The YM theory can thus be seen as a \lq\lq regulated version\rq\rq\ of a truly topological field theory.\footnote{Notice that, although YM theory is clearly dependent on the metric of \(\Sigma\), in 2 dimensions it shows a very \lq\lq weak\rq\rq\ dependence to it. In fact, in this dimensionality the action can be simplified as (suppressing the constants and the Lie algebra inner product) \[
\begin{aligned}
\int F\wedge\star F &= \int g^{\mu\rho}g^{\nu\sigma}F_{\mu\nu}F_{\rho\sigma} \sqrt{|g|}d^2x = \int (F_{01})^2 g^{\mu\rho}g^{\nu\sigma}\varepsilon_{\mu\nu}\varepsilon_{\rho\sigma} \sqrt{|g|}d^2x = \\ &= \int (F_{01})^2 |g^{-1}|\sqrt{|g|}d^2x = \int (F_{01})^2 |g|^{-1/2}d^2x  ,
\end{aligned}\]
so the metric does not appear through its components, but only in the invariant quantity \(\det(g)\).} 
At least classically, it is intuitive from the EoM with respect to \(\phi\) that the only contribution to the classical solutions comes from the moduli space of flat connections, where \(F^A=0\) up to gauge transformations. It is not trivial, though, to infer that this is all the theory has to offer also at the quantum level, that is essentially the result we showed in general terms in the last section, via the localization principle applied to the path integral \(Z(\epsilon\to 0)\).

In this section we discuss, following \cite{witten-2dYM}, how the localization principle can be translated in the language of TFT, in order to give a more physical interpretation of the abstract mathematical results that we discussed in finite dimensions. In particular, we will see that the BF model partition function can be recovered as an expectation value in a TFT, and that this ensures its localization properties onto the moduli space of flat connections. The regulated version at \(\epsilon\neq 0\) will not follow precisely this behavior, as we already know that higher extrema of the YM action contribute to the partition function \(Z(\epsilon)\), but these will have a nice interpretation in terms of the moduli space.

\subsection*{Intermezzo: TFT}

This is a good moment to explain briefly in more general terms what one usually means by TFT, and how localization enters in this subject. Traditionally, TFT borrows the language of BRST formalism for quantization of gauge theories, as many examples of topological theories arise from that context. Recall that the standard BRST quantization procedure is based on the definition of a differential, the \lq\lq BRST charge\rq\rq\ \(Q\), that acts on an extended graded field space, whose grading counts the \lq\lq ghost number\rq\rq\ (in other words, \(Q\) is an operator of degree \(\mathrm{gh}(Q)=+1\)). The BRST charge represents an infinitesimal supersymmetry transformation, that squares to zero in the gauge-fixed theory. On the Hilbert space, physical states are those of ghost number zero, and that are annihilated by the BRST charge,
\begin{equation}
Q|\text{phys}\rangle = 0, \quad \mathrm{gh}|\text{phys}\rangle = 0 .
\end{equation}
The action of \(Q\) on the field space is often denoted with a Poisson bracket-like notation, \(\Phi \mapsto \delta_Q\Phi := - \lbrace Q,\Phi\rbrace\) for any field \(\Phi\). By gauge invariance of the vacuum, \(Q|0\rangle = 0\), and so for any operator \(\mathcal{O}\) one has that \(\langle 0|\lbrace Q,\mathcal{O}\rbrace |0\rangle = 0\).

For a QFT being \lq\lq topological\rq\rq, in physics one usually means that all its quantum properties are independent from a choice of a metric on the base space \(M\).\footnote{Here the word \lq\lq topological\rq\rq\ is somewhat overused. Mathematically, a topological space consists of a set \(M\) and a \emph{topology} \(\mathcal{O}_M\), that is roughly the set of all \lq\lq open neighborhoods\rq\rq\ in \(M\). In QFT one almost always works on base spaces that have the structure of a \emph{manifold} of some kind (smooth, complex, ...), so that it allows for the presence of an \emph{atlas} \(\mathcal{A}_M\) of charts that identifies it locally as \(\mathbb{R}^n\) for some (constant) \(n\). A manifold is thus a triple \((M,\mathcal{O}_M,\mathcal{A}_M)\), and the choice of a metric is only on top of this structure. So metric-independence does not generically mean that the QFT describes only the topology of \(M\), but it can depend on the choice of a smooth (or complex, ...) structure on \(M\).} This is rephrased in the requirement that the partition function of the theory should be metric-independent. Assuming the path integral measure to be metric-independent and \(Q\)-invariant (so that the BRST symmetry is non anomalous), the variation with respect to the metric of the partition function is
\begin{equation}
\delta_g Z \propto \int_\mathcal{F} [D\Phi] e^{-S[\Phi]} \delta_g S[\Phi] = \int_\mathcal{F} [D\Phi] e^{-S[\Phi]}\left( \int_M d^n x \sqrt{|g|} \delta g^{\mu\nu} T_{\mu\nu}\right) \propto \langle 0|T_{\mu\nu}|0\rangle ,
\end{equation}
so a suitable definition of TFT is the one that requires the energy-momentum tensor \(T_{\mu\nu}\) to be a BRST variation, \(T_{\mu\nu} = \lbrace Q,V_{\mu\nu}\rbrace\) for some operator \(V_{\mu\nu}\). This would ensure \(\delta_g Z = 0\) for what we said above.

Collecting the above remarks, we can give the following \lq\lq working definition\rq\rq\ \cite{birmingham-TFT}. A \emph{Topological Field Theory} is a QFT defined over a \(\mathbb{Z}\)-graded field space \(\mathcal{F}\), with a nilpotent operator \(Q\) (\textit{i.e.}\ a \emph{cohomological vector field} on \(\mathcal{F}\)), and a \(Q\)-exact energy-momentum tensor \(T_{\mu\nu} = \lbrace Q,V_{\mu\nu}\rbrace\), for some \(V_{\mu\nu}\in C^\infty(\mathcal{F})\). \emph{Physical states} are defined to be elements of the \(Q\)-cohomology of \(\mathcal{F}\) in degree zero, \(|\text{phys}\rangle \in H^0(\mathcal{F}, Q)\).\footnote{We stress that, if we see \(\mathcal{F}\) as a graded extension of an original field space \(\mathcal{F}_0\) acted upon by gauge transformations, the \(Q\)-cohomology of \(\mathcal{F}\) is exactly the analogous of the gauge-equivariant cohomology \(H_{gauge}^*(\mathcal{F}_0)\), computed in the Cartan model with Cartan differential \(Q\).}

\begin{rmk}\begin{itemize}
\item \(Q\) is called \lq\lq BRST charge\rq\rq\ or \lq\lq operator\rq\rq, but in general it can be every supersymmetry charge (as it is a cohomological vector field). We saw examples in the last chapter where \(Q = \delta_{susy} + \delta_{BRST}\), where \(\delta_{BRST}\) is the actual gauge-supersymmetry, and \(\delta_{susy}\) is a Poincaré-supersymmetry.

\item The one above is a good \lq\lq working definition\rq\rq\ for most examples, but it is not completely adequate in all cases. Indeed, there are examples of QFT where \(T_{\mu\nu}\) fails to be BRST-exact, but nonetheless one can still establish the topological nature of the model. We do not need to treat any example of this kind, so we refer to \cite{birmingham-TFT} for more details.

\item We already encountered an example of TFT in Section \ref{sec:index}, \textit{i.e.}\ supersymmetric QM. There, we called it topological because its partition function was determined by topological invariants of the base space, here we point out that it indeed fits into this general discussion. In fact, in \eqref{index-Q-exact-S} we saw that its action can be expressed as \(S = \lbrace Q, \Sigma\rbrace\), where \(Q\equiv Q_{\dot{x}}\) was the \(U(1)\)-Cartan differential on the loop space, and \(\Sigma\) some other loop space functional. The \(Q\)-exactness of the energy-momentum tensor follows from this.

\item If the energy-momentum tensor is \(Q\)-exact, in particular the Hamiltonian satisfies 
\begin{equation}
\langle \text{phys}| H |\text{phys}\rangle = \langle \text{phys}| T_{00} |\text{phys}\rangle = 0 .
\end{equation}
This means that the energy of any physical state is zero, and thus the TFT does not contain propagating degrees of freedom. It can only describe \lq\lq topological\rq\rq\ properties of the base space.
\end{itemize}
\end{rmk}

TFTs fall in two main broad categories. The first one is constituted by the so-called TFT of \emph{Witten-type} (or also \emph{cohomological TFT}). Their defining property is to have a \(Q\)-exact quantum action,
\begin{equation}
S_q = \lbrace Q, V\rbrace
\end{equation}
for some operator \(V\). As the case of supersymmetric QM, the energy-momentum tensor is automatically \(Q\)-exact too,
\begin{equation}
T_{\mu\nu} = \frac{2}{\sqrt{|g|}} \left\lbrace Q, \delta V/\delta g_{\mu\nu}\right\rbrace .
\end{equation}
From the equivariant point of view, these  theories have a very simple localization property. The action is representative of the trivial  \(Q\)-equivariant cohomology class, so the partition function can be in fact written as
\begin{equation}
\label{YM-8}
Z \propto \int_\mathcal{F}[D\Phi] e^{-tS_q[\Phi]}
\end{equation}
for any \(t\in \mathbb{R}\), by the standard argument of supersymmetric localization. In this case \(\mathcal{F}\) is non-compact and so simply taking \(t=0\) is not really allowed, but the limit \(t\to\infty\) is perfectly defined, producing the path integral localization onto the space of solutions of the classical EoM of \(S_q\). This means that all TFT of Witten-type are completely determined by their semiclassical approximation!

The second main class of TFT is called of \emph{Schwarz-type} (or also \emph{quantum TFT}). In this case one starts with a classical action \(S\) that is metric-independent, so that the classical energy-momentum tensor is zero. The usual BRST quantization of this action produces a quantum action of the type \(S_q = S + \lbrace Q,V\rbrace\) for some \(V\), and again
\begin{equation}
T_{\mu\nu} = \frac{2}{\sqrt{|g|}} \left\lbrace Q, \delta V/\delta g_{\mu\nu}\right\rbrace ,
\end{equation}
since the classical piece does not contribute. Chern-Simons theory and BF theory are of this type. Regarding the second remark above, we mention that some Schwarz-type TFTs fail to respect the \(Q\)-exactness property of \(T_{\mu\nu}\), for example in the case of non-Abelian BF theories in dimension \(d>3\) (this happens essentially because one can ensure nilpotency of the BRST charge only on-shell, and moreover \(Q\) has to be defined in a metric-dependent way).

As we anticipated before we are interested in BF theories, an example of Schwarz-type TFT, and in particular in their 2-dimensional realization. On a base space \(\Sigma\) of generic dimension \(d\), the classical action of the BF theory is
\begin{equation}
S_{BF}[A,B] := -\int_\Sigma \mathrm{Tr}( B\wedge F^A ),
\end{equation}
where \(F^A\) is the curvature 2-form of a connection \(A\), and \(B\) is a \((d-2)\)-form with values in the adjoint bundle. The BRST charge is the usual gauge-supersymmetry. In the 2-dimensional case, \(B\equiv i\phi\in \Omega^0(\Sigma;\mathrm{ad}(P))\), and the naive BRST quantization of this theory has no problems (as we mentioned, in more than 3 dimensions the definition of \(Q\) have to be modified and things complicate). The quantum action is 
\begin{equation}
\begin{aligned}
S_q[A,B,\pi,b,c] &= -\int_\Sigma \mathrm{Tr}\left( BF^A + \pi(\nabla^A\star A) + b(\nabla^A\star\nabla^A)c \right) \\
&= S_{BF}[A,B] + \lbrace Q,V\rbrace ,  \\
\text{with} \quad V &:= -\int_\Sigma \mathrm{Tr}\left( b(\nabla^A\star A)\right) ,
\end{aligned}
\end{equation}
where we introduced a ghost \(c\), an anti-ghost \(b\) and an auxiliary field \(\pi\) with the BRST transformations properties
\begin{equation}
\delta_Q A = \nabla^A c , \quad \delta_Q c = -\frac{1}{2}[c,c] , \quad \delta_Q b = \pi , \quad \delta_Q \pi =0, \quad \delta_Q B = -[B,c] .
\end{equation}
Even though this theory is not of Witten-type, so it has no direct localization onto the subspace of classical solutions, it turns out that this is still the case. Traditionally, this is shown finding a suitable redefinition of coordinates in field space that trivializes the bosonic sector of the action (meaning that there are no derivatives acting on bosonic fields), and whose Jacobian cancels in the path integral with the 1-loop determinant obtained by integrating out the fermions (in this case, the ghosts fields). This is called \emph{Nicolai map}, and for the 2-dimensional BF model is given by the redefinitions \cite{birmingham-TFT,blau-TFT}
\begin{equation}
\xi(A) := F^A, \qquad \eta(A):= \nabla^{A_c} \star A_q ,
\end{equation}
where \(A:= A_c + A_q\) is the expansion of the gauge field around a classical (on-shell) solution \(A_c\). Assuming the fermions being integrated out, the path integrals over \(B\) and \(\pi\) exactly identify the space of zeros of \(\xi,\eta\), that is the moduli space of solutions to \(F^A=0\) up to gauge transformations, \(\mathcal{A}_0\). We do not pursue this direction further, but summarize the approach taken in \cite{witten-2dYM}, more related to localization.

\subsection*{Cohomological approach}

Thinking in equivariant cohomological terms, one can expect to show the localization of the BF model (and of its \lq\lq regulated\rq\rq\ version, \emph{i.e.}\ YM theory) finding a suitable \lq\lq localization 1-form\rq\rq\ \(V'\). This has to be such that the deformation of the action given  by \(t\lbrace Q, V'\rbrace\) induces the path integral to localize in this subspace when \(t\to\infty\), analogously to what we did for example in Section \ref{sec:QM}, and also to the discussion of the last section in the finite-dimensional case. This is exactly what is shown in \cite{witten-2dYM}, where the partition function of the model of interest is found as an expectation value inside a cohomological TFT, proving automatically its localization behavior.

This cohomological TFT is constructed in a way such that the action of the BRST operator \(Q\) coincides with the Cartan differential \(d_C\) arising in the symplectic formulation of 2-dimensional YM theory. Using the more common supergeometric language of Sections \ref{sec:QM}, \ref{sec:index} and recalled in \eqref{YM-susy-language} (in finite dimensions), we move from the field space \(\mathcal{A}(P)\times Lie(\mathcal{G}(P))\) to \(\mathcal{F}:= \Pi T\mathcal{A}(P)\times Lie(\mathcal{G}(P))\), introducing the graded coordinates \((A_\mu^a, \psi_\mu^a, \phi^a)\) and regarding \(d_C\) as a supersymmetry (BRST) transformation,
\begin{equation}
\label{YM-7}
d_C \equiv - \lbrace Q,\cdot\rbrace = \int_\Sigma d\Sigma \left( \psi_\mu^a \frac{\delta}{\delta A^a_\mu} - i(\underline{\phi}^a)^\mu \frac{\delta}{\delta \psi_\mu^a} \right) ,
\end{equation}
where we recall from \eqref{YM-fund-vf} that the fundamental vector field associated to the action of a Lie algebra element \(\phi\in Lie(\mathcal{G}(P))\) is \(\underline{\phi} = \nabla \phi\). The BRST transformation of every field follows the rule \(\delta_Q \Phi = -\lbrace Q,\Phi\rbrace \equiv d_C \Phi\), and  on the coordinates we have
\begin{equation}
\delta_Q A = \psi , \qquad \delta_Q \psi = -i\nabla\phi , \qquad \delta_Q \phi = 0.
\end{equation}
The ghost numbers of the elementary fields are \(\mathrm{gh}(A,\psi,\phi)=(0,1,2)\). Other multiplets could be added as well, but this is the basic one we start with.

Before describing in more detail the particular cohomological theory and its relation with the physical YM theory, we summarize the strategy that has been followed. One starts with a suitable cohomological TFT with action
\begin{equation}
S_c = \lbrace Q, V\rbrace ,
\end{equation}
for some operator \(V\), that ensures the localization onto the moduli space \(\mathcal{A}_0\) of flat connections.
Then the mapping to the physical theory is done by the common equivariant localization procedure. The TFT is deformed adding a cohomologically trivial localizing action,
\begin{equation}
\label{YM-10}
S(t) = S_c + t \lbrace Q, V'\rbrace = \lbrace Q, V + tV'\rbrace ,
\end{equation}
for some gauge invariant operator \(V'\) that forces only the interesting YM multiplet \((A,\psi,\phi)\) to survive. Since the field space here is non-compact, some additional care must be taken in claiming the \(t\)-independence of the deformed theory. In particular, the new term must not introduce new fixed points of the \(Q\)-symmetry, that would contribute to the localization locus of the resulting theory. These fixed points, if presents in the theory at \(t\neq 0\), can be interpreted as \lq\lq flowing from infinity\rq\rq\ in the moduli space (since this is, as just remarked, non-compact).\footnote{This will be exactly the case in going to the YM theory with \(\epsilon\neq 0\), where the new fixed points are just the higher extrema of the action \(S=\frac{1}{2}(F,F)\).} If this is not the case, one can infer properties of the \lq\lq physical\rq\rq\ theory at \(t=\infty\) by making computations in the cohomological one at \(t=0\).\footnote{Notice that this is logically the opposite of what one does usually in using localization. Here the \lq\lq easy\rq\rq\ theory is the cohomological one at \(t=0\), while the more difficult (but more interesting) is the one at \(t\neq 0\).}

\paragraph*{The cohomological theory}

The cohomological theory considered in \cite{witten-2dYM} makes use of the additional multiplets \((\lambda,\eta)\) and \((\chi,H)\) with the transformation properties
\begin{equation}
\delta_Q \lambda = \eta , \qquad \delta_Q \eta = -i[\phi,\lambda], \qquad \delta_Q \chi = -i H , \qquad \delta_Q H = [\phi,\chi] .
\end{equation}
The extended field space and ghost numbers are
\begin{equation}
\begin{array}{rccccc}
\mathcal{F} = & \underbrace{\Pi T\mathcal{A}(P)\times Lie(\mathcal{G}(P))} & \times & \underbrace{\left(\Omega^0(\Sigma;\mathrm{ad}(P))\right)^2} & \times & \underbrace{\left(\Omega^0(\Sigma;\mathrm{ad}(P))\right)^2} \\
 & (A,\psi,\phi) & & (\lambda,\eta) & & (\chi,H) \\
\mathrm{gh} = & (0,1,2) & & (-2,-1) & & (-1,0) .
\end{array}
\end{equation}
Of course the definition of the BRST operator will be extended from \eqref{YM-7}, but we do not need it explicitly. The operator \(V\) defining the TFT is chosen to be
\begin{equation}
\label{YM-9}
V = \frac{1}{h^2} \int_\Sigma d\Sigma\ \mathrm{Tr}\left( \frac{1}{2}\chi(H-2(\star F^A)) + g^{\mu\nu} (\nabla^A_\mu \lambda)\psi_\nu \right) ,
\end{equation}
where \(h\in \mathbb{R}\) is a parameter from which the theory is completely independent (it is analogous to \(1/t\) appearing in \eqref{YM-8}), that can be interpreted as the \lq\lq coupling constant\rq\rq\ of the TFT. Computing the cohomological action \(S_c\) one sees that the \(H\) field plays an auxiliary role, and can be eliminated setting \(H=\star F\). Analyzing then the theory in the \(h\to 0\) limit (its semiclassical \lq\lq exact\rq\rq\ approximation), the localization locus is identified by the BRST-fixed point (we refer to \cite{witten-2dYM} for the details)
\begin{equation}
\delta_Q \chi = 0 \quad (\Rightarrow F=0), \qquad \delta_Q \psi = 0 \quad (\Rightarrow \nabla\phi=0) .
\end{equation}
This is analogous to the usual situation in Poincaré-supersymmetric theories, as discussed in Chapter \ref{cha:loc-susy}, where the localization locus was always identified by the subcomplex of BPS configurations, given by the vanishing of the variation of the fermions. Here the \lq\lq fermions\rq\rq\ are the fields with odd ghost number. This means that the final moduli space contains \(\mathcal{A}_0\), plus maybe some contributions from the zero-modes of the other bosonic (even ghost number) fields, \(\lambda\) and \(\phi\).

\paragraph*{\(t\neq 0\) deformation}

The deformation \eqref{YM-10} can be made in order to reduce effectively the field content of the theory to the YM multiplet only, \((A,\psi,\phi)\). In particular, to eliminate the non-trivial presence of the field \(\lambda\) from the contributions to the localization locus one can consider
\begin{equation}
V' = -\frac{1}{h^2}\int_\Sigma d\Sigma\ \mathrm{Tr}\left(\chi\lambda \right).
\end{equation}
Computing the deformed action \(S(t)\), one sees that for \(t\neq 0\) all the additional fields \(H,\chi,\lambda,\eta\) can be integrated out (again, some more details on the technical passages can be found in \cite{witten-2dYM}). In particular, the EoM for \(H\) and \(\lambda\) are
\begin{equation}
\label{YM-11}
H = 0, \qquad \lambda = -\frac{1}{t}(\star F) .
\end{equation}
This is already the sign that the localizing term \(V'\) qualitatively changed the localization property of the theory. In fact, we see that for \(t=0\) we do not have an algebraic equation for \(\lambda\), but \eqref{YM-11} reduces to the solution \(F=0, H=0\) of the cohomological theory. The theories defined  by \(S(t)\) and \(S_c\) may be thus different, but the failure of their equivalence can only come from new components of the moduli space that flow in from infinity for \(t \neq 0\); the contribution of the \lq\lq old\rq\rq\ component must be independent of \(t\). Taking the limit \(t\gg 1\), the dominant contribution to the deformed action is (suppressing the \(A\)-dependence)
\begin{equation}
\begin{aligned}
S(t\gg 1) &= -\frac{1}{t}\left\lbrace Q, \int_\Sigma d\Sigma\ \mathrm{Tr}\left( \psi^\mu \nabla_\mu (\star F)\right) \right\rbrace \\
&= \frac{1}{t}\int_\Sigma d\Sigma\ \mathrm{Tr}\Bigl( i\nabla_\mu \phi \nabla^\mu (\star F) - (\star F)[\psi_\mu,\psi^\mu] + \nabla_\mu \psi^\mu \epsilon^{\nu\sigma}\nabla_\nu\psi_\sigma \Bigr) .
\end{aligned}
\end{equation}
The main point is that the \(\phi\)-EoM is actually equivalent to the YM equation \(\nabla\star F =0\). This means that the moduli space of the deformed theory contains the moduli space of the standard YM theory, and indeed includes all the higher extrema corresponding to non-flat connections. These solutions with \(F\neq 0\) have \(\lambda \sim - 1/t\), and thus their contribution to the path integral goes roughly as \(\exp(-1/t)\), as expected. When \(t=0\), the cohomological theory is recovered and the only contribution to the moduli space is given by the flat connections.

\paragraph*{Connection with 2-dimensional YM theory}

We already argued that the deformed TFT gained all the YM spectrum \lq\lq flowing from infinity in the moduli space\rq\rq, but it remains to see how one can get practically the YM (and BF) partition function from the theory defined by \(S(t)\). This is simply obtained by another deformation of the exponential in the partition function: we notice that the YM action is gauge invariant, so it is meaningful to compute the expectation value of \(e^{S_{YM}}\) in the TFT. Thus we consider an exponential operator of the form
\begin{equation}
\begin{gathered}
 \exp\left( \omega_0 + \epsilon\Theta \right)  \\
\text{with}\quad \omega_0 := \int_\Sigma \mathrm{Tr}\left( i \phi F + \frac{1}{2}\psi\wedge\psi\right), \qquad \Theta := \frac{1}{2}\int_\Sigma \mathrm{Tr}(\phi\star\phi) .
\end{gathered}
\end{equation}
Since the quantity
\begin{equation}
\left\langle \exp\left( \omega_0 + \epsilon\Theta\right) \right\rangle_t \propto \int_{\Pi T\mathcal{A}(P)\times Lie(\mathcal{G}(P))} DA D\psi D\phi\ \exp\left( \omega_0 + \epsilon\Theta - S(t)\right)
\end{equation}
is well defined for \(t\to\infty\), we can actually take \(t=\infty\) and drop \(S(\infty)=0\) (recall that the path integral is independent on the actual value of \(t\)), getting exactly the YM partition function
\begin{equation}
\left\langle \exp\left( \omega_0 + \epsilon\Theta\right) \right\rangle_t \propto \int_{\Pi T\mathcal{A}(P)\times Lie(\mathcal{G}(P))} DA D\psi D\phi\ \exp\left( \omega_0 + \epsilon\Theta\right) \propto Z(\epsilon) ,
\end{equation}
up to some normalization constant.

In the limit \(\epsilon \to 0\), when only the BF model survives, the path integral over \(\phi\) produces the constraint \(\delta(F)\), so localizing the expectation value onto the space of flat connections. This means that, although we started from different theories \(S_c \cong S(0)\) and \(S(t)\), this particular expectation value satisfies
\begin{equation}
\left\langle \exp\left( \omega_0 \right) \right\rangle_t = \left\langle \exp\left( \omega_0 \right) \right\rangle_{t=0}
\end{equation}
and the BF model is recovered as an expectation value in the cohomological theory. This gives another interpretation to the topological behavior of the BF model, and a measure of the failure of 2-dimensional YM theory in being topological.

Concluding, we only point out that the the operators \(\omega_0\) and \(\Theta\) are precisely the infinite-dimensional realization in this example of the general expressions in \eqref{YM-5}. In fact, the symplectic 2-form on \(\mathcal{A}(P)\)
\begin{equation}
\Omega = \int_\Sigma \mathrm{Tr}(\psi\wedge\psi)
\end{equation}
only serves to have a formal interpretation of the measure \(DAD\psi e^{\Omega}\), since the field \(\psi\) is really a spectator in the action
\begin{equation}
S[A,\psi,\phi] = - \int_\Sigma \mathrm{Tr}\left( i \phi F +\frac{\epsilon}{2}\phi\star\phi + \frac{1}{2}\psi\wedge\psi \right).
\end{equation}

\section{Localization of 2-dimensional YM theory} 

As we said at the beginning of the chapter, 2-dimensional YM theory is an exactly solvable theory, whose partition function can be expressed in closed form, for example by group characters expansion methods \cite{Cordes-2dYMTFT,witten-2dYM}. This makes it possible to compare results from the localization formalism, and obtain a new geometric interpretation of the already present solution of the theory. In general, its partition function on a Riemann surface \(\Sigma\) of genus \(g\), with a simply-connected gauge group \(G\), is given as
\begin{equation}
\label{YM-char-expansion}
Z(\epsilon) = (\mathrm{vol}(G))^{2g-2} \sum_R \frac{1}{\dim(R)^{2g-2}}e^{-\epsilon \tilde{C}_2(R)} ,
\end{equation}
where the sum runs over the representations \(R\) of \(G\), and \(\tilde{C}_2(R)\) is related to the quadratic Casimir \(C_2(R) := \sum_a \mathrm{Tr}_R(T_aT_a)\) of the representation \(R\) by some normalization constant. For not simply-connected \(G\) this formula has to be slightly modified (see \cite{witten-2dYM}).\footnote{Simply-connectedness implies that the principal \(G\)-bundle \(P\to\Sigma\) has to be trivial. When one drops this condition, the triviality is not ensured and contributions to the formula appear due to singular points in \(\Sigma\) for the connection.} We are not interested in reviewing the proof of \eqref{YM-char-expansion} in general, but we present a quick argument for a simple example, that already contains the logic behind it.

\begin{ex}[YM theory with genus \(g=1\)] 
We quickly motivate the result for the YM partition function on a genus 1 surface. We can think of this surface as a disk with boundary, that is homeomorphic to a sphere with one hole. The radial direction is identified with the interval \([0,T]\), and the angular coordinate as \([0,L]\) with the edges identified.

It is more natural to compute the partition function of the theory in the Hamiltonian formulation,
\begin{equation}
Z = \mathrm{Tr}_{\mathcal{H}}\mathcal{P}\exp\left( -\int_0^T dt\ H(t) \right)
\end{equation}
up to possible normalization factors, where \(\mathcal{H}\) is the Hilbert space of the system. To this end, let us consider the canonical quantization of the YM action. To make sense of the partition function we must fix a gauge, and we do this  setting \(A_t = 0\) (temporal gauge). In this gauge the action simplifies as
\begin{equation}
S[A] = -\frac{1}{2\epsilon} \int dxdt\ \mathrm{Tr}(F_{01}^2) = \frac{1}{2\epsilon} \int dxdt\ (\partial_t A^a_x)^2 ,
\end{equation}
where we expanded \(A_\mu = A^a_\mu T_a\) with respect to the generators of \(\mathfrak{g}\), and suppressed the inner product implicitly summing over the Lie algebra indices. We see that the only non-zero canonical momentum is \(\Pi^a_x = (1/\epsilon)\partial_t A^a_x\), acting on the Hilbert space as \(\Pi^a_x(t,x) \mapsto \frac{\delta}{\delta A_x^a(t,x)}\). The canonical Hamiltonian in temporal gauge is thus
\begin{equation}
H(t) = \epsilon \int_0^L dx\ (\Pi^a_x(t,x))^2 \mapsto \epsilon \int_0^L dx\ \frac{\delta}{\delta A_x^a(t,x)} \frac{\delta}{\delta A_x^a(t,x)} .
\end{equation}

The Hilbert space \(\mathcal{H}\) can be considered to consist of gauge invariant functions \(\Psi(A)\). The only gauge invariant data obtained from the gauge field at any point \(p\in \Sigma\) is its holonomy,
\begin{equation}
U_p[A] := \mathcal{P}\exp \left(\oint_{C(p)} A\right) \quad\in G ,
\end{equation}
where \(C(p)\) is a loop about \(p\). In the case of the one holed-sphere, all loops are homotopic to the one on the boundary,  so \(\Psi\in \mathcal{H}\) must be an invariant function of 
\begin{equation}
U = \mathcal{P}\exp \int_0^L dx A_1 ,
\end{equation}
and independent of \(t\in[0,T]\). Any invariant function must be expandible in characters of representations of \(G\), so \(\Psi(U) = \sum_R \Psi_R \mathrm{Tr}_R(U)\), where \(\Psi_R\in \mathbb{C}\) and \(\mathrm{Tr}_R(U)\) is the Wilson loop in the representation \(R\) of \(G\). We notice that the basis functions \(\chi_R(U) := \mathrm{Tr}_R(U)\) diagonalize the Hamiltonian, since 
\begin{equation}
H\chi_R(U) = \epsilon \mathrm{Tr}_R\int_0^1 dx\ T_a T_a\ \mathcal{P}\exp \int_0^L dx A_1 = \epsilon L C_2(R) \chi_R(U) ,
\end{equation}
where \(C_2(R) := \mathrm{Tr}_R(T_a T_a)\) is the quadratic Casimir in the representation \(R\), a time-independent eigenvalue of \(H\). Via this diagonalization the partition function is easily computed,
\begin{equation}
Z = \sum_R e^{-TL\epsilon c(R)} ,
\end{equation}
matching \eqref{YM-char-expansion} for \(g=1\). All the geometric information about \(\Sigma\) that enters in \(Z\) is its total area \(TL\), and any other local property. Notice that in the topological limit \(\epsilon\to 0\), the Hamiltonian vanishes (as the theory has no propagating degrees of freedom) and the partition function simplifies further.
\end{ex}

From the localization formalism discussed in the last sections, we expect the partition function to be of the type
\begin{equation}
\label{YM-12}
Z(\epsilon) = Z_0(\epsilon) + \sum_n Z_n(\epsilon) ,
\end{equation}
with \(Z_0(\epsilon)\) representing the contribution from the moduli space \(\mathcal{A}_0\) of flat connections, such that \(Z_0(0)\sim \mathrm{vol}(\mathcal{A}_0)\), and the other \(Z_n(\epsilon)\) coming from contributions of the higher extrema of the YM action, such that \(Z_n(\epsilon)\sim \exp(-1/\epsilon)\) in the weak coupling limit \(\epsilon \ll 1\). Using cohomological arguments, in \cite{witten-2dYM} (also nicely reviewed in \cite{witten-CS}) it was shown how to recover the general features of \eqref{YM-char-expansion}, and in particular how to interpret it in terms of an \(\epsilon\)-expansion at weak coupling, in relation to the expected form \eqref{YM-12}. The detailed derivation is cumbersome and requires some more technical background, so we refer to the article for it, but the logic is essentially the same as for the discussion at the end of Section \ref{sec:YM-localization-formula}. The strategy is the following. Any solution to the YM EoM identifies a disconnected region \(\mathcal{S}_n\subset \mathcal{A}(P)\). For every such region, one fixes a small neighborhood \(N_n\) around \(\mathcal{S}_n\) that equivariantly retracts onto it. The technically difficult passage is to perform the integral over the \lq\lq normal directions\rq\rq\ to \(\mathcal{S}_n\) in \(N_n\), and then reduce it on the moduli space \(\mathcal{A}_n := \mathcal{S}_n/\mathcal{G}(P)\). The main difficulty is that in general the MWM theorem (or its equivariant counterpart) does not work, since the action of \(\mathcal{G}(P)\) is not generally free on \(\mathcal{S}_n\) (also for \(n=0\)), as we assumed in writing down \eqref{YM-5} for the \(\mathcal{S}_0 \equiv \mu^{-1}(0)\) component. For the higher extrema, this is readily seen by the fact that the equation 
\begin{equation}
\nabla (\star F) = 0 \qquad \text{with}\ F\neq 0
\end{equation}
identifies a vacuum \(f:=(\star F)\) as a preferred element of \(\mathfrak{g}\) (being it covariantly constant over \(\Sigma\)), and thus the gauge group is spontaneously broken to a subgroup \(G_f\subseteq G\). The action of the whole gauge group thus cannot be free on this subspace, and the quotient \(\mathcal{S}_n/\mathcal{G}(P)\) is singular. Via a suitable choice of localization 1-form one is still able to extract information by this integral over the normal directions, and in particular to compute the \(\epsilon\)-dependence of the higher extrema contributions.

We limit ourselves now to the comparison of the exact result \eqref{YM-char-expansion} applied to the case \(G=SU(2)\),
%and \(G=SO(3)\)
with the expectation \eqref{YM-12} obtained by cohomological arguments. For this gauge group, the character expansion of the partition function results
\begin{equation}
Z(\epsilon) = \frac{1}{(2\pi^2)^{g-1}}\sum_{n=1}^{\infty} \frac{\exp(-\epsilon \pi^2 n^2)}{n^{2g-2}} .
\end{equation}
Simply taking \(\epsilon=0\), we see that this is finite and proportional to a Riemann zeta-function, but to explore better the \(\epsilon\)-dependence it is convenient to consider
\begin{equation}
\label{YM-13}
\frac{\partial^{g-1}Z}{\partial \epsilon^{g-1}} = \left(-\frac{1}{2}\right)^{g-1} \sum_{n=1}^{\infty} \exp\left(-\epsilon\pi^2 n^2\right) = \frac{(-1)^{g-1}}{2^g} \left( -1 + \sum_{n\in\mathbb{Z}} \exp\left(-\epsilon\pi^2 n^2\right)\right) .
\end{equation}
This is not quite in the expected form, since the exponentials in the sum go to zero for \(\epsilon\to 0\) but not as \(\exp(-1/\epsilon)\). We can bring this expression closer to the desired result using the Poisson summation formula
\begin{equation}
\sum_{n\in\mathbb{Z}} f(n) = \sum_{k\in\mathbb{Z}} \hat{f}(k) \qquad \text{with}\ \hat{f}(k) := \int_{-\infty}^{+\infty} f(x)e^{-2\pi ikx} , 
\end{equation}
where \(f\) is a function and \(\hat{f}\) its Fourier transform, and rewriting the sum of exponentials in \eqref{YM-13} as
\begin{equation}
\frac{\partial^{g-1}Z}{\partial \epsilon^{g-1}} =
\frac{(-1)^{g-1}}{2^g} \left( -1 + \sqrt{\frac{1}{\pi\epsilon}} \sum_{k\in\mathbb{Z}} \exp\left(-\frac{k^2}{\epsilon}\right)\right) .
\end{equation}
This is exactly the result that could be obtained via integration over normal coordinates in the localization framework (see \cite{witten-CS}, eq.\ (4.102)), but fundamentally differs from our expectation, since for \(\epsilon\to 0\) the contribution from the flat connections (with \(k=0\))  is singular for the presence of the square root. This means that the partition function is not really a polynomial in \(\epsilon\) for small couplings, but an expression of the form
\begin{equation}
Z(\epsilon) = \sum_{m=0}^{g-2} a_m \epsilon^m + a_{g-3/2}\epsilon^{g-3/2} + \text{exponentially small terms} . 
\end{equation}
The singularity in \(Z(\epsilon\to 0)\) arises because, for gauge group \(SU(2)\), the subspace \(\mu^{-1}(0)\) is singular and the MWM theorem does not apply. 

A simpler situation would occur considering the gauge group \(SO(3)\) (which is not simply connected) and a non-trivial principal bundle over \(\Sigma\). In this case, the character expansion of the partition function requires some modifications with respect to \eqref{YM-char-expansion}, the result being
\begin{equation}
Z(\epsilon) = \frac{1}{(8\pi^2)^{g-1}}\sum_{n=1}^\infty \frac{(-1)^{n+1} \exp(-\pi^2 \epsilon n^2)}{n^{2g-2}} .
\end{equation}
Following the same idea as above, we look at the \((g-1)^{th}\) derivative
\begin{equation}
\frac{\partial^{g-1}Z}{\partial \epsilon^{g-1}} = \frac{(-1)^g}{8^{g-1}} \sum_{n=1}^\infty (-1)^{n} \exp(-\pi^2 \epsilon n^2) = \frac{(-1)^g}{8^{g-1}}\frac{1}{2}\left( -1 + \sum_{n\in\mathbb{Z}} (-1)^{n} \exp(-\pi^2 \epsilon n^2)\right) ,
\end{equation}
and we rewrite the sum using the Poisson summation formula, getting
\begin{equation}
\frac{\partial^{g-1}Z}{\partial \epsilon^{g-1}} = \frac{(-1)^g}{2\cdot 8^{g-1}} \left( -1 + \sqrt{\frac{1}{\pi\epsilon}} \sum_{k\in\mathbb{Z}}  \exp\left(-\frac{\left(k+1/2\right)^2}{\epsilon} \right)\right).
\end{equation}
This time we see that the contribution for \(k=0\) from the moduli space of flat connections is finite for \(\epsilon\to 0\), and the whole \(\partial^{g-1}Z / \partial \epsilon^{g-1}\) is constant up to exponentially small terms. This means that the partition function at weak coupling \(Z(\epsilon\to 0)\) is a regular polynomial of degree \(g-1\) in \(\epsilon\), up to exponentially decaying terms,
\begin{equation}
Z(\epsilon) = \sum_{m=0}^{g-2} a_m \epsilon^m + O(\epsilon^{g-1}),
\end{equation}
and it reflects the fact that, for a non-trivial \(SO(3)\)-bundle, \(\mu^{-1}(0)\) is smooth and acted on freely by \(G\). These two quick examples capture the way this localization framework can give a very geometric interpretation to the \( \epsilon \)-expansion of the partition function, and its dependence on the classical geometry of the moduli space.

Finally, we point out that an analogous treatment was done more recently in \cite{witten-CS} to analyze in this cohomological framework Chern-Simons theory on 3-dimensional \emph{Seifert manifolds}. A Seifert manifold is a smooth object that can be described as an \(\mathbb{S}^1\)-bundle over a 2-dimensional orbifold, and this feature makes it possible to dimensionally reduce the Chern-Simons theory along the direction of the circle \(\mathbb{S}^1\) to a 2-dimensional YM theory over a singular base space. It turns out that  the localization locus of the resulting theory receives contributions only from the flat connections over the total space. This is in accordance with the fact that Chern-Simons theories are by themselves TFT (of Schwarz-type). 2-dimensional YM theories have been studied extensively in the past years, and many interesting results were obtained thanks to their non-perturbative solvability. For example, exact results for Wilson loops expectation values and their relation with higher dimensional supersymmetric theories were studied in \cite{griguolo-bassetto-1998,griguolo-bassetto-2009,bassetto-2011}. A relation with certain topological string theories and supersymmetric black hole entropy computations were analyzed in \cite{caporaso-YM-topological-string}. A duality between higher-dimensional supersymmetric gauge theories and deformations of 2-dimensional YM theory was revisited in \cite{szabo-qdeformation,szabo-YMloc}. Localization techniques play an important role in all those cases.

%%%%%%%%%%%%%%%%%%%%%%%%%%%%%%%%%%%%%%%%%%%%%%%%% cap 5: conclusioni
\chapter{Conclusion}

In this thesis we reviewed and summarized the main features of the formalism of equivariant cohomology, the powerful localization theorems first introduced by Atiyah-Bott and Berline-Vergne, and the principles that allow to formally apply these integration formulas to QFT. From the physical point of view, the equivariant (or supersymmetric) localization principle gives a systematic approach to understand when the \lq\lq semiclassical\rq\rq\ approximation of the path integral, describing the partition function or an expectation value in QFT, can give an exact result for the full quantum dynamics. We discussed the applicability of these techniques in the context of \emph{supersymmetric} theories. These are characterized by a space of fields that is endowed with a graded structure and the presence of some symmetry operator whose \lq\lq square\rq\rq\ gives a standard \lq\lq bosonic\rq\rq\ symmetry of the action functional. This supersymmetry operator is interpreted as a differential acting on the subspace of symmetric configurations in field space, and its cohomology describes the field theoretical analog of the \(G\)-equivariant cohomology of a \(G\)-manifold. 

After having introduced the general features of the mathematical theory of equivariant cohomology and equivariant localization, we reviewed the concepts in supergeometry that allow for the construction of supersymmetric QFT, and that constitute the correct framework to translate the mathematical theory in the common physical language. Since many recent applications of the localization principle aimed at the computations of path integrals in supersymmetric QFT on curved spaces, we included a discussion of the main tools needed to define supersymmetry in such instances. Then we collected  some examples from the literature of application of the supersymmetric localization principle to path integrals in QFT of diverse dimensions. The common feature of these examples is that, via a suitable \lq\lq cohomological\rq\rq\ deformation of the action functional, it is possible to reduce the infinite-dimensional path integral to a finite-dimensional one that represents its semiclassical limit, as stressed above. We described cases in which this reduction relates the partition function to topological invariants of the geometric structure underlying the theory, namely the cases of supersymmetric QM (a 1-dimensional QFT) and the weak coupling limit of 2-dimensional Yang-Mills theory (its \lq\lq topological\rq\rq\ limit). We also reviewed the more recent applications to the computations of the expectation values of supersymmetric Wilson loops in 3- and 4-dimensional gauge theories, namely Supersymmetric Chern-Simons theory and Supersymmetric Yang-Mills theory defined on the 3- and the 4-sphere. In these cases, the path integral results to be equivalently described by some 0-dimensional QFT with a Lie algebra as target space, called \lq\lq matrix model\rq\rq.

In the last few decades, the literature concerning the applications of supersymmetric localization has grown exponentially, and many other advanced examples of its use in the physics context have been found. From the point of view of supersymmetric QFT, a consistent slice of the state-of-the-art on the subject can be found in \cite{LocQFT}, including computations analogue to the one we showed for Wilson loop expectation values or topological invariants over more complicated geometries. From the point of view of Quantum (Super)Gravity, these techniques have found applications in the computations of the Black Hole quantum entropy \cite{Dabholkar_2011,zaffaroni2019lectures}. In many circumstances, localization allows for the analysis of properties of QFT at strong coupling, an otherwise prohibited region of study with conventional perturbative techniques. This feature can be used also to test a class of conjectural dualities between some types of gauge theories and string theories, the so-called AdS/CFT correspondences \cite{zarembo-ads/cft-loc}. Concerning the subject of Wilson loops in 3-dimensional Chern-Simons theories and their relations to matrix models and holography, for which localization has played an important role, a recent reference that concisely reviews the state-of-the-art is \cite{roadmap-wilson-loops}.

\appendix

\chapter{Some differential geometry}
\label{app:diff_geo}

\section{Principal bundles, basic forms and connections}
\label{app:principal-bundles}

Here we recall some notions about principal bundles that can be useful to follow the discussion, especially of the first chapters of this thesis. Principal bundles are the geometric construction behind the concepts of \emph{covariant derivatives} and \emph{connections} in gauge theory or General Relativity, for example. If \(G\) is a Lie group, a principal \(G\)-bundle is a smooth bundle \(P\xrightarrow{\pi} M\) such that
\begin{enumerate}[label=(\roman*)]
\item \(P\) is a (right) \(G\)-manifold;
\item the \(G\)-action on \(P\) is \emph{free};
\item as a bundle, \(P\to M\) is isomorphic to \(P\to P/G\), where the projection map is canonically defined as \(p\mapsto [p]\).
\end{enumerate}
Notice that since the \(G\)-action is free, a principal \(G\)-bundle is a fiber bundle with typical fiber \(G\), and by the third property it is at least \emph{locally trivial}, \textit{i.e.}\ over every open set \(U\subseteq M\) it looks like \(G\times U\). Morphisms of principal bundles are naturally defined as maps between bundles that preserve the \(G\)-structure, so \(G\)-equivariant maps. A principal bundle is \emph{trivial} if it is isomorphic through a principal bundle isomorphism to the trivial product bundle \(G\times M\). A useful fact is that the triviality of a principal bundle is completely captured by the existence of a \emph{global} section \(\sigma:M\to P\) such that \(\pi\circ \sigma = id_M\). Since every principal bundle is locally trivial, than local sections can always be chosen and they constitute a so-called \emph{local trivialization of the bundle}.

The main example of principal bundle that occurs in the geometric construction of spacetime is the \emph{frame bundle} \(LM\) over some \(n\)-dimensional smooth manifold \(M\). At every point \(p\in M\), the elements of the fiber \(L_pM\) are the \emph{frames} at \(p\), \textit{i.e.}\ all the possible bases \(e=(e_1,\cdots,e_n)\) for the tangent space \(T_pM\). \(LM\) has a natural \(GL(n,\mathbb{R})\) right action that corresponds to the rotation of the basis, \(e\cdot g := (e_k g^k_1,\cdots, e_k g^k_n)\) for \(g\in GL(n,\mathbb{R})\). In gauge theories, the \emph{structure group} (or sometimes \emph{gauge group}) \(G\) of the theory is the Lie group acting on the right on a principal \(G\)-bundle.

Since the fibers of the principal bundle are essentially the Lie group \(G\), tangent vectors on \(P\) can come from its Lie algebra \(\mathfrak{g}\). This leads to the following definition.
\begin{defn}
The \emph{vertical sub-bundle} \(VP\) of the tangent bundle \(TP\) is the disjoint union \[
VP := \bigsqcup_{p\in P} V_p P , \qquad \text{with} \
V_p P := \mathrm{Ker}(\pi_{*p}) = \lbrace X\in T_pP | \pi_*(X) = 0 \rbrace \subset T_p P . \]
Analogously but for differential forms, the \emph{basic forms} inside \(\Omega(P)\) are those forms \( \omega \in \mathrm{Im}(\pi^*)\), so that it exists an \(\alpha\in \Omega(M)\) such that \(\omega = \pi^* \alpha\). The space of basic forms is denoted \(\Omega(P)_{bas}\).
\end{defn}
The vertical vectors in every \(V_p P\) are in one to one correspondence with the Lie algebra elements in \(\mathfrak{g}\), through the Lie algebra homomorphism
\begin{equation}
X\in \mathfrak{g} \mapsto \underline{X} := \left.\frac{d}{dt}\right|_{t=0} \left(\cdot e^{tX}\right)^*
\end{equation}
that maps Lie algebra elements to the corresponding \emph{fundamental vector fields}.\footnote{The choice of the sign at the exponential differs from the one in \eqref{eq:fundam-vf} because here we are considering a right action.} Fundamental vector fields satisfy the following properties:
\begin{enumerate}[label=(\roman*)]
\item \([\underline{X},\underline{Y}] = \underline{[X,Y]}\);
\item the integral curve of \(\underline{X}\) through \(p\in M\) is 
\begin{equation}
\begin{aligned}
     \gamma_p :  \mathbb{R} &\to M \\
                 t &\mapsto \gamma_p (t) = p \cdot e^{tX} ;
     \end{aligned}
\end{equation}
\item denoting with \(r_g\) the right action of \(g\in G\), 
\begin{equation}
\left(r_g\right)_* (\underline{X}_p) = \left( \underline{Ad_{g^{-1}*}(X)} \right)_{r_g(p)} .
\end{equation}
\end{enumerate}

As for the vertical vector fields being encoded in the Lie algebra \(\mathfrak{g}\), also the basic forms can be characterized in terms on the (infinitesimal) action of \(\mathfrak{g}\) on \(\Omega(P)\). This can be seen introducing the following definitions.
\begin{defn}
A differential form \(\omega\in \Omega(P)\) is said to be \emph{\(G\)-invariant} if it is preserved by the \(G\) action: \[
\omega = (r_g)^* \omega \qquad \forall g\in G.\]
The space of \(G\)-invariant forms is commonly denoted \(\Omega(P)^G\). A differential form is called \emph{horizontal} if it is annihilated by vertical vector fields,
\[ \iota_X \omega = 0 \qquad \forall X\in \Gamma(VP).\]
\end{defn}
The properties of being invariant and horizontal can be also stated infinitesimally with respect to the action of the Lie algebra \(\mathfrak{g}\). If we define
\begin{equation}
\mathcal{L}_X := \mathcal{L}_{\underline{X}} , \qquad \iota_X := \iota_{\underline{X}} , \qquad \forall X\in \mathfrak{g},
\end{equation}
then an invariant form is characterized by \(\mathcal{L}_X \omega = 0\) for every \(X\in \mathfrak{g}\), and a horizontal form by \(\iota_X \omega = 0\) for every \(X\in \mathfrak{g}\). This makes the concepts of invariant and horizontal elements independent from the principal bundle structure, so that they can be defined by this characterization for every \(\mathfrak{g}\)-dg algebra, as in Section \ref{sec:Weil-model}. Also basic forms can be defined for every \(\mathfrak{g}\)-dg algebra, combining the definitions of invariant and horizontal forms, thanks to the following theorem:
\begin{thm}[Characterization of basic forms]
Schematically,
\[ \text{invariant}+\text{horizontal} \Leftrightarrow \text{basic}.\]
\end{thm}
\begin{proof} For notational convenience only, let us consider 1-forms.
\begin{itemize}
\item[\((\Leftarrow)\)] If \(\omega=\pi^*\alpha\) is basic, then
\[ (r_g)^*\omega = (r_g)^*\pi^*\alpha = (\pi \circ r_g)^* \alpha = \omega ,\] 
since the principal bundle is locally trivial. So \(\omega\) is also invariant. For a vertical vector \(X\in VP\), \[
\iota_X \omega = (\pi^*\alpha)(X) = \alpha (\pi_*(X)) = 0 .\] 
So \(\omega\) is also horizontal.
\item[\((\Rightarrow)\)] Let \(\omega\in \Omega^1(P)\) be horizontal and invariant. Since \(\pi\) is surjective, for every vector \(X\in T_pP\) there exists \(Y=\pi_*(X)\in T_{\pi(p)}M\). We can define \(\alpha\in \Omega^1(M)\) such that, at every \(x\in M\) \[
\alpha_{x}(Y) := \omega_p(X) \qquad \text{for} \ p\in \pi^{-1}(x),\]
and thanks to the horizontality and invariance of \(\omega\) we can check that this form is well defined, \textit{i.e.}\ independent from the choice of point \(p\) in the fiber \(\pi^{-1}(x)\) and from the choice of vector \(X\) such that \(\pi_*(X)=Y\). In fact, if \(X'\in T_p P\) is another vector such that \(\pi_*(X')=Y\), then \(\pi_*(X-X')=0\) so \((X-X')\in V_p P\). By horizontality, \(\omega(X-X')=0 \Rightarrow \omega(X)=\omega(X')\), so \(\alpha\) is independent from the choice of vector. Moreover, if \(p'\in \pi^{-1}(x)\) is another point in the fiber, there exists a \(g\in G\) such that \(r_g(p) = p'\), so by \(G\)-invariance \(\omega_{p'} = \omega_p\) and thus \(\alpha\) is independent from the choice of point in the fiber.
\end{itemize}
\end{proof}

\begin{prop} The differential \(d\) closes on the subspace \(\Omega(P)_{bas}\) of basic forms, defining a proper subcomplex. This extends to any \(\mathfrak{g}\)-dg algebra. \end{prop}
\begin{proof}
Consider \(\omega\in \Omega(P)_{bas}\), and its differential \(d\omega\). We characterize basic forms by being horizontal and \(G\)-invariant. By Cartan's magic formula the Lie derivative commutes with the differential, so for every \(X\in\mathfrak{g}\), \(\mathcal{L}_X d\omega=d(\mathcal{L}_X \omega) = 0\). Thus \(d\omega\) is still \(G\)-invariant. Also, \(\iota_X d\omega = \mathcal{L}_X\omega - d \iota_X \omega = 0\). Thus \(d\omega\) is still horizontal, and so basic.
\end{proof}

We recall now the definition of connection and curvature on principal bundles, from which one inherits covariant derivatives on associated vector bundles.
\begin{defn} An \emph{(Ehresmann) connection} on a principal \(G\)-bundle \(P\to M\) is an \emph{horizontal distribution} \(HP\), \textit{i.e.}\ a smooth choice at every point \(p\in P\) of vector subspaces \(H_pP\subset T_pP\) such that
\begin{enumerate}[label=(\roman*)]
\item \(T_p P = V_pP \oplus H_pP\);
\item \((r_g)_* (H_p P) = H_{p\cdot g}P\) (\(G\)-equivariance of the horizontal projection).
\end{enumerate}
Given a horizontal distribution \(HP\), every vector \(X\in T_pP\) decomposes into an \emph{horizontal} and a \emph{vertical} part, \[
X = hor(X) + ver(X) .\]
A \emph{connection 1-form} on \(P\to M\) is Lie algebra-valued 1-form \(A\in \Omega^1(P)\otimes\mathfrak{g})\) such that
\begin{enumerate}[label=(\roman*)]
\item for any \(X\in \mathfrak{g}\), \(\iota_X A = A(\underline{X}) = X\) (\emph{vertical} 1-form);
\item for any \(g\in G\), \((r_g)^*A = (Ad_{g^{-1}*}\circ A)\) (\(G\)-equivariance).
\end{enumerate}
\end{defn}
The choice of a horizontal distribution is equivalent to the choice of a connection 1-form on \(P\), since at every \(p\in P\) one can use \(A\) as a projection onto the vertical subspace \(V_pP\cong \mathfrak{g}\), and \(\pi_*\) as a projection onto the horizontal subspace, identifying \(V_p P := \mathrm{Ker}(\pi_{*p})\) and \(H_p P := \mathrm{Ker}(A_p)\). This choice is smooth  and \(G\)-equivariant since \(A\) is, by definition. Notice that the splitting \(HP\oplus VP\) induces a splitting \(\Omega^1(P) = \Omega^1_{hor}(P) \oplus \Omega^1_{ver}(P)\), and that we can identify the \lq\lq space of connection 1-forms\rq\rq\ as
\begin{equation}
\mathcal{A}(P) := \lbrace A\in (\Omega^1_{ver}(P)\otimes \mathfrak{g}) | A\ \text{is }G\text{-equivariant}\rbrace.
\end{equation}
It is easy to see that for every \(A,A'\in \mathcal{A}(P)\), their difference is not a connection, and in fact it is an \emph{horizontal} \(\mathfrak{g}\)-valued 1-form,
\begin{equation}
(A-A')\in \mathfrak{a}:= \lbrace a\in (\Omega^1_{hor}(P)\otimes \mathfrak{g}) | a\ \text{is }G\text{-equivariant}\rbrace .
\end{equation}
This means that every connection \(A\) can be written as another connection \(A'\) plus a horizontal form, or in other words that \((\mathcal{A}(P),\mathfrak{a})\) can be seen as a natural \emph{affine space}, modeled on the infinite-dimensional vector space \(\mathfrak{a}\). As for any affine space, one can think of the space of connections as an infinite-dimensional smooth manifold, with tangent spaces \(T_A\mathcal{A}(P) \cong \mathfrak{a}\) at every \(A\in \mathcal{A}(P)\).

\begin{defn}
The \emph{covariant exterior derivative} on \(\Omega(P)\) is \(D:= d\circ hor^*\). The \emph{curvature} of a connection 1-form \(A\) is \[
F := DA = dA(hor(\cdot),hor(\cdot)) \in \Omega^2(P)\otimes \mathfrak{g} .\]
\end{defn}
The curvature \(F\) satisfies the following properties:
\begin{enumerate}[label=(\roman*)]
\item by definition, \(F\) is horizontal: \(\iota_X F=0\) for every \(X\in \mathfrak{g}\);
\item by \(G\)-equivariance of \(A\), \(F\) is \(G\)-equivariant too;
\item it obeys the structural equation 
\begin{equation}
\label{appeq-1}
F = dA + \frac{1}{2}[A\stackrel{\wedge}{,}A]
\end{equation}
where \([A\stackrel{\wedge}{,}A] = f^a_{bc} A^b\wedge A^c \otimes T_a\) with respect to a basis \(\lbrace T_a\rbrace\) of \(\mathfrak{g}\) and the structure constants \(f^a_{bc}\);
\item it obeys the second Bianchi identity,
\begin{equation}
DF = 0 \qquad \text{or} \qquad dF = [F\stackrel{\wedge}{,}A] .
\end{equation}
\end{enumerate}

One can consider the very trivial construction of a principal \(G\)-bundle as \(G\to pt\), where \(P=G\times pt \cong G\). Here the right \(G\)-action is simply the diagonal action (trivial on \(pt\), induced by the natural action on \(G\)). On this bundle there is a canonical choice of connection 1-form, the \emph{Maureer-Cartan (MC) form} \(\Theta\in \Omega^1(G)\otimes \mathfrak{g}\). For every vector  \(X\in T_g G\) at some \(g\in G\), there is a Lie algebra element \(A\in T_e G\cong\mathfrak{g}\) such that \(X = l_{g*}(A)\), and the MC form is defined by
\begin{equation}
\Theta_g(X) := l_{g^{-1}*} (X) = A .
\end{equation}
One can check that this form is indeed \(G\)-equivariant, and it is obviously vertical, giving a connection 1-form. Moreover it satisfies the \emph{Maurer-Cartan equation}
\begin{equation}
d\Theta + \frac{1}{2}[\Theta\stackrel{\wedge}{,}\Theta] = 0,
\end{equation}
so that by \eqref{appeq-1} we see that its curvature is zero. On the trivial principal \(G\)-bundle \(G\times M \to M\) one can always define a connection 1-form by pulling back the MC connection along the projection \(\pi_1: G\times M \to G\). In the general case, the principal bundle \(P\) is locally trivial, so in any local patch \(G\times U_\alpha\to U_\alpha\) one can pull back the MC connection and use a suitable partition of unity to glue together the local pieces to a global connection 1-form on \(P\). This shows that any principal bundle allows for a connection. The curvature \(F\) of the chosen connection \(A\) measures, in a sense, the deviation of \(A\) from being the Maurer-Cartan connection.

\medskip
As said before, a connection on a principal \(G\)-bundle allows for the definition of a covariant derivative on \emph{associated vector bundles}. An associated vector bundle to the principal \(G\)-bundle \(P\xrightarrow{\pi} M\) is a vector bundle constructed over \(M\) with some typical fiber \(V\) (a vector space) that has a (left) \(G\)-action compatible with the one on \(P\). Precisely, the associated bundle is \(P_V \xrightarrow{\pi_V} M\), where
\begin{equation}
\begin{aligned}
&P_V := \faktor{(P\times V)}{\sim_G} \qquad \text{with}\ (p,v) \sim_G (p\cdot g, g^{-1}\cdot v)\ \forall (p,v) \in P\times V, g\in G, \\
&\pi_V ([p,v]) := \pi(p) ,
\end{aligned}
\end{equation}
and it has indeed typical fiber \(V\). In the case of the frame bundle \(P=LM\), one can construct the tangent bundle \(TM\), the cotangent bundle \(T^*M\) and all the tensor bundles as associated to \(LM\). In fact, for the tangent bundle for example, the typical fiber is \(V:= \mathbb{R}^n\) and the \(GL(n,\mathbb{R})\)-action is \((g\cdot v)^k = g^k_j v^j\). This encodes the change of basis rule if we see vectors as elements \([e,v]\in LM_{\mathbb{R}^n}\),
\begin{equation}
[e,v] \equiv  e_k  v^k  \sim_{GL} [e\cdot g, g^{-1}\cdot v] \equiv e_j g^j_k (g^{-1})^k_j v^j = e_k  v^k .
\end{equation}
On the associated vector bundle, a \emph{field} (in physics terms) is a (local, at least) section \(\phi:M\to P_V\), that can be always seen locally as a \(V\)-valued function on every \(U\subseteq M\), \(\tilde{\phi}:U\to V\), so that \(\phi(x) = [p,\tilde{\phi}(x)]\) for some chosen \(p\in \pi^{-1}(x), x\in U\).
Another example of this concept that came up in Chapter \ref{cha:equivariant cohomology} is the \emph{homotopy quotient} \(M_G := (M\times EG)/G\) of a \(G\)-manifold \(M\). This is precisely the associated bundle with fiber \(V=M\) to the principal \(G\)-bundle \(EG\to BG\) (however, this is not an associated \emph{vector} bundle, since \(M\) is not a vector space in general).

As we said at the beginning of this appendix, every principal bundle is locally trivial, so that it exists a set of \emph{local trivializations} \(\lbrace U_\alpha, \varphi_\alpha:\pi^{-1}(U_\alpha) \to U_\alpha\times G\rbrace\), where \(\lbrace U_\alpha\rbrace\) covers \(M\), and \(\varphi_\alpha\) is \(G\)-equivariant. This means that \(\varphi_\alpha(p) = (\pi(p), g_\alpha(p))\) for some \(G\)-equivariant map \(g_\alpha :\pi^{-1}( U_\alpha)\to G\),\footnote{In this case \(G\)-equivariance means \(g_\alpha(p\cdot h)=g_\alpha(p)\cdot h\).} that makes every fiber diffeomorphic to \(G\). To this local trivialization, one can canonically associate a family of local sections \(\lbrace \sigma_\alpha : U_\alpha\to \pi^{-1}(U_\alpha)\rbrace\), determined by the maps \(\varphi_\alpha\) so that for every \(m\in U_\alpha\), \(\varphi_\alpha(\sigma_\alpha(m)) = (m,e)\), where \(e\in G\) is the identity element. In other words, \(g_\alpha\circ \sigma_\alpha: U_\alpha\to G\) is the constant function over the local patch \(U_\alpha\subseteq M\) that maps every point to the identity. Conversely, a local section \(\sigma_\alpha\) allows us to identify the fiber over \(m\) with \(G\). Indeed, given any \(p\in \pi^{-1}(m)\), there is a unique group element \(g_\alpha(p)\in G\) such that \(p = \sigma_\alpha(m)\cdot g_\alpha(p)\). Using these canonical local data, the connection \(A\) and the curvature \(F\) can be pulled back on \(M\) giving the local \emph{gauge field} \(A^{(\alpha)} := \sigma_\alpha^*(A)\) and \emph{field strength} \(F^{(\alpha)} := \sigma_\alpha^*(F)\). The \emph{covariant derivative} along the tangent vector \(X\in TM\) of a local \(V\)-valued function \(\tilde{\phi}:U_\alpha \to V\) is defined as
\begin{equation}
\nabla_X \tilde{\phi} := d\tilde{\phi}(X) + A^{(\alpha)}(X) \cdot \tilde{\phi} ,
\end{equation}
where the second term denotes the action of the Lie algebra on \(V\), that for matrix groups coincides with the action of \(G\). We denote schematically the covariant derivative as \(\nabla = d+ A\) on a generic associated vector bundle. When \(V=\mathfrak{g}\) we have the so-called \emph{adjoint bundle}, often denoted \(\text{ad}(P)\), that is in one-to-one correspondence with the space \(\mathfrak{a}\) above, of horizontal and \(G\)-equivariant Lie algebra-valued forms on \(P\). On this special associated bundle, the covariant derivative acts with the infinitesimal adjoint action  of \(\mathfrak{g}\),
\begin{equation}
\nabla = d + [A,\cdot] .
\end{equation}
By the horizontal property of the curvature, we see that \(F\) can be regarded as a 2-form on \(M\) with values in \(\text{ad}(P)\).\footnote{Being horizontal means that pulling it back on the base space, we do not lose information on the 2-form. In fact, the local representation of the curvature \(F^{(\alpha)}=\sigma_\alpha^*(F)\) still transforms covariantly also as a \(\mathfrak{g}\)-valued 2-form over \(M\). Strictly speaking, the covariant derivative on the adjoint bundle acts on this local representative. Notice that the gauge field \(A^{(\alpha)}\) instead looks only locally as an element of the adjoint bundle, but globally it does not respect the \lq\lq right\rq\rq\ transformation property, and indeed it comes from a global 1-form on \(P\) that is not horizontal, but vertical.} Then the Bianchi identity can be rewritten in terms of the covariant derivative,
\begin{equation}
\nabla F = dF + [A,F] = [F,A] + [A,F]=0 .
\end{equation}

As the last piece of information, we recall the meaning of \emph{gauge transformations} from the perspective of the principal bundle. Locally, we can think of them as local actions of the gauge group \(G\), so that a gauge transformation is a map that associates to every point \(x\in M\) an element \(g(x)\in G\), acting on the local field strength in the adjoint representation. At the level of the principal bundle, this can be viewed more formally defining the group \(\mathcal{G}(P)\subset\mathrm{Diff}(P)\) of principal bundle maps of the type\footnote{As  a principal bundle map it is by definition \(G\)-equivariant, \(\Psi(p\cdot g) = \Psi(p)\cdot g\), and it commutes with the projection, \(\pi(\Psi(p))=\pi(p)\), for every \(p\in P\).}
\begin{equation}
\begin{tikzcd}[column sep=small]
P \arrow[rr,"\Psi"] \arrow{rd}[swap]{\pi} && P \arrow[dl,"\pi"] \\
 & M .
\end{tikzcd}
\end{equation}
We notice right-away that, from the local point of view, this can indeed be identified with the space of sections \(\Omega^0\left(M;\mathrm{Ad}(P)\right)\) of the bundle \(\mathrm{Ad}(P)\),  associated to \(P\) with typical fiber \(G\) and \(G\)-action defined by conjugation (the adjoint representation of \(G\) on itself).\footnote{Notice that this looks like \(P\to\Sigma\) as a fiber bundle, since both are locally trivial with fiber \(G\). It is only the \(G\) action that distinguishes them. On \(P\) we have a right action, on \(\mathrm{Ad}(P)\) we have the left action on the fibers \(g\cdot f := g f g^{-1}\).}

\begin{proof}[Proof of \(\mathcal{G}(P)\cong \Omega^0(M;\mathrm{Ad}(P))\)]
We can see that associated to every element \(\Psi\in \mathcal{G}(P)\) there is a unique class of local sections \(\lbrace \psi_\alpha :U_\alpha \to G\rbrace\) that transforms in the adjoint representation, and vice versa.
\begin{itemize}
\item[\((\Rightarrow)\)] In every local patch \(U_\alpha\), let us define the map \(\tilde{\psi}_\alpha:\pi^{-1}(U_\alpha)\to G\) such that \[ \tilde{\psi}_\alpha(p):= g_\alpha(\Psi(p))\ g_\alpha(p)^{-1},\]
where \(g_\alpha\) is the trivialization map inside \(U_\alpha\). By equivariance of \(\Psi\) and \(g_\alpha\), \(\tilde{\psi}_\alpha\) is \(G\)-invariant, so it depends only on the base point \(\pi(p)\). Thus we can define \(\psi_\alpha :U_\alpha\to G\) such that \[ \psi_\alpha(x) := \tilde{\psi}_\alpha(p) \qquad \text{for some}\ p\in \pi^{-1}(x) .\]
In changing local patch, this transforms in the adjoint representation. In fact, if \(x\in U_\alpha \cap U_\beta\) \[
\begin{aligned}
\psi_\beta(x) &= g_\beta(\Psi(p))\ g_\beta(p)^{-1} \\
&= \left[ g_\beta(\Psi(p))  g_\alpha(\Psi(p))^{-1}\right] g_\alpha(\Psi(p)) g_\alpha(p) \left[ g_\alpha(p)^{-1} g_\beta(p)\right] \\
&= g_{\alpha\beta}(x) \psi_\alpha(x) g_{\alpha\beta}(x)^{-1} ,
\end{aligned}
\]
where in the last passage we recognized that \( g_\beta(p)g_\alpha(p)^{-1}\) is \(G\)-invariant and thus can be written as a map \(g_{\alpha\beta}:U_\alpha \cap U_\beta\to G\) that depends only on the base point \(\pi(p)\), and we used that \(\pi\circ \Psi = \pi\).

\item[\((\Leftarrow)\)] Starting from a class of local sections \(\lbrace \psi_\alpha\rbrace\), we define the \(G\)-invariant maps \(\tilde{\psi}_\alpha := \psi_\alpha\circ \pi\). Then we can obtain \(\Psi\) by \lq\lq inverting\rq\rq\ the above definition in every patch and gluing them together,
\[ \Psi(p) := \sigma_\alpha(p) \cdot \left( \tilde{\psi}_\alpha(p) g_\alpha(p) \right). \]
\end{itemize}
\end{proof}

The group of gauge transformations \(\mathcal{G}(P)\) acts naturally on the space \(\mathcal{A}(P)\) via pull-back,
\begin{equation}
\Psi\cdot A := \Psi^*(A) \qquad \forall A\in \mathcal{A}(P),\Psi\in \mathcal{G}(P).
\end{equation}
If we consider the local gauge field \(A^{(\alpha)}\), one can prove that the trivialization of the gauge-transformed connection follows the usual rule
\begin{equation}
(\Psi\cdot A)^{(\alpha)} = \psi_\alpha^{-1} A^{(\alpha)} \psi_\alpha + \psi^{-1}_\alpha (d\psi_\alpha).
\end{equation}
From this local expression, it is easy to find the representation of \(Lie(\mathcal{G}(P))\) on \(T\mathcal{A}(P)\). In fact, writing \(\Psi\) as \(\exp(X)\) for some \(X\in Lie(\mathcal{G}(P))\cong \Omega^0(M;\mathrm{ad}(P))\), we can recognize the associated fundamental vector field as
\begin{equation}
\underline{X}_A = \left.\frac{d}{dt}\right|_{t=0} e^{-tX} \cdot A = dX + [A,X] = \nabla^A X .
\end{equation}
This make us see the usual \lq\lq infintesimal variation\rq\rq\ \(\delta_X A\) as a tangent vector \(\delta_X A \equiv \underline{X}_A \in T_A \mathcal{A}(P)\) at the point \(A\in \mathcal{A}(P)\).

\section{Spinors in curved spacetime}
\label{app:spin-geom}

In QFT, \emph{fermionic} particles are described geometrically by \emph{spinors}, \textit{i.e.}\ fields that transform under the Lorentz algebra in representations whose angular momentum is \emph{half-integer}. At the level of Lie groups, they transform thus in representations of the double-cover of the rotation group of spacetime, \(SO(1,d-1)\) (or \(SO(d)\) in the Euclidean case),
\begin{equation}
SO(1,d-1) \cong \faktor{Spin(1,d-1)}{\mathbb{Z}_2} .
\end{equation}
For simplicity, let us denote the dimension by \(d\) for the rest of the section, since the discussion is valid both for the Euclidean and the Lorentzian signature. In Minkowski spacetime \((\mathbb{R}^d,\eta)\), there exists a preferred class of \emph{global} coordinate systems, the global \lq\lq inertial frames\rq\rq, where the metric is diagonal
\begin{equation}
\eta_{\mu\nu} = diag(-1,+1,\cdots,+1)
\end{equation}
and that are preserved by the Lorentz transformations. Working only with such special type of coordinate systems, one can introduce and work with spinors as living in double-valued representations of the Lorentz algebra, and transforming as
\begin{equation}
\begin{aligned}
&\text{vector fields}: \qquad V^\mu \mapsto \Lambda^\mu_\nu V^\nu , \\
&\text{spinor fields}: \qquad \Psi_\alpha \mapsto S(\Lambda)^\beta_\alpha \Psi_\beta ,
\end{aligned}
\end{equation}
where, if \(\Lambda = \exp(i\omega_{\mu\nu}M^{\mu\nu})\), \(S(\Lambda) = \exp(i\omega_{\mu\nu}\Sigma^{\mu\nu})\). When we move on to the description of a generically curved spacetime \(M\), there is a priori no such choice of \lq\lq preferred\rq\rq\ coordinate systems, and a general coordinate transformation (GCT) is generated by a diffeomorphism \(M\to M\), reflecting on the tangent spaces as \(GL(d,\mathbb{R})\) basis changes. \(SO(d)\) injects as a subgroup of the General Linear group, but \(Spin(d)\) does not, since it is a double cover, so it is not clear a priori how GCTs act on spinor fields. Tensor fields are naturally present in the fully covariant formalism as fields over the manifold \(M\), but to define spinors one has to introduce further structure.

The solution to this puzzle is really to (try to) mimic the same idea applied the the Minkowski case, and employ the presence of a (pseudo-)Riemannian metric \(g\) on \(M\). On the metric manifold \((M,g)\) all the tangent bundles arise as associated bundles to the \emph{frame bundle} \(LM\), that is a principal \(GL(d,\mathbb{R})\)-bundle over \(M\). Using the presence of a metric on \(M\), one can restrict the frame bundle to a principal \(SO(d)\)-bundle, by considering only those frames \(e=(e_1,\cdots,e_d)\) such that, at a given point
\begin{equation}
g(e_i,e_j) = \eta_{ij},
\end{equation}
where \(\eta\) is the \lq\lq flat\rq\rq\ Minkowski (or Euclidean) metric. This reduction defines the so called \emph{orthonormal frame bundle} \(LM^{(SO)}\xrightarrow{\pi} M\). A section of this bundle is an orthonormal frame, or \emph{tetrad}. It is customary to denote with Latin indices the expansion of every vector field with respect to an orthonormal frame, and with Greek indices the expansion with respect to a generic (for example chart-induced) frame:\footnote{Latin indices are sometimes called \lq\lq flat\rq\rq, and Greek ones \lq\lq curved\rq\rq. If one needs to raise and lower indices, flat indices are understood to be multiplied by the diagonalized metric \(\eta_{ij}\), curved indices by \(g_{\mu\nu}\).}
\begin{equation}
V = V^i e_i = V^\mu \frac{\partial}{\partial x^\mu} \qquad \text{for}\ V\in \Gamma(TM).
\end{equation}
The choice of an orthonormal frame is encoded in the choice of a \emph{vielbein}, or \emph{solder form} on \(M\), that is a linear identification of the tangent bundle with the typical fiber \(\mathbb{R}^d\):
\begin{equation}
\begin{aligned}
E : TM &\to \mathbb{R}^d \\
 V &\mapsto E(V) := (\tilde{e}^i(V))_{i=1,\cdots,d} ,
\end{aligned}
\end{equation}
where \((\tilde{e}^i)\) is the dual frame to a chosen orthonormal frame \((e_i)\). Notice that the choice of a metric is in one to one correspondence with the choice of a vielbein, since
\begin{equation}
g(\cdot,\cdot) = \langle E(\cdot) ,E(\cdot)\rangle,
\end{equation}
where \(\langle\cdot,\cdot\rangle\) is the canonical inner product on \(\mathbb{R}^d\) with the chosen signature. In a chart-induced basis, \(g_{\mu\nu} = e^i_\mu e^j_\nu \eta_{ij}\), where we denoted the components of the vielbein \((E(\partial_\mu))^i \equiv e^i_\mu\). The \lq\lq inverse vielbein\rq\rq\ at any point is the matrix \(e^\mu_i\) such that \(e^i_\mu e^\mu_j = \delta^i_j\).

Once this orthonormal reduction is made, one can define spinor bundles as associated bundles to a principal \(Spin(d)\)-bundle, that must be compatible with the orthonormal frame bundle. This is made precise by defining the presence of a spin-structure on \(M\).
\begin{defn} A \emph{spin-structure} on \((M,g)\) is a principal \(Spin(d)\)-bundle \(Spin(M)\xrightarrow{\pi_S} M\), together with a principal bundle map\footnote{Recall that a principal bundle map by definition commutes with the projections, \(\pi(\Phi(p))=\pi_S(p)\).} \begin{equation*}
\begin{tikzcd}[column sep=small]
Spin(M)\arrow[rr,"\Phi"] \arrow{dr}[swap]{\pi_S}  && LM^{(SO)} \arrow[dl,"\pi"] \\
 & M 
\end{tikzcd}
\end{equation*}
with respect to the double-cover map \(\varphi:Spin(d) \to SO(d)\). This means that the equivariance condition is \[
\Phi(s\cdot g) = \Phi(s)\cdot \varphi(g) \qquad \forall s\in Spin(M), g\in Spin(d).\]
A section of \(Spin(M)\to M\) is called \emph{spin-frame}.
\end{defn}
We notice that the equivariance condition in this definition is just the formal requirement that spinors and tensors transform all together with compatible rotations by the action of the respective groups. Although the above restriction of the frame bundle to the orthonormal frame bundle can always be done in presence of a metric on \(M\), a spin-structure does not necessarily exist, and if it does it is not necessarily unique. There can be topological obstructions to this process that can be characterized in terms of the cohomology of \(M\).\footnote{In particular, it turns out that a spin-structure exists if and only if the second Stiefel–Whitney class of \(M\) vanishes \cite{marsh}.}

By this construction, and from the canonical Levi-Civita covariant derivative \(\nabla\) on \((M,g)\), we can induce a connection 1-form on the orthonormal frame bundle and on the spin-frame bundle, and thus have a compatible covariant derivative on associated spinor bundles. Let us recall that the Levi-Civita connection on \((M,g)\) is the unique metric-compatible and torsion free connection, \textit{i.e.}
\begin{equation}
\begin{array}{lcl}
\nabla_X g = 0 & \Leftrightarrow & X(g(Y,Z)) = g(\nabla_XY,Z)+ g(Y,\nabla_XZ) ,\\
T = 0 & \Leftrightarrow & \nabla_XY - \nabla_YX = [X,Y] .
\end{array}
\end{equation}
This covariant derivative is associated to the gauge field \(\Gamma\in \Omega^1(M)\otimes \mathfrak{gl}(n,\mathbb{R})\) such that \(\Gamma^\rho_{\mu\nu} := (\nabla_\mu (\partial_\nu))^\rho\). Simply restricting to orthonormal frames, one can induce a connection 1-form on \(LM^{(SO)}\), \(\omega\in \Omega^1(LM^{(SO)})\otimes \mathfrak{so}(d)\) such that in any trivialization induced by a local frame \((U\subset M, e:U\to LM^{(SO)})\) the gauge field has components
\begin{equation}
\omega(X)^i_j := (\nabla_X (e_i))^j = X^\mu e^j_\nu (\nabla_\mu e_i)^\nu = X^\mu e^j_\nu\left( \partial_\mu e^\nu_i + \Gamma^\nu_{\mu\sigma} e^\sigma_i \right)  \quad \text{or} \quad \omega(X)_{ij} = g(\nabla_X e_i,e_j) ,
\end{equation}
and it can be written as \(\omega^{(U)}:= e^*\omega = \frac{1}{2} \omega_{ij} M^{ij}\), where \(M^{ij}\) are the generators of \(\mathfrak{so}(d)\). Given a spin-structure as in the above definition, we can induce a \emph{compatible spin-connection} \(\tilde{\omega}\in \Omega^1(Spin(M))\otimes \mathfrak{so}(d)\) by pulling back \(\omega\), \(\tilde{\omega} := \Phi^*\omega\).\footnote{Notice that \(Lie(Spin(d))\cong Lie(SO(d))\cong \mathfrak{so}(d)\).} In a given patch \(U\subset M\), if \(s:U\to Spin(M)\) is a local spin-frame and \(e:= \Phi\circ s\) is the associated tangent frame, the local gauge fields representing the spin-connection and the Levi-Civita connection coincide,
\begin{equation}
\tilde{\omega}^{(U)} := s^*\tilde{\omega} = (\Phi\circ s)^* \omega = \omega^{(U)} ,
\end{equation}
so in particular the local components of the compatible spin-connection are defined as 
\begin{equation}
\tilde{\omega}(X)^i_{j} = (\nabla_X (e_i))^j.
\end{equation}

The covariant derivative on an associated spinor bundle is defined as usual. Let \(V\) be the typical fiber, acted upon by the representation \(\rho:Spin(d)\to GL(V)\). Then for every local \(V\)-valued function \(\psi:U\to V\),
\begin{equation}
\nabla_X \psi = d\psi(X) + \frac{1}{2}\omega(X)_{ij} \rho(M^{ij})\cdot \psi .
\end{equation}
If in particular we take the fundamental representation of \(Spin(d)\), \textit{i.e.}\ \(\psi\) is a \emph{Dirac spinor}, the generators are \(\rho(M^{ij}) = \Sigma^{ij} := \frac{1}{4}[\gamma^i,\gamma^j]\), where \(\gamma^i\) are the Dirac matrices. Thus,
\begin{equation}
\nabla_\mu \psi = \partial_\mu \psi + \frac{1}{8}\omega_{\mu ij} [\gamma^i,\gamma^j]\cdot \psi .
\end{equation}

We quote the fact that, in general, one is not forced to consider a spin-connection that is compatible with the Levi-Civita connection.\footnote{For example in SUGRA it is sometimes convenient to work with torsion-full spin connections.} However, in this work we always implicitly define covariant derivatives on spinors via a compatible spin-connections. A discussion about spinors in curved spacetime can be found also in \cite{wald-book}.

\chapter{Mathematical background on equivariant cohomology}
\label{app:eq_coho}

\section{Equivariant vector bundles and equivariant characteristic classes}
\label{app:char-classes}

We recall the definitions of characteristic classes on principal bundles \cite{tu-differential_geometry} and then their equivariant version when the bundle supports a \(G\)-action for some Lie group \(G\). Consider a principal \(H\)-bundle \(P\xrightarrow{\pi}M\) with connection 1-form \(A\), and curvature \(F\). Both are forms on \(P\) with values in the Lie algebra \(\mathfrak{h}\). A \emph{polynomial} on \(\mathfrak{h}\) is an element \(f\in S(\mathfrak{h}^*)\), and it is called \emph{invariant polynomial} if it is invariant with respect to the adjoint action of \(H\) on \(\mathfrak{h}\),
\begin{equation}
f(Ad_{*h}X) = f(X) \qquad \forall X\in \mathfrak{h}, h\in H.
\end{equation}
For example, if \(H\) is a matrix group, the adjoint action is simply \(Ad_{*h}X = hXh^{-1}\). If \(f\) is an invariant polynomial of degree \(k\), then \(f(F)\) is an element of \(\Omega^{2k}(P)\). Explicitly, with respect to a basis \((T_a)_{a=1,\cdots,\dim\mathfrak{h}}\) of \(\mathfrak{h}\) and the dual basis \((\alpha^a)_{a=1,\cdots,\dim\mathfrak{h}}\) of \(\mathfrak{h}^*\), if \(F = F^a T_a\) and \(f = f_{a_1\cdots a_k}\alpha^{a_1} \cdots \alpha^{a_k}\), then
\begin{equation}
f(F) = f_{a_1\cdots a_k} F^{a_1}\wedge \cdots\wedge F^{a_k}.
\end{equation}
The above form has three remarkable properties:
\begin{enumerate}[label=(\roman*)]
\item \(f(F)\) is a basic form on \(P\), \textit{i.e.}\ it exists a \(2k\)-form \(\Lambda\in \Omega^{2k}(M)\) such that \(f(F) = \pi^* \Lambda\);
\item \(d \Lambda=0\), or equivalently \(df(F)=0\);
\item the cohomology class \([\Lambda]\in H^{2k}(M)\) is independent on the connection \(F\).
\end{enumerate}
The cohomology class \([\Lambda]\) on \(M\) is called \emph{characteristic class} of \(P\) associated to the invariant polynomial \(f\). Denoting \(\text{Inv}(\mathfrak{h})\subseteq S(\mathfrak{h}^*)\) the algebra of invariant polynomials on \(\mathfrak{h}\), the map
\begin{equation}
\begin{aligned}
w: \text{Inv}(\mathfrak{h}) &\to H^*(M) \\
f &\mapsto [\Lambda]
\end{aligned}
\end{equation}
is called \emph{Chern-Weil homomorphism}.

If one is considering a vector bundle \(E\to M\) associated to the principal \(H\)-bundle \(P\to M\), here the connection 1-form \(A\) and the curvature \(F\) are represented only locally via \(\mathfrak{h}\)-valued forms on \(M\). Under a change of trivialization the local connection does not transform covariantly, but the local curvature does (by conjugation), so the invariant polynomial \(f(F)\) is independent on the frame and it defines a global form on \(M\). The definition of characteristic classes could be thus given in terms of the local curvature of a vector bundle, without changing the result.

We need mainly three examples of characteristic classes, associated to the invariant polynomials \(\mathrm{Tr}, \det\) and \(\mathrm{Pf}\), that corresponds for matrix groups to the standard trace, determinant and pfaffian. These are the \emph{Chern character} \begin{equation}
\mbox{ch}(F) := \mathrm{Tr}\left( e^F\right) ,
\end{equation}
the \emph{Euler class} 
\begin{equation}
e(F) := \mathrm{Pf}\left(\frac{F}{2\pi}\right) ,
\end{equation}
and the \emph{Dirac \(\hat{A}\)-genus}
\begin{equation}
\hat{A}(F) := \sqrt{\det{\left[\frac{\frac{1}{2}F}{\sinh\left( \frac{1}{2}F\right)} \right]}} .
\end{equation}

\medskip
Now we turn the discussion to the case of \(G\)-equivariant bundles  \cite{bott-tu-classes,berline-vergne-loc,szabo}. 
\begin{defn} A \emph{\(G\)-equivariant vector bundle} is a vector bundle \(E\xrightarrow{\pi} M\), such that:
\begin{enumerate}[label=(\roman*)]
\item both \(E\) and \(M\) are \(G\)-spaces and \(\pi\) is \(G\)-equivariant;
\item \(G\) acts linearly on the fibers.
\end{enumerate}
A principal \(H\)-bundle \(P\xrightarrow{\pi} M\) is \(G\)-equivariant if 
\begin{enumerate}[label=(\roman*)]
\item both \(E\) and \(M\) are \(G\)-spaces and \(\pi\) is \(G\)-equivariant;
\item the \(G\)-action commutes with the \(H\)-action on \(P\).
\end{enumerate}
\end{defn}

Usually, a connection \(A\) on a \(G\)-equivariant principal bundle is required to be \emph{\(G\)-invariant}, that is \(\mathcal{L}_X A=0\) for every \(X\in \mathfrak{g}\). If \(G\) is compact, this choice is always possible by averaging any connection over \(G\) to obtain a \(G\)-invariant one \cite{bott-tu-classes}. Since the \(G\)- and the \(H\)-actions commute, a principal \(H\)-bundle \(P\xrightarrow{\pi} M\) induces another principal \(H\)-bundle \(P_G\xrightarrow{\pi_G} M_G\) over the homotopy quotient \(M_G\). Topologically, the \emph{equivariant characteristic classes} of \(P\xrightarrow{\pi} M\) are the ordinary characteristic classes of \(P_G\xrightarrow{\pi_G} M_G\), thus defining elements in the \(G\)-equivariant cohomology \(H_G^*(M)\). From the differential geometric point of view, they can be derived as \emph{equivariantly closed extensions} of the ordinary characteristic classes in the Cartan model. In particular, in \cite{bott-tu-classes,berline-vergne-loc} it was shown that the equivariant characteristic class associated to an invariant polynomial \(f\) is represented by \(f(F^{\mathfrak{g}})\), where
\begin{equation}
F^{\mathfrak{g}} = 1\otimes F + \phi^a \otimes \mu_a
\end{equation}
is the equivariant extension of the curvature \(F\) on the principal \(H\)-bundle. \(\phi^{a=1,\cdots,\dim\mathfrak{g}}\) are the generators of \(S(\mathfrak{g}^*)\) in the Cartan model, and
the map \( \mu:\mathfrak{g}\to \Omega(P;\mathfrak{h}) \) such that 
\begin{equation}
\mu_X := -\iota_X A = -A(\underline{X})
\end{equation}
is called \emph{moment map}, with analogy to the symplectic case. We denoted \(\mu_a \equiv \mu_{T_a}\) with \(T_{a=1,\cdots,\dim\mathfrak{g}}\) the basis of \(\mathfrak{g}\) dual to \(\phi^a\). If we define \(\nabla = d+ A\) the covariant derivative, that in the adjoint bundle acts as \(\nabla \omega = d\omega + [A\stackrel{\wedge}{,}\omega]\), we notice that we can obtain the above equivariant curvature in the Cartan model from the \emph{equivariant covariant derivative}
\begin{equation}
\nabla^{\mathfrak{g}} := 1\otimes \nabla - \phi^a \otimes \iota_a
\end{equation}
that is completely analogous to the definition of the Cartan differential \eqref{eq:Cartan-differential}. With this definition, the equivariant curvature can be expressed as
\begin{equation}
F^\mathfrak{g} = (\nabla^\mathfrak{g})^2 + \phi^a \otimes  \mathcal{L}_a ,
\end{equation}
where the last piece takes care of the non-nilpotency of the Cartan differential on generic differential forms, and moreover it satisfies an equivariant variation of the Bianchi identity
\begin{equation}
(\nabla^\mathfrak{g} F^\mathfrak{g}) = 0 .
\end{equation}
Notice that, if we assume the connection \(A\) to be \(G\)-invariant, the moment map \(\mu\) indeed satisfies a moment map equation with respect to the curvature \(F\) (see Section \ref{subsec:eq-cohom-symplectic}),
\begin{equation}
\nabla \mu_X = -\iota_X F \qquad \forall X\in\mathfrak{g}.
\end{equation}

Once a suitable equivariant extension of the curvature \(F^\mathfrak{g}\) is known, the particular equivariant characteristic classes are simply a modification of the old ones, so the equivariant version of the above Chern character, Euler class and Dirac \(\hat{A}\)-genus are given by
\begin{equation}
\mathrm{ch}_G(F) := \mathrm{Tr}\left(e^{F^\mathfrak{g}}\right) ,\quad e_G(F) := \mathrm{Pf}\left(\frac{F^\mathfrak{g}}{2\pi}\right) , \quad \hat{A}_G(F) := \sqrt{\det{\left[\frac{\frac{1}{2}F^\mathfrak{g}}{\sinh\left( \frac{1}{2}F^\mathfrak{g}\right)} \right]}} ,
\end{equation}
respectively.

\section{Universal bundles and equivariant cohomology}
\label{app:univ-bundles}

In this section we motivate the well-definiteness of equivariant cohomology of Section \ref{sec:equivariant-cohomology}, starting from the definition of the space \(EG\). Proofs for the various propositions we are going to state informally and/or without proof can be found for example in \cite{tu-equiv_cohom, hatcher,Bott-equiv-cohom}. We should mention that the mathematically correct approach to this subject works considering only \emph{CW complexes}. These are special types of topological spaces that can be constructed by \lq\lq attaching deformed disks\rq\rq\ to each other \cite{hatcher}. We only quote that any smooth manifold can be given the structure of a CW complex, so that in the smooth setting we do not need to bother with this subtlety.\footnote{This is a result of Morse theory, see \cite{tu-equiv_cohom} and references therein.}

\begin{defn} A principal \(G\)-bundle \(\pi : EG \rightarrow BG\) is called \emph{universal G-bundle} if:
\begin{enumerate}[label=(\roman*)]
\item for any principal \(G\)-bundle \(P\rightarrow X\), there exists a map \(h:X\rightarrow BG\) such that \(P\cong h^*(EG)\)  (the pull-back bundle of \(EG\) through \(h\));
\item if \(h_0, h_1:X\rightarrow BG\) are such that \(h_0^*(EG)\cong h_1^*(EG)\), then the two maps are homotopic. 
\end{enumerate} 
The base space \(BG\) is called \emph{classifying space}. \end{defn}

The classifying property (i) required of \(EG\) means that for every principal \(G\)-bundle there is a copy of it sitting inside \(EG\rightarrow BG\). The important fact is that existence can be proven for a large class of interesting cases, the argument going as follows. First recall that \emph{homotopic maps pull back to isomorhic bundles}, \textit{i.e.}\ if \(E\rightarrow B\) is a vector bundle, \(X\) a paracompact space, then
\begin{equation}
\label{tu-1}
g,h:X\rightarrow B\ \text{homotopic maps}\ \Rightarrow g^*(E)\cong h^*(E) .
\end{equation}
Then the property (ii) in the definition above states that if \(E\rightarrow B\) is a universal bundle, \(\Rightarrow\) is replaced by \(\Leftrightarrow\). Now let, for any paracompact space \(X\),
\begin{equation}
P_G(X):=\left\lbrace \text{isomorphism classes of principal G-bundles over X}\right\rbrace
\end{equation}
and for some space \(BG\) (to be identified with the classifying space),
\begin{equation}
[X,BG]:= \left\lbrace \text{homotopy classes of maps }X \rightarrow BG\right\rbrace. 
\end{equation} 
Notice that the definition of \(P_G(X)\) is totally independent from the notion of universal \(G\)-bundle. Considering then the map
\begin{equation}
\begin{aligned}
  \phi : [X,BG] &\rightarrow P_G(X) \\
  [h:X\rightarrow BG] &\mapsto h^*(EG) ,
  \end{aligned}
\end{equation}
by \eqref{tu-1} we have that it is well-defined (independent from the representatives). The conditions (i) and (ii) are equivalent to surjectivity and injectivity of \(\phi\), so finally \(P_G(X)\cong[X,BG]\). Since \(P_G(X)\) exists, this proves the existence of the classifying space \(BG\) and of the universal bundle \(EG\rightarrow BG\).\footnote{In the language of category theory, we could say that \(P_G(\cdot)\) is a contravariant functor,  \emph{representable} through \([\cdot,BG]\).}

We can now motivate the well-definiteness of the Borel construction for equivariant cohomology. A fundamental result for this is that  \emph{a principal \(G\)-bundle is a universal bundle if and only if its total space is (weakly) contractible}.\footnote{A \emph{weakly contractible} space is a topological space whose homotopy groups are all trivial. Clearly any contractible space is weakly contractible. It is a fact that every CW complex that is weakly contractible is also contractible \cite{hatcher}, so for our purposes the two concepts coincide.} 
The contractibility of \(EG\) makes its cohomology trivial, so that, since \((M\times EG)\sim M\), we have \(H^*(M\times EG)\cong H^*(M)\). When we take the homotopy quotient, the product by \(EG\) acts as a \lq\lq regulator\rq\rq\ of the resulting cohomology. In fact, if the action of \(G\) on \(M\) is free, such that \(M\to M/G\) is a principal \(G\)-bundle, one can prove that for any (weakly) contractible \(G\)-space \(E\)
\begin{equation}
\faktor{(M\times E)}{G} \sim \faktor{M}{G} ,
\end{equation}
where \(\sim\) here stands for \lq\lq weakly homotopic\rq\rq.\footnote{Two spaces are \emph{weakly homotopic} if they have the same homotopy groups. Again, homotopy equivalence implies weak homotopy equivalence, and for CW complexes these two concepts coincide.} In general, even if the \(G\)-action is not free, two homotopy quotients with respect to different (weakly) contractible \(G\)-spaces \(E\) and \(E'\) are (weakly) homotopy equivalent,
\begin{equation}
\faktor{(M\times E)}{G} \sim \faktor{(M\times E')}{G} .
\end{equation}
Another known fact is that \emph{weakly homotopic spaces have the same (co)homology groups, for all coefficients}, generalizing \eqref{cohomology-1}. Putting together these properties, we have that 
\begin{equation}
H^*\left(\faktor{(M\times E)}{G}\right) \cong H^*\left( \faktor{(M\times E')}{G}\right) ,
\end{equation}
so that  the resulting cohomology  is independent of the choice of contractible principal \(G\)-bundle. The homotopy quotient thus well-defines the \(G\)-equivariant cohomology of \(M\), producing an \lq\lq homotopically correct\rq\rq\ version of its orbit space. As pointed out in Section \ref{sec:equivariant-cohomology}, when the \(G\)-action is free on \(M\) this reproduces the naive definition of cohomology of the quotient space \(M/G\).

\subsection*{Every compact Lie group has a universal bundle}

In Section \ref{sec:equivariant-cohomology} we gave the example of the universal bundle for the circle, \(EU(1)=\mathbb{S}^\infty\) and \(BU(1)=\mathbb{C}P^\infty\). One can generalize this construction to concretely define a universal bundle for any compact Lie group \(G\). This is because any such Lie group  embeds into \(U(n)\) or \(O(n)\) (the maximal compact subgroups of \(GL(n,\mathbb{C})\) and \(GL(n,\mathbb{R})\)), for some \(n\), and for them one can construct universal bundles explicitly. As a subgroup, \(G\) will act freely on the given universal bundle. Then one can take this to be its universal bundle too.

A class of principal \(O(n)\)- or \(U(n)\)-bundles is given by the so-called \emph{Stiefel manifolds}. A Stiefel manifold \(V_k(\mathbb{F}^n)\) is the set of all \emph{orthonormal \(k\)-frames} in \(\mathbb{F}^n\), where \(\mathbb{F}=\mathbb{R},\mathbb{C}\), and the orthonormality is defined with respect to the canonical Euclidean or sesquilinear inner products. A \emph{\(k\)-frame} is an ordered  set \((v_1,\cdots,v_k)\) of \(k\) linearly independent vectors in \(\mathbb{F}^n\). Notice that when \(k=1\), \(V_1(\mathbb{C}^n)\) is the set of all unit vectors in \(\mathbb{C}^n\cong \mathbb{R}^{2n}\), \textit{i.e.}\ the \((2n-1)\)-sphere. The latter is acted freely by the group \(U(1)\) by diagonal multiplication, and analogously the Stiefel manifold \(V_k(\mathbb{C}^n)\) is acted freely by \(U(k)\), that essentially rotates the vectors of the \(k\)-frames. Analogously, \(V_k(\mathbb{R}^n)\) is acted freely by \(O(k)\). Thus we have the generalization of the sequence of principal \(U(k)\)- and \(O(k)\)-bundles, that in the limit \(n\to\infty\) produces the contractible universal bundles \(EU(k)=V_k(\mathbb{C}^\infty)\) and \(EO(k)=V_k(\mathbb{R}^\infty)\). The Stiefel manifolds can thus be seen as a \lq\lq higher dimensional versions\rq\rq\ of the spheres, in the sense of the following consideration: 
\begin{equation}
V_k(\mathbb{R}^n) \cong \faktor{O(n)}{O(n-k)}.
\end{equation}
Comparing with Example \ref{ex:group-actions}, where we remarked that \(\mathbb{S}^{n-1}\cong O(n)/O(n-1)\), we see that indeed \(\mathbb{S}^{n-1} \equiv V_1(\mathbb{R}^n)\). The \((n-1)\)-sphere is just the set of unit vectors in \(\mathbb{R}^n\), so orthonormal 1-frames.

The base spaces \(G_k(\mathbb{C}^n) = V_k(\mathbb{C}^n)/U(k)\) and \(G_k(\mathbb{R}^n) = V_k(\mathbb{R}^n)/O(k)\) are the sets of equivalence classes of \(k\)-frames, that identify \(k\)\emph{-hyperplanes} through the origin inside \(\mathbb{C}^n\) or \(\mathbb{R}^n\). These manifolds are called \emph{Grassmannians}. The \emph{infinite Stiefel manifold} \(V_k(\mathbb{F}^\infty)\)  and the \emph{infinite Grassmannian} \(G_k(\mathbb{F}^\infty)\) are thus the total space of the universal bundle and the classifying space for the unitary and orthogonal groups \(U(k)\) and \(O(k)\), and generalize the universal bundle \(\mathbb{S}^\infty\to \mathbb{C}P^\infty\) of the circle.

As  recalled above, any compact Lie group \(G\) can be embedded as a closed subgroup of an orthogonal group (or a unitary group). This means that \(G\) also acts freely on \(V_k(\mathbb{F}^\infty)\) for some \(k\), and in turn \(V_k(\mathbb{F}^\infty)\to V_k(\mathbb{F}^\infty)/G\) is a principal \(G\)-bundle, whose total space is a contractible space. This gives the universal bundle for any compact Lie group \(G\).

\subsection*{Module structure of equivariant cohomology}
We end this section with a more algebraic comment about the construction of equivariant cohomology. Notice first that 
\begin{equation}
\label{tu-2}
pt_G = \faktor{(pt\times EG)}{G} \cong BG \quad \Rightarrow \quad H_G^*(pt)\cong H^*(BG),
\end{equation}
so the equivariant cohomology of a point is the standard cohomology of the classifying space \(BG\), generalizing Example \ref{ex:equivariant-cohomologies}. Thus the equivariant cohomology \(H_G^*(\cdot)\) inherits analogous functorial properties to the standard (singular) cohomology of the last section, with respect to the ring \(H^*(BG)\) instead of the coefficient ring \(A\cong H^*(pt; A)\). To see this, let us first notice that a  \(G\)-equivariant function \(f:M\to N\) between the two \(G\)-spaces \(M,N\) induces a well-defined map between the two homotopy quotients,
\begin{equation}
\begin{aligned}
   f_G :  M_G &\rightarrow N_G \\
           [m,e] &\mapsto [f(m),e].
\end{aligned}
\end{equation}
This induced map inherits many properties from \(f\):
\begin{enumerate}[label=(\roman*)]
   \item if \(f\) is injective (surjective), then \(f_G\) is injective (surjective);
   \item if \(id:M\rightarrow M\) is the identity, then \(id_G:M_G\rightarrow M_G \) is the identity;
   \item \( (h\circ f)_G = h_G \circ f_G\);
   \item if \(f:M\rightarrow N\) is a fiber bundle with fiber \(F\), then \( f_G:M_G\rightarrow N_G \) is also a fiber bundle with fiber \(F\).
\end{enumerate}
As pointed out in Section \ref{sec:cohomology-review}, a map between two topological spaces induces a map (in the opposite direction) between the associated singular cohomologies, so
\begin{equation}
f_G^* : \left(H^*(N_G)\equiv H^*_G(N)\right) \to \left( H^*(M_G)\equiv H^*_G(M)\right).
\end{equation}
Defining thus a trivial map \(\phi :M\to pt\), we see from \eqref{tu-2} that the induced homomorphism \(\phi^*_G : H^*(BG)\to H^*_G(M)\) makes the equivariant cohomology \(H^*_G(M)\) naturally into a \(H^*(BG)\)-module!\footnote{Recall that singular cohomology has a ring structure.} Also, in general \(f_G^* : H^*_G(N) \to  H^*_G(M)\) is a \(H^*(BG)\)-module homomorphism.\footnote{In category theory terminology, we could say that the Borel construction \((\cdot)_G\) is a covariant functor from the category of \(G\)-spaces to \textbf{Top} (or \textbf{Man}), and \(H_G^*(\cdot)\) is a contravariant functor between \textbf{Top} (or \textbf{Man}) and the category of \(H^*(BG)\)-modules.} Notice that the cohomology of the classifying space \(BG\) is usually very simple, as we pointed out in Section \ref{sec:cartan} via its associated Weil model.

There is a curious difference between standard cohomology and equivariant cohomology regarding the associated coefficient rings. In the former case, it is clear from the various examples in Section \ref{sec:cohomology-review} that the coefficient ring \(\mathbb{R}\cong H^*(pt)\) always embeds into the cohomology \(H^*(M)\) (also for other commutative rings). In the case of equivariant cohomology, on the other hand, the coefficient ring \(H_G^*(pt) = H^*(BG) = S(\mathfrak{g}^*)^G\) does not, since the map \(\phi_G^*\) above is not injective in general, as it is clear also from the example of \(H^*_{U(1)}(\mathbb{S}^1) = \mathbb{R}\). It turns out that the condition for \(H^*(BG)\) to embed in \(H^*_G(M)\) is that \(G\) acts on \(M\) \emph{with fixed points}. We can argue briefly why this is the case. Let \(p\in M\) be a fixed point. The inclusion \(i:\lbrace p\rbrace\to M\) is \(G\)-equivariant since the action on \(p\) is trivial, so there is a well-defined map \(i_G: pt_G=BG\to M_G\). This is easily checked to be a section of the bundle \(M_G\xrightarrow{\pi} BG\), with respect to the projection map \(\pi([m,e]) := [e]\in BG\). The identity \(\pi\circ i_G = id_{M_G}\) lifts to the pull-backs in the opposite direction: \(i_G^*\circ \pi^* = id\) on \(H_G^*(pt)=H^*(BG)\). This means that the map \(\pi^*: H^*(BG)\to H_G^*(M)\) has a left-inverse, and thus it is injective. This property can be seen in the example of the \(U(1)\)-equivariant cohomology of the 2-sphere. In this case there are two fixed points, and indeed \(H^*(BU(1))=\mathbb{R}[\phi]\) embeds in \(H^*_{U(1)}(\mathbb{S}^2)= \mathbb{R}[\phi]\oplus \mathbb{R}[\phi]y\), where \(y\) can be identified in the Cartan model with the equivariantly closed extension of the volume form, \(y\equiv [\tilde{\omega}]\).

\section{Fixed point sets and Borel localization}
\label{app:borel-loc}

We now spend a few words about a procedure that we used many times without many worries, that is to \lq\lq algebraically localize\rq\rq\ the space of equivariant differential forms \(\Omega(M)^{U(1)}[\phi]\) with respect to the indeterminate \(\phi\), setting it to \(\phi = -1\). This localization was useful to simplify the notation in many occasions, but it really has a non-trivial deeper meaning. In fact, it allows to show in a more algebraic way that the \(G\)-equivariant cohomology of the smooth \(G\)-manifold \(M\) is encoded in the fixed point set \(F\) of the \(G\)-action, at least when \(G\) is a torus. The fundamental theorem concerning this point is the so-called \emph{Borel localization theorem}, that sometimes allows to obtain the ring structure of the equivariant cohomology of the manifold from that of its fixed point set. We consider the case of a circle action here.

First, let us recall what localization in algebra means. If \(R\) is a commutative ring, the \emph{localization} of \(R\) with respect to a closed subset \(S\subseteq R\) is a way to formally introduce a multiplicative inverse for every element of \(S\) in \(R\), so to introduce \emph{fractions} in \(R\), analogously to what one does in the construction of the rational numbers \(\mathbb{Q}\) from the integers \(\mathbb{Z}\). This procedure makes the former commutative ring into a field (in the algebraic sense).
Since we are interested in \(U(1)\)-equivariant cohomologies, let us consider an \(\mathbb{R}[\phi]\)-module \(N\), and practically define the \emph{localization of N with respect to \(\phi\)} as
\begin{equation}
N_\phi \cong \left\lbrace\left. \frac{x}{\phi^n} \right| x\in N, n\in \mathbb{N} \right\rbrace ,
\end{equation}
identifying elements in \(N_\phi\) as
\begin{equation}
\frac{x}{\phi^n} \sim \frac{y}{\phi^m} \Leftrightarrow \exists k\in\mathbb{N}: \phi^k(\phi^m x - \phi^n y) = 0 \ \text{in}\ N.
\end{equation}
The simplest example of such a localized module is just \(\mathbb{R}[\phi]_\phi \cong \mathbb{R}[\phi^{-1},\phi]\), \textit{i.e.}\ the Laurent polynomials in \(\phi\). Notice that there is always an \(\mathbb{R}[\phi]\)-module homomorphism that makes \(N\) inject into \(N_\phi\), \(i:N\to N_\phi\) such that \(i(x):= x/\phi^0\). If \(f:N\to M\) is an \(\mathbb{R}[\phi]\)-module homomorphism, then there is a well-defined induced homomorphism between the localized modules \(f_\phi:N_\phi \to M_\phi\) such that \(f(x/\phi^n) := f(x)/\phi^n\). The important algebraic property of localization for what concerns this discussion is that \emph{it commutes with cohomology}: if \((A,d)\) is a differential complex,
\begin{equation}
A^{(0)} \xrightarrow{d} A^{(1)} \xrightarrow{d} \cdots , \qquad d^2 = 0 ,
\end{equation}
where \(A^{(i)}\) are \(\mathbb{R}[\phi]\)-modules, then also \((A_\phi,d_\phi)\) is a differential complex, and
\begin{equation}
\label{borel-1}
H^*(A,d)_\phi \cong H^*(A_\phi,d_\phi) .
\end{equation}

Quite analogously, from Example \ref{ex:U(1)-cohom} onward we substitute the \emph{indeterminate} \(\phi\in S(\mathfrak{u}(1)^*)\) with a \emph{variable}, and then set it to the value \(\phi=-1\) for notational convenience. Stated more formally, we start from the Cartan model of \(U(1)\)-equivariant differential forms \(\Omega(M)^{U(1)}[\phi]\), that has clearly an \(\mathbb{R}[\phi]\)-module structure, and localize it to \(\Omega(M)^{U(1)}[\phi]_\phi\), so introducing \(\phi\) also at the denominator. This puts \(\phi\) on the same footing as a real variable, so that we are allowed to fix it to some value, for convenience only. Notice that operations like \eqref{eq:inverse-of-localiz-form}, where we \lq\lq invert\rq\rq\ an equivariant form, are allowed only in the localized module \(\Omega(M)^{U(1)}[\phi]_\phi\), where the division by \(\phi\) is meaningful. From the result \eqref{borel-1}, we understand that this localization of the Cartan model does not spoil the resulting equivariant cohomology \(H_{U(1)}^*(M)\), because the two operations commute.\footnote{Notice that \(H_{U(1)}^*(M)\) has generically an \(\mathbb{R}[\phi]\)-module structure, by the discussion in Appendix \ref{app:univ-bundles} and the application of the Weil model (see Section \ref{sec:cartan}) \(H^*(BG)\cong S(\mathfrak{u}(1)^*) \cong \mathbb{R}[\phi]\).}

The Borel localization theorem relates really the localized equivariant cohomologies of the \(U(1)\)-manifold \(M\) and of its fixed point locus \(F\). To understand what this has to say about the actual equivariant cohomology of \(M\), we recall first some other algebraic facts. A \emph{torsion} element in a module \(N\) over a ring \(R\), is an element \(x\in N\) such that \(\exists r\neq 0 \in R: rx=0\). If \(N\) is an \(\mathbb{R}[\phi]\) module, the element \(x\) is said to be \(\phi\)-\emph{torsion} if it exists some power of \(\phi\) that annihilates it: \(\phi^k x = 0\) for some \(k\in \mathbb{N}\). The module \(N\) is \(\phi\)-torsion if every one of its elements is \(\phi\)-torsion. It is easy to see that\footnote{Just consider that in the localized module \(x \sim \frac{\phi^k}{\phi^k}x\), so if \(x\) is \(\phi\)-torsion it is equivalent to 0 in \(N_\phi\).} 
\begin{equation}
N\ \text{is }\phi\text{-torsion} \quad \Leftrightarrow \quad N_\phi = 0. 
\end{equation}
Applying this to the case of \(N=H^*_{U(1)}(M)\), we can see that the equivariant cohomology in the case of a free \(U(1)\)-action on \(M\) is \(\phi\)-torsion. In fact, if the action is free, we can easily compute \(H_{U(1)}^*(M) = H^*(M/U(1))\), so that \(H_{U(1)}^k(M)=0\) in some degree \(k> \dim(M/U(1))\). This means that \(\phi^k \cdot H_{U(1)}^*(M) = 0\) for some \(k\) high enough. The first argument in Section \ref{sec:eq_loc_princ} in fact is the proof that more is true: \(H^*_{U(1)}(M)\) is \(\phi\)-torsion if the \(U(1)\)-action is \emph{locally free} on \(M\), since we found essentially \((H^*_{U(1)}(M))_\phi = 0\) as the Poincaré lemma, after having introduced \(\phi\) at the denominator.\footnote{In Section \ref{sec:eq_loc_princ} \(M\) is the manifold without its fixed point set, there called \(\tilde{M}\).} This motivates the following theorem, that states that, \emph{up to torsion}, the \(U(1)\)-equivariant cohomology of \(M\) is concentrated on its fixed point set. A proof can be found in \cite{guillemin-sternberg-book, tu-equiv_cohom}.
\begin{thm}[Borel localization]
Let \(U(1)\) act smoothly on the manifold \(M\), with compact fixed point set \(F\). The inclusion \(i:F\hookrightarrow M\) induces an isomorphism of algebras over \(\mathbb{R}[\phi]\),
\[ i^*_\phi : H^*_{U(1)}(M)_\phi \to H^*_{U(1)}(F)_\phi .
\]
\end{thm}

This theorem is an \lq\lq abstract version\rq\rq\ of the localization theorems described in Chapter \ref{cha:loc theorems}, and intuitively gives another way to see that they have to be true, without having to travel through all the smooth algebraic models and the integration theory that we described in due time. It shows that localization is something present at a very low level of structure, originating just from the topological nature of equivariant cohomology.

\section{Equivariant integration and Stokes' theorem}
\label{app:integr}

In this section we define what it means to integrate a \(G\)-equivariant differential form \(\omega\in \Omega_G(M)\) over a smooth, oriented \(G\)-manifold \(M\) of dimension \(dim(M)=n\) and we report an extended version of Stokes' theorem that applies in the equivariant setup. Let \(G\) be a connected Lie group acting (smoothly) on the left on  \(M\), being \( \{\phi^a\}_{a=1,\cdots , dim(\mathfrak{g})} \) a basis for \(\mathfrak{g}^*:=Lie(G)^*\).
If the equivariant form \(\omega\) is of degree \(k\), we can express it as
\begin{equation}
\omega = \omega^{(k)} + \omega^{(k-2)}_a \phi^a + \omega^{(k-4)}_{ab} \phi^a \phi^b + \cdots = \sum_{p\geq 0} \omega^{(k-2p)}_{a_1 \cdots a_p} \phi^{a_1} \cdots \phi^{a_p}
\end{equation}
where the coefficients are differential forms on \(M\), tensor products have been suppressed and we require \(\omega\) to be \(G\)-invariant. The natural way to define integration of such objects is obtained just making the integral \(\int_M\) act on the coefficients \(\omega^{(k-2p)}_{a_1\cdots a_p}\) of the \(\phi\)-expansion of \(\omega\). In this way, one obtains a map
\begin{equation}
\int_M : \Omega_G(M)\to S(\mathfrak{g}^*)\equiv \mathbb{R}[\phi^a] .
\end{equation}
Thanks to the equivariant Stokes' theorem (to be stated later), this descends also in equivariant cohomology, \(\int_M: H^*_G(M)\to  S(\mathfrak{g}^*)\), analogously to the standard (non-equivariant) case.

\begin{defn} The integral on \(M\) of the \(G\)-equivariant form \(\omega\) of \(\textrm{deg}(\omega)=k\) is defined as 
   \[
   \int_M \omega := \sum_{p\geq 0} \left( \int_M \omega^{(k-2p)}_{a_1 \cdots a_p} \right) \phi^{a_1} \cdots \phi^{a_p} .
   \]
\end{defn}
Notice that if \(n\) and \(k\) are of different parity, the integral is automatically zero. If instead \(k=n+2m\) for some \(m\in \mathbb{Z}\), then 
  \begin{equation}
  \int_M \omega = \begin{cases} \left( \int_M \omega^{(n)}_{a_1 \cdots a_m} \right) \phi^{a_1} \cdots \phi^{a_m}  &  k\geqslant n \\
  0  &  k< n .\end{cases}
  \end{equation}
In particular, if we have a top form \(\omega\in\Omega(M)\) on M and \(\tilde{\omega}\) is any equivariant extension of \(\omega\) in \(\Omega_G(M)\), then we can deform the integral
\begin{equation}
\int_M \omega = \int_M \tilde{\omega}
\end{equation}
without changing its value.

We can then prove the equivariant version of the Stokes' theorem.
\begin{thm} Let \(G\) be a connected Lie group acting (smoothly) on the left on a smooth manifold \(M\) with boundary \(\partial M\). If \(\omega \in\Omega_G(M)\) of \(\textrm{deg}(\omega)=k\), then
  \[
  \int_M d_C \omega = \int_{\partial M} \omega
  \]
where \(d_C = 1 \otimes d + \phi^a \otimes \iota_{a} \) is the Cartan differential and \(\iota_a \equiv \iota_{T_a}\), with \(\{T_a\}\) a basis of \( \mathfrak{g}^*\) dual to \(\lbrace \phi^a\rbrace\). \end{thm}
\begin{proof} The proof follows from the direct evaluation and the standard Stokes' theorem. If the integral is not zero, selecting the component of top-degree, \[
\left.(d_C\omega)\right|_{(n)} = d\omega^{(n-1)}_I\phi^I + (\iota_{\lbrace a}\omega^{(n+1)}_{J\rbrace})\phi^a\phi^J ,\]
where \(I=(a_1,\cdots,a_{(k-n)/2})\) and \(J=(a_1,\cdots,a_{(k-n+1)/2})\). The second term vanishes since \(\omega^{(n+1)}=0\) by dimensionality. So the integral of \(d_C\omega\) is the integral of the first term, on which we can use the standard version of Stokes' theorem, and getting the statement of the theorem. \end{proof}

%%%%%%%%%%%%%%%%%%%%%%%%%%%%%%%%%%% BIBLIOGRAFIA %%%%%%%%%%%%%%%%%%%%%%%%%%%%%%%%%%%%%%%%%
%\let\cleardoublepage\clearpage %serve a non lasciare la pagina bianca in mezzo
%\bibliographystyle{unsrt}
%\bibliography{biblio}
\cleardoublepage
%\phantomsection
\addcontentsline{toc}{chapter}{Bibliography}
\printbibliography

@book{
 tu-equiv_cohom,
 author = {Tu, L.W.},
 publisher = "{Princeton University Press}",
 %address = "{New Jersey, United States}",
 title = {Introductory Lectures on Equivariant Cohomology},
 ISBN = {9780691191744},
 year = {2020},
 series = "{Annals of Mathematics Studies}"
}

@book{
brocker-book,
 author = {Br\"{o}cker, T. and tom Dieck, T.},
 publisher = "{Springer-Verlag}",
 %address={Berlin, Heidelberg},
 title = "{Representations of Compact Lie Groups}",
 doi = {10.1007/978-3-662-12918-0},
 year = {1985}
}

@book{
    guillemin-sternberg-book,
    author = "Guillemin, V.W. and Sternberg, S.",
    title = "{Supersymmetry and equivariant de Rham theory}",
    publisher = "Springer",
    %address = "Berlin",
    year = "1999"
}

@book{
    nakahara,
    author = "Nakahara, M.",
    title = "{Geometry, topology and physics}",
    publisher = "Institute of Physics Publishing",
    %address = "London",
    year = "2003"
}

@book{
marsh,
author={Marsh, A.},
title="{Mathematics for Physics}",
publisher={World Scientific Publishing},
year={2018},
doi={10.1142/10816}
}

@book{
bott-tu-differential_forms,
author={Bott, R. and Tu, L.W.},
title={Differential Forms in Algebraic Topology},
publisher="{Springer-Verlag}",
%address={Berlin Heidelberg},
year=1982,
doi={10.1007/978-1-4757-3951-0}
}

@book{
 borel-2,
 ISBN = {9780691090948},
 URL = {http://www.jstor.org/stable/j.ctt1bd6jxd},
 author = {Borel, A.},
 publisher = {Princeton University Press},
 series = {Annals of Mathematics Studies},
 title = {Seminar on Transformation Groups},
 year = {1960}
}

@BOOK
{
   tu-intro_to_mfd,
   AUTHOR    = "Tu, L.W.",
   TITLE     = "An introduction to manifolds",
   PUBLISHER = "Springer",
   YEAR      = "2010"
}

@book{
tu-differential_geometry,
title = {Differential Geometry},
subtitle = {Connections, Curvature, and Characteristic Classes},
author = {Tu, L.W.},
year = {2017},
doi = {10.1007/978-3-319-55084-8},
publisher = {Springer}
}

@book{
    berline-heatkernels,
    title="{Heat Kernels and Dirac Operators}",
    author="Berline, N. and Getzler, E. and Vergne, M.",
    year="2004",
    Publisher= "Springer"
}

@book
{
   hatcher,
   AUTHOR    = "Hatcher, A.",
   TITLE     = "Algebraic topology",
   publisher = "Cambridge University Press",
   YEAR      = "2002"
}

@BOOK
{
   rotman-algebra,
   AUTHOR    = "Rotman, J.J.",
   TITLE     = "Advanced modern algebra",
   PUBLISHER = "Prentice Hall",
   YEAR      = "2003"
}

@book{
    wess-bagger,
    author = "Wess, J. and Bagger, J.",
    title = "{Supersymmetry and supergravity}",
    publisher = "Princeton University Press",
    year = "1992"
}

@BOOK
{
   varadarajan-susy_math,
   AUTHOR    = "Varadarajan, V.",
   TITLE     = "Supersymmetry for mathematicians: an introduction",
   PUBLISHER = "Courant Lecture Notes in Mathematics, American Mathematical Society",
   YEAR      = "2004",
   doi      = "10.1090/cln/011"
}

@book{
    peskin-schroeder,
    author = "Peskin, M.E. and Schroeder, D.V.",
    title = "{An Introduction to quantum field theory}",
    isbn = "978-0-201-50397-5",
    publisher = "Addison-Wesley",
    %address = "Reading, USA",
    year = "1995"
}

@book{
      srednicki,
      author        = "Srednicki, M.",
      title         = "{Quantum Field Theory}",
      publisher     = "Cambridge University Press",
      year          = "2007",
      url           = "https://web.physics.ucsb.edu/~mark/qft.html",
}

@book{
    wald-book,
    author = "Wald, R.M.",
    title = "{General Relativity}",
    doi = "10.7208/chicago/9780226870373.001.0001",
    publisher = "Chicago University Press",
    year = "1984"
}

@book{
  diFrancesco-CFT,
  title="{Conformal Field Theory}",
  author={Di Francesco, P. and Mathieu, P. and S{\'e}n{\'e}chal, D.},
  isbn={978-0-387-94785-3},
  doi={10.1007/978-1-4612-2256-9},
  year={1997},
  publisher={Springer-Verlag}
  %address={New York}
}

@book{
marsden-ratiu,
author = {Marsden, J.E. and Ratiu, T.S.},
title = {Introduction to Mechanics and Symmetry: A Basic Exposition of Classical Mechanical Systems},
doi = {10.1007/978-0-387-21792-5},
year = {1998},
isbn = {1441931430},
publisher = {Springer Publishing Company, Incorporated}
}

@book{
audin,
author = {Audin, M.},
title = {Torus Actions on Symplectic Manifolds},
year = {2004},
doi = {10.1007/978-3-0348-7960-6},
publisher = {Springer Basel},
isbn = {978-3-7643-2176-5}
}

@book{
dasilva,
author = {Da Silva, A.C.},
title = {Lectures on Symplectic Geometry},
year = {2008},
doi = {10.1007/978-3-540-45330-7},
isbn = {978-3-540-42195-5},
publisher = {Springer-Verlag},
%address={Berlin, Heidelberg}
}

@book{
    guillemin-sternberg,
    author = "Guillemin, V. and Sternberg, S.",
    title = "{Symplectic techniques in physics}",
    publisher = {Cambridge University Press},
    year = "1990"
}

@article{
borel-1,
author="Borel, A.",
title="Sur la cohomologie des espaces fibrés principaux et des
espace homogènes des groupes de Lie compacts",
journal="{Annals of Mathematics}",
volume="57",
%pages="115--207",
year="1953"
}

@article{
cartan-1,
title="{Notions d'algèbre différentielle; application aux groupes de Lie et aux variétés où opère un groupe de Lie}",
author="Cartan, H.",
journal="Colloque de topologie (espaces fibrés)",
%address="Bruxelles",
year="1951"
%pages="15--27"
}

@article{
cartan-2,
title="{La transgression dans un groupe de Lie et dans un espace fibré principal}",
author="Cartan, H.",
journal="Colloque de topologie (espaces fibrés)",
%address="Bruxelles",
year="1951"
%pages="57--71"
}

@article{
    Bott-equiv-cohom,
    author = "Bott, R.",
    %editor = "Shapiro, M.M. and Silberberg, R. and Wefel, J.P.",
    title = "{An Introduction to Equivariant Cohomology}",
    doi = "10.1007/978-94-011-4542-8\_3",
    journal = "NATO Sci. Ser. C",
    volume = "530",
    %pages = "35--57",
    year = "1999"
}

@article{
    bott-tu-classes,
    title={Equivariant characteristic classes in the Cartan model},
    author={Bott, R. and Tu, L.W.},
    year={2001},
    eprint={math/0102001},
    archivePrefix={arXiv}
    %primaryClass={math.DG}
}

@article{
    szabo,
    title={Equivariant Localization of Path Integrals},
    author={Szabo, J.R.},
    year={1996},
    eprint={hep-th/9608068},
    archivePrefix={arXiv}
    %primaryClass={hep-th}
}

@article{
   duistermaat-heckman,
   AUTHOR  = "Duistermaat, J.J. and Heckman, G.J.",
   TITLE   = "On the variation in the cohomology of the symplectic form of the reduced phase space",
   JOURNAL = "Invent. Math.",
   VOLUME  = "69",
   %number  = "2",
   %pages   = "259–268",
   YEAR    = "1982"
}

@article{
atiyah-bott-localization,
title = "The moment map and equivariant cohomology",
journal = "Topology",
volume = "23",
%number = "1",
%pages = "1 - 28",
year = "1984",
doi = "https://doi.org/10.1016/0040-9383(84)90021-1",
url = "http://www.sciencedirect.com/science/article/pii/0040938384900211",
author = "Atiyah, M.F. and Bott, R."
}

@article{
berline-vergne-loc,
author = "Berline, N. and Vergne, M.",
doi = "10.1215/S0012-7094-83-05024-X",
journal = "Duke Mathematical Journal",
%journal = "Duke Math. J.",
%number = "2",
%pages = "539--549",
publisher = "Duke University Press",
title = "Zeros d’un champ de vecteurs et classes caracteristiques equivariantes",
url = "https://doi.org/10.1215/S0012-7094-83-05024-X",
volume = "50",
year = "1983"
}

@article{
guillemin-prato,
title = "{Heckman, Kostant, and Steinberg formulas for symplectic manifolds}",
journal = "Advances in Mathematics",
volume = "82",
%number = "2",
%pages = "160 - 179",
year = "1990",
issn = "0001-8708",
doi = "10.1016/0001-8708(90)90087-4",
author = "Guillemin, V. and Prato, E."
}

@article{
    Jeffrey-nonabLoc,
    title={Localization for nonabelian group actions},
    author={Jeffrey, L.C. and Kirwan, F.C.},
    year={1993},
    eprint={alg-geom/9307001},
    archivePrefix={arXiv}
    %primaryClass={alg-geom}
}

@article{
   LocQFT,
   title={Localization techniques in quantum field theories},
   volume={50},
   ISSN={1751-8121},
   DOI={10.1088/1751-8121/aa63c1},
   %number={44},
   journal={Journal of Physics A: Mathematical and Theoretical},
   %publisher={IOP Publishing},
   author={Pestun, V. and Zabzine, M. and Benini, F. and Dimofte, T. and Dumitrescu, T.T. and Hosomichi, K. and Kim, S. and Lee, K. and Le Floch, B. and Mariño, M. and et al.},
   year={2017},
   eprint = {1608.02952},
   archivePrefix={arXiv}
   %pages={440301}
}

@article
{
   cremonesi,
   AUTHOR    = "Cremonesi, S.",
   TITLE     = "An introduction to localization and supersymmetry in curved space",
   journal = "PoS",
   volume = "Modave2013",
   YEAR      = "2013"
}

@article{
   itamar-wilson_loop_3dCS,
   title={Exact results for Wilson loops in superconformal Chern-Simons theories with matter},
   volume={2010},
   eprint = "0909.4559",
   archivePrefix = "arXiv",
   ISSN={1029-8479},
   DOI={10.1007/jhep03(2010)089},
   %number={3},
   journal={Journal of High Energy Physics},
   author={Kapustin, A. and Willett, B. and Yaakov, I.},
   year={2010}
}

@article{
   pestun-article,
   title={Localization of Gauge Theory on a Four-Sphere and Supersymmetric Wilson Loops},
   volume={313},
   eprint = "0712.2824",
   archivePrefix = "arXiv",
   ISSN={1432-0916},
   DOI={10.1007/s00220-012-1485-0},
   %number={1},
   journal={Communications in Mathematical Physics},
   author={Pestun, V.},
   year={2012},
   %pages={71–129}
}

@article{
   witten-2dYM,
   title={Two dimensional gauge theories revisited},
   volume={9},
   eprint = "hep-th/9204083",
   archivePrefix = "arXiv",
   ISSN={0393-0440},
   DOI={10.1016/0393-0440(92)90034-x},
   %number={4},
   journal={Journal of Geometry and Physics},
   author={Witten, E.},
   year={1992},
   %pages={303–368}
}

@article{
   Cordes-2dYMTFT,
   title={Lectures on 2D yang-mills theory, equivariant cohomology and topological field theories},
   volume={41},
   eprint = "hep-th/9411210",
   archivePrefix = "arXiv",
   ISSN={0920-5632},
   DOI={10.1016/0920-5632(95)00434-b},
   %number={1-3},
   journal={Nuclear Physics B - Proceedings Supplements},
   author = {Cordes, S. and Moore, G. and Ramgoolam, S.},
   year={1995},
   %pages={184–244}
}

@article{
Witten-susy_morse,
author = "Witten, E.",
doi = "10.4310/jdg/1214437492",
journal = "Journal of Differential Geometry",
%journal = "J. Differential Geom.",
%number = "4",
%pages = "661--692",
title = "Supersymmetry and Morse theory",
volume = "17",
year = "1982"
}

@article{
    birmingham-TFT,
    author = "Birmingham, D. and Blau, M. and Rakowski, M. and Thompson, G.",
    title = "{Topological field theory}",
    %number = "CERN-TH-6045-91",
    doi = "10.1016/0370-1573(91)90117-5",
    journal = "Physics Reports",
    volume = "209",
    %pages = "129--340",
    year = "1991"
}

@article{
blau-TFT,
title = "Topological gauge theories of antisymmetric tensor fields",
journal = "Annals of Physics",
volume = "205",
%number = "1",
%pages = "130 - 172",
year = "1991",
issn = "0003-4916",
doi = "10.1016/0003-4916(91)90240-9",
author = "Blau, M. and Thompson, G."
}

@article{
Kankaanrinta-proper,
title = "Proper smooth G-manifolds have complete G-invariant Riemannian metrics",
journal = "Topology and its Applications",
volume = "153",
%number = "4",
%pages = "610 - 619",
year = "2005",
issn = "0166-8641",
doi = "10.1016/j.topol.2005.01.034",
author = "Kankaanrinta, M."
}

@article{
   palo,
   title={Symplectic geometry of supersymmetry and nonlinear sigma model},
   volume={321},
   eprint= {hep-th/9311110},
   archivePrefix = "arXiv",
   ISSN={0370-2693},
   DOI={10.1016/0370-2693(94)90327-1},
   %number={1-2},
   journal={Physics Letters B},
   author={Palo, K.},
   year={1994},
   %pages={61–65}
}

@article{
    witten-CS,
    title={Non-Abelian Localization For Chern-Simons Theory},
    author={Beasley, C. and Witten, E.},
    year={2005},
    eprint={hep-th/0503126},
    archivePrefix={arXiv},
    %primaryclass={hep-th}
}

@article{
   festuccia-susy-curved,
   title={Rigid supersymmetric theories in curved superspace},
   volume={2011},
   eprint = {1105.0689v2},
   archivePrefix={arXiv},
   ISSN={1029-8479},
   DOI={10.1007/jhep06(2011)114},
   %number={6},
   journal={Journal of High Energy Physics},
   author={Festuccia, G. and Seiberg, N.},
   year={2011}
}

@article{
   Closset-susy-curved3d,
   title={Supersymmetric field theories on three-manifolds},
   volume={2013},
   eprint={1212.3388v3},
   archivePrefix={arXiv},
   ISSN={1029-8479},
   DOI={10.1007/jhep05(2013)017},
   %number={5},
   journal={Journal of High Energy Physics},
   author={Closset, C. and Dumitrescu, T.T. and Festuccia, G. and Komargodski, Z.},
   year={2013}
}

@article{
   Kehagias-susy-curved,
   title={Global supersymmetry on curved spaces in various dimensions},
   volume={873},
   ISSN={0550-3213},
   DOI={10.1016/j.nuclphysb.2013.04.010},
   eprint={1211.1367},
   archivePrefix={arXiv},
   %number={1},
   journal={Nuclear Physics B},
   author={Kehagias, A. and Russo, J.G.},
   year={2013},
   %pages={116–136}
}

@article{
   Dabholkar_2011,
   title={Quantum black holes, localization, and the topological string},
   volume={2011},
   ISSN={1029-8479},
   DOI={10.1007/jhep06(2011)019},
   eprint={1012.0265},
   archivePrefix={arXiv},
   %number={6},
   journal={Journal of High Energy Physics},
   author={Dabholkar, A. and Gomes, J. and Murthy, S.},
   year={2011}
}

@article{
Friedan-AS_index_thm,
title = "{Supersymmetric derivation of the Atiyah-Singer index and the chiral anomaly}",
journal = "Nuclear Physics B",
volume = "235",
%number = "3",
%pages = "395 - 416",
year = "1984",
issn = "0550-3213",
doi = "10.1016/0550-3213(84)90506-6",
author = "Friedan, D.  and Windey, P."
}

@article{
    alvarez-susyAS,
    author = "Alvarez-Gaume, L.",
    %editor = "Brittin, W.E. and Gustafson, K.E. and Wyss, W.",
    title = "{Supersymmetry and the Atiyah-Singer Index Theorem}",
    doi = "10.1007/BF01205500",
    journal = "Communications in Mathematical Physics",
    volume = "90",
    %pages = "161",
    year = "1983"
}

@article{
    hietamaki-susyQM-AS,
    author = "Hietamaki, A. and Morozov, A.Y. and Niemi, A.J. and Palo, K.",
    title = "{Geometry of N=1/2 supersymmetry and the Atiyah-Singer index theorem}",
    doi = "10.1016/0370-2693(91)90481-5",
    journal = "Physics Letters B",
    volume = "263",
    %pages = "417--424",
    year = "1991"
}

@article{
    zaffaroni2019lectures,
    title={Lectures on {AdS} Black Holes, Holography and Localization},
    author={Zaffaroni, A.},
    year={2019},
    journal={Living Reviews in Relativity},
    volume={23},
    eprint={1902.07176},
    archivePrefix={arXiv},
    %primaryclass={hep-th}
}

@misc{
    Figueroa-majorana,
    title={Majorana Spinors},
    author={Figueroa-O'Farrill, J.},
    howpublished={Notes available at \url{https://www.maths.ed.ac.uk/~jmf/Teaching/Notes.html}}
}

@article{
   zarembo-ads/cft-loc,
   title={Localization and {AdS/CFT} correspondence},
   volume={50},
   ISSN={1751-8121},
   url={http://dx.doi.org/10.1088/1751-8121/aa585b},
   DOI={10.1088/1751-8121/aa585b},
   eprint={1608.02963},
   archivePrefix={arXiv},
   %number={44},
   journal={Journal of Physics A: Mathematical and Theoretical},
   author={Zarembo, K.},
   year={2017},
   %pages={443011}
}

@article{
    figueroa-susylectures,
    title={BUSSTEPP Lectures on Supersymmetry},
    author={Figueroa-O'Farrill, J.},
    year={2001},
    eprint={hep-th/0109172},
    archivePrefix={arXiv}
}

@article{
brink-SYM,
title = "Supersymmetric Yang-Mills theories",
journal = "Nuclear Physics B",
volume = "121",
%number = "1",
%pages = "77 - 92",
year = "1977",
issn = "0550-3213",
doi = "10.1016/0550-3213(77)90328-5",
author = "Brink, L. and Schwarz, J.H. and Scherk, J."
}

@article{ 
   Alekseevsky-killing-sp,
   title={Killing spinors are {Killing vector fields in Riemannian} supergeometry},
   volume={26},
   ISSN={0393-0440},
   DOI={10.1016/s0393-0440(97)00036-3},
   eprint={dg-ga/9704002},
   archivePrefix={arXiv},
   %number={1-2},
   journal={Journal of Geometry and Physics},
   author={Alekseevsky, D.V. and Cortés, V. and Devchand, C. and Semmelmann, U.},
   year={1998},
   %pages={37–50}
}

@article{
    baum-killing-sp,
    author = "Baum, H.",
    title = "{Conformal Killing spinors and special geometric structures in Lorentzian geometry: A Survey}",
    %booktitle = "{Workshop on Special Geometric Structures in String Theory}",
    eprint = "math/0202008",
    archivePrefix = "arXiv",
    year = "2002"
}

@article{
   Dumitrescu-supercurrents,
   title={Supercurrents and brane currents in diverse dimensions},
   volume={2011},
   ISSN={1029-8479},
   DOI={10.1007/jhep07(2011)095},
   eprint={1106.0031},
   archivePrefix={arXiv},
   %number={7},
   journal={Journal of High Energy Physics},
   author={Dumitrescu, T.T. and Seiberg, N.},
   year={2011}
}

@article{
stelle-west-old-sugra,
title = "Minimal auxiliary fields for supergravity",
journal = "Physics Letters B",
volume = "74",
%number = "4",
%pages = "330 - 332",
year = "1978",
issn = "0370-2693",
doi = "10.1016/0370-2693(78)90669-X",
author = "Stelle, K.S.  and West,  P.C. "
}

@article{
sohnius-west-new-sugra,
title = "An alternative minimal off-shell version of {N=1} supergravity",
journal = "Physics Letters B",
volume = "105",
%number = "5",
%pages = "353 - 357",
year = "1981",
issn = "0370-2693",
doi = "10.1016/0370-2693(81)90778-4",
author = "Sohnius, M.F.  and West, P.C."
}

@article{
   Zaffaroni-susy-curved-holography,
   title={Supersymmetry on curved spaces and holography},
   volume={2012},
   ISSN={1029-8479},
   DOI={10.1007/jhep08(2012)061},
   eprint={1205.1062},
   archivePrefix={arXiv},
   %number={8},
   journal={Journal of High Energy Physics},
   author={Klare, C. and Tomasiello, A. and Zaffaroni, A.},
   year={2012}
}

@article{
   Dumitrescu-curved,
   title={Exploring curved superspace},
   volume={2012},
   ISSN={1029-8479},
   DOI={10.1007/jhep08(2012)141},
   eprint={1205.1115},
   archivePrefix={arXiv},
   %number={8},
   journal={Journal of High Energy Physics},
   author={Dumitrescu, T.T. and Festuccia, G. and Seiberg, N.},
   year={2012}
}

@article{
   Marino-locCS,
   title={Lectures on localization and matrix models in supersymmetric {Chern–Simons-matter} theories},
   volume={44},
   ISSN={1751-8121},
   DOI={10.1088/1751-8113/44/46/463001},
   eprint={1104.0783},
   archivePrefix={arXiv},
   %number={46},
   journal={Journal of Physics A: Mathematical and Theoretical},
   author={Mariño, M.},
   year={2011},
   %pages={463001}
}

@article{
   cattaneo-supergeometry,
   title={Introduction to supergeometry},
   volume={23},
   ISSN={1793-6659},
   url={http://dx.doi.org/10.1142/S0129055X11004400},
   DOI={10.1142/s0129055x11004400},
   eprint={1011.3401},
   archivePrefix={arXiv},
   %number={06},
   journal={Reviews in Mathematical Physics},
   author={Cattaneo, A.S. and Sch\"{a}tz, F.},
   year={2011},
   %pages={669–690}
}

@article{
   batchelor,
   title={Two approaches to supermanifolds},
   author={Batchelor, M.},
   journal={Trans. Amer. Math. Soc.},
   vol={258},
   year={1980},
   doi={10.1090/S0002-9947-1980-0554332-9}
}

@article{
    witten-superintegration,
    title={Notes On Supermanifolds and Integration},
    author={Witten, E.},
    year={2012},
    eprint={1209.2199},
    archivePrefix={arXiv}
}

@article{
    niemi-equiv-morse,
    title={Equivariant {Morse} theory and quantum integrability},
    author={A. J. Niemi and K. Palo},
    year={1994},
    eprint={hep-th/9406068},
    archivePrefix={arXiv},
    %primaryclass={hep-th}
}

@article{
   niemi-tirkkonen,
   title={Cohomological partition functions for a class of bosonic theories},
   volume={293},
   ISSN={0370-2693},
   DOI={10.1016/0370-2693(92)90893-9},
   %number={3-4},
   eprint={hep-th/9206033},
   archivePrefix={arXiv},
   journal={Physics Letters B},
   author={Niemi, A.J. and Tirkkonen, O.},
   year={1992},
   %pages={339–343}
}

@article{
   Bianchi-framing-loc-CS,
   title={Framing and localization in {Chern-Simons} theories with matter},
   volume={2016},
   ISSN={1029-8479},
   DOI={10.1007/jhep06(2016)133},
   %number={6},
   eprint={1604.00383},
   archivePrefix={arXiv},
   journal={Journal of High Energy Physics},
   author={Bianchi, M.S. and Griguolo, L. and Leoni, M. and Mauri, A. and Penati, S. and Seminara, D.},
   year={2016}
}

@article{
   roadmap-wilson-loops,
   title={Roadmap on Wilson loops in 3d Chern–Simons-matter theories},
   volume={53},
   ISSN={1751-8121},
   DOI={10.1088/1751-8121/ab5d50},
   eprint = "1910.00588",
   archivePrefix = "arXiv",
   %number={17},
   journal={Journal of Physics A: Mathematical and Theoretical},
   author={Drukker, N. and Trancanelli, D. and Bianchi, L. and Bianchi, M.S. and Correa, D.H. and Forini, V. and Griguolo, L. and Leoni, M. and Levkovich-Maslyuk, F. and Nagaoka, G. and et al.},
   year={2020},
   %pages={173001}
}

@article{
atiyah-singer-thm,
author = "Atiyah, M. F. and Singer, I. M.",
journal = "Bulletin of the American Mathematical Society",
%number = "3",
%pages = "422--433",
title = "The index of elliptic operators on compact manifolds",
url = "https://projecteuclid.org:443/euclid.bams/1183525276",
volume = "69",
year = "1963"
}

@article{
 wolf-chern-Dirac_op,
 ISSN = {00222518, 19435258},
 URL = {http://www.jstor.org/stable/24890502},
 author = {Wolf, J.A. and Chern, S.S.},
 journal = {Indiana University Mathematics Journal},
 %number = {7},
 %pages = {611--640},
 title = "Essential Self Adjointness for the {Dirac} Operator and Its Square",
 volume = {22},
 year = {1973}
}

@article{
    cecotti-functional,
    author = "Cecotti, S. and Girardello, L.",
    title = "Functional Measure, Topology and Dynamical Supersymmetry Breaking",
    %reportnumber = "HUTP-81/A052",
    doi = "10.1016/0370-2693(82)90947-9",
    journal = "Physics Letters B",
    volume = "110",
    %pages = "39",
    year = "1982"
}

@incollection{
     Atiyah-circular-symmetry,
     author = {Atiyah, M.F.},
     title = {Circular symmetry and stationary-phase approximation},
     booktitle = "Colloque en l'honneur de {Laurent Schwartz} - Volume 1",
     series = {Ast\'erisque},
     publisher = "Soci\'et\'e math\'ematique de {France}",
     %number = {131},
     year = {1985},
     %pages = {43-59},
     %zbl = {0578.58039},
     %mrnumber = {816738},
     url = {http://www.numdam.org/item/AST_1985__131__43_0}
}

@article{
morozov-supersymplectic-QFT,
title = "Supersymplectic geometry of supersymmetric quantum field theories",
journal = "Nuclear Physics B",
volume = "377",
%number = "1",
%pages = "295 - 338",
year = "1992",
issn = "0550-3213",
doi = "https://doi.org/10.1016/0550-3213(92)90026-8",
author = "Morozov, A.Y. and  Niemi, A.J. and Palo, K."
}

@article{
   willet-loc3d,
   title={Localization on three-dimensional manifolds},
   volume={50},
   ISSN={1751-8121},
   url={http://dx.doi.org/10.1088/1751-8121/aa612f},
   DOI={10.1088/1751-8121/aa612f},
   %number={44},
   eprint={1608.02958},
   archivePrefix={arXiv},
   journal={Journal of Physics A: Mathematical and Theoretical},
   author={Willett, B.},
   year={2017},
   %pages={443006}
}

@article{
    witten-knot,
    author = "Witten, E.",
    title = "{Quantum Field Theory and the Jones polynomial}",
    %reportNumber = "IASSNS-HEP-88-33",
    doi = "10.1007/BF01217730",
    journal = "Communications in Mathematical Physics",
    volume = "121",
    %pages = "351--399",
    year = "1989"
}

@article{
   Gomis-loc-tHooft,
   title="Exact results for {’t Hooft} loops in Gauge theories on $S^4$",
   volume={2012},
   ISSN={1029-8479},
   url={http://dx.doi.org/10.1007/JHEP05(2012)141},
   DOI={10.1007/jhep05(2012)141},
   eprint={1105.2568},
   archivePrefix={arXiv},
   %number={5},
   journal={Journal of High Energy Physics},
   author={Gomis, J. and Okuda, T. and Pestun, V.},
   year={2012}
}

@article{
    gomis-tHooft,
    title="{Quantum 't Hooft operators and S-duality in N=4 super Yang-Mills}",
    author={Gomis, J. and  Okuda, T. and Trancanelli, D.},
    year={2009},
    eprint={0904.4486},
    archivePrefix={arXiv},
    %primaryclass={hep-th}
}

@article{
    okuda-trancanelli-spectral-curves,
    author = "Okuda, T. and Trancanelli, D.",
    title = "Spectral curves, emergent geometry, and bubbling solutions for {Wilson loops}",
    eprint = "0806.4191",
    archivePrefix = "arXiv",
    %primaryClass = "hep-th",
    doi = "10.1088/1126-6708/2008/09/050",
    journal = "Journal of High Energy Physics",
    volume = "09",
    year = "2008"
}

@article{ 
    marino-MM,
    title="{Les Houches lectures on matrix models and topological strings}",
    author={Mari\~{n}o, M.},
    year={2004},
    eprint={hep-th/0410165},
    archivePrefix={arXiv},
    %primaryclass={hep-th}
}

@article{
    seiberg-nonrenormalization,
    author = "Seiberg, N.",
    title = {Supersymmetry and Nonperturbative beta Functions},
    %reportNumber = "IASSNS-HEP-88-5",
    doi = "10.1016/0370-2693(88)91265-8",
    journal = "Physics Letters B",
    volume = "206",
    %pages = "75--80",
    year = "1988"
}

@article{
    berkovits-10dSYM-offshell,
    author = "Berkovits, N.",
    title = "{A ten-dimensional superYang-Mills action with off-shell supersymmetry}",
    eprint = "hep-th/9308128",
    archivePrefix = "arXiv",
    %reportNumber = "KCL-TH-93-11",
    doi = "10.1016/0370-2693(93)91791-K",
    journal = "Phys. Lett. B",
    volume = "318",
    %pages = "104--106",
    year = "1993"
}

@article{
   drukker-wilsonloop,
   title={An exact prediction of {N=4 supersymmetric Yang–Mills} theory for string theory},
   volume={42},
   eprint = "hep-th/0010274",
   archivePrefix = "arXiv",
   ISSN={1089-7658},
   url={http://dx.doi.org/10.1063/1.1372177},
   DOI={10.1063/1.1372177},
   %number={7},
   journal={Journal of Mathematical Physics},
   author={Drukker, N. and Gross, D.J.},
   year={2001},
   %pages={2896–2914}
}

@article{
   erickson-wilsonloop,
   title={Wilson loops in supersymmetric {Yang–Mills} theory},
   volume={582},
   ISSN={0550-3213},
   url={http://dx.doi.org/10.1016/S0550-3213(00)00300-X},
   DOI={10.1016/s0550-3213(00)00300-x},
   eprint={hep-th/0003055},
   archivePrefix={arXiv},
   %number={1-3},
   journal={Nuclear Physics B},
   author={Erickson, J.K. and Semenoff, G.W. and Zarembo, K.},
   year={2000},
   %pages={155–175}
}

@article{
   gaiotto-CStheory,
   title="{Notes on superconformal Chern-Simons-Matter theories}",
   volume={2007},
   eprint={0704.3740},
   archivePrefix={arXiv},
   ISSN={1029-8479},
   url={http://dx.doi.org/10.1088/1126-6708/2007/08/056},
   DOI={10.1088/1126-6708/2007/08/056},
   %number={08},
   journal={Journal of High Energy Physics},
   author={Gaiotto, D. and Yin, X.},
   year={2007},
   %pages={056–056}
}

@article{
    dhoker-susy-notes,
    title="{Supersymmetric gauge theories and the AdS/CFT correspondence}",
    author={D'Hoker, E. and Freedman, D.Z.},
    year={2002},
    eprint={hep-th/0201253},
    archivePrefix={arXiv},
    %primaryclass={hep-th}
}

@article{
WWtheorem,
title = "Limits on massless particles",
journal = "Physics Letters B",
volume = "96",
%number = "1",
%pages = "59 - 62",
year = "1980",
issn = "0370-2693",
doi = "10.1016/0370-2693(80)90212-9",
author = "Weinberg, S. and Witten, E."
}

@article{
duru-Hatom-path-integral,
title = "Solution of the path integral for the H-atom",
journal = "Physics Letters B",
volume = "84",
%number = "2",
%pages = "185 - 188",
year = "1979",
issn = "0370-2693",
doi = "10.1016/0370-2693(79)90280-6",
author = "Duru, I.H.  and Kleinert, H."
}

@article{
   Maldacena-ads_cft, 
   title="The Large {N} Limit of Superconformal Field Theories and Supergravity",
   volume={38},
   eprint={hep-th/9711200},
    archivePrefix={arXiv},
   ISSN={0020-7748},
   url={http://dx.doi.org/10.1023/A:1026654312961},
   DOI={10.1023/a:1026654312961},
   %number={4},
   journal={International Journal of Theoretical Physics},
   author={Maldacena, J.},
   year={1999},
   %pages={1113–1133}
}

@article{
   Maldacena-wilson-loops,
   title="Wilson Loops in Large-{N} Field Theories",
   volume={80},
   ISSN={1079-7114},
   DOI={10.1103/physrevlett.80.4859},
   eprint={hep-th/9803002},
    archivePrefix={arXiv},
   %number={22},
   journal={Physical Review Letters},
   author={Maldacena, J.},
   year={1998},
   %pages={4859–4862}
}

@article{
   Rey_2001,
   title="Macroscopic strings as heavy quarks: Large-{N} gauge theory and {anti-de Sitter} supergravity",
   volume={22},
   ISSN={1434-6052},
   url={http://dx.doi.org/10.1007/s100520100799},
   DOI={10.1007/s100520100799},
   eprint={hep-th/9803001},
    archivePrefix={arXiv},
   %number={2},
   journal={The European Physical Journal C},
   publisher={Springer Science and Business Media LLC},
   author={Rey, S.-J. and Yee, J.-T.},
   year={2001},
   %pages={379–394}
}

@article{
   Correa-bremsstr,
   title="An exact formula for the radiation of a moving quark in {N}=4 super {Yang Mills}",
   volume={2012},
   ISSN={1029-8479},
   url={http://dx.doi.org/10.1007/JHEP06(2012)048},
   DOI={10.1007/jhep06(2012)048},
   eprint={1202.4455},
    archivePrefix={arXiv},
   %number={6},
   journal={Journal of High Energy Physics},
   author={Correa, D. and Henn, J. and Maldacena, J. and Sever, A.},
   year={2012}
}

@article{
   Aharony-deconfinement,
   title={The deconfinement and {Hagedorn} phase transitions in weakly coupled large {N} gauge theories},
   volume={5},
   ISSN={1631-0705},
   url={http://dx.doi.org/10.1016/j.crhy.2004.09.012},
   DOI={10.1016/j.crhy.2004.09.012},
   eprint={hep-th/0310285},
    archivePrefix={arXiv},
   %number={9-10},
   journal={Comptes Rendus Physique},
   author={Aharony, O. and Marsano, J. and Minwalla, S. and Papadodimas, K. and Van Raamsdonk, M.},
   year={2004},
   %pages={945–954}
}

@article{
   Ferrara-central-extension,
   title={Central extensions of supersymmetry in four and three dimensions},
   volume={423},
   ISSN={0370-2693},
   url={http://dx.doi.org/10.1016/S0370-2693(97)01586-4},
   DOI={10.1016/s0370-2693(97)01586-4},
   %number={3-4},
   eprint={hep-th/9711116},
    archivePrefix={arXiv},
   journal={Physics Letters B},
   author={Ferrara, S. and Porrati, M.},
   year={1998},
   %pages={255–260}
}

@article{
   Gorsky-central-extension,
   title="{More on the tensorial central charges in N=1 supersymmetric gauge theories: BPS wall junctions and strings}",
   volume={61},
   ISSN={1089-4918},
   url={http://dx.doi.org/10.1103/PhysRevD.61.085001},
   DOI={10.1103/physrevd.61.085001},
   eprint={hep-th/9909015},
    archivePrefix={arXiv},
   %number={8},
   journal={Physical Review D},
   author={Gorsky, A. and Shifman, M.},
   year={2000}
}

@article{
    okuda-instantonsN4,
    author = "Okuda, T. and Pestun, V.",
    title = "{On the instantons and the hypermultiplet mass of N=2* super Yang-Mills on $S^4$}",
    eprint = "1004.1222",
    archivePrefix = "arXiv",
    %primaryclass = "hep-th",
    %reportNumber = "ITEP-TH-04-10",
    doi = "10.1007/JHEP03(2012)017",
    journal = "Journal of High Energy Physics",
    volume = "03",
    %pages = "017",
    year = "2012"
}

@article{
    nekrasov-instantons,
    title={{Seiberg-Witten} Prepotential From Instanton Counting},
    author={Nekrasov, N.A.},
    year={2002},
    eprint={hep-th/0206161},
    archivePrefix={arXiv}
    %primaryclass={hep-th}
}

@article{
kalkman-BRST,
author = "Kalkman, J.",
journal = "Communications in Mathematical Physics",
%journal = "Comm. Math. Phys.",
%number = "3",
%pages = "447--463",
title = "BRST model for equivariant cohomology and representatives for the equivariant Thom class",
doi = "10.1007/BF02096949",
volume = "153",
year = "1993"
}

@article{
 atiyah-bott-YM,
 ISSN = {00804614},
 journal = {Philosophical Transactions of the Royal Society of London},
 Series = {A},
 author = {Atiyah, M. F. and Bott, R.},
 %number = {1505},
 %pages = {523--615},
 title = "{The Yang-Mills Equations over Riemann Surfaces}",
 volume = {308},
 year = {1983}
}

@article{
   Drukker_2007,
   title="{More supersymmetric Wilson loops}",
   volume={76},
   ISSN={1550-2368},
   url={http://dx.doi.org/10.1103/PhysRevD.76.107703},
   DOI={10.1103/physrevd.76.107703},
   eprint={0704.2237},
    archivePrefix={arXiv},
   %number={10},
   journal={Physical Review D},
   author={Drukker, N. and Giombi, S. and Ricci, R. and Trancanelli, D.},
   year={2007}
}

@article{
   Drukker_2008,
   title="Wilson loops: From {4D supersymmetric Yang-Mills theory to 2D Yang-Mills theory}",
   volume={77},
   ISSN={1550-2368},
   url={http://dx.doi.org/10.1103/PhysRevD.77.047901},
   DOI={10.1103/physrevd.77.047901},
   %number={4},
   journal={Physical Review D},
   author={Drukker, N. and Giombi, S. and Ricci, R. and Trancanelli, D.},
   year={2008}
}

@article{
   ABJM_2008,
   title="{N = 6 superconformal Chern-Simons-matter theories, M2-branes and their gravity duals}",
   volume={2008},
   ISSN={1029-8479},
   url={http://dx.doi.org/10.1088/1126-6708/2008/10/091},
   DOI={10.1088/1126-6708/2008/10/091},
   %number={10},
   eprint={0806.1218},
    archivePrefix={arXiv},
   journal={Journal of High Energy Physics},
   author={Aharony, O. and Bergman, O. and Jafferis, D. Louis and Maldacena, J.},
   year={2008},
   %pages={091–091}
}

@article{
    trancanelli-ABJM-wilson-loop,
    author = "Drukker, N. and Trancanelli, D.",
    title = "{A Supermatrix model for N=6 super Chern-Simons-matter theory}",
    eprint = "0912.3006",
    archivePrefix = "arXiv",
    doi = "10.1007/JHEP02(2010)058",
    journal = "Journal of High Energy Physics",
    volume = "02",
    year = "2010"
}

@article{
    marino-ABJM,
    author = "Mari\~{n}o, M. and Putrov, P.",
    title = "Exact Results in {ABJM} theory from Topological Strings",
    eprint = "0912.3074",
    archivePrefix = "arXiv",
    %primaryClass = "hep-th",
    doi = "10.1007/JHEP06(2010)011",
    journal = "Journal of High Energy Physics",
    volume = "06",
    year = "2010"
}

@article{
mathai-quillen,
title = "{Superconnections, Thom classes, and equivariant differential forms}",
journal = "Topology",
volume = "25",
%number = "1",
%pages = "85 - 110",
year = "1986",
issn = "0040-9383",
doi = "10.1016/0040-9383(86)90007-8",
url = "http://www.sciencedirect.com/science/article/pii/0040938386900078",
author = "Mathai, V. and Quillen, D."
}

@article{
    marsden-weinstein,
    author = "Marsden, J. and Weinstein, A.",
    title = "{Reduction of symplectic manifolds with symmetry}",
    doi = "10.1016/0034-4877(74)90021-4",
    journal = "Reports on Mathematical Physics",
    volume = "5",
    %number = "1",
    %pages = "121--130",
    year = "1974"
}

@incollection{
meyer,
title = "Symmetries and Integrals in Mechanics",
booktitle = "Dynamical Systems",
publisher = "Academic Press",
%pages = "259 - 272",
year = "1973",
isbn = "978-0-12-550350-1",
doi = "10.1016/B978-0-12-550350-1.50025-4",
author = "Meyer, K."
}

@article{
    caporaso-YM-topological-string,
    title={Topological Strings, Two-Dimensional {Yang-Mills} Theory and {Chern-Simons} Theory on Torus Bundles},
    author={Caporaso, N. and Cirafici, M. and Griguolo, L. and Pasquetti, S. and Seminara, D. and Szabo, R.J.},
    year={2006},
    eprint={hep-th/0609129},
    archivePrefix={arXiv}
}

@article{
    szabo-qdeformation,
    author = "Szabo, R.J. and Tierz, M.",
    title = "{q-deformations of two-dimensional Yang-Mills theory: classification, categorification and refinement}",
    eprint = "1305.1580",
    archivePrefix = "arXiv",
    %primaryClass = "hep-th",
    doi = "10.1016/j.nuclphysb.2013.08.001",
    journal = "Nuclear Physics B",
    volume = "876",
    year = "2013"
}

@article{
   szabo-YMloc,
   title={Five-dimensional cohomological localization and squashed q-deformations of two-dimensional {Yang-Mills} theory},
   volume={2020},
   ISSN={1029-8479},
   url={http://dx.doi.org/10.1007/JHEP06(2020)036},
   DOI={10.1007/jhep06(2020)036},
   eprint={2003.09411},
    archivePrefix={arXiv},
   journal={Journal of High Energy Physics},
   author={Santilli, L. and Szabo, R.J. and Tierz, M.},
   year={2020}
}

@article{
    griguolo-bassetto-1998,
    author = "Bassetto, A. and Griguolo, L.",
    title = "{Two-dimensional QCD, instanton contributions and the perturbative Wu-Mandelstam-Leibbrandt prescription}",
    eprint = "hep-th/9806037",
    archivePrefix = "arXiv",
    doi = "10.1016/S0370-2693(98)01319-7",
    journal = "Physics Letters B",
    volume = "443",
    year = "1998"
}

@article{
    griguolo-bassetto-2009,
    author = "Bassetto, A. and Griguolo, L. and Pucci, F. and Seminara, D. and Thambyahpillai, S. and Young, D.",
    title = "{Correlators of supersymmetric Wilson-loops, protected operators and matrix models in N=4 SYM}",
    eprint = "0905.1943",
    archivePrefix = "arXiv",
    %primaryClass = "hep-th",
    doi = "10.1088/1126-6708/2009/08/061",
    journal = "Journal of High Energy Physics",
    volume = "08",
    year = "2009"
}

@article{
   bassetto-2011,
   title={Quantum {’t Hooft Loops of SYM N=4 as instantons of YM2} in Dual Groups {SU(N) and SU(N)/ZN}},
   volume={98},
   ISSN={1573-0530},
   url={http://dx.doi.org/10.1007/s11005-011-0480-2},
   DOI={10.1007/s11005-011-0480-2},
   eprint={1011.0638},
    archivePrefix={arXiv},
   journal={Letters in Mathematical Physics},
   author={Bassetto, A. and Thambyahpillai, S.},
   year={2011}
}

\end{document}